\documentclass[review]{elsarticle}
\usepackage[letterpaper, margin=1.in, top=1.1in]{geometry}
\usepackage{lineno,hyperref}
\modulolinenumbers[5]

\usepackage{bm}
\usepackage{graphicx}
\usepackage{subfig}
\usepackage[mathscr]{euscript}
\usepackage[overload]{empheq}
\usepackage{pifont}
\newcommand{\cmark}{\ding{51}}%
\newcommand{\xmark}{\ding{55}}%

\makeatletter
\def\Cline#1#2{\@Cline#1#2\@nil}
\def\@Cline#1-#2#3\@nil{%
  \omit
  \@multicnt#1%
  \advance\@multispan\m@ne
  \ifnum\@multicnt=\@ne\@firstofone{&\omit}\fi
  \@multicnt#2%
  \advance\@multicnt-#1%
  \advance\@multispan\@ne
  \leaders\hrule\@height#3\hfill
  \cr}

\usepackage{ etoolbox, xparse}   
\DeclarePairedDelimiterX{\abs}[1]\lvert\rvert{\ifblank{#1}{\,\cdot\,}{#1}}
\let\oldabs\abs
\def\abs{\futurelet\testchar\MaybeOptArgAbs}
\def\MaybeOptArgAbs{\ifx[\testchar\let\next\OptArgAbs
\else \let\next\NoOptArgAbs\fi \next}
\def\OptArgAbs[#1]#2{\oldabs[#1]{#2}}
\def\NoOptArgAbs#1{\ifblank{#1}{\oldabs{}}{\oldabs[\big]{#1}}}

\DeclarePairedDelimiterX{\set}[1]\{\}{\setargs{#1}}
\NewDocumentCommand{\setargs}{>{\SplitArgument{1}{;}}m}
{\setargsaux#1}
\NewDocumentCommand{\setargsaux}{mm}
{\IfNoValueTF{#2}{#1}{\nonscript\,#1\nonscript\;\delimsize\vert\nonscript\:\allowbreak #2\nonscript\,}}
\let\oldset\set
\def\set{\futurelet\testchar\MaybeOptArgSet}
\def\MaybeOptArgSet{\ifx[\testchar \let\next\OptArgSet
\else \let\next\NoOptArgSet \fi \next}
\def\OptArgSet[#1]#2{\oldset[#1]{#2}}
\def\NoOptArgSet#1{\OptArgSet[\big]{#1}}

\newtheorem{theorem}{Theorem}

\newtheorem{proposition}[theorem]{Proposition}

\newenvironment{proof}[1][Proof]{\begin{trivlist}
\item[\hskip \labelsep {\bfseries #1}]}{\end{trivlist}}

\begin{document}

\begin{frontmatter}
\title{On the verification of CFD solvers of all orders of accuracy on curved wall-bounded domains and for realistic RANS flows}

\author[]{Farshad Navah\corref{mycorrespondingauthor}}
\ead{farshad.navah@mail.mcgill.ca}

\author[]{Siva Nadarajah}
\address{Mechanical Engineering Department, \\McGill University,
Montr\'{e}al, Qu\'{e}bec, Canada, H3A 0C3}

\cortext[mycorrespondingauthor]{Corresponding author}




\begin{abstract}
{This paper aims at extending the code verification methodology to high-order accurate scheme implementations on curved wall-bounded domains as well as for realistic turbulent flows such as the flat plate boundary layer modelled by the Reynolds-averaged Navier-Stokes (RANS) equations. Two new manufactured solutions (MSs) are introduced with demonstrated ability  to verify the treatment of slip and no-slip boundary conditions in high-order frameworks on curved domains. These MSs serve as well to discuss the impact of the method of computation of boundary normals on the order of accuracy (OOA) of the solution at the wall. Furthermore, two turbulent boundary layer MSs from literature, devised to mimic the genuine features of RANS-modelled flows in the vicinity of wall, are compared in terms of their suitability to achieving high-order accuracy. A number of useful concepts in verification are explored through these cases such as the limit values of the Spalart-Allmaras (SA) turbulence model source terms at the wall, the verification of the modified vorticity term of the modified SA model, the grid sensitivity of wall-bounded turbulent flows, the inadequacy of substituting solution verification to code verification and the effect of non-dimensionalization of the solution on the minimization of iterative errors via residual convergence. In all cases, demonstrations are carried for orders of accuracy up to the sixth.}
\end{abstract}


\end{frontmatter}



\section{Introduction}
This paper is the continuation of a previous work \cite{Navah2017a} on code verification for high-order accurate computational fluid dynamics (CFD) solvers which elaborated on this important concept in terms of motivations, theoretical foundations, methodology and applications to free flows, i.e., flows  not bounded by walls. The verification for this type of problem, due to its simplicity, constitutes a first major step in the cumulative process of gathering evidence that a given high-order solver delivers the expected performance in terms of accuracy per computational effort, thus justifying the  endeavour invested into its development. Nevertheless, the verification for free flows via trigonometric MSs does not provide information on the ability of the code in tackling high-order representation of curved domains in presence of wall boundary conditions in inviscid and laminar flows. Furthermore, it does not permit to investigate the solver's performance in delivering high orders of accuracy in the solution of realistic RANS flows such as the flat-plate boundary layer. Providing a framework to address these issues constitutes hence the principal motivation of the current work.

The demonstrations are conducted via a numerical setup composed of compressible RANS-SA \cite{Allmaras-et-al_2012} equations in conservative form, discretized by the correction procedure via flux reconstruction (CPR) scheme \cite{Huynh2007,Huynh_2009a,Wang-et-al_2011a}, also known as flux reconstruction (FR), which unifies many compact high-order methods under the same formulation, along with the discontinuous Galerkin (DG) correction functions.

The article is structured as follows: the first section introduces the main motivation and contributions of the paper. The theoretical background is discussed in the second section, followed in the third by the presentation of the governing equations. The fourth and the fifth sections are respectively dedicated to the compact high-order numerical method and to verification cases and methodology and the last section concludes the article.

\subsection{Contributions}
The first goal here is to enable the verification of high-order accurate simulation codes for the solution of wall-bounded flows on curved domains. To this end:
\begin{itemize}
\item A manufactured solution (MS) for inviscid problems, spanning simultaneously flow regimes from subsonic to supersonic, is presented and the effect of wall normals computation via isoparametric mapping, on the orders of accuracy of the reflecting boundary condition of slip wall is discussed; 
\item A similar MS is proposed for the verification of laminar flows in incompressible and compressible subsonic regimes at once on curved domains along with  no-slip (adiabatic) wall condition and the verification sensitivity to viscous terms is discussed.
\end{itemize}

In a second phase, the focus is moved to realistic turbulent manufactured cases on Cartesian domains. Two existing wall-bounded MSs from literature are considered and their application is broadened by assessing their adequacy for the verification of high-order implementations. These MSs act as well as  a framework to delve into the following ideas pertinent to verification:
\begin{itemize}
\item The limit values of the SA source terms with indeterminate forms at the wall;
\item The verification of the modified vorticity term of the modified SA equation;
\item The sensitivity of realistic turbulent wall-bounded solutions to grid stretching, in terms of its effects on discretization error, residual convergence and orders of accuracy (OOAs) of solution variables and drag coefficient;
\item The impact of non-dimensionalization in minimizing iterative errors via residual convergence;
\item The inadequacy of substituting solution verification to code verification.
\end{itemize}

In all cases, the suitability of the MSs in realizing OOAs up to the sixth is demonstrated.
All manufactured cases presented in this work are made available (see \cite{navah017_github}) by providing an IPython \cite{PER-GRA:2007} notebook and a C routine, thus enabling the reproduction of the verification methodology in any code.

\section{Theoretical background}
We refer the reader to \cite{Navah2017a} for the description of the theoretical background in verification and validation (V\&V) and more precisely the presentation of the terminology and explanation of implicated concepts, the comparison of the method of analytical solutions versus the method of manufactured solutions and the motivations behind the adoption of the latter, as well as an extensive review of the literature on V\&V with special attention to applications in high-order frameworks.

\section{Governing equations}\label{sec:goveq}

We consider the steady-state, compressible RANS governing equations, under the assumption of  a Newtonian and calorically perfect gas, closed by the conservative form of the original and revised SA models of turbulence \cite{Allmaras-et-al_2012}. The choice of the SA model is motivated by its simplicity as well as proven effectiveness for the representation of aerodynamic flows. The role of the revised SA model is to complement the original model by introducing conditional variations to some of its terms in order to ensure the stability of the solution process whenever  model's working variable takes negative values. This often occurs in coarse spatial discretizations of the boundary layer of high-order solutions.

The governing partial differential equations (PDEs) are expressed by the following general formulation for advective-diffusive problems:
\begin{equation}
\partial_t(Q_k) +\partial_j(F^{inv}_{kj}) - \partial_j(F^{vis}_{kj}) = S_k,
\label{eq:NS}
\end{equation}
where $Q_k$ represents a state variable which is the solution of the $k^{th}$ partial differential equation with $k \in [1\,..\,{N_\mathrm{eq}}]$ where  $N_\mathrm{eq}=N_\mathrm{d}+3$ is the number of equations based on  $N_\mathrm{d}$  space dimensions; $F^{inv}_{kj}$ and $F^{vis}_{kj}$ respectively express the $k^{th}$ inviscid (advective) and viscous (diffusive) fluxes for $j\in [1\,..\,{N_\mathrm{d}}]$,  $S_k$  denotes the source term for equation $k$ and the repeated indices are summed over following Einstein's convention. The substitution of the following expressions and values in Eq. \eqref{eq:NS} yields the expanded form of the conservation laws:

\begin{itemize}
\item[-] Conservation of mass ($k=1$)
\begin{equation}
Q_k=\rho,\qquad\qquad F^{inv}_{kj}=\rho u_j,\qquad F^{vis}_{kj}=0,\qquad S_k=0;\label{eq:cont}
\end{equation}
\item[-] Conservation of momentum ($k \in [2\,..\,{N_\mathrm{d}+1}]$ and $i=k-1$)
\begin{equation}
Q_k=\rho u_i,\qquad F^{inv}_{kj}=\rho u_ju_i +p \delta_{ij},\qquad F^{vis}_{kj}=\tau_{ij},\qquad S_k=0;\label{eq:mom}
\end{equation}
\item[-] Conservation of energy ($k=N_\mathrm{d}+2$)
\begin{equation}
Q_k=\rho E, \qquad  F^{inv}_{kj}=\rho u_j H, \qquad  F^{vis}_{kj}=u_i \tau_{ij} +\omega_j, \qquad  S_k=0;\label{eq:ener}
\end{equation}
\item[-] Transport of the turbulent working variable ($k=N_\mathrm{d}+3$)
\begin{equation}
\begin{split}
Q_k=\rho \tilde{\nu}, \qquad F^{inv}_{kj}=\rho u_j \tilde{\nu},\qquad F^{vis}_{kj}=\frac{1}{\sigma}(\mu+\rho \tilde{\nu}f_n)\,\partial_j \tilde{\nu}, \qquad
S_k=  \rho \,\mathcal{P} - \rho \,\mathcal{D} + \rho \,\mathcal{T} \\+\frac{c_{b2}}{\sigma} \rho \,\partial_j \tilde{\nu} \, \partial_j \tilde{\nu} 
 - \frac{1}{\sigma}(\nu+\tilde \nu f_n)\, \partial_j (\rho \,\partial_j \tilde \nu) \label{eq:sa}.\
\end{split}
\end{equation}
%
\end{itemize}
The quantities appearing in equations \eqref{eq:cont} through \eqref{eq:sa} are defined as follows: $\rho$ represents the density, ${\bm{u}} = {\bf{e}}_i u_i$ denotes the velocity vector with ${\bf{e}}_i$ being the $i^{th}$ orthonormal basis vector of the Euclidean spatial system, $E$ is the total energy per mass defined as $E=e+\frac{1}{2} (u_iu_i)$ where $e$ is the internal energy defined as $e=\frac{R}{\gamma-1}T$, for a calorically perfect gas, with $R$ being the gas constant and $T$ the temperature. The total enthalpy is expressed by $H=E + \frac{p}{\rho}$ with $p$ denoting the pressure which is related to the energy via the ideal gas law such that
\begin{equation}
{\footnotesize p=(\gamma-1)\rho\left(E-\frac{1}{2}(u_iu_i)\right)},
\label{eq:state}
\end{equation}
where $\gamma$ is the specific heat ratio ($\gamma=1.4$ for air).

In Eq. \eqref{eq:mom}, $\tau_{ij}$ expresses the components of the viscous stress tensor, $\underline{\underline{\tau}}$, which for compressible Newtonian fluids read
\begin{equation*}
\tau_{ij}=2 \, \mu_{\mathrm{eff}} \, S_{ij},
\quad \mathrm{with} \;\;
S_{ij} = \frac{1}{2}(\partial_i u_j + \partial_j u_i)-\frac{1}{3}\partial_k u_k \delta_{ij},
\end{equation*}
where the effective viscosity is denoted by $\mu_{\mathrm{eff}}$ and defined as the sum of the dynamic viscosity, $\mu$, and the eddy viscosity, $\mu_t$, say, $\mu_{\mathrm{eff}}= \mu+\mu_t$; and $\delta_{ij}$ represents the Kronecker delta. Note that we assume the dynamic viscosity to be spatially constant throughout this article.

In Eq. \eqref{eq:ener}, $\omega_j =\lambda_{\mathrm{eff}}\, \partial_j T$ expresses the $j^{th}$ component of the heat flux vector where $\lambda_{\mathrm{eff}}$ is the effective thermal conductivity defined as $\lambda_{\mathrm{eff}} = \lambda + \lambda_t$, with $\lambda=\frac{\gamma R}{(\gamma-1)}\frac{\mu}{\mathrm{Pr}}$, the molecular conductivity,  and $\lambda_t=\frac{\gamma R}{(\gamma-1)} \frac{\mu_t}{\mathrm{Pr}_t}$, the eddy conductivity. The laminar and turbulent Prandtl numbers are respectively set to $\mathrm{Pr}=0.7$ and $\mathrm{Pr}_t=0.9$ unless specified.


In the SA model, $\mu_t$ is the  eddy (turbulent) viscosity, expressed by

\begin{subequations}
\begin{empheq}
[left={\mu_t= \rho \nu_t =\empheqlbrace}]{alignat=2}
&\rho \tilde{\nu} f_{v1}\qquad & \tilde \nu \geq 0, \label{eq:mu_t+}\\
&0                      \qquad & \tilde \nu < 0,    \label{eq:mu_t-}
\end{empheq}
\end{subequations}
where
\begin{equation*}
f_{v1}=\frac{\chi^3}{\chi^3+c_{v1}^3},
\quad
\chi=\tilde{\nu}/\nu,
\quad
c_{v1}=7.1,
\end{equation*}
and $\tilde{\nu}$ is the working variable of the SA model which is equivalent to a turbulent kinematic viscosity.

The term $\mathcal{P}$ in Eq. \eqref{eq:sa} stands for turbulence production defined by
\begin{subequations}
\begin{empheq}
[left={\mathcal{P}=\empheqlbrace}]{alignat=2}
& c_{b1}(1-f_{t2})\tilde s \tilde \nu\qquad &  \tilde \nu \geq 0, \label{eq:prod+}\\
& c_{b1}(1-c_{t3})       s \tilde \nu\qquad &  \tilde \nu < 0,    \label{eq:prod-}
\end{empheq}
\end{subequations}
where $c_{b1} = 0.1355$, the laminar suppression term is denoted by $f_{t2}=c_{t3}\,\mathrm{exp}(-c_{t4}\,\chi^2)$ with $c_{t3}=1.2$ and $c_{t4}=0.5$, $s=\lvert\varepsilon_{ijk}\partial_j u_k \rvert$ is the vorticity magnitude with $\varepsilon_{ijk}$ expressing the Levi-Civita permutation symbol, and $\tilde s$ is the modified vorticity which reads
\begin{subequations}
\begin{empheq}
[left={\tilde s=\empheqlbrace}]{alignat=2}
& s+\bar{s}      \qquad &  \bar{s} \geq -c_{v2} s, \label{eq:S+}\\
& s+\frac{s(c_{v2}^2s+c_{v3}\bar{s})}{(c_{v3}-2c_{v2})s-\bar{s}}\qquad &   \bar{s} < -c_{v2}s, \label{eq:S-}
\end{empheq}
\end{subequations}
where
\begin{gather*}
\begin{aligned}
&\bar{s}=\frac{\tilde{\nu} f_{\nu 2}}{\kappa^2 d_w^2},\quad f_{v2} = 1 - \frac{\chi}{1+\chi f_{v1}},
&c_{v2}=0.7, \quad c_{v3}=0.9, \quad \kappa = 0.41,
       \end{aligned}
\end{gather*}
and $d_w$ is the distance to the closest wall. In Eq. \eqref{eq:sa}, the destruction term, $\mathcal{D}$, is defined as
\begin{subequations}
\begin{empheq}
[left={\mathcal{D=}\empheqlbrace}]{alignat=2}
& \left(c_{w1} f_w -\frac{c_{b1}}{\kappa^2}f_{t2}\right)\frac{ {\tilde\nu}^2 }{d_w^2}\qquad &  \tilde \nu \geq 0, \label{eq:D+}\\
& - c_{w1} \frac{ {\tilde\nu}^2 }{d_w^2}  \qquad &   \tilde \nu < 0, \label{eq:D-}
\end{empheq}
\end{subequations}
where
\begin{equation*}
\begin{aligned}
c_{w1} = \frac{c_{b1}}{\kappa^2}+\frac{1+c_{b2}}{\sigma},\phantom{\frac{\frac{1}{1}}{\frac{1^2}{1^\frac{1}{2}}}}
c_{b2} = 0.622, \quad \sigma =2/3, \;\,  \mathrm{and} \;\, f_w = g\left(\frac{1+c_{w3}^6}{g^6+c_{w3}^6}\right)^{1/6}.
\end{aligned}
\label{eq:dist}
\end{equation*}
The reader is referred to \cite{Allmaras-et-al_2012} for the full definition of the trip term, $\mathcal{T}$, in Eq. \eqref{eq:sa} which is employed to mimic the effect of a forced transition. A value of $\mathcal{T}=0$ is considered throughout this work. The remaining closure functions and constants of the SA model are
\begin{subequations}
\begin{empheq}
[left={f_n=\empheqlbrace}]{alignat=2}
& 1                                  \qquad &  \tilde \nu \geq 0, \label{eq:fn+}\\
& \frac{c_{n1}+\chi^3}{c_{n1}-\chi^3}\qquad &   \tilde \nu < 0, \label{eq:fn-}
\end{empheq}
\end{subequations}
\begin{equation*}
\begin{aligned}
g = r + c_{w2}(r^6-r),
\quad
r = \mathrm{min} \left(\frac{\tilde \nu}{\tilde s \kappa^2 d_w^2}, r_\mathrm{lim} \right),
\quad
r_\mathrm{lim}=10,
\quad
c_{n1}=16,
\quad
c_{w2} = 0.3,
\quad \mathrm{and}\quad
c_{w3} = 2.
\end{aligned}
\end{equation*}

The original SA model is represented by equations \ref{eq:sa}, \ref{eq:mu_t+}, \ref{eq:prod+}, \ref{eq:S+}, \ref{eq:D+} and \ref{eq:fn+} whereas the modified SA model corresponds to equations \ref{eq:sa}, \ref{eq:mu_t-}, \ref{eq:prod-}, \ref{eq:S-}, \ref{eq:D-} and \ref{eq:fn-}.

\section{Compact high-order numerical method}
The numerical framework is made of the CPR scheme via DG correction functions with Gauss-Legendre-Lobatto (GLL) solution nodes. The reader is refered to \cite{Navah2017a} for a detailed description of the scheme, the advective and diffusive numerical fluxes and the Riemann and viscous boundary conditions (BCs). We complete the information necessary to reproduce this work by presenting the treatment of the wall boundaries.


\subsection{Boundary conditions}
The notations employed in this section are adopted from \cite{Navah2017a}. The superscripts $\cdot^-$ and $\cdot^+$  respectively refer to the internal and external traces of an interface quantity with respect to an element and $\bm{n}={\bf{e}}_q n_q$ with $q\in [1\,..\,{N_\mathrm{d}}]$ designates the normal vector. Furthermore, ghost nodes are defined at domain boundaries, serving the prescription of boundary conditions by letting external interface quantity at the boundary take the desired boundary value: $\cdot^+ = \cdot^\mathrm{BC}$.
\subsubsection{Wall BC}
We consider the adiabatic wall condition, the treatment of which is split into  inviscid and viscous steps as,
\begin{itemize}
\item Inviscid (slip condition): the boundary states of all dependent variables, except those of momentum equations, are set to their inner counterparts, i.e., $Q_k^\mathrm{BC}=Q_k^-$ for $k\notin [2\,.. \,N_\mathrm{d}+1]$, whereas the boundary velocity vector is defined by cancelling the impacting component, $(u_i^- n_i^-) \, {\bf{e}}_q n_q^-$, of the inner velocity, ${\bf{e}}_qu_q^-$, while preserving its tangential component, ${\bf{e}}_qu_q^--(u_i^- n_i^-) \, {\bf{e}}_q n_q^-$. We thus obtain 
\begin{equation}
Q_k^\mathrm{BC}=(\rho u_q)^\mathrm{BC}=(\rho u_q)^--2((\rho u_i)^- n_i^-) \, n_q^-,
\label{eq:slip_BC}
\end{equation}
for $k \in [2\, .. \, N_\mathrm{d}+1]$ and $q=k-1$. 

The boundary state, $Q_k^\mathrm{BC}$, is provided to the inviscid numerical flux of Eq. (26) of \cite{Navah2017a} which is modified by prescribing a null dissipation on the wall, i.e., $D=0$.

\item Viscous (no-slip condition): we first set $Q_k^\mathrm{BC}=Q_k^-$ (hence $\hat{Q}_k = Q_k^-$) for $k\in \{1\,,\,N_\mathrm{d}+2\}$ and $Q_k^\mathrm{BC}=-Q^{-}$ (hence $\hat{Q}_k=0$) for $k\notin \{1\,,\,N_\mathrm{d}+2\}$ in the computation of full and partial derivative corrections in Eqs. (30) and (32) of \cite{Navah2017a}. As for the viscous numerical flux of Eq. (31) in \cite{Navah2017a}, we consider $\mathcal{H}_k^{vis}=F^{vis}_{kj}\left(Q_k^-,{\bf{e}}_q(\overline{\partial_q Q_k})^-\right) n_j$ for all $k$ at the adiabatic wall boundary except for $k=N_\mathrm{d}+2$ (conservation of energy) that instead receives $\mathcal{H}_{N_\mathrm{d}+2}^{vis}=0$.
\end{itemize}

\subsection{Solution process}
The objective is to minimize iterative and round-off errors such that the discretization error is isolated as the major source of numerical error and can hence be the focus of analyses carried out in this article. To this end, we employ an exact linearization of the system of RANS  equations along with the original and modified SA models. This linearization, verified to be accurate up to double precision, serves in the Newton's iterative method, initialized by the manufactured solution and pursued until the discrete residuals are permanently reduced to a minimal value.

\subsection{Treatment of curved domains}
The elements in the physical domain are defined by a mapping to a reference element in the computational domain via Lagrange interpolation functions of the same degree (isoparametric) as the ones used to represent the discrete elemental solution. To account for this mapping, we adopt the expression of the CPR scheme in the physical domain described by Eq. (37) of \cite{Wang-et-al_2011a}.

\section{Verification cases and methodology}

Table \ref{tb:MS_prop_2} presents the solver features verified by each of the four manufactured cases presented in this section.

\begin{table}
\centering
\begin{tabular}{ l|l|c|c|c|c|c  }
\textbf{Property} &  \textbf{Feature}       & MS-1      & MS-2    & MS-3    & MS-4   &  \textbf{Cum.}   \\ \hline\hline
Re                & Inviscid        &   \cmark  &  \cmark & \cmark  & \cmark &  \cmark \\ \Cline{2-7}{0.5pt}
                  & Viscous         &   \xmark  &  \cmark & \cmark  & \cmark &  \cmark \\ \Cline{2-7}{0.5pt}
                  & Turbulent       &   \xmark  &  \xmark & \cmark  & \cmark &  \cmark \\ \Cline{1-7}{0.5pt}
Ma                & Supersonic      &   \cmark  &  \xmark & \cmark  & \xmark &  \cmark \\ \Cline{2-7}{0.5pt}
                  & Transonic       &   \cmark  &  \xmark & \cmark  & \xmark &  \cmark \\ \Cline{2-7}{0.5pt}  
                  & Subsonic        &   \cmark  &  \cmark & \cmark  & \cmark &  \cmark \\ \Cline{1-7}{0.5pt}
Boundary          & Riemann         &   \cmark  &  \cmark & \cmark  & \cmark &  \cmark \\ \Cline{2-7}{0.5pt}
Conditions        & Viscous         &   \xmark  &  \cmark & \cmark  & \cmark &  \cmark \\ \Cline{2-7}{0.5pt}
                  & Slip Wall       &   \cmark  &  \cmark & \cmark  & \cmark &  \cmark \\ \Cline{2-7}{0.5pt}
                  & No-slip Wall    &   \xmark  &  \cmark & \cmark  & \cmark &  \cmark \\ \Cline{1-7}{0.5pt}
Mapping           & Curved Domain   &   \cmark  &  \cmark & \xmark  & \xmark &  \cmark
\end{tabular}
\caption{List of solver capabilities verified by manufactured solutions for wall-bounded flows and curved domains}
\label{tb:MS_prop_2}
\end{table}

\subsection{Inviscid flows on curved wall-bounded domain - MS-1}
We introduce MS-1 that examines the implementation of CFD solvers for inviscid transonic regime, Riemann BC, slip wall boundary condition and curved domains. 
The domain curvature is defined by a deformation of the initial domain of $(\mathcal{X}, \mathcal{Y}) \in \Omega_0 = [0.5,1.0]\times [0.0,0.5]$, where Cartesian grids are created, to the target domain of $(x,y) \in \Omega$ by
\begin{equation}
\begin{split}
x &= \frac{4}{3}\left(\mathcal{X}^2 -\frac{1}{4} \right) + 1,\\
y &= \mathcal{Y}+ a \,\mathrm{sin}\left( 2 \pi x\right),
\end{split}
\label{eq:MS-1_map}
\end{equation}
where $a=0.05$. The physical domain and a typical grid of MS-1 are presented in Fig. \ref{fig:MS-1_domain}. The transformation in Eq. \eqref{eq:MS-1_map} fulfills three purposes: 
\begin{itemize}
\item The elemental transformation Jacobian  is not trivially constant. It in fact can not be exactly represented by finite polynomial expansions thanks to the presence of a sinusoidal term.
\item The untransformed coordinate $\mathcal{Y}$ can be expressed as an explicit function of $x$ and $y$ coordinates, viz., $\mathcal{Y}=y-a\,\mathrm{sin}(2\pi x)$, defining isolines of constant distance from the bottom of the domain.
\item For $\mathcal{Y}=0$, the wall coordinates can be parametrized by $y$ as an explicit function of $x$, thus allowing to derive an exact expression for the unit outward normal vector at the wall that reads  
\[{\bm{n}}_w = \frac{1}{\sqrt{(\frac{\partial y}{\partial x})^2 + (-1)^2}}\left(\frac{\partial y}{\partial x} \,{\bf{e}}_1 -{\bf{e}}_2 \right).\]
\end{itemize}

\begin{figure}[!hbt]
\centering
\includegraphics[trim = 4mm 2mm 40mm 8mm, clip,width=0.43\linewidth]{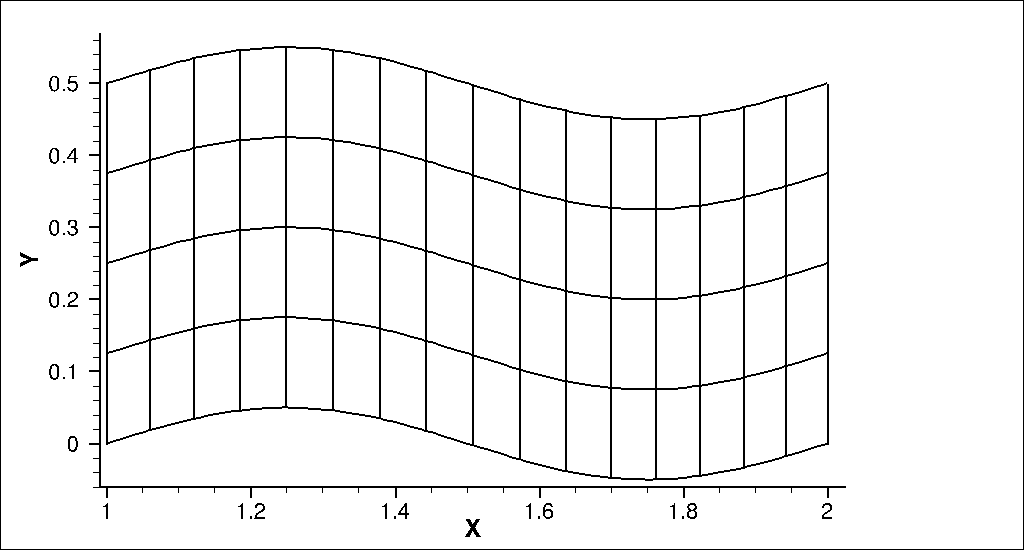}
\caption{Domain and typical grid of MS-1}
\label{fig:MS-1_domain}
\end{figure}

The Riemann BC is applied on all sides of the domain except on the bottom where the slip wall is enforced. A grid set is generated by sequentially removing every other grid line from the finest grid of $1024 \times 256$ quadrangular elements for which the element sizes are distributed following a  geometric series with fixed expansion ratios of $r_\mathcal{X} = 0.993$ and $r_\mathcal{Y} = 1.0$ in the corresponding undeformed coordinates directions. The transformation \eqref{eq:MS-1_map} is then applied to resulting grid and solution node coordinates to retrieve the tessellation of the deformed domain.

The manufactured solution fields are
\begin{equation}
\begin{aligned}
\rho^{\mathrm{MS}} &= \rho_0 + \mathcal{Y}^2,   &\\
u^{\mathrm{MS}}  &= u_{w} + \mathcal{Y},   &  \\
v^{\mathrm{MS}}  &= \frac{\partial y}{\partial x} \, u^{\mathrm{MS}},   &  \\
p^{\mathrm{MS}}  &= p_0  + \mathcal{Y}^2, & 
\label{eq:trigo_MS-1}
\end{aligned}
\end{equation}
where $\rho_0=1.0$, $p_0=1.0$ and $u_{w}=1.0$ is the horizontal velocity component at the wall.

The following remarks can be made with regards to the manufactured solution \eqref{eq:trigo_MS-1}:
\begin{itemize}
\item The manufactured density and pressure are defined with a null normal gradient at the wall as would be expected in a physically realistic case.
\item The velocity field is defined parallel to the wall ensuring the fulfilment of the non-penetration condition along the slip wall. This can be easily verified since by construction, ${\bm{u}}^{\mathrm{MS}} \cdot {\bm{n}}_w = u^{\mathrm{MS}} \,\frac{\partial y}{\partial x} \,+\, \frac{\partial y}{\partial x} \,u^{\mathrm{MS}} (-1) = 0$.
\end{itemize}

In a first attempt, the grid convergence study for MS-1 resulted in OOAs  for all variables similar to the ones in  Fig. \ref{fig:MS-1_isoparnormal}, featuring the recovery of the theoretical order, $\mathcal{O}\left(h^{\text{P}+1}\right)$, by the $L_1$ norm only, whereas the $L_2$ and $L_\infty$ norms respectively yield $\mathcal{O}\left(h^{\text{P}+1/2}\right)$ and  $\mathcal{O}\left(h^{\text{P}}\right)$. Through investigation, the root of this suboptimal outcome was traced back to the manner in which the slip wall boundary condition is prescribed. Indeed in Eq. \eqref{eq:slip_BC}, the outer state at the wall depends on the inner wall normal defined by the spatial transformation metrics between the computational and physical domains that depend in turn on the first order derivatives of the shape functions \cite{Wang-et-al_2011a}. The OOAs of the inner normal is therefore bounded by those of the derivatives of the shape functions, that for an isoparametric mapping, i.e., when the shape function and solution interpolation function are chosen to be the same, converge according to $\mathcal{O}\left(h^{\text{P}}\right)$. Employing the inexact normal results in the discretization error to reduce more slowly at the wall with mesh refinement compared to the rest of the domain. The maximum error in the domain is hence set by the local error at the wall as shown in Fig. \ref{fig:MS-1_error_ro_nrml}(a). This explains as well the disagreement between the OOAs based on different $L$ norms. The $L_1$ norm reports the optimal orders, since in its computation, the irregularity at wall bordering nodes is outweighed by the large number of regular contributions from the nodes of the rest of the domain as all deviations from the average are equally weighted by a unitary coefficient. Nevertheless, the higher sensitivity of the $L_2$ norm to deviations from the average and the ability of the $L_\infty$ norm to detect local inconsistencies show their advantages by revealing this particular issue.

The suboptimal error reduction rates can be rectified by considering modified slip conditions as in \cite{Krivodonova_2006}, by adopting a superparametric representation of the geometry as in \cite{Zwanenburg2017} or by using the analytical normal to the wall instead of the normal computed by an isoparametric mapping in Eq. \eqref{eq:slip_BC}.  We adopt the latter approach here since we have defined the spatial domain of MS-1 such that the exact outward normal to the wall is explicitly available. This technique remedies the accuracy of the boundary state of momentum equation, $Q_k^{BC}=(\rho u_q)^{BC}$
(for $k \in [2\, .. \, N_\mathrm{d}+1]$ and $q=k-1$), that now only depends on terms of at least the same degree as the inner solution.  We thus recover lower maximum error levels, distributed in the domain in Fig.\ref{fig:MS-1_error_ro_nrml}(b), instead of being clustered at the wall. Consequently, the theoretical OOAs are recovered by all $L$ norms (see Figs. \ref{fig:Err_allE_allP_MS-1} and \ref{fig:Orders_MS-1}). Let's however note the slower appearance of the asymptotic range of $\rho$ for orders higher than two in Fig. \ref{fig:Orders_MS-1}.

\begin{figure}[!hbt]
\centering
\subfloat[$\rho$]{
\includegraphics[trim = 11mm 3mm 18mm 12mm, clip,width=0.33\linewidth]
{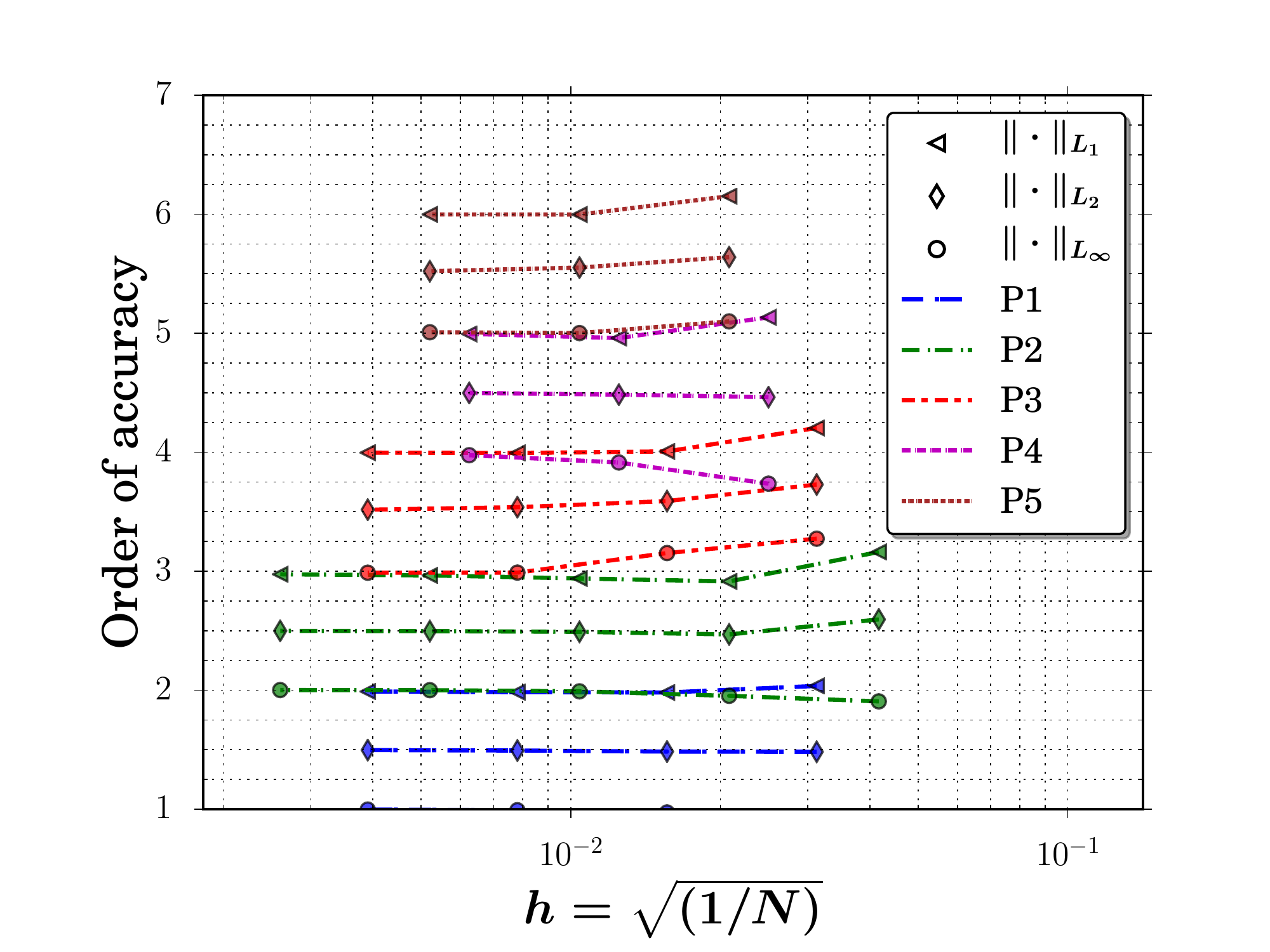}}
~~~
\subfloat[$\rho u$]{
\includegraphics[trim = 11mm 3mm 18mm 12mm, clip,width=0.33\linewidth]
{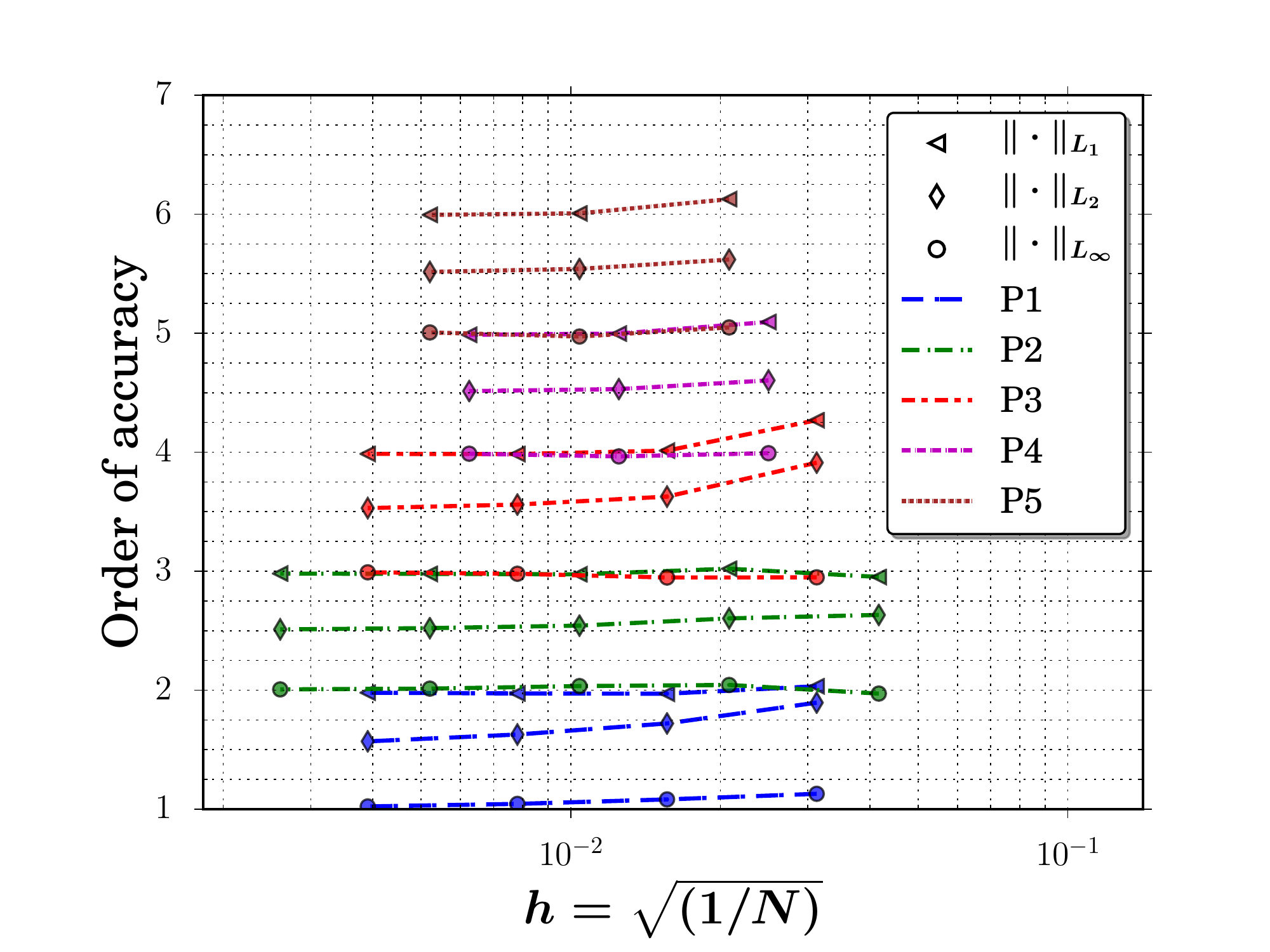}}
\caption{Evolution of the OOAs in $L_1$, $L_2$ and $L_\infty$ norms versus mesh refinement for MS-1 with the inexact normal in the slip wall condition and polynomial degrees $\mathrm{P}1$--$\mathrm{P}5$}
\label{fig:MS-1_isoparnormal}
\end{figure}

\begin{figure}[!hbt]
\centering
\subfloat[Inexact ${\bm{n}}_w$]{
\includegraphics[trim = 0mm 0mm 0mm 0mm, clip,width=0.4\linewidth]
{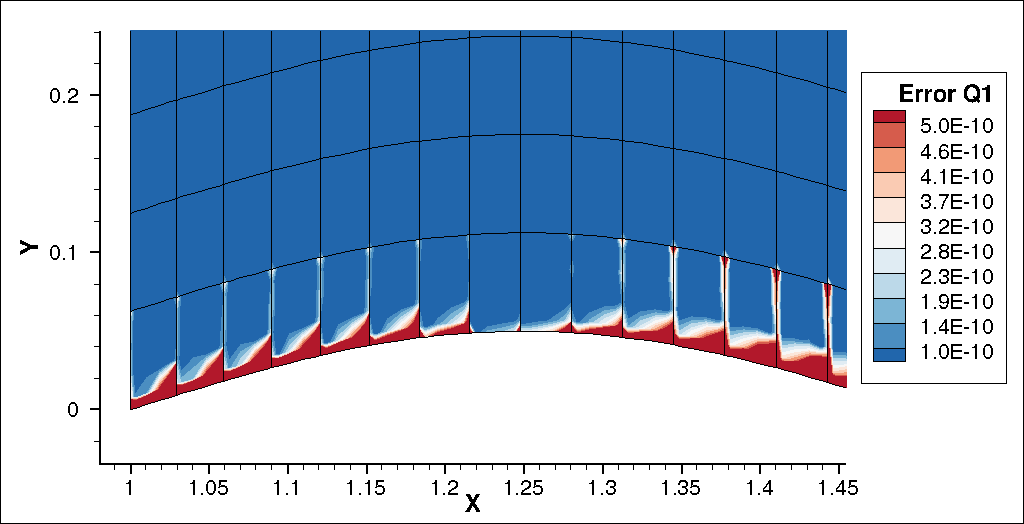}}
~~~
\subfloat[Exact ${\bm{n}}_w$]{
\includegraphics[trim = 0mm 0mm 0mm 0mm, clip,width=0.4\linewidth]
{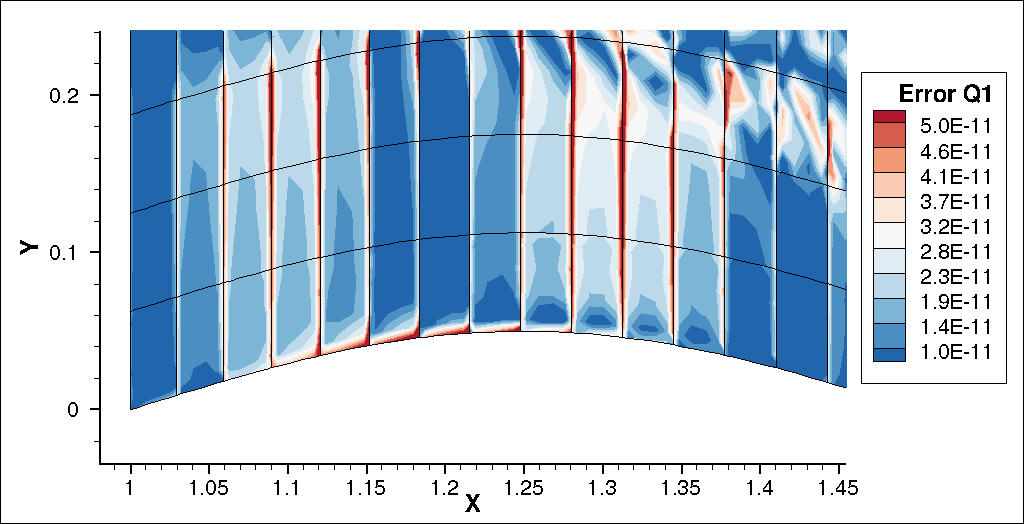}}
\caption{Error distribution in the vicinity of the slip wall for $\rho$ of MS-1 with the slip wall condition computed via the inexact and exact normals}
\label{fig:MS-1_error_ro_nrml}
\end{figure}

\subsection{Laminar flows on curved wall-bounded domain - MS-2}
MS-2 is devised here specially for the verification of wall-bounded laminar flows on curved domains and in incompressible as well as subsonic regimes at once. 

\begin{figure}[!hbt]
\centering
\includegraphics[trim = 4mm 2mm 40mm 8mm, clip,width=0.43\linewidth]{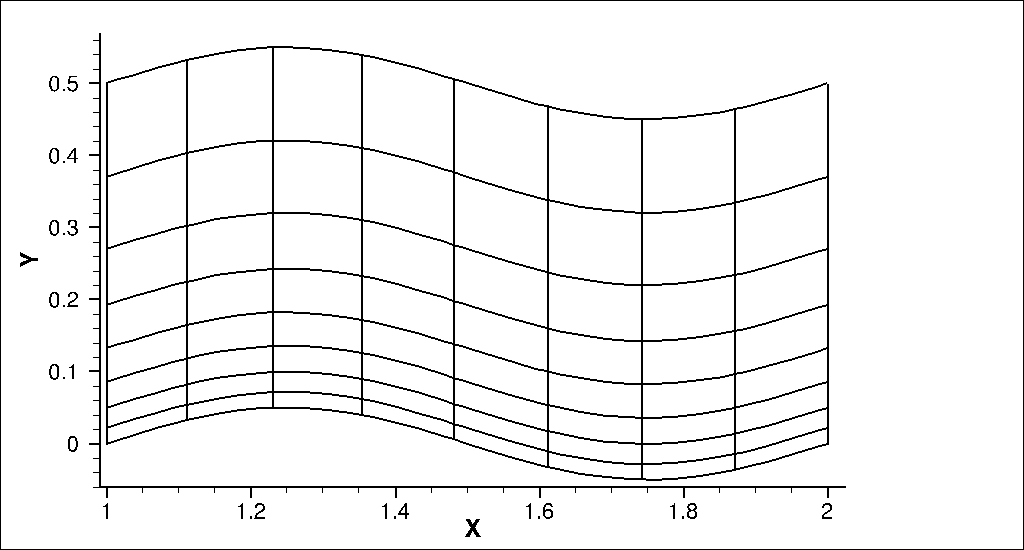}
\caption{Domain and typical grid of MS-2}
\label{fig:MS-2_domain}
\end{figure}

The domain of MS-2 is the same as that of MS-1, defined via the transformation  \eqref{eq:MS-1_map}. However, the grids are generated by the sequential coarsening of a fine grid of $1024 \times 1024$  quadrangular elements with geometric series expansion ratios of $r_\mathcal{X} = 0.9995$ and $r_\mathcal{Y} = 1.002$ in untransformed coordinates directions. This results in a clustering of anisotropic elements close to the wall boundary, tailored to capture the velocity profile growth. A typical grid is illustrated in Fig. \ref{fig:MS-2_domain}. The  no-slip wall BC is applied on the bottom while all other sides of the domain are treated by the Riemann and viscous BCs. The  manufactured fields are defined by:
\begin{equation}
\begin{aligned}
\rho^{\mathrm{MS}} &\equiv \rho_0 + \frac{\mathcal{Y}^2}{\rho_c^2},   &\\
u^{\mathrm{MS}}  &= \mathrm{erf}(\eta),   &  \\
v^{\mathrm{MS}}  &= \frac{\partial y}{\partial x} \, u^{\mathrm{MS}},   &  \\
p^{\mathrm{MS}}  &= p_0  + \mathcal{Y}^2, & 
\label{eq:trigo_MS-2}
\end{aligned}
\end{equation}
where $\rho_0=1.2$, ${\rho_c=3.0}$,  $\mathrm{erf}$ designates the error function,  $\eta$ is a similarity variable defined as
\begin{equation}
\eta = \sigma \frac{\mathcal{Y}}{\sqrt{\frac{4}{3}(\mathcal{X}^2+\frac{1}{2})}},
\end{equation}
with $\sigma=2.0$ and finally $p_0=2.0$. The horizontal velocity field is inspired from that of manufactured solutions in \cite{Eca-et-al_2007} and the dynamic viscosity is set to $\mu=10^{-5}$ to mimic an aerodynamic boundary layer flow with a relatively large Reynolds number. Nevertheless, let's note that low values of $\mu$ translate to a modest sensitivity of the verification process to errors in the viscous terms of the governing equations as discussed in \cite{Navah2017a}. In order to assess the sensitivity of the verification via MS-2 to bugs in the viscous terms, we introduce an inconsistency in the first component of the shear stress tensor via $(1 + d\alpha)\,F^{vis}_{11}$ where $d\alpha$ is gradually increased from its original value of zero until a modification in the orders of accuracy is detected. Thus, MS-2 revealed the alteration by the P4 and P5 discretizations as soon as $d\alpha \approx 10^{-4}$ as shown in Fig. \ref{fig:MS-2_bug}. This indicates that although MS-2 is sensitive to the presence of bugs in the viscous terms of the governing equations, its sensitivity is rather low and hence it is not fully adequate for the detection of minor bugs in the diffusive terms.  We therefore strongly recommend to employ MS-2 only once the code under examination has been verified for the free laminar flow of MS-3 of \cite{Navah2017a} which features a demonstrated balancing of the viscous versus inviscid terms and hence enables the detection of minor bugs in either. A prior verification via MS-3 of \cite{Navah2017a} allows therefore to safely proceed to MS-2 of the present study where the primary focus is on the verification of the no-slip wall condition on curved domains. 

The grid convergence of errors and OOAs for MS-3 in $L$ norms (Figs. \ref{fig:Err_allE_allP_MS-2} and \ref{fig:Orders_MS-2}) are indicative of absence of bugs for the numerical framework at hand, although negligible fluctuations exist in the OOAs for some polynomial degrees (ex: $\rho$ in P1 and P2).

\begin{figure}[!hbt]
\centering
\subfloat[Discretization errors]{
\includegraphics[trim = 7mm 3mm 18mm 12mm, clip,width=0.33\linewidth]
{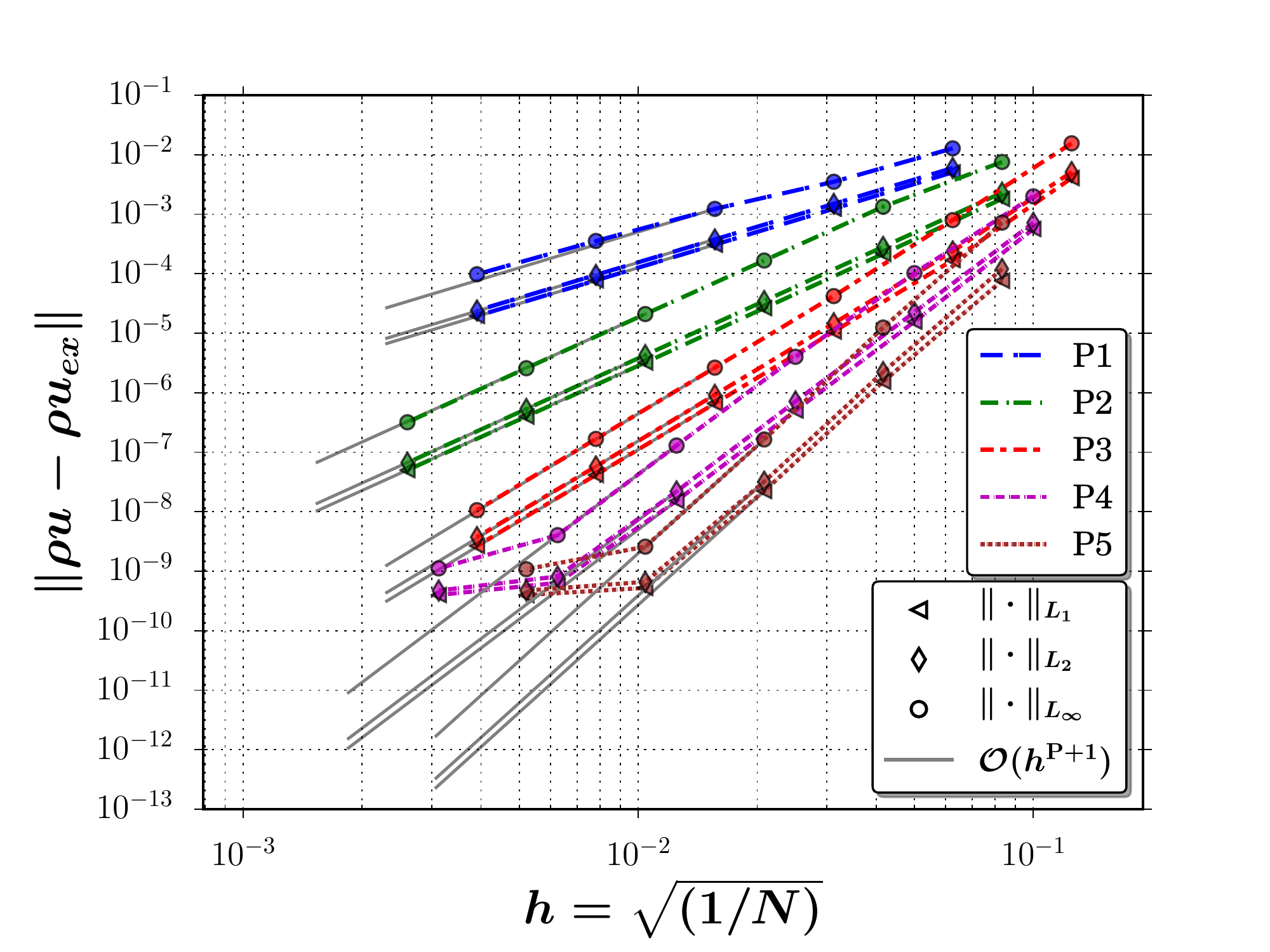}}
~~~
\subfloat[Orders of accuracy]{
\includegraphics[trim = 7mm 3mm 18mm 12mm, clip,width=0.33\linewidth]
{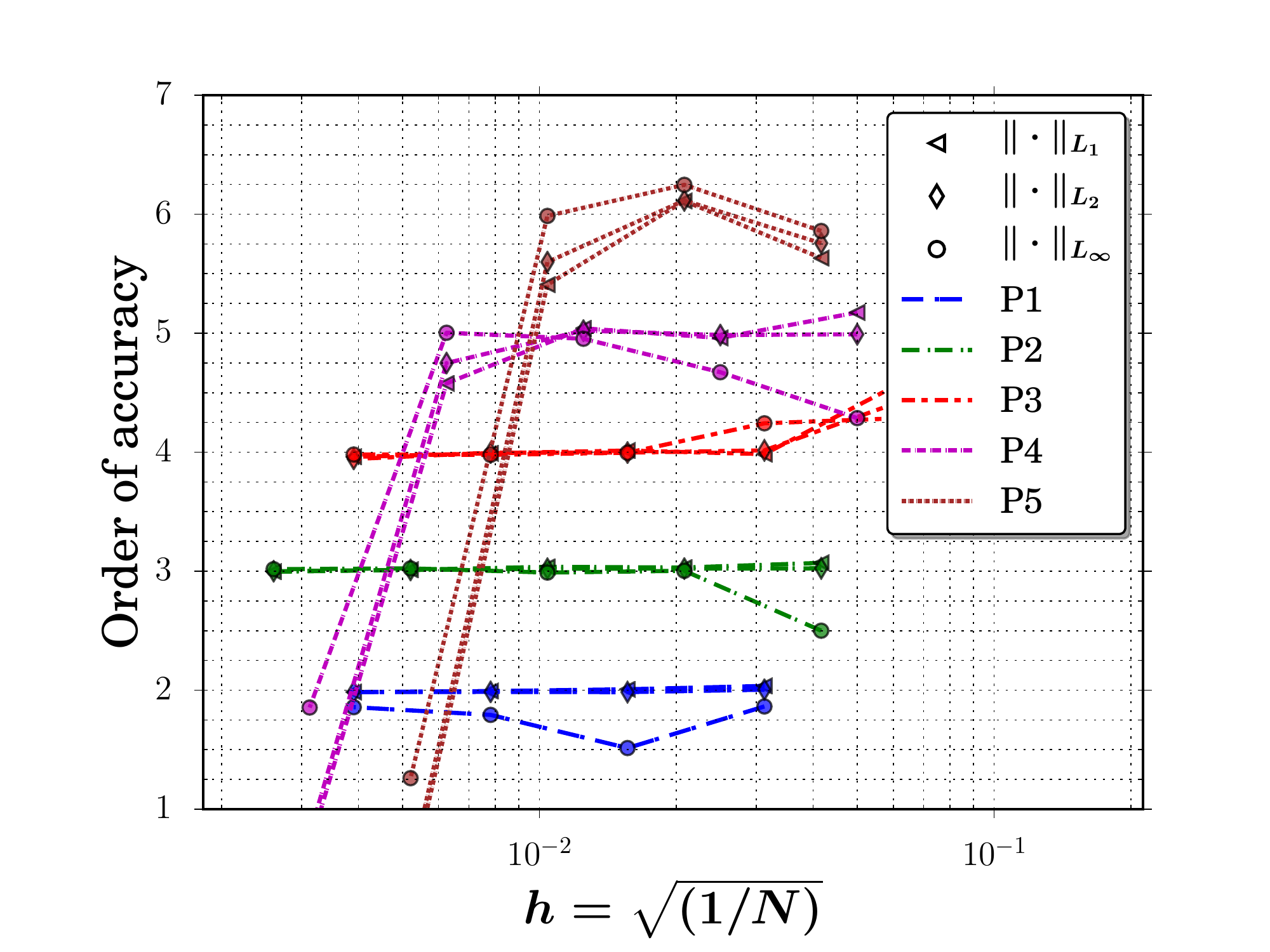}}
\caption{Evolution of the discretization error and OOAs in $L$ norms versus mesh refinement for MS-2  and polynomial degrees P1--P5 under the effect of the term $(1 + d\alpha)\,F^{vis}_{11}$ where $d\alpha \approx 10^{-4}$}
\label{fig:MS-2_bug}
\end{figure}

\subsection{Wall-bounded turbulent flows - MS-3 and MS-4}
The manufactured cases for turbulent flow such as those presented in \cite{Navah2017a} 
are based on trigonometric functions and hence do not mimic real flows; consequently, although the conducted budget analyses show that a sufficient level of balance between groups of terms of their forcing functions is achieved, this balance does not necessarily resemble the one occurring in the operation mode of the model in practical scenarios. Producing such a balance is the task of physically realistic MSs and serves to reinforce the conclusions of the verification campaign. In the case of RANS models, an example of a physically realistic MS is a boundary layer flow. Such MSs allow as well to examine the wall BC implementation for the turbulence model, to verify the rate of convergence of aerodynamic outputs and to assess the performance of error estimators and adaptive strategies in realistic scenarios. 

\paragraph{MS-3}
E\c{c}a and colleagues presented \cite{Eca-et-al_2007,Eca2007a,Eca2007} a wall-bounded RANS-based MS mimicking a realistic boundary layer, although with some deficiencies such as the lack of logarithmic layer in the velocity profile which furthermore corresponds to that of a laminar flow rather than a turbulent one,  in terms of shape factor and shear stress at the wall \cite{Eca2012}. Also, among the three presented $\tilde \nu$ profiles, the second order of accuracy was achieved only for one \cite{Eca2007}, labelled "MS2" in \cite{Eca-et-al_2007}, which features a quadratic dependency on the $y$ coordinate for $y\ll1$ versus the expected linear dependency from the law of the wall of the SA model \cite{Allmaras-et-al_2012}. Furthermore, this MS still fails to replicate the expected behaviour of the SA's production and destruction terms in the vicinity of   wall \cite{Oliver-et-al_2012}. We here extend its application to the verification of discretizations with up to the sixth OOA and we refer to it as MS-3, defined by the following dimensionless fields:
\begin{equation}
\begin{aligned}
\rho^{\mathrm{MS}} &= \rho_0   + \rho_x \,\mathrm{sin}(a_{\rho_{x}}    \pi   x / L)  + \rho_y \,\mathrm{cos}(a_{\rho_{y}}   \pi   y / L)  + \rho_{xy} \,\mathrm{cos}(a_{\rho_{x y}}   \pi   x / L)\, \mathrm{cos}(a_{\rho_{x y}}   \pi   y / L),&\\
u^{\mathrm{MS}}  &= \mathrm{erf}(\eta),&\\
v^{\mathrm{MS}}  &= \frac{1}{\sigma \sqrt{\pi}}\left(1-e^{-\eta^2}\right),&\\
p^{\mathrm{MS}}  &= p_0 + 0.5 \,\mathrm{ln}(2x-x^2+0.25)\,\mathrm{ln}(4y^3-3y^2+1.25) , & \\
\tilde{\nu}^{\mathrm{MS}}  &= \tilde{\nu}_\mathrm{max}\,\eta_\nu^2 \,e^{1-\eta_v^2}, & 
\label{eq:MS-3}
\end{aligned}
\end{equation}
where $\eta= \frac{\sigma y}{x}$ is a similarity variable, $\sigma=4$, $\tilde{\nu}_\mathrm{max}=10^3\, \nu_0$ with $\nu_0=\mu=10^{-6}$, $\eta_\nu= \frac{\sigma_\nu y}{x}$ with $\sigma_\nu=2.5\,\sigma=10$ and $p_0=0$ in the original version of this MS. The density field in \eqref{eq:MS-3} is defined by $\rho_0=1.1$, $\rho_x=0.2$, $\rho_y=0.5$, $\rho_{xy}=0.2$,   $a_{\rho_x}=1.0$, $a_{\rho_y}=2.0$, $a_{\rho_{xy}}=3.0$ and $L=1$. Finally, the domain of MS-3 is defined by $\Omega = [0.5,1.0]\times [0.0,0.5]$. 

We have applied two modifications to the original version of this MS:
\begin{itemize}
\item The unitary density field of the original MS meant for the verification of incompressible solvers is replaced by the trigonometric field in \eqref{eq:MS-3} to enable the verification of compressible RANS equations. The field is devised to satisfy the adiabatic boundary condition via null normal gradients at the wall; 
\item The original pressure field featured a region of null values in the domain that resulted in unphysical undershoots due to Gibbs oscillations on coarse grids. By setting $p_0=1$ in \eqref{eq:MS-3}, the pressure is lifted by a unit, thus avoiding these numerical difficulties.
\end{itemize}

The solution fields of MS-3 are presented in Fig. \ref{fig:MS-3}. The domain is discretized by a fine grid of $1024 \times 3072$ elements, the sizes of which are defined by a geometric series expansion in each space coordinates with ratios of $r_x=1.0$ and $r_y=1.0017$. A set of ten grids is generated by sequentially removing every other grid line from this fine grid. The results of the grid convergence study are presented in Figs. \ref{fig:Err_allE_allP_MS-3} and \ref{fig:Orders_MS-3} for $L$ norms and in Figs. \ref{fig:Err_allE_allP_H_MS-3} and \ref{fig:Orders_H_MS-3} for $H_1$ semi-norm of uncorrected and fully corrected derivatives. These figures show an especially late (in terms of $h$) occurrence of the asymptotic range for most variables and polynomial degrees. This is also manifested in $H_1$ semi-norm, particularly for $\rho$ and $\rho u$. In a first step, we ensure that the modification to the density field that we have introduced is not the reason behind this outcome. To this end, the OOAs are recomputed for MS-3 with a unitary density field. The results of this test in terms of $\rho$ and $\rho \tilde \nu$ variables, presented in Fig. \ref{fig:MS-3_orders_cons_dens}, attest that the modification to the density solution is indeed not responsible for the delayed monotonic ranges as this still occurs with the original definition of manufactured density.
\begin{figure}[!hbt]
\centering
\subfloat[$\rho$]{
\includegraphics[trim = 11mm 3mm 18mm 12mm, clip,width=0.33\linewidth]
{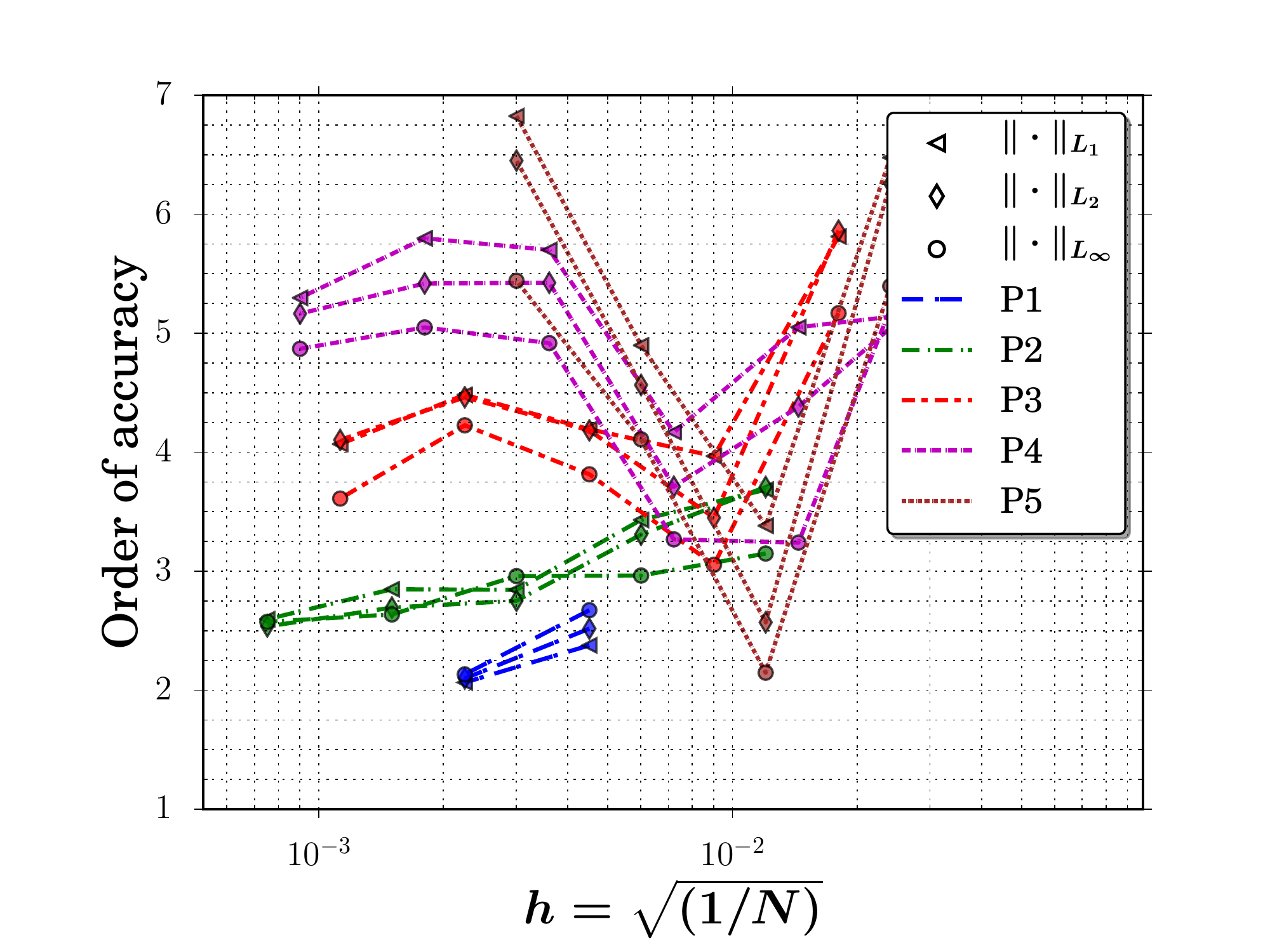}}
~~~
\subfloat[$\rho \tilde \nu$]{
\includegraphics[trim = 11mm 3mm 18mm 12mm, clip,width=0.33\linewidth]
{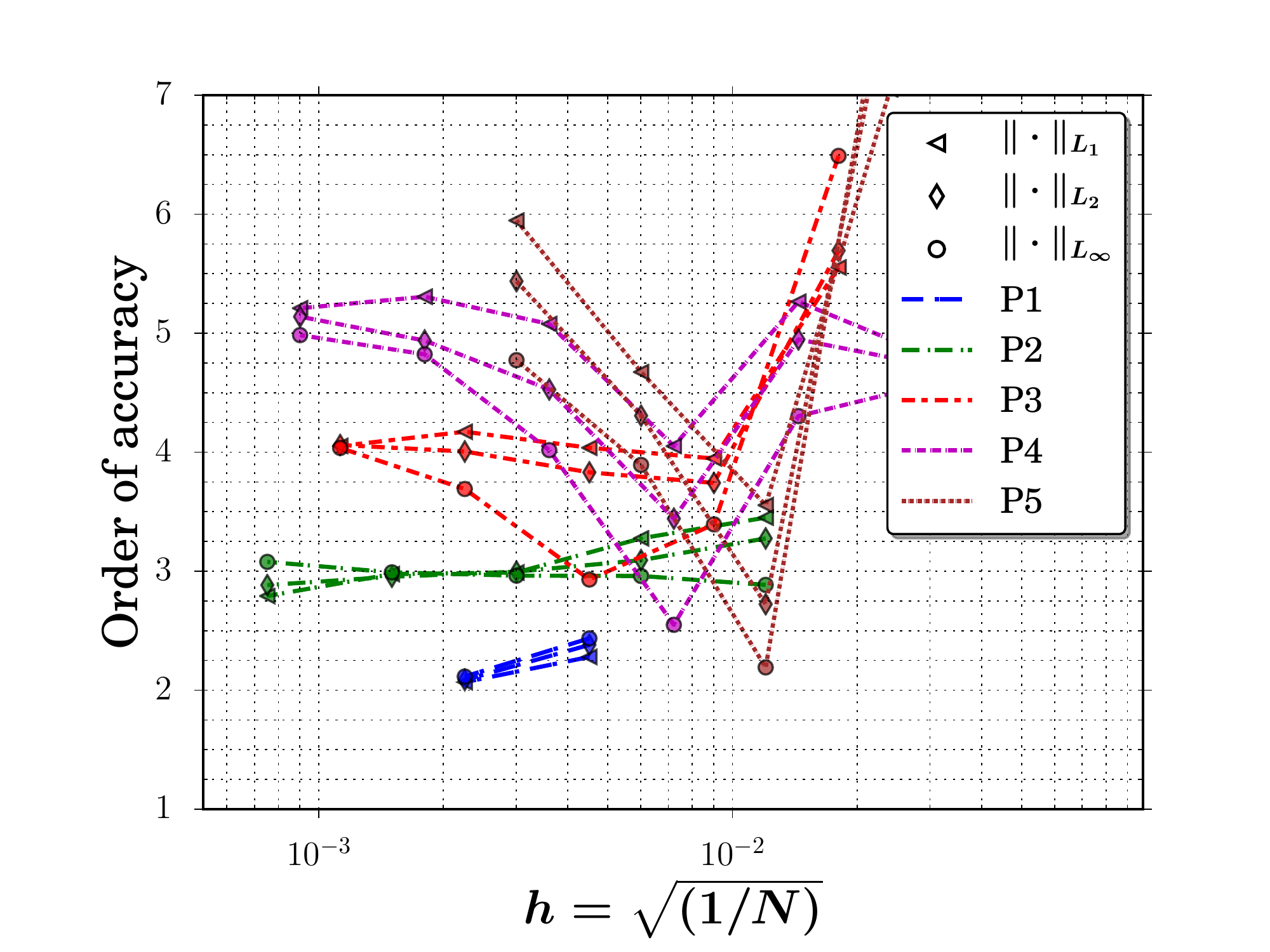}}
\caption{Evolution of the OOAs in $L$ norms versus mesh refinement for MS-3 along with a unitary density field   for $\rho$ and $\rho \tilde \nu$ variables and polynomial degrees P1--P5}
\label{fig:MS-3_orders_cons_dens}
\end{figure}

In our search for the reason behind this effect, we also look into the $\tilde \nu$ field and associated source terms. By deactivating the turbulent portion of MS-3 and hence applying it in the NS mode via $\tilde \nu=0$ in the RANS equations, through Eq. \eqref{eq:mu_t+}  and while keeping $\mu=10^{-6}$, more regular OOAs with an earlier asymptotic range are obtained on the same grid set as shown in Fig. \ref{fig:MS-3_NS}. One reason for the delay is the undershoot of $\tilde \nu$ values in large regions of the domain where it is defined very close to zero, thus activating the modified version of the SA model and hence mismatching the forcing function defined based on a positive $\tilde \nu^{MS}$. This especially occurs on coarse grids and for lower polynomial degrees (P1 is the most affected) in this case, thus delaying the appearance of the monotonic convergence. Another  reason behind the delayed monotonic convergence resides in the fact that by activating the $\tilde \nu$ field, a region of large discretization error appears in the middle of the domain as shown in Fig. \ref{fig:MS-3_NS_err}, which due to its extent, dominates the error norms. Accelerating the convergence would hence require a grid set that clusters the degrees of freedom (DOFs) not only at the wall but in this region as well. Another possibility is to use adaptive techniques as in \cite{Navah2009}.

The delayed asymptotic convergence of MS-3 along with the mentioned lack of compliance with the expected solutions of the RANS-SA model motivates us to consider as well another MS for RANS-modelled flows.

\begin{figure}[!hbt]
\centering
\subfloat[$\rho$]{
\includegraphics[trim = 11mm 3mm 18mm 12mm, clip,width=0.33\linewidth]
{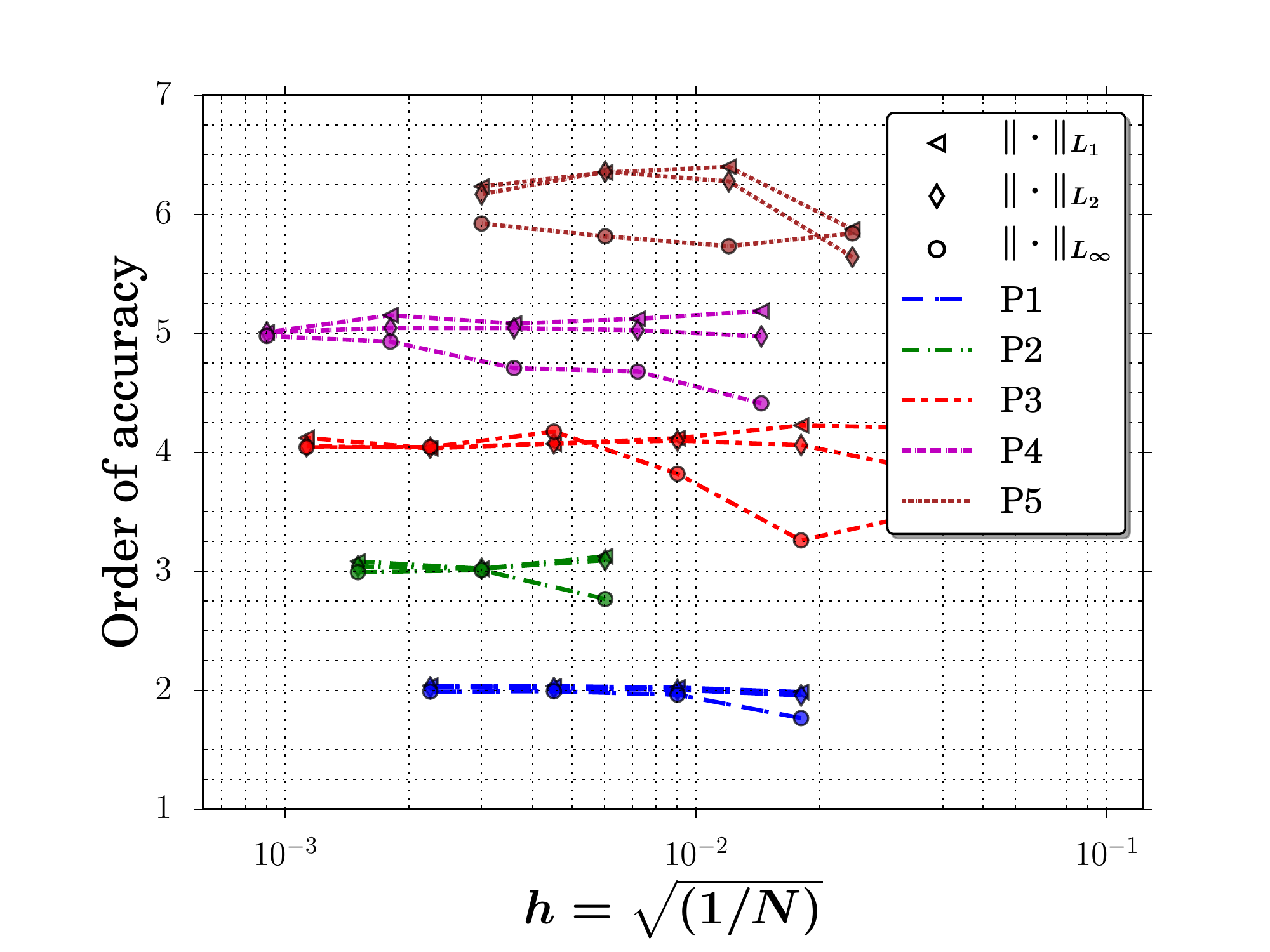}}
~~~
\subfloat[$\rho u$]{
\includegraphics[trim = 11mm 3mm 18mm 12mm, clip,width=0.33\linewidth]
{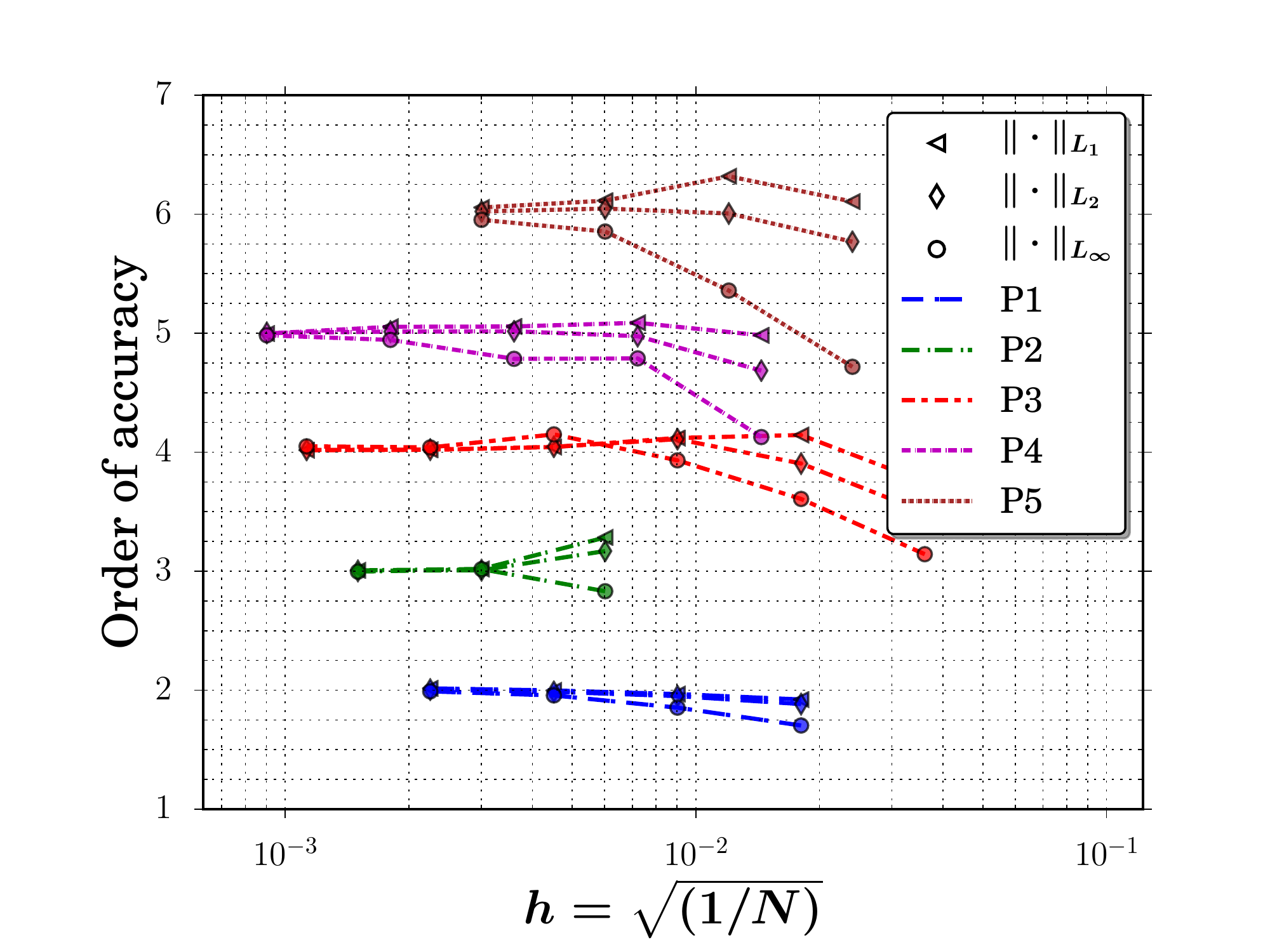}}
\caption{Evolution of the OOAs in $L$ norms versus mesh refinement for MS-3 in the NS mode (with $\tilde \nu=0$ and $\mu=10^{-6}$)  for $\rho$ and $\rho u$ variables and polynomial degrees P1--P5}
\label{fig:MS-3_NS}
\end{figure}

\begin{figure}[!hbt]
\centering
\subfloat[$\mu_\mathrm{eff}=\mu$]{
\includegraphics[width=0.3\linewidth]
{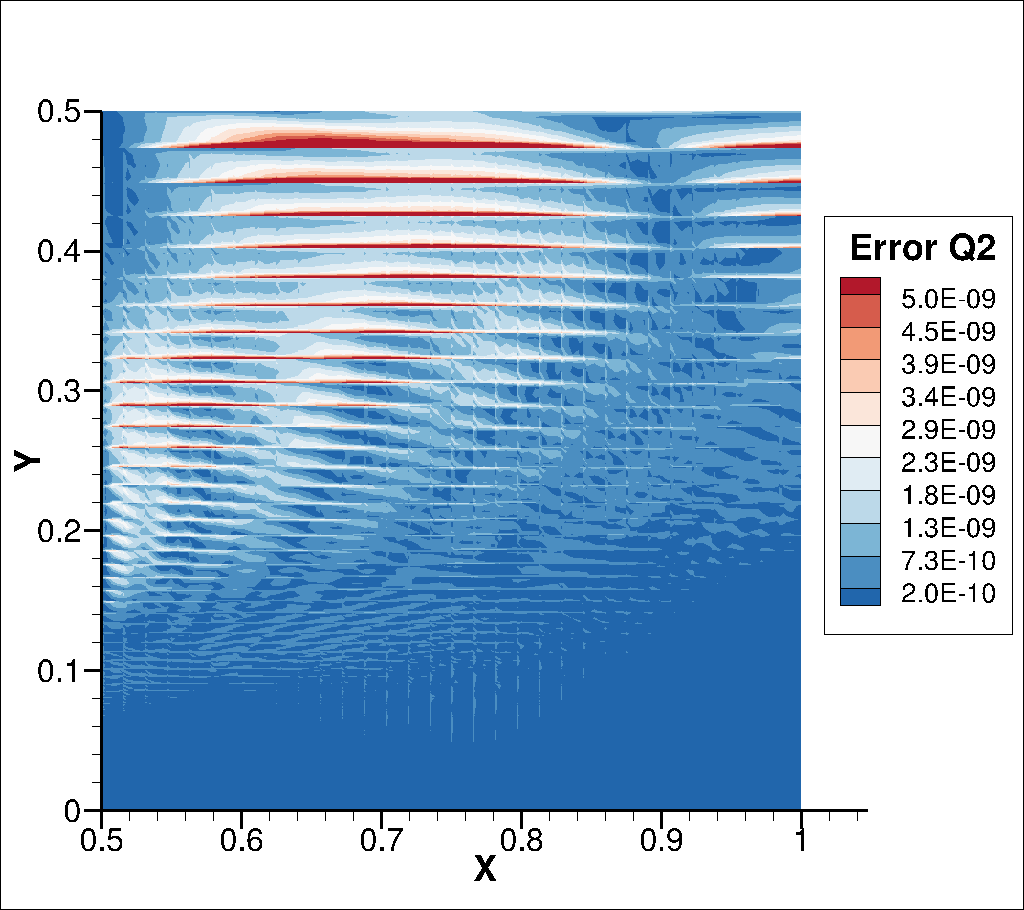}}
~~~
\subfloat[$\mu_\mathrm{eff}=\mu + \mu_t$]{
\includegraphics[width=0.3\linewidth]
{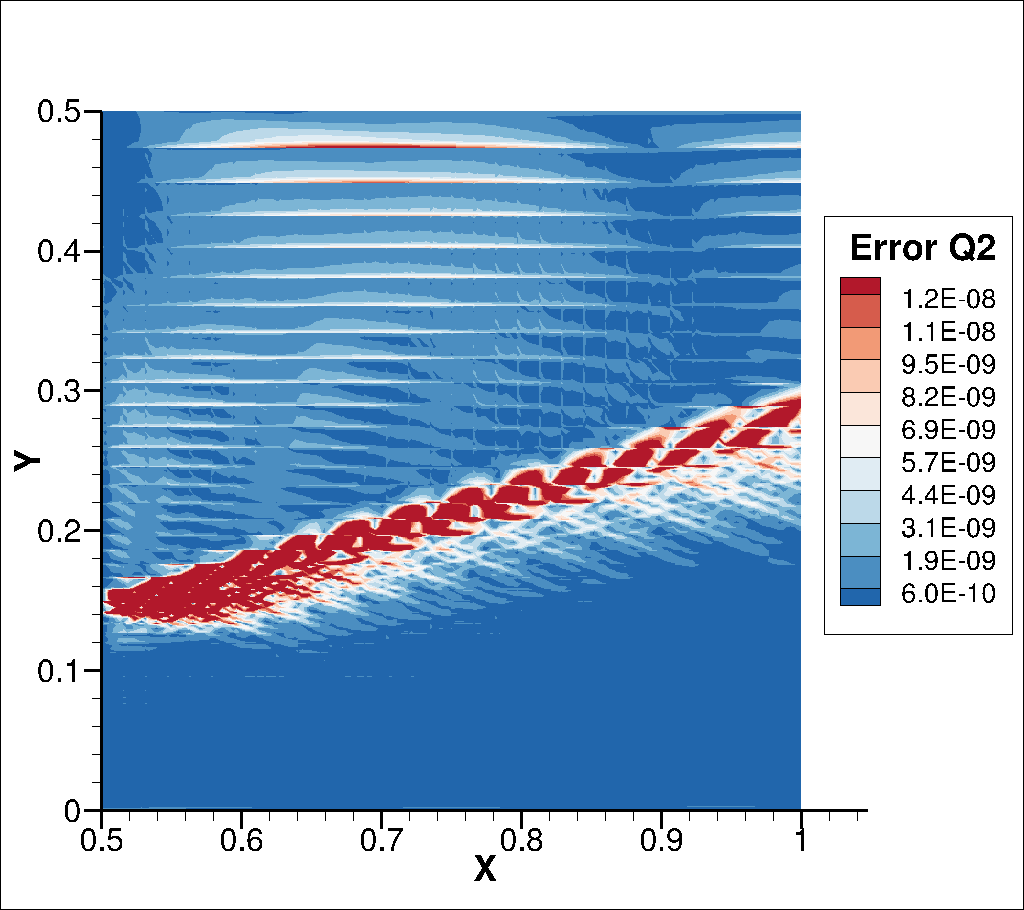}}
\caption{Distribution of the discretization error for the variable $\rho u$ of MS-3 and polynomial degree P4 in NS and RANS-SA modes}
\label{fig:MS-3_NS_err}
\end{figure}

\paragraph{MS-4}
Oliver and colleagues reviewed \cite{Oliver-et-al_2012} the physically realistic manufactured cases presented by E\c{c}a et al. in the year 2007 \cite{Eca-et-al_2007,Eca2007a,Eca2007} and discussed problematic aspects of these MSs that cause a departure from the expected RANS behaviour. They presented \cite{Oliver-et-al_2012} as well a compressible zero-pressure-gradient flat plate MS based on well-established correlations for boundary layer flow and the near-wall $\tilde \nu$ solutions of the SA model. This MS, labelled here MS-4, thus mimics the inner portion, that is, the  viscous sublayer and the logarithmic layer, of a RANS-modelled boundary layer. Oliver and colleagues demonstrated \cite{Oliver-et-al_2012} as well the application of this MS for the verification of a second order finite element solver.

We extend here the application of MS-4 to the verification of high-order schemes. Furthermore, the proper treatment of SA source terms at the no-slip wall in high-order frameworks is presented. The verification of the modified vorticity term of the modified SA equations is discussed. An analysis of the drag coefficient convergence is conducted and the effect of the mesh on the OOAs is presented. Finally,  a non-dimensional version of this MS is introduced and the effect of the non-dimensionalization on the conditioning of the linear system is studied.

We refer the reader to \cite{Oliver-et-al_2012} for explanations on the construction of MS-4 and only present its primitive fields here that are expressed by
\begin{equation}
\begin{aligned}
\rho^{\mathrm{MS}} &= \frac{p_0}{R\,T^{\mathrm{MS}}},   &\\
u^{\mathrm{MS}}  &= \frac{u_\infty}{A}\mathrm{sin} \left( \frac{A}{u_\infty}u_{eq}\right),&\\
v^{\mathrm{MS}}  &= -\eta_v \frac{d u_\tau}{dx}y,   &  \\
p^{\mathrm{MS}}  &= p_0 , & \\
\tilde{\nu}^{\mathrm{MS}}  &= \kappa u_\tau y -\alpha y^2, & 
\label{eq:trigo_MS-4}
\end{aligned}
\end{equation}
where
\begin{equation}
\begin{gathered}
R=287.0 \, (\mathrm{J/(kg\,K})),  \quad p_0=1\times 10^4 \, (\mathrm{N/m^2}), \\ T_\infty =250.0\, (\mathrm{K}), \quad \mu=1\times 10^{-4} \,(\mathrm{kg/(m\,s)}),\quad \alpha= 5.0\,(\mathrm{1/s}),
\label{eq:vars_MS-4-a}
\end{gathered}
\end{equation}
and furthermore 
\begin{equation}
\begin{gathered}
{Ma}_\infty = 0.8, \quad T^{\mathrm{MS}}=T_\infty \left[ 1+r_T\frac{\gamma-1}{2}{Ma}_\infty^2 \left( 1- \left( \frac{u^\mathrm{MS}}{u_\infty} \right)^2 \right) \right], \quad
r_T=0.9,\\ u_\infty = {Ma}_\infty\sqrt{\gamma \,R\, T_\infty},   
\quad
A=\sqrt{1-T_\infty/T_w}, \quad T_w = T_\infty \left[ 1+r_T \frac{\gamma-1}{2} {Ma}_\infty^2 \right], \quad
u_{eq} = u_\tau u_{eq}^+, \\
u_\tau = \sqrt{\frac{\tau_w}{\rho_w}} = u_\infty \sqrt{\frac{c_f}{2}},  \quad c_f = \frac{1}{F_c}c_{f,inc} \left(\frac{1}{F_c} \mathrm{Re}_x \right), \quad
F_c = \frac{T_w/T_\infty -1}{(\mathrm{sin}^{-1}A)^2}, \\
c_{f,inc}(\mathrm{Re}_x) = C_{cf}\mathrm{Re}_x^{-1/7}, \quad C_{cf} = 0.027, 
\quad \mathrm{Re}_x = \rho_\infty u_\infty x/\mu,\\
\rho_\infty= \frac{p_0}{R\,T_\infty}, \quad u_{eq}^+ = \frac{1}{\kappa}\mathrm{log}\left( 1+ \kappa y^+  \right) + C_1 \left[ 1 - e^{-y^+/\eta_1} - \frac{y^+}{\eta_1}e^{-y^+b} \right], \\
\quad C_1=-(1/\kappa)\mathrm{log}(\kappa)+C, \quad \eta_1 = 11.0, \quad b=0.33,\quad C=5.0,\\
y^+=\frac{y}{l_v}, \quad l_v = \frac{\nu_w}{u_\tau}, \quad \nu_w = \frac{\mu}{\rho_w}, \quad \rho_w = \frac{p_0}{R\,T_w}, \quad \mathrm{and}\quad \eta_v = 30.0.
\label{eq:vars_MS-4-b}
\end{gathered}
\end{equation}

We set all the remaining constants and parameters necessary for the determination of the forcing functions to the values specified in Section \ref{sec:goveq}. The domain is defined by $\Omega = [0.5,0.55]\times [0.0,0.034]$ for which the Reynolds number varies between $176,690  \le \mathrm{Re}_x \le 194,359$. The boundary conditions are the no-slip adiabatic wall at the bottom and the Riemann and viscous BCs on all other frontiers of the domain. The fields of MS-4 are presented in Fig. \ref{fig:MS-4}.

\subsubsection{Asymptotic values of the SA source terms at the wall}
With regards to the value of some of the SA source terms at the wall (see Fig. \ref{fig:MS-4_budget} for the SA budget of MS-4  versus $y$ at $x=0.525$), there is a subtle aspect worth of mention. More precisely, the production and destruction terms of the original SA model, respectively expressed by Eqs. \eqref{eq:prod+} and \eqref{eq:D+}, scale with quotients of the form $\propto \frac{\tilde \nu^{a_1}}{d_w^{a_2}}$  where $a_1$ and $a_2$ are real exponents.  These quotients could numerically take NaN values at the wall, for example in situations where a weak boundary condition is employed to enforce $\tilde \nu \approx 0$, as it is commonly the case in compact high-order frameworks, resulting in $\tilde \nu  \neq 0$ at the wall, whereas an exact computation of the wall distance yields ${d_w}=0$. To alleviate this, we rather impose the limit values of the production and destruction terms at the wall which respectively read 
\begin{equation}
(\rho \, \mathcal{P}) |_{\mathrm{wall}} = -c_{b1} \,c_{t3}\,\tau_\mathrm{wall} + c_{b1}\, \tau_\mathrm{wall},
\end{equation}
and
\begin{equation}
(\rho \, \mathcal{D}) |_{\mathrm{wall}} = -c_{b1} \,c_{t3}\,\tau_\mathrm{wall} + c_{w1}\,\kappa^2\, \tau_\mathrm{wall},
\end{equation}
where $\tau_\mathrm{wall}= t_{k-1}\,F^{vis}_{kj}|_{\mathrm{wall}}\, n_j$ with $k \in [2\,..\,{N_\mathrm{d}+1}]$, $j$ $\in [1\,..\,N_\mathrm{d}]$ and $t_{k-1}$ and $n_j$ denote the components of respectively the unit tangential and outward normal vectors at the wall. The tangent vector is the one on the plane defined by the vectors ${\bf{e}}_{k-1}\,F^{vis}_{kj}|_{\mathrm{wall}}\, n_j$ and ${\bf{e}}_j n_j$ and its direction is aligned with the free stream.

This results in a smooth transition of the production and destruction quantities in the domain to their asymptotic values at the wall, thus avoiding numerical invalidity and solution inaccuracy. Let's note that for code verification purposes, the MS forcing functions should as well account for these wall limit values to enable the verification of the asymptotic behaviour of these terms.

\subsubsection{Verification of the modified vorticity term}

None of the trigonometric MSs for turbulent flow presented in \cite{Navah2017a} verifies the term \eqref{eq:S-} which is part of the modified SA version but can be activated in regions of $\tilde \nu > 0$ and affect both the production and the destruction terms. The condition that enables this term is rather $\overline s<-c_{v2}s$. The distributions of $\overline s$, $-c_{v2}s$ and $\overline s+c_{v2}s$ for MS-4 along $x=0.525$ are illustrated in Fig. \ref{fig:MS-4_Sbar} demonstrating that this MS does contain a region, viz. $4.3\times10^{-4} \leq y \leq 6.6\times10^{-4}$, where $\overline s+c_{v2}s < 0$. Numerical tests confirmed that MS-4 is indeed sensitive to the presence of an alteration in the implementation of $\frac{(1+d\alpha)\,s(c_{v2}^2s+c_{v3}\bar{s})}{(c_{v3}-2c_{v2})s-\bar{s}}$ (see Eq. \eqref{eq:S-}) where $d\alpha$ is changed from its genuine value of $d\alpha=0$ (the unaltered term) to $d\alpha=1\times 10^{-7}$, in which case the OOAs drop from their expected values, thus revealing the presence of a bug in the implementation. 

\begin{figure}[!hbt]
\centering
\includegraphics[trim = 2mm 4mm 24mm 21mm, clip,width=0.33\linewidth]
{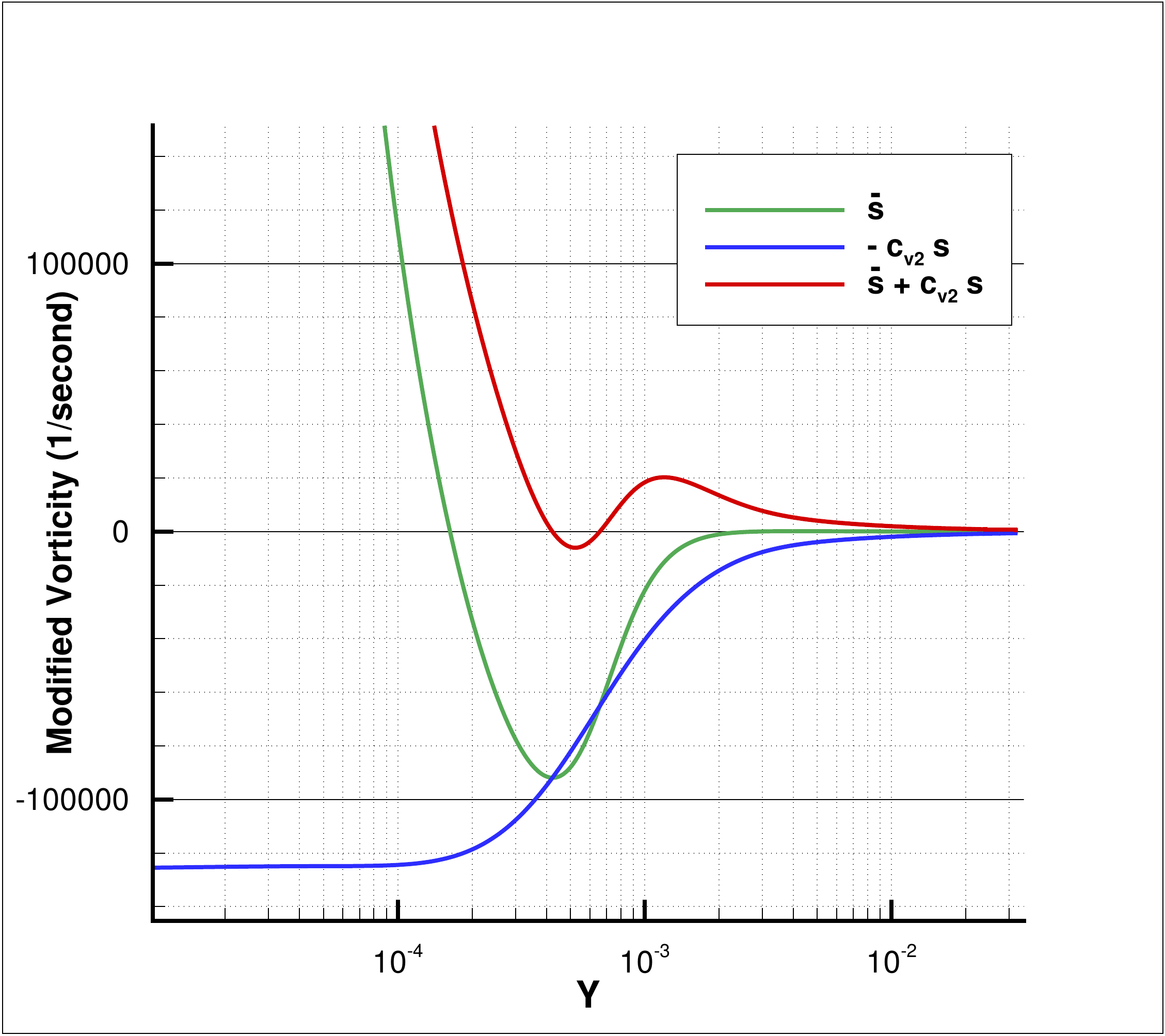}
\caption{Distributions of $\overline s$, $-c_{v2}s$ and $\overline s+c_{v2}s$ for MS-4 along $x=0.525$}
\label{fig:MS-4_Sbar}
\end{figure}

\subsubsection{Grid sensitivity of wall-bounded turbulent flows}
\label{sec:grisenswall}

The effects of mesh stretching and the height of the first element at the wall on the performance of high-order codes are studied here via MS-4 and by considering three grid sets, labelled A, B and C, each of which is generated by sequentially merging couples of adjacent element rows and columns of the respective fine grid, resulting in 10 mesh levels: L0 to L9 in the order of increasing refinement. The fine grids of all sets share the same number of elements, $1024 \times 3072$, along with a uniform element size distribution in the $x$ direction. As for the $y$ direction, geometric series  expansion ratios of $r_y^\mathrm{A}=1.00275822561$, $r_y^\mathrm{B}=1.0016969186$ and $r_y^\mathrm{C}=1.00111787853$ are utilized to establish mesh size distributions that cluster the elements at the wall. These ratios are furthermore chosen such that the normalized heights, $y^+ =\frac{y\,\sqrt{\tau_\mathrm{wall}}/\rho_\mathrm{wall}}{\nu_\mathrm{wall}}$, of the first element at the wall and at $x=0.525$ of the mesh L9 of the sets B and C are respectively 16 and 64 times that of the set A. The progression of exact $y^+$ values versus mesh refinement is compared in Fig. \ref{fig:MS-4_gridsets} for the three grid sets.

\begin{figure}[!hbt]
\centering
\subfloat[$y^+$]{
\includegraphics[trim = 2mm 4mm 18mm 10mm, clip,width=0.33\linewidth]
{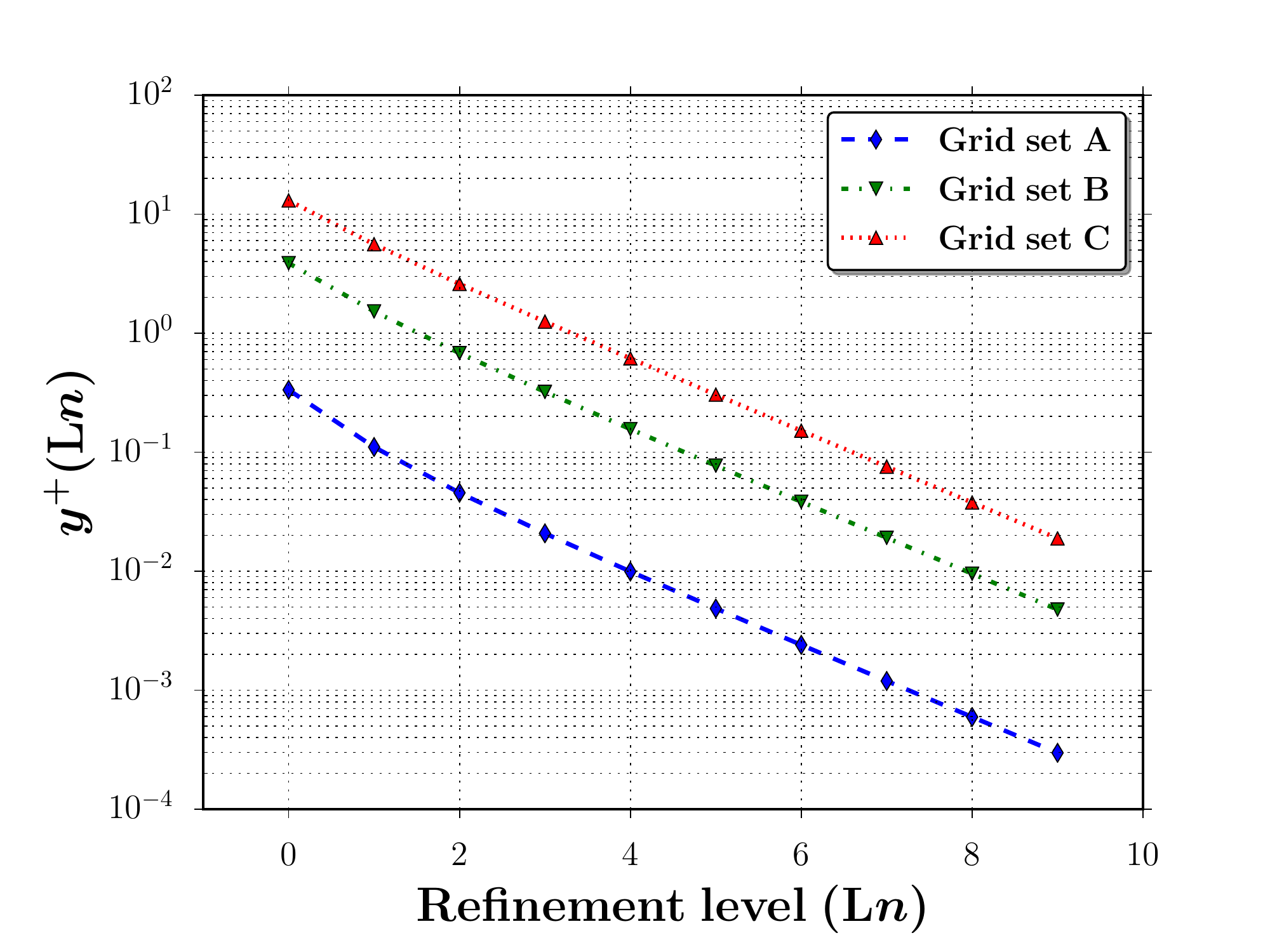}}~~~
\subfloat[$y^+$ refinement rate]{
\includegraphics[trim = 2mm 4mm 18mm 10mm, clip,width=0.33\linewidth]
{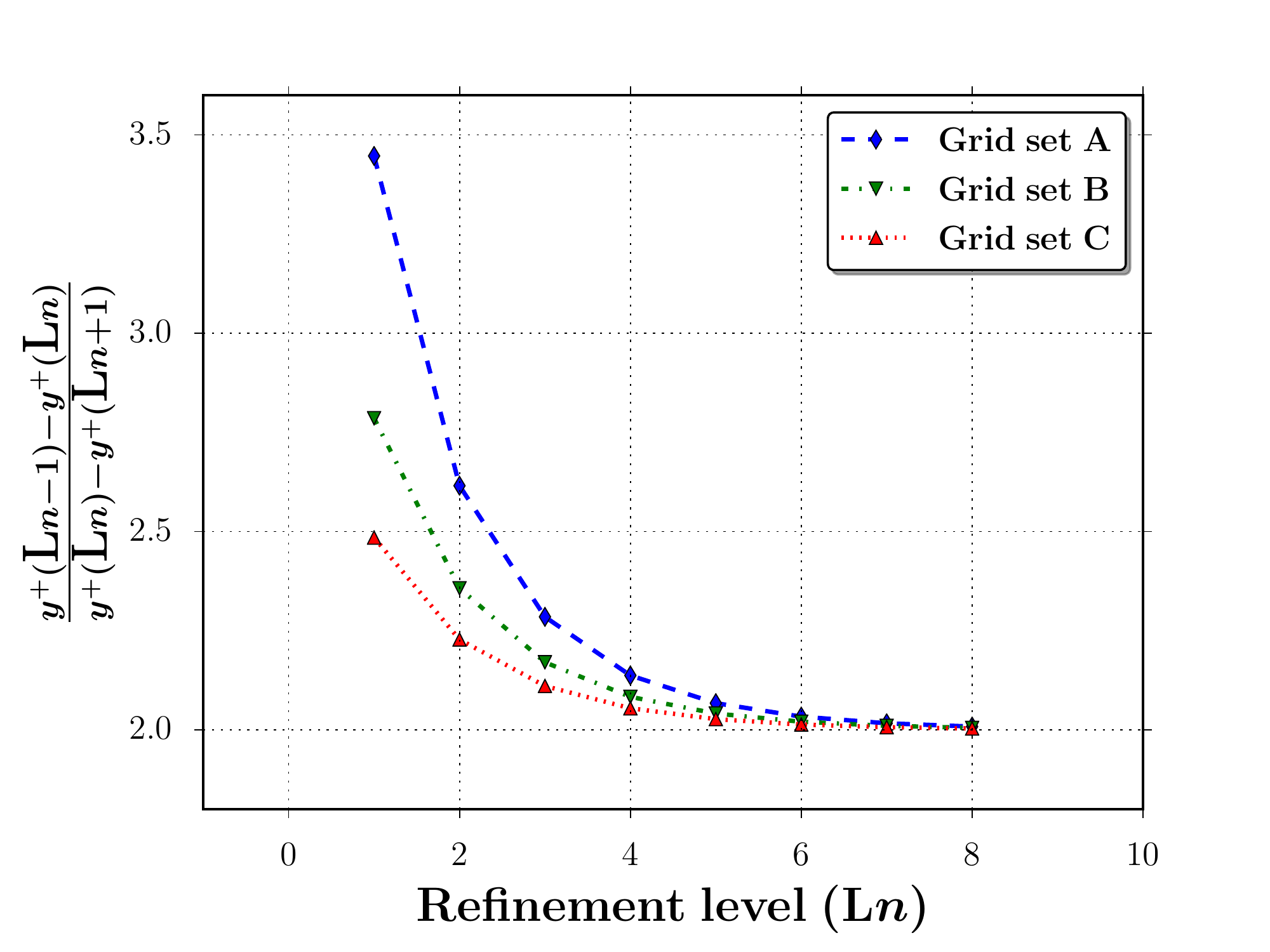}}
\caption{Comparison of exact $y^+$ and its rate of refinement based on the $1^{st}$ element height at the wall at $x=0.525$ for grid sets A, B and C of MS-4}
\label{fig:MS-4_gridsets}
\end{figure}

The results of MS-4 in terms of discretization errors and OOAs versus mesh refinement for the grid set B are presented in Figs. \ref{fig:Err_allE_allP_MS-4} and \ref{fig:Orders_MS-4} for the $L$ norms and in Figs. \ref{fig:Err_allE_allP_H_MS-4} and \ref{fig:Orders_H_MS-4} for the $H$ norms. In general, the error levels drop consistently with both mesh and polynomial refinements as expected. The rates of convergence are nevertheless not perfectly monotonic and minor discrepancies with the theoretical values are observed for some variables and norms especially for high-order discretizations. Similar results were obtained in the laminar mode, i.e., by deactivating the SA field of MS-4 and considering $\mu_t=0$; thus suggesting that the observed discrepancies are not caused in this case by the additional complexity of the SA equation.  Compared to other state variables, the OOAs of $\rho\tilde \nu$ are indeed very satisfactory in $L$ norms, while exhibiting a rather slow but steady convergence in $H$ norms. Amongst all variables, the results of $\rho v$ are the most irregular with regards to the expected OOAs. 

At this point, it can be concluded that the high-order code can not be considered verified solely based on the results of MS-4 on grid set B, since the obtained orders are not sufficiently and steadily close to expected values for all variables and all norms. However, by taking into account the evidence accumulated by verification via manufactured cases of \cite{Navah2017a} as well as by  MS-1 through MS-3 of the present study, it is very unlikely that the observed irregularities be caused by an implementation error. A possible explanation could rather be provided by the grid sensitivity study conducted previously \cite{Navah2017a} for the trigonometric manufactured cases where it is shown that the inadequate application of grid stretching to solutions exempt of sharp gradients delays the appearance of the asymptotic range. One such solution field is for example $\rho v$ of MS-4 featuring very mild gradients compared to those of $\rho u$ for which  the orders are closer to the expected values and exhibit a rather monotonic evolution. A comparison of the OOAs of the latter solution field in $L$ norms between the three grid sets is provided in Fig. \ref{fig:Orders_MS-4_sets} where it can be observed that the grid set A and B produce similar orders but that the grid set C yields a delayed asymptotic range except for the P1 discretization, which thus proves to be the least sensitive to grid stretching amongst all considered polynomial degrees in this case. 
%

\begin{figure}[!hbt]
\centering
\subfloat[Grid set A]{
\includegraphics[trim = 16mm 3mm 18mm 13mm, clip,width=0.33\linewidth]{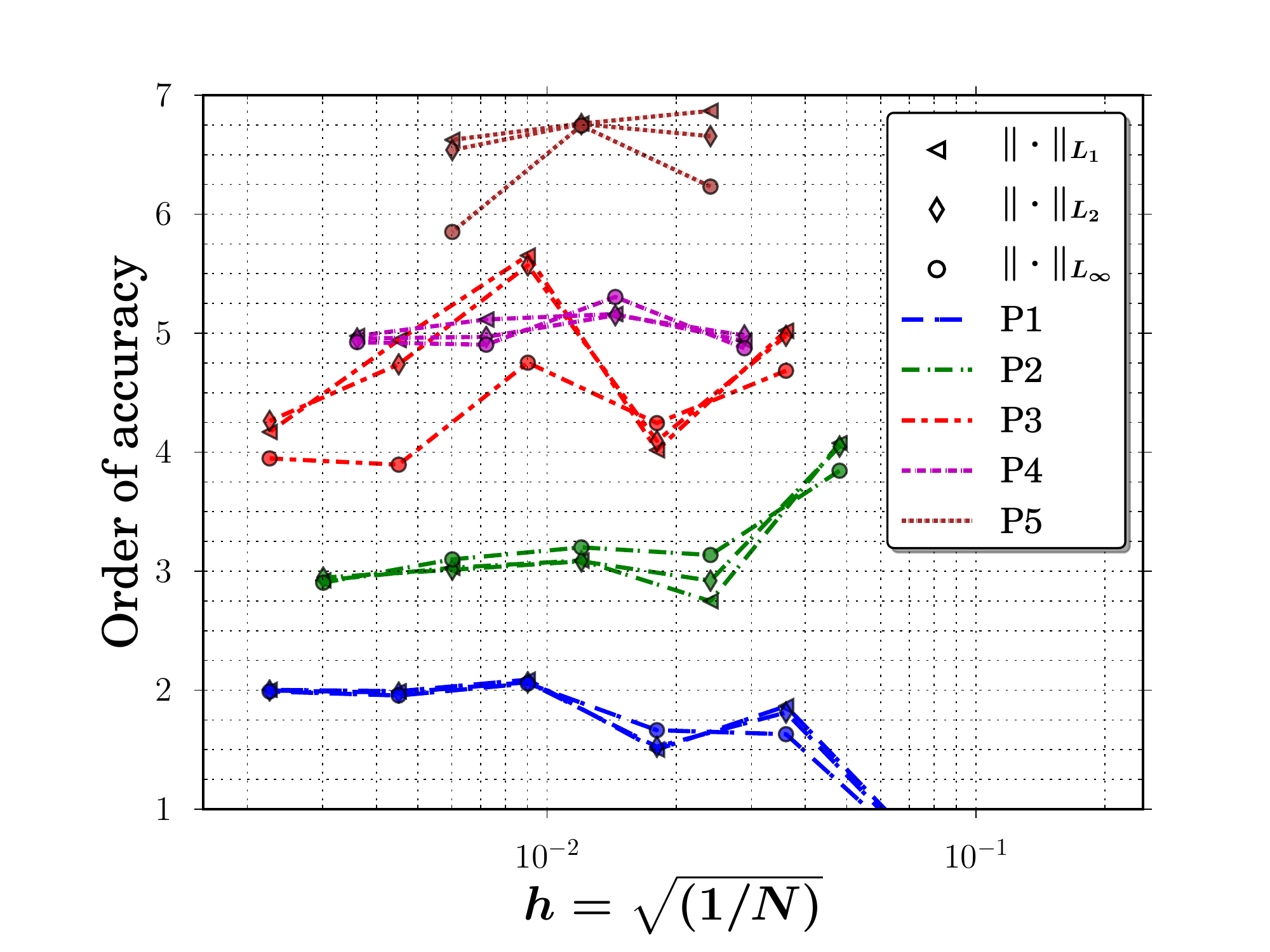}}~~~
\subfloat[Grid set B]{
\includegraphics[trim = 16mm 3mm 18mm 13mm, clip,width=0.33\linewidth]{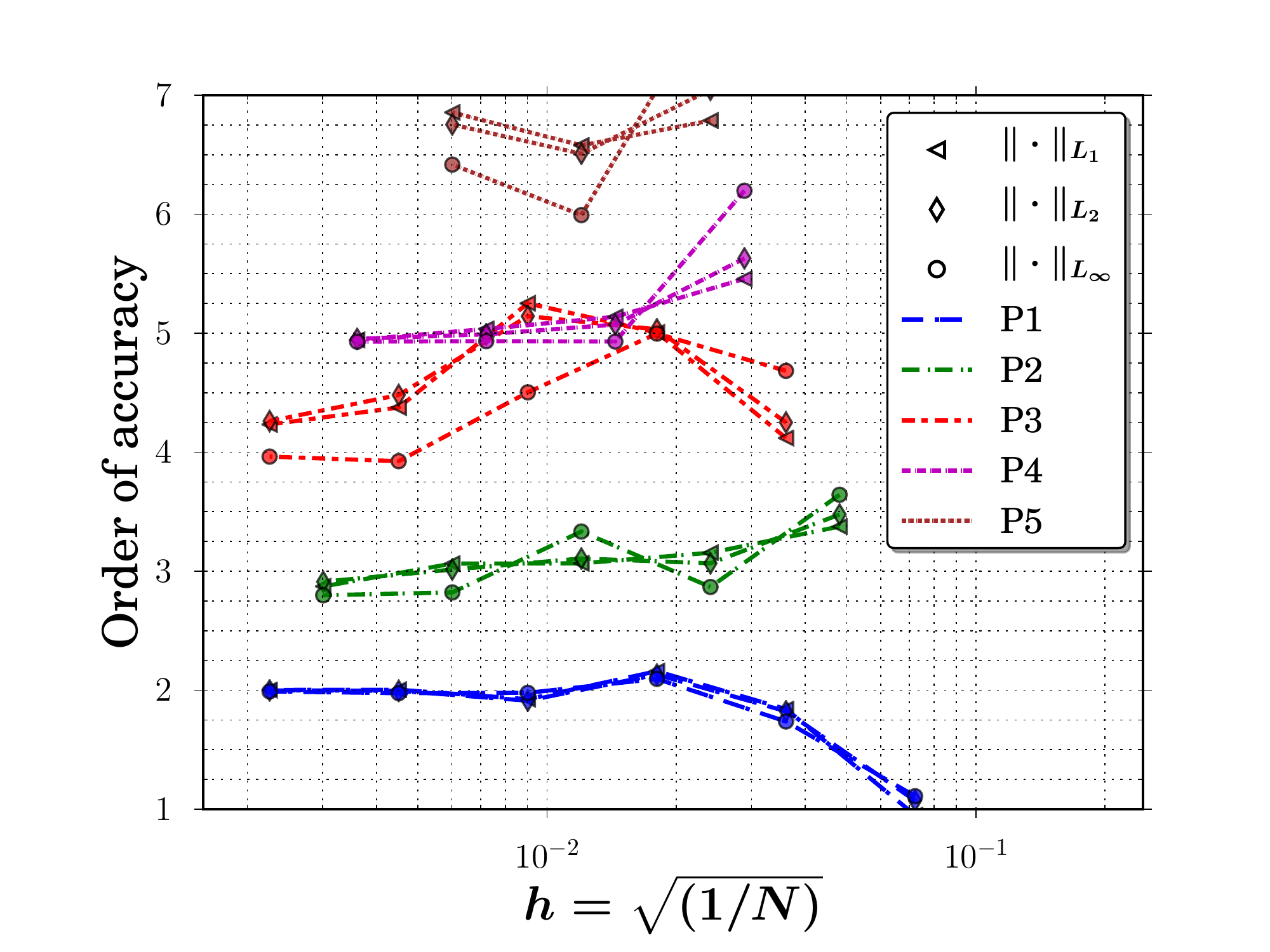}}
\vfill
\subfloat[Grid set C]{
\includegraphics[trim = 16mm 3mm 18mm 13mm, clip,width=0.33\linewidth]{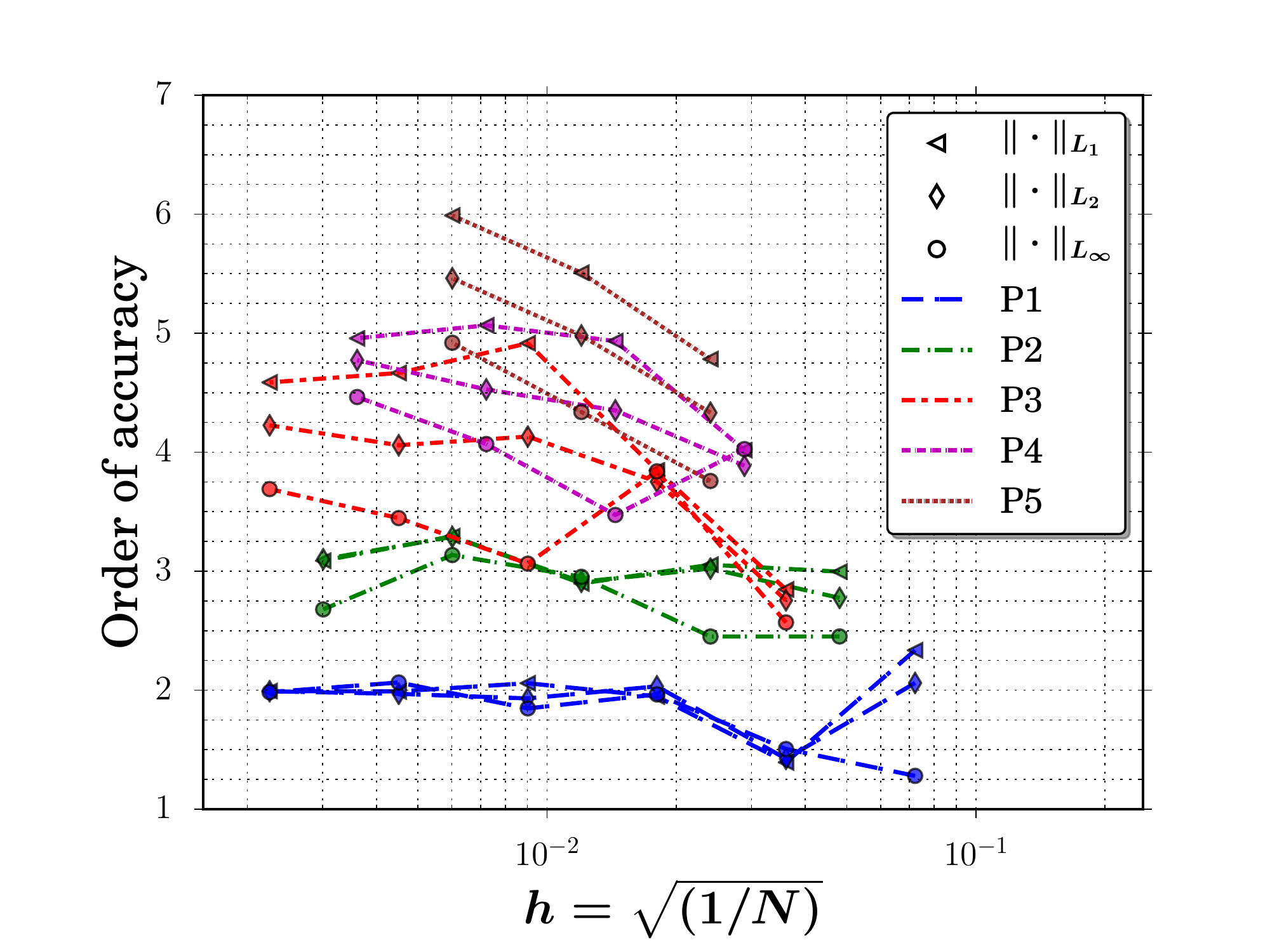}}
\caption{Evolution of the OOAs in $L_1$, $L_2$ and $L_\infty$ norms versus mesh refinement for $\rho u$ of dimensional MS-4 and polynomial degrees $\mathrm{P}1$--$\mathrm{P}5$ on grid sets A, B and C}
\label{fig:Orders_MS-4_sets}
\end{figure}

\begin{figure}[!hbt]
\centering
\subfloat[Grid set A]{
\includegraphics[trim = 5mm 4mm 4mm 10mm, clip,width=0.36\linewidth]{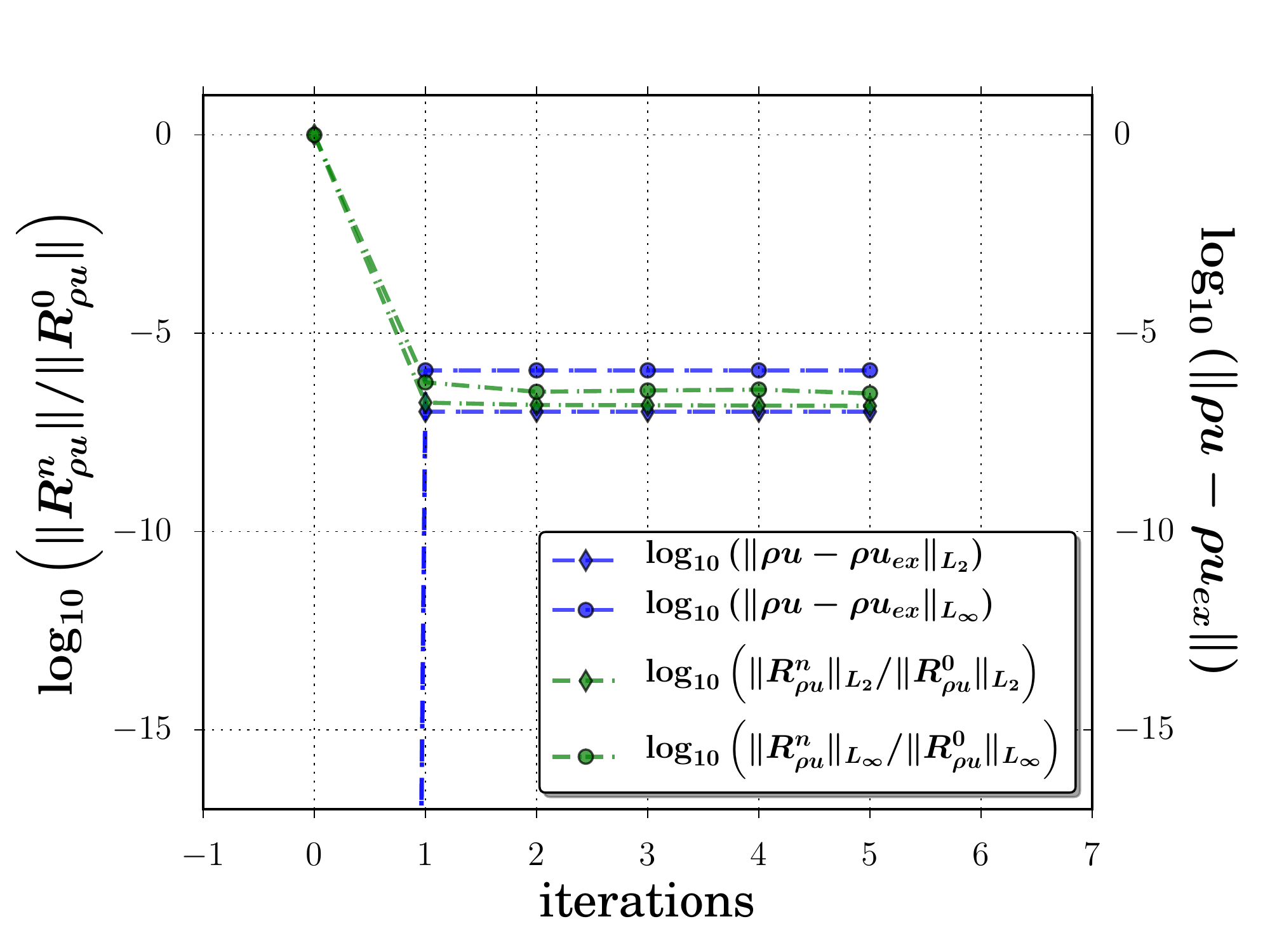}}
~~~
\subfloat[Grid set B]{
\includegraphics[trim = 5mm 4mm 4mm 10mm, clip,width=0.36\linewidth]{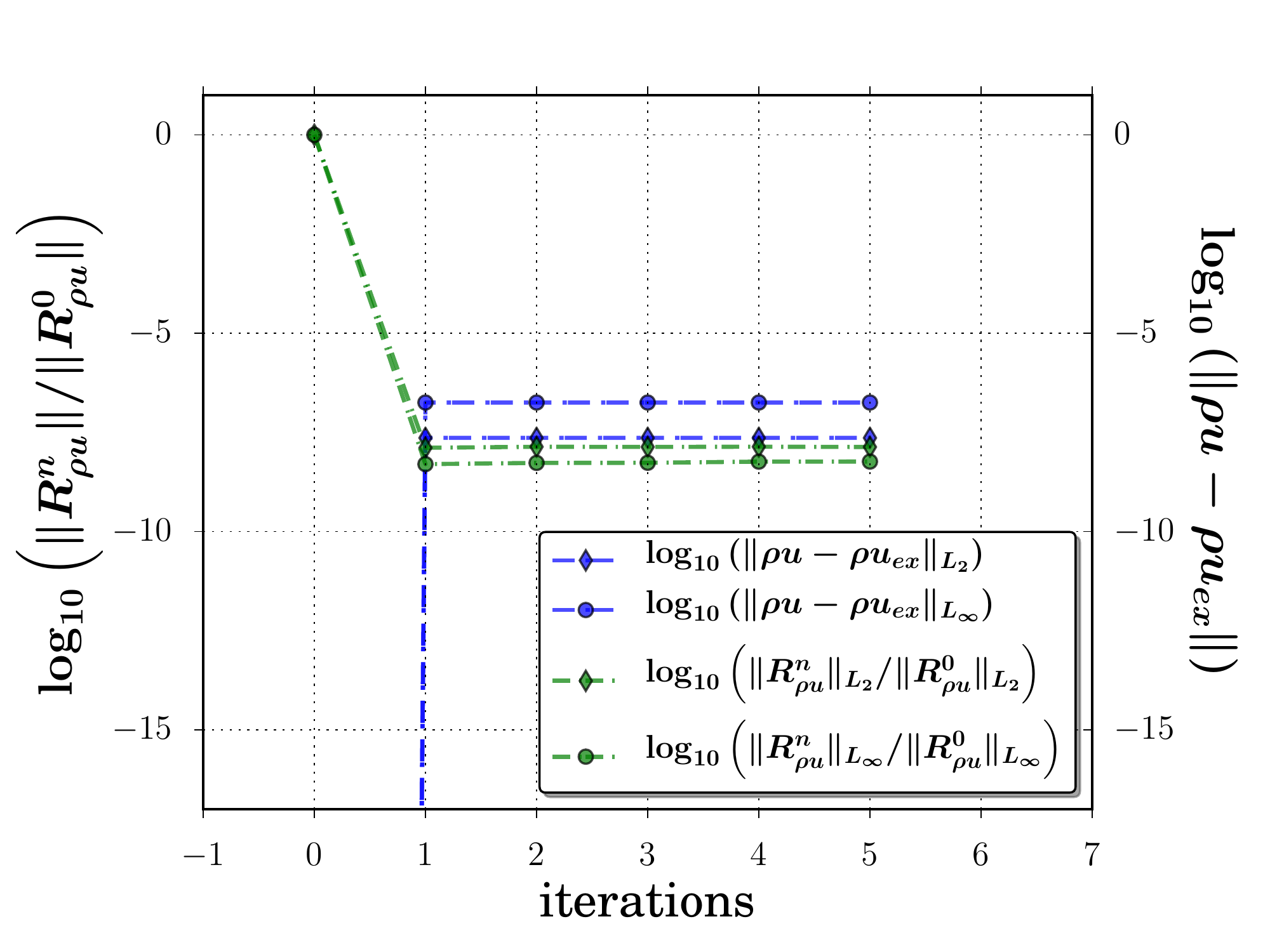}}
\vfill
\subfloat[Grid set C]{
\includegraphics[trim = 5mm 4mm 4mm 10mm, clip,width=0.36\linewidth]{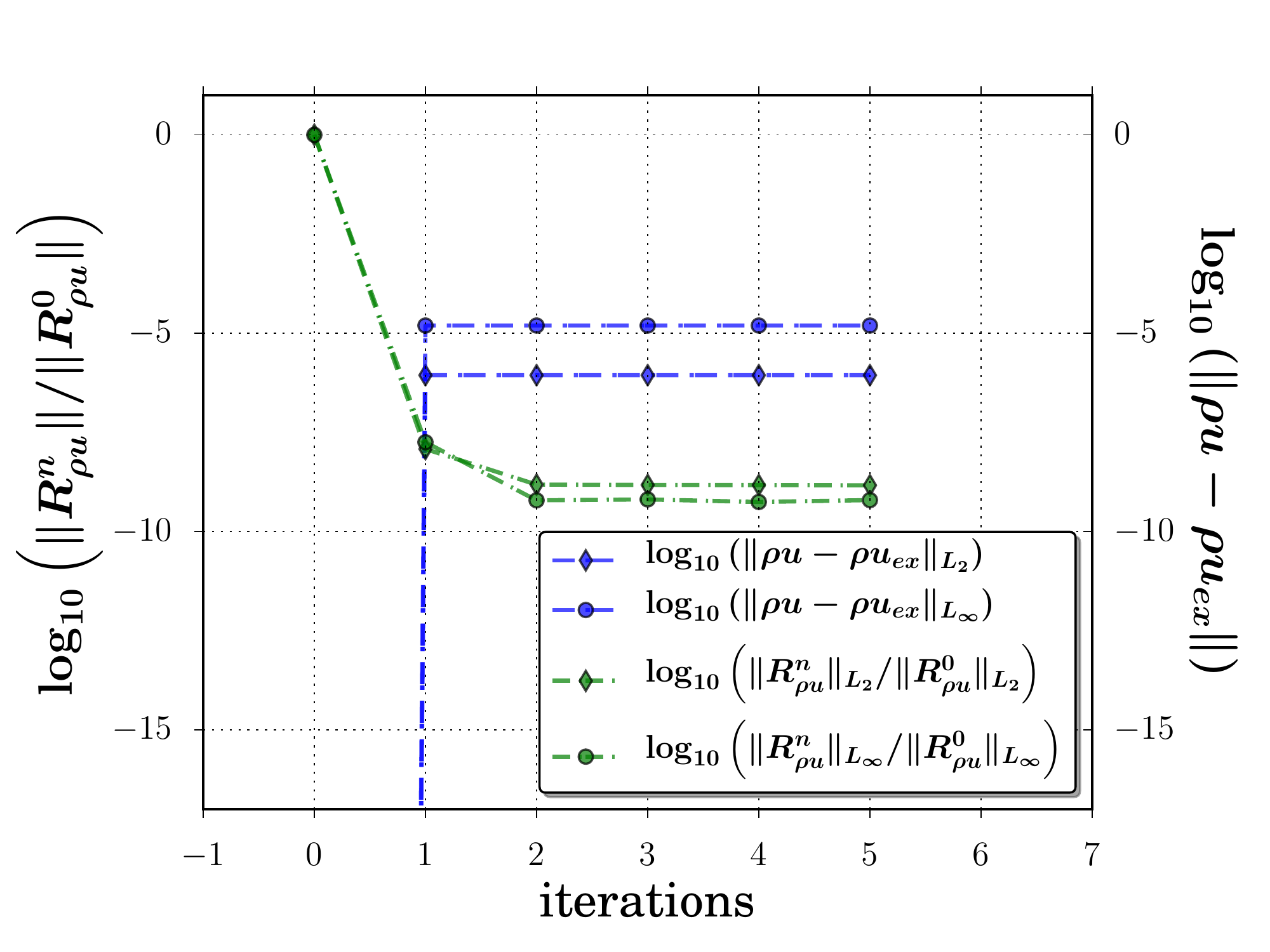}}
\caption{Residual convergence and discretization error of $\rho u$ in $L_2$ and $L_\infty$ norms versus number of Newton iterations for the dimensional MS-4 with polynomial degree $\mathrm{P}3$ and refinement level L4 of grid sets A, B and C}
\label{fig:Orders_MS-4_sets_conv}
\end{figure}

The grid clustering at the wall decreases from the grid sets A to B to C as their corresponding geometric expansion ratios approach the value of unity in the same order. This is also reflected in the increase of their corresponding exact $y^+$ values (see Fig. \ref{fig:MS-4_gridsets}). This clustering affects the level of achievable numerical convergence as well as how the discretization error is tackled by the grid. These two effects are reflected in the plots of Fig. \ref{fig:Orders_MS-4_sets_conv} comparing the evolution of the relative residual and that of the discretization error of $\rho u$ versus the number of Newton iterations for P3 and grid level L4. These plots show that the grid clustering limits the residual minimization as the final residuals in both $L_2$ and $L_\infty$ norms decrease as the grid distribution approaches uniformity in the $y$ direction. The discretization error norms at convergence on the other hand, decrease from grid set A to B to re-increase from B to C, signifying that the grid distribution B approximates the optimal element size distribution in the $y$ direction the best. In sum, these observations suggest that the numerical stiffness increases with mesh clustering at the wall but that the clustering intensification does not reduce the overall discretization error after a certain point.

The analysis of the grid stretching effect can be further enriched by considering the computation of output functionals relevant to engineering design such  as the drag coefficient, $C_d$. The verification of the latter is enabled by the choice of a physically realistic wall-bounded manufactured solution that provides the exact value of the drag coefficient, which for MS-4 is
\begin{equation}
C_d^{\mathrm{ex}} =\frac{\int_\mathrm{wall} \,\mu_\mathrm \,\left(\frac{\partial u}{\partial y} \right)^{\mathrm{MS}} \,dx}{p_\mathrm{dyn}\int_\mathrm{wall} \,dx}=3.6013213414944\times 10^{-03},
\end{equation}
where $p_\mathrm{dyn}=0.5\,p_0$ is the reference dynamic pressure. It thus is possible to evaluate the exact $C_d$ error as $\mathcal{E}^{\mathrm{ex}}_{C_d}=|C_d - C_d^\mathrm{ex} |$. 

\begin{figure}[!hbt]
\centering
\subfloat[Grid set A]{
\includegraphics[trim = 2mm 1mm 18mm 10mm, clip,width=0.33\linewidth]
{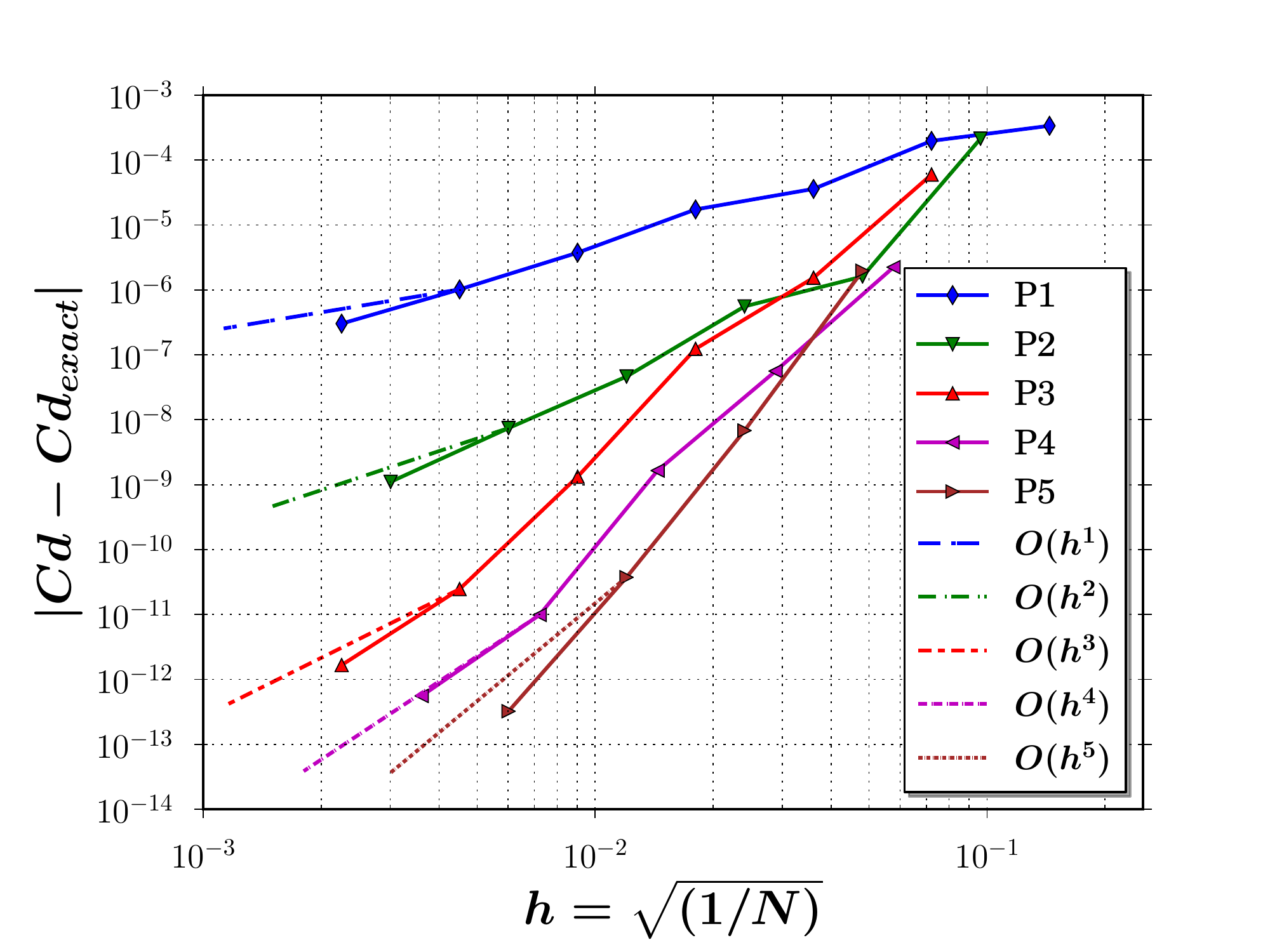}}
~~~
\subfloat[Grid set B]{
\includegraphics[trim = 2mm 1mm 18mm 10mm, clip,width=0.33\linewidth]
{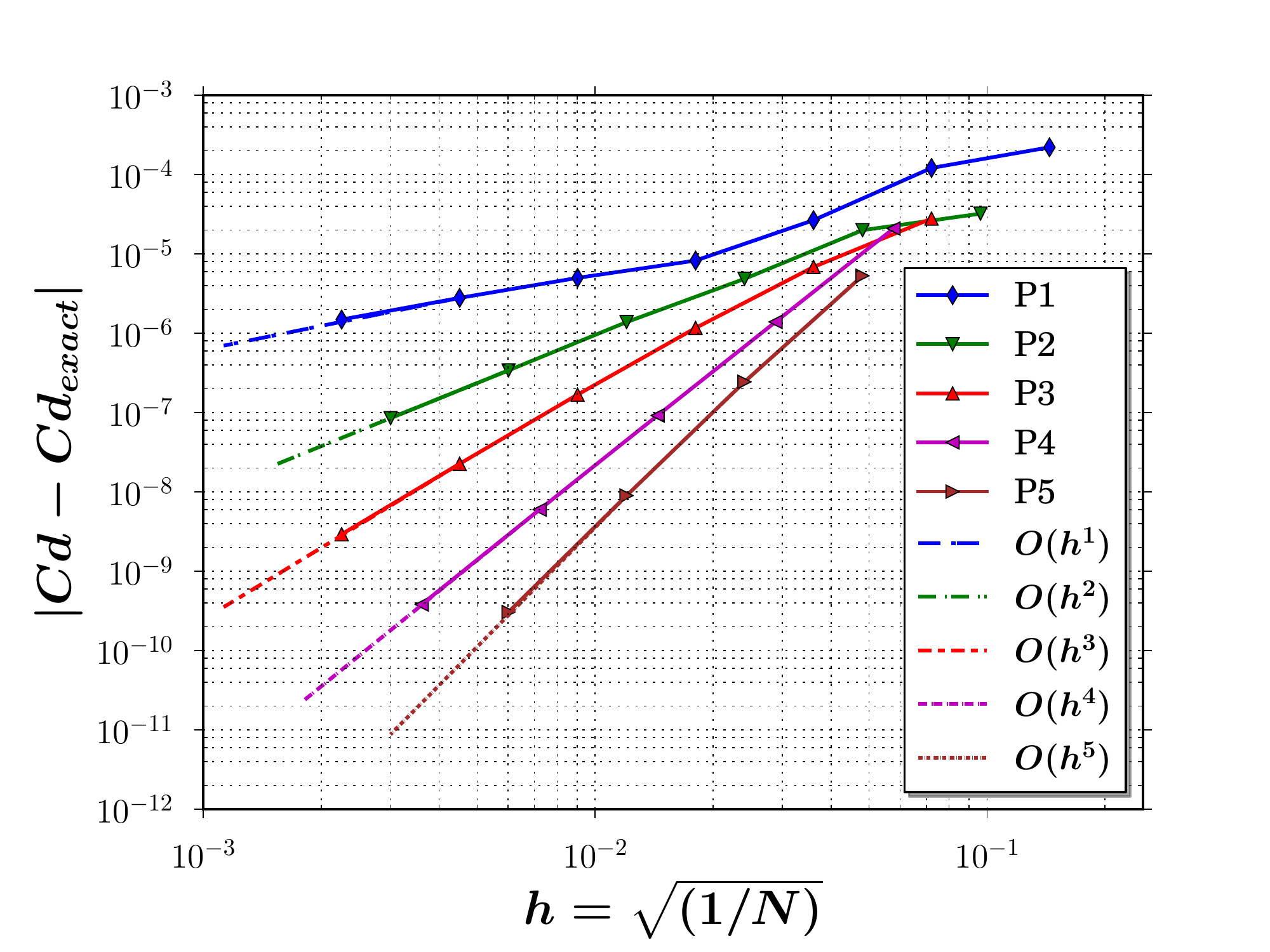}}
\vfill
\subfloat[Grid set C]{
\includegraphics[trim = 2mm 1mm 18mm 10mm, clip,width=0.33\linewidth]
{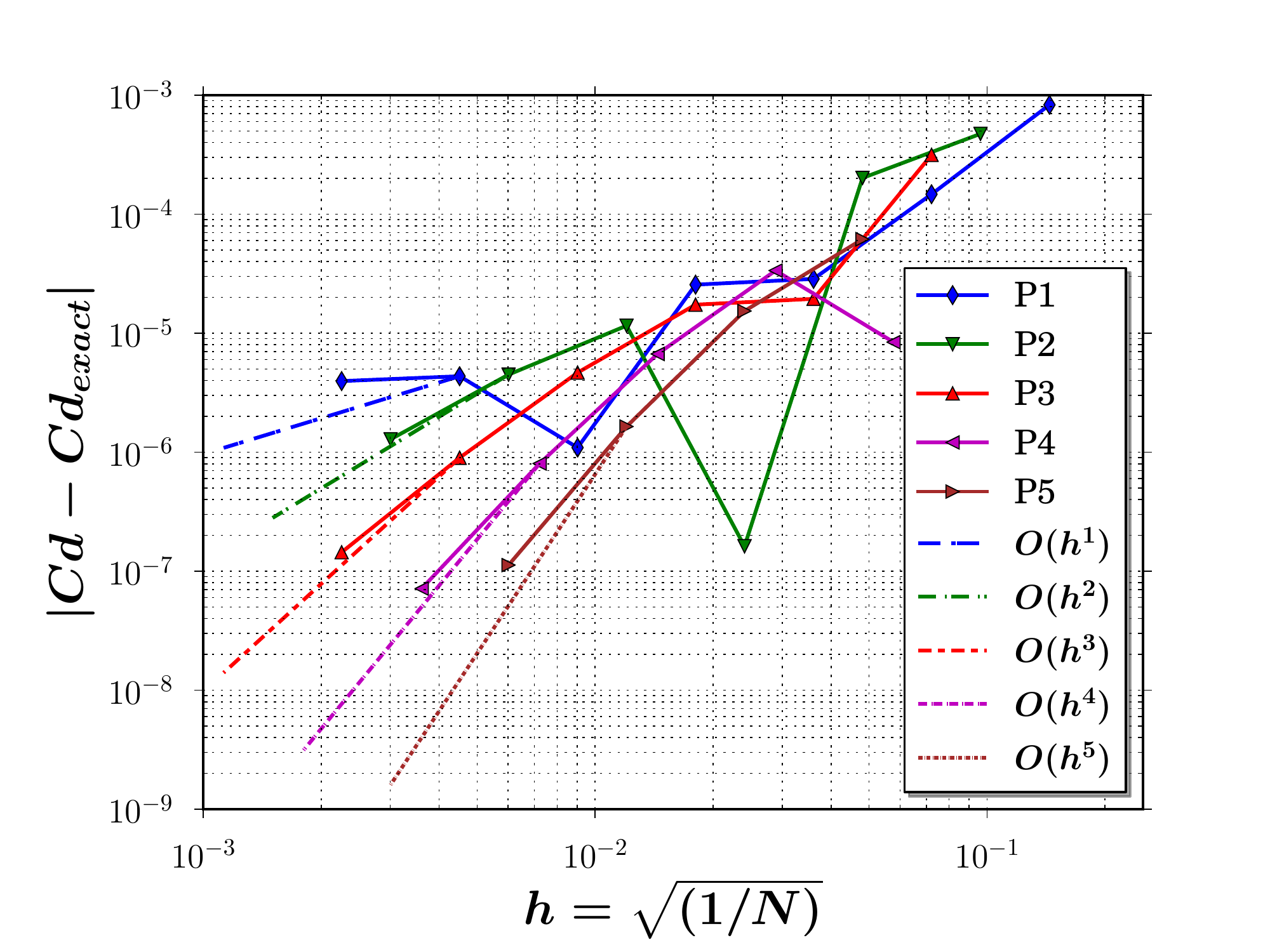}}
\caption{Evolution of the exact drag coefficient errors versus mesh refinement for MS-4 and polynomial degrees $\mathrm{P}1$--$\mathrm{P}5$ on grid sets A, B and C}
\label{fig:Ms-CD_sets}
\end{figure}

The evolution of $\mathcal{E}^{\mathrm{ex}}_{C_d}$ versus mesh refinement is plotted in Fig. \ref{fig:Ms-CD_sets} for all grid sets. We expect an OOAs of $\mathcal{O}(h^\mathrm{P})$ for $C_d$ solutions in agreement with the values reported for high-order DG solutions of a flat-plate case \cite{Nguyen-et-al_2007,Bassi-et-al_2014,Crivellini-et-al_2013a} and high-order CPR solutions of an airfoil case \cite{Shi2015}. Amongst the three considered grid sets, only the grid set B achieves the expected OOAs along with a monotonic convergence for all P, whereas the grid sets A and C respectively produce non-monotonic super- and sub-convergences. The super-convergence of the grid set A is explained by its relatively large value of geometric series expansion ratio that is related to how fast the element size grows in the $y$ direction. For a higher growth rate, an element resulting from merging two consecutive rows of elements at the wall inherits a cumulative off-wall size larger than what the mere doubling of the size of the smaller of the two parents would have produced. This causes a refinement ratio of the first element height at the wall  that is larger than the desired value of 2.0, expected from doubling the number of elements in each spatial direction. This effect is depicted in Fig. \ref{fig:MS-4_gridsets} (b) that shows the refinement ratio of $y^+$ of the first element height at the wall, demonstrating that amongst all, the grid set A has the highest refinement ratio on the coarsest grids. This highest refinement rate explains the apparent super-convergence of the drag coefficients produced by the grid set A. Let's note that although the expected convergence rates in $C_d$ are not achieved by the grid set A, it provides the most accurate values of this output functional compared to other grid sets. 
As for the grid set C, the sub-convergence is due to the relative coarseness of its grids, in the vicinity of the wall, which fail to attain the asymptotic range for the considered values of $h$, although grids L2 through L9 of this set feature $y^+<5$ (see Fig. \ref{fig:MS-4_gridsets}). The latter value stems from the common standards of engineering grid generation practices. These standards hence seem insufficient for achieving the expected convergence rate of aerodynamics quantities. The sub-convergence of $C_d$ by the grid set C also corroborates with the observed non-asymptotic convergence in the $\rho u$ errors for this grid set (Fig. \ref{fig:Orders_MS-4_sets}).

Based on these observations, it can be concluded that the verification of the drag coefficient requires grids with smaller $y^+$ values than the engineering standard of $y^+<5$, necessary to correctly capture the velocity profile for low-order solvers. In the case of high-order methods, the difference is even larger since it is shown \cite{Drosson-et-al_2013} that a value of $y^+=64$ produces an acceptable velocity profile in a P4 DG solution.

\subsubsection{Solution verification versus code verification}
Based on the observations accumulated thus far, it can be concluded that monotonic convergence of aerodynamic output functionals such as the drag coefficient of realistic (and by extension real) wall-bounded flows is conditional to the adequacy of the grid set in terms of specific combinations of first wall element $y^+$ and its refinement ratios, fulfilled here for example, by the grid set B only. Hence, employing an inappropriate grid set could lead to discrepancies between expected and  observed OOAs of such output functionals. It therefore is prudent to avoid fully replacing code verification by solution verification of these outputs via error estimation techniques such as Richardson extrapolation \cite{Roy2005}. Solution verification should instead be adopted jointly and complementarily to code verification, via analytical/manufactured solutions for which, the expected OOAs for flow variables are produced by the solver under examination. 

Another reason to avoid substituting solution verification to code verification is that output functionals are often integral values on a select number of boundaries and hence do not account for the discretization error distribution on the entirety of the domain. As a result, the presence of local inconsistencies could go undetected if they do not directly affect the considered output. To illustrate this, we have introduced a spurious alteration in the value of the Riemann BC at the top of the domain by considering $({\rho u})^\mathrm{BC} = (1 + d\alpha)\,({\rho u})^{\mathrm{MS}}\vert_{\Gamma_\mathrm{top}}$ where $d\alpha$ is modified from its original null value to $d\alpha = 1 \times 10^{-8}$. This results in the solution contamination in the superior region of the domain as shown by the distribution of the discretization error in $\rho u$ in Fig. \ref{fig:Ms-Q1_topBC} where the transport of the error from the top boundary into the domain is depicted. The error due to the incorrect BC persists throughout mesh refinement, thus affecting the OOAs in $L_1$, $L_2$ and $L_\infty$ norms for P3 through P5 in Fig. \ref{fig:Ms-Q1ordervsCD_topBC} where on the other hand, the drag coefficient is shown to fail in detecting the presence of the erroneous boundary condition for all polynomial degrees. 

Finally, in practical cases where the exact value of the output functional is not known, the verification would rely on estimation techniques such as the Richardson extrapolation which is sensitive to deviations from convergence  monotonicity as it relies on three consecutive solutions in the asymptotic range in order to estimate the zero-mesh-size solution, in contrast to two solutions required to compute the OOAs in the present manufactured case for example. This sensitivity constitutes another argument in avoiding the replacement of proper code verification by solution verification of output functionals in turbulent RANS-modelled flows, the most simplest of which often proves challenging in terms of generating grid sequences that deliver asymptotic grid convergence \cite{pelletier_2008}.

\begin{figure}[!hbt]
\centering
\includegraphics[trim = 5mm 1mm 3mm 7mm, clip,width=0.45\linewidth]
{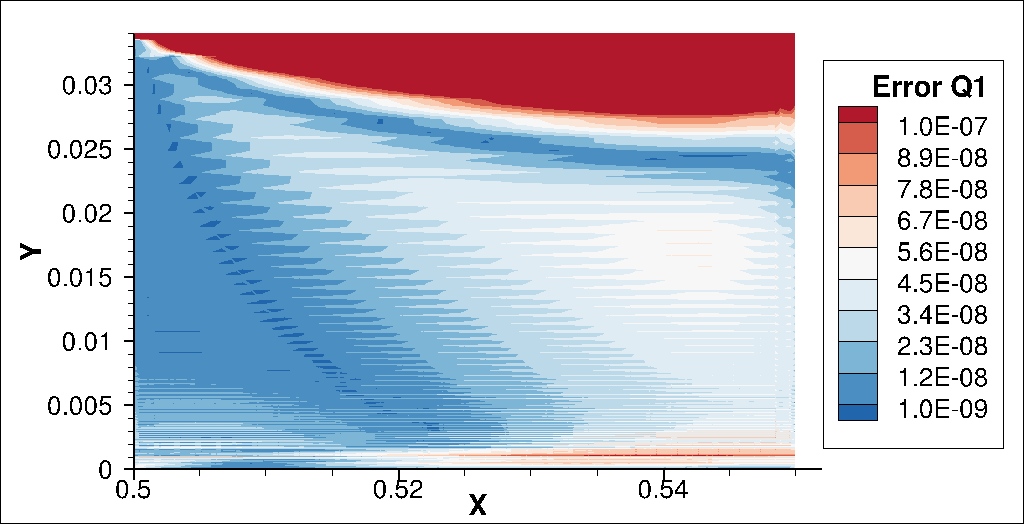}
\caption{Spatial distribution of the discretization error in $\rho u$ under the effect of a spurious boundary condition at the top of the domain}
\label{fig:Ms-Q1_topBC}
\end{figure}

\begin{figure}[!hbt]
\centering
\subfloat[$\rho u$ OOAs in $L$ norms]{
\includegraphics[trim = 2mm 1mm 18mm 10mm, clip,width=0.33\linewidth]
{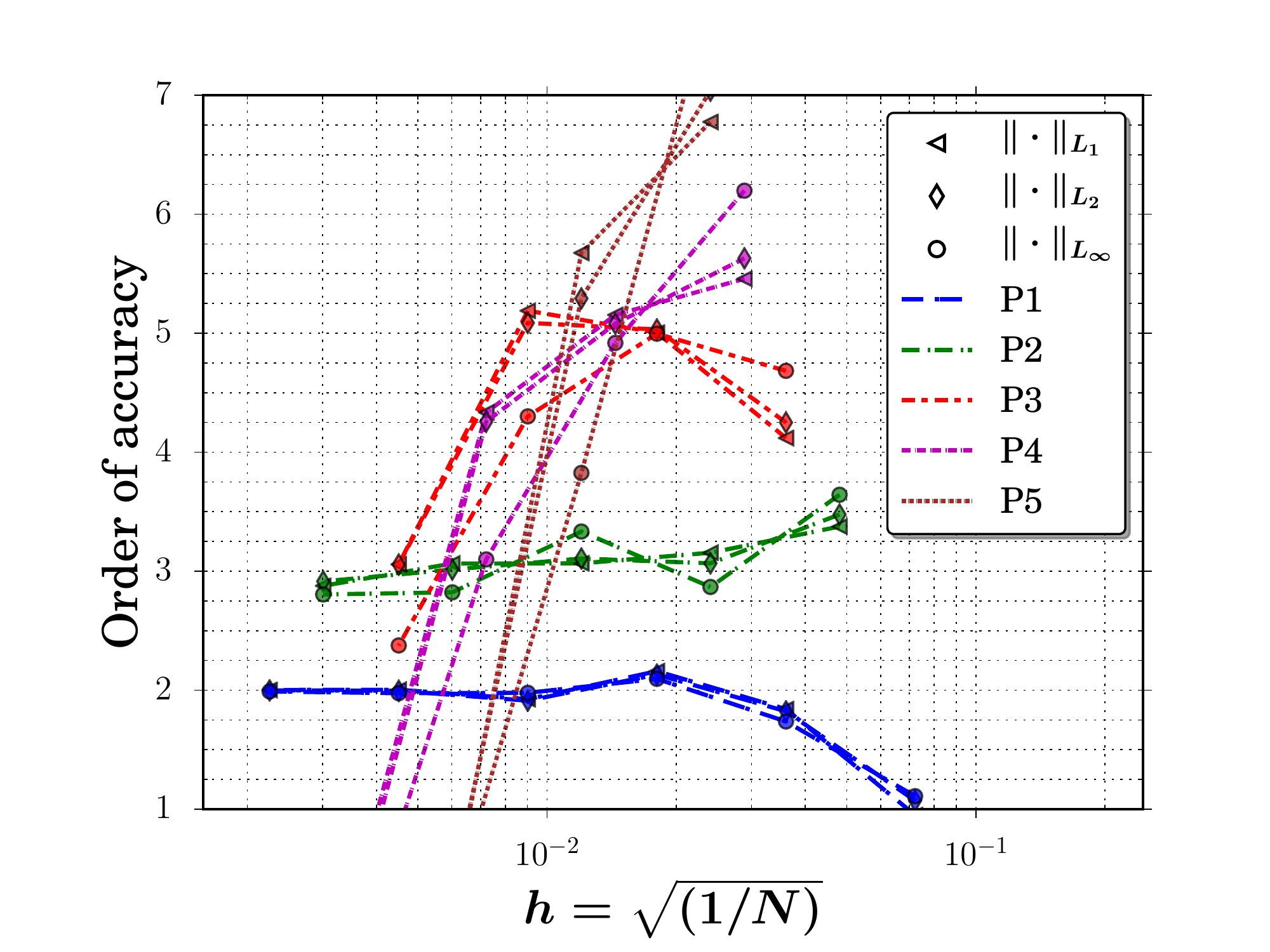}}~~~
\subfloat[$C_d$ error]{
\includegraphics[trim = 2mm 1mm 18mm 10mm, clip,width=0.33\linewidth]
{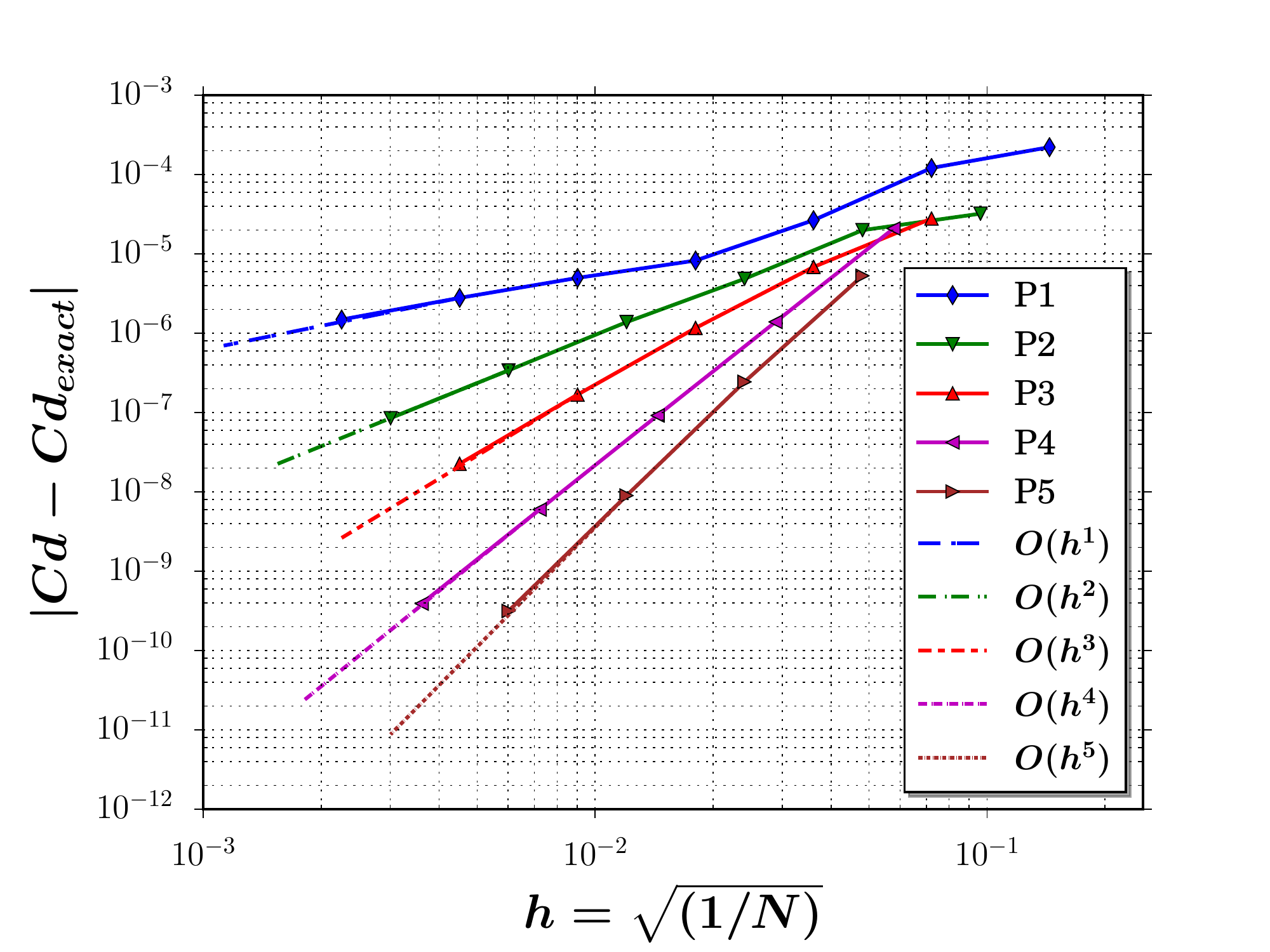}}
\caption{Comparison of the performance of global versus local (wall output) error metrics in detecting localized implementation inconsistencies for MS-4 and polynomial degrees $\mathrm{P}1$--$\mathrm{P}5$}
\label{fig:Ms-Q1ordervsCD_topBC}
\end{figure}

%
%

\subsubsection{Non-dimensionalization effect}
Implicit solution methods are particularly effective for tackling stiff problems. A representative example of this class of techniques is the Newton's method which requires the linearization of the discrete set of non-linear residual equations, 
\begin{equation}
\bm{\mathscr{R}}(\bm{Q})=\bm{0},
\label{eq:res_eqs}
\end{equation}
where  $\bm{Q},\bm{\mathscr{R}}\in{\rm I\!R}^{N_\mathrm{dof}}$ with $N_\mathrm{dof}$ being the total number of DOFs, by computing the Jacobian matrix, $\underline{\underline{\mathscr{A}}}(\bm{Q})$, defined as $\mathscr{A}_{ij} = \partial \mathscr{R}_i(\bm{Q})/\partial Q_j$ with $i,j \in [1\,..\,N_\mathrm{dof}]$, in order to solve an algebraic system of equations,
\begin{equation}
\underline{\underline{\mathscr{A}}}(\bm{Q}^n)\, \bm{\Delta Q}^n=-\bm{\mathscr{R}}(\bm{Q}^n),
\label{eq:alg_sys}
\end{equation}
such that the solution of Eq. \eqref{eq:res_eqs} is iteratively reached via successive solution updates of the form $\bm{Q}^{n+1} = \bm{Q}^n + \bm{\Delta Q}^n$ until $\bm{\mathscr{R}}(\bm{Q}^{n+1})$ is sufficiently close to zero.
In practice, the Jacobian matrix is often approximated by $\underline{\underline{\mathscr{A}}}^\prime=(\underline{\underline{\mathscr{A}}} + \underline{\underline{d \mathscr{A}}})$  and the right hand side of Eq. \eqref{eq:alg_sys} by $\bm{\mathscr{R}}^\prime=\bm{\mathscr{R}} + \bm{d\mathscr{R}}$. The discrepancies, $d \mathscr{A}_{ij}$, could be minor if only due to round-off errors as in the case of an \textit{exact linearization} or they could be of considerable amplitude as in the case of an \textit{inexact linearization} where the Jacobian is computed by an approximative approach such as finite differences. In all cases, a high condition number amplifies the propagation of $d \mathscr{A}_{ij}$ and $d \mathscr{R}_i$ into $\Delta Q_j$. The proof of this is an intermediary result of the Theorem 1.2.3 of \cite{Bjorck2015} which we present here for the sake of completeness.  Let's consider the algebraic system 
\begin{equation}
\underline{\underline{\mathscr{A}}}^\prime\,\bm{\Delta {Q}}'=-\bm{\mathscr{R}}^\prime,
\label{eq:alg_sys2}
\end{equation}
where $\bm{\Delta {Q}}' = \bm{\Delta {Q}} + \bm{d \Delta {Q}}$ is an approximation of the solution of Eq. \eqref{eq:alg_sys}. Equation \eqref{eq:alg_sys2} can be expressed as
\begin{equation}
(\underline{\underline{\mathscr{A}}} + \underline{\underline{d \mathscr{A}}})\,(\bm{\Delta {Q}} + \bm{d\Delta {Q}})=-(\bm{\mathscr{R}} + \bm{d\mathscr{R}}).
\label{eq:alg_sys3}
\end{equation}
Subtracting Eq. \eqref{eq:alg_sys} from \eqref{eq:alg_sys3}, assuming det$(\underline{\underline{\mathscr{A}}})\neq0$ and a few algebraic operations yield
\begin{equation}
\bm{d\Delta {Q}}=\underline{\underline{\mathscr{A}}}^{-1}(-\underline{\underline{d \mathscr{A}}}\,\bm{\Delta {Q}}' - \bm{d\mathscr{R}}).
\label{eq:alg_sys4}
\end{equation}
Furthermore, by considering the inequality $\| \underline{\underline{\mathscr{X}}} \, \underline{\underline{\mathscr{Y}}}  \| \leq \| \underline{\underline{\mathscr{X}}}\|\,\| \underline{\underline{\mathscr{Y}}}\|$ in a proper norm, we obtain
\begin{equation}
\frac{\| \bm{d\Delta {Q}}\|}{\|\bm{\Delta {Q'}}\|} \leq \|\underline{\underline{\mathscr{A}}}^{-1}\|\,\|\underline{\underline{\mathscr{A}}}\| \left( \frac{\|\underline{\underline{d \mathscr{A}}}\|}{\|\underline{\underline{\mathscr{A}}}\|} + \frac{\|\bm{d\mathscr{R}}\|}{\|\underline{\underline{\mathscr{A}}}\|\, \|\bm{\Delta {Q}}'\|}  \right),
\label{eq:alg_sys5}
\end{equation}
where $\|\underline{\underline{\mathscr{A}}}^{-1}\|\,\|\underline{\underline{\mathscr{A}}}\|$ is the condition number of $\underline{\underline{\mathscr{A}}}$, thus demonstrating that the sensitivity of the solution to errors in residual equation and the Jacobian matrix can be attenuated if conditioning is improved (condition number reduced closer to unity). By propagating unwanted perturbations into the solution, ill-conditioning furthermore deteriorates the performance of iterative approaches such as the Newton's method or the GMRES algorithm respectively employed to solve the non-linear system \eqref{eq:res_eqs} and the linearized system \eqref{eq:alg_sys}. As a result, the iterative solution process could stall or even diverge, thus preventing a sufficient reduction in round-off and iterative errors via residual minimization \cite{Roy2005}. Since it is crucial to effectively minimize these sources of error for verification purposes, it hence is desirable to have well-conditioned systems in addition to employing an exact linearization.

\begin{proposition} 
A proper scaling of the solution acts as a preconditioning to the linear system \eqref{eq:alg_sys}. 
\end{proposition}

\begin{proof}
Let's consider
\begin{equation}
\breve{Q}_i=\frac{Q_{(i)}}{\mathring{Q}_{(i)}},
\label{eq:nondim_Q}
\end{equation}
where $\breve{Q}_i$ and ${\mathring{Q}_i}$ respectively denote a non-dimensional solution value and its corresponding dimensional reference  quantity and the parenthesized indices are not summed over. The change of variable \eqref{eq:nondim_Q} results in a new set of non-linear residual operators, $\bm{\mathscr{\breve R}}\in{\rm I\!R}^{N_\mathrm{dof}}$ with the following property:
\begin{equation}
\bm{\mathscr{\breve R}}(\bm{\breve Q})=\bm{\mathscr{R}}(\bm{Q})=\bm{0}.
\label{eq:res_eqs_nd}
\end{equation}
The corresponding linearized system is
\begin{equation}
\underline{\underline{\mathscr{\breve A}}}(\bm{\breve Q})\, \bm{\Delta \breve Q}=-\bm{\mathscr{\breve R}}(\bm{\breve Q}),
\label{eq:alg_sys_nd}
\end{equation}
where
\begin{equation}
\breve{\mathscr{A}}_{ij}= \frac{\partial \mathscr{\breve{R}}_{i}(\bm{\breve Q})}{\partial \breve{Q}_{(j)}}= \frac{\partial \mathscr{R}_{i}(\bm{Q})}{\partial Q_{(j)}}  \frac{\partial Q_{(j)}}{\partial \breve{Q}_{(j)}} =\frac{\partial \mathscr{{R}}_{i}(\bm{ Q})}{\partial {Q}_{(j)}}\,{\mathring{Q}_{(j)}}=\mathscr{A}_{i(j)}{\mathring{Q}_{(j)}} .
\label{eq:nondim_A}
\end{equation}
The equivalence between the systems \eqref{eq:alg_sys_nd} and \eqref{eq:alg_sys} is established through
\begin{equation}
\underline{\underline{\mathscr{A}}}(\bm{Q})\, \underline{\underline{\mathscr{P}}}^{-1} \,  \underline{\underline{\mathscr{P}}} \,\bm{\Delta Q}=-\bm{\mathscr{R}}(\bm{Q}),
\label{eq:alg_sys_prec}
\end{equation}
where $\underline{\underline{\mathscr{P}}}$ is a diagonal matrix, the components of which are defined by ${\mathscr{P}}_{ij}= \frac{\delta_{i(j)}}{\mathring{Q}_{(j)}}$. It is easy to notice that \eqref{eq:alg_sys_prec} is in fact a right preconditioned system if solved sequentially via
\begin{equation}
\underline{\underline{\mathscr{A}}}(\bm{Q})\, \underline{\underline{\mathscr{P}}}^{-1} \, \bm{\mathscr{S}} =-\bm{\mathscr{R}}(\bm{Q}),\label{eq:alg_sys_prec1}
\end{equation}
for $\bm{\mathscr{S}}$ first and then for $\bm{\Delta Q}$ by
\begin{equation}
\underline{\underline{\mathscr{P}}} \,\bm{\Delta Q} =\bm{\mathscr{S}}.
\label{eq:alg_sys_prec2}
\end{equation}

The proof in thus complete by showing that a proper non-dimensionalization can hence act as an intrinsic preconditioning mechanism and thereby contribute to reducing the propagation of numerical perturbations into the solution.
\end{proof}

To further validate this result numerically, we devise a non-dimensionalized version of MS-4 by considering the following reference values arising from dimensional analysis:
\begin{equation}
\begin{gathered}
\mathring{R}=R, \quad 
\mathring{p}=p_0, \quad
\mathring{T}=T_\infty, \quad
\mathring{U}= \sqrt{\mathring{R}\,\mathring{T}},\\
\mathring{\mu}=\frac{\mathring{L}\,\mathring{p}}{\mathring{U}}\quad \mathrm{and}\quad
\mathring{L}= 1.0 \, (\mathrm{m}),
\end{gathered}
\end{equation}
such that the non-dimensional version is set by the same parameters as for the dimensional version in Eqs. \eqref{eq:vars_MS-4-a} and \eqref{eq:vars_MS-4-b} except for
\begin{equation}
\begin{gathered}
R=1.0, \quad 
p_0=1.0, \quad
T_\infty=1.0, \quad
\mu=2.67861904719577\times 10^{-6}, \\
\mathrm{and}\quad
\alpha= 1.86663348236639\times 10^{-2},
\end{gathered}
\end{equation}
where the superscript $\breve{\cdot}$ is omitted for the sake of simplicity.
\begin{figure}[!hbt]
\centering
\includegraphics[trim = 5mm 4mm 1mm 3mm, clip,width=0.35\linewidth]
{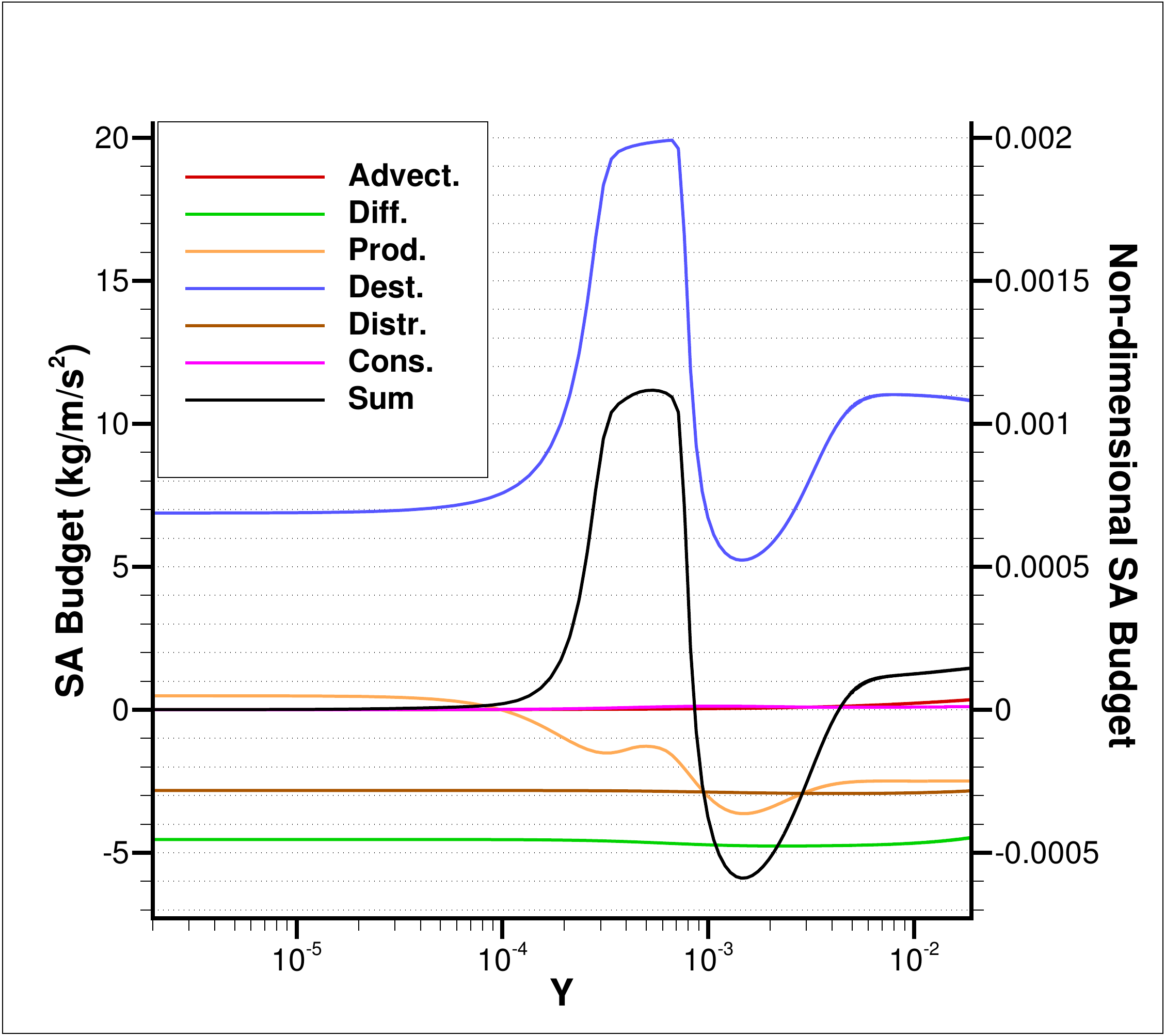}
\caption{Comparison of the SA forcing function budget between the dimensional and non-dimensional versions of MS-4 along $x=0.525$}
\label{fig:MS-4_budget}
\end{figure}

Figure \ref{fig:MS-4_budget} shows the budget of the SA source terms along $x=0.525$ by comparing the dimensional versus non-dimensional versions of MS-4. We first notice that the advection and conservation terms have negligible contributions to the total forcing function. We hence recommend that this MS be combined to MS-4 from \cite{Navah2017a} in order to provide a more reliable diagnostic on the verification of all SA model terms. It can also be noted in the Figure \ref{fig:MS-4_budget} that all the non-dimensional terms are scaled with regards to their dimensional counterparts by a factor corresponding to $\mathring{p}$. In other words, the non-dimensionalization of MS-4 results in $\breve{\mathscr{R}}(\bm{\breve Q})= C_0\,{\mathscr{R}}(\bm{\breve Q})=\,{\mathscr{R}}(\bm{Q})$ where $C_0=p_0$ for the SA equation. 

To examine the effect of non-dimensionalization on the behavior of iterative methods, the systems \eqref{eq:alg_sys} and \eqref{eq:alg_sys_prec1} are constituted and solved via GMRES algorithm. The relative residual of the linear system $\underline{\underline{{A}}} \, \bm{{x}}=\bm{b}$ is computed at each iteration by $\frac{\|\bm{b}- \underline{\underline{{A}}} \, \bm{{x}}\|}{\|\bm{b} \|}$. The results for an unpreconditioned GMRES are presented in Figs. \ref{fig:MS-4_cond_P1} and \ref{fig:MS-4_cond_P5} for respectively P1 and P5 polynomials. For each polynomial degree, a coarse and a fine grid is considered. In all cases, the beneficial effect of the non-dimensionalization on the convergence of the GMRES algorithm is evident as it allows to minimize the residual to lower levels compared to the original system for which the iterative procedure either stalls or diverges. The non-dimensionalization effect is also studied here for a GMRES algorithm preconditioned globally by a block Jacobi method and locally by an incomplete lower upper factorization  
at each block (associated to each of the four processors used to partition the domain) \cite{petsc-efficient}. The results, illustrated also in Figs. \ref{fig:MS-4_cond_P1} and \ref{fig:MS-4_cond_P5}, show that although the preconditioning does allow the dimensional system to converge and reduce the gap between the dimensional and non-dimensional versions, the discrepancy is nevertheless still non-negligible. We hence recommend the non-dimensionalization of the system of governing equations to mitigate the propagation of numerical errors into the solution, especially for stiff high-order turbulent problems.

\begin{figure}[!hbt]
\centering
\subfloat[$8 \times 24$ elements]{
\includegraphics[trim = 1mm 4mm 13mm 14mm, clip,width=0.36\linewidth]
{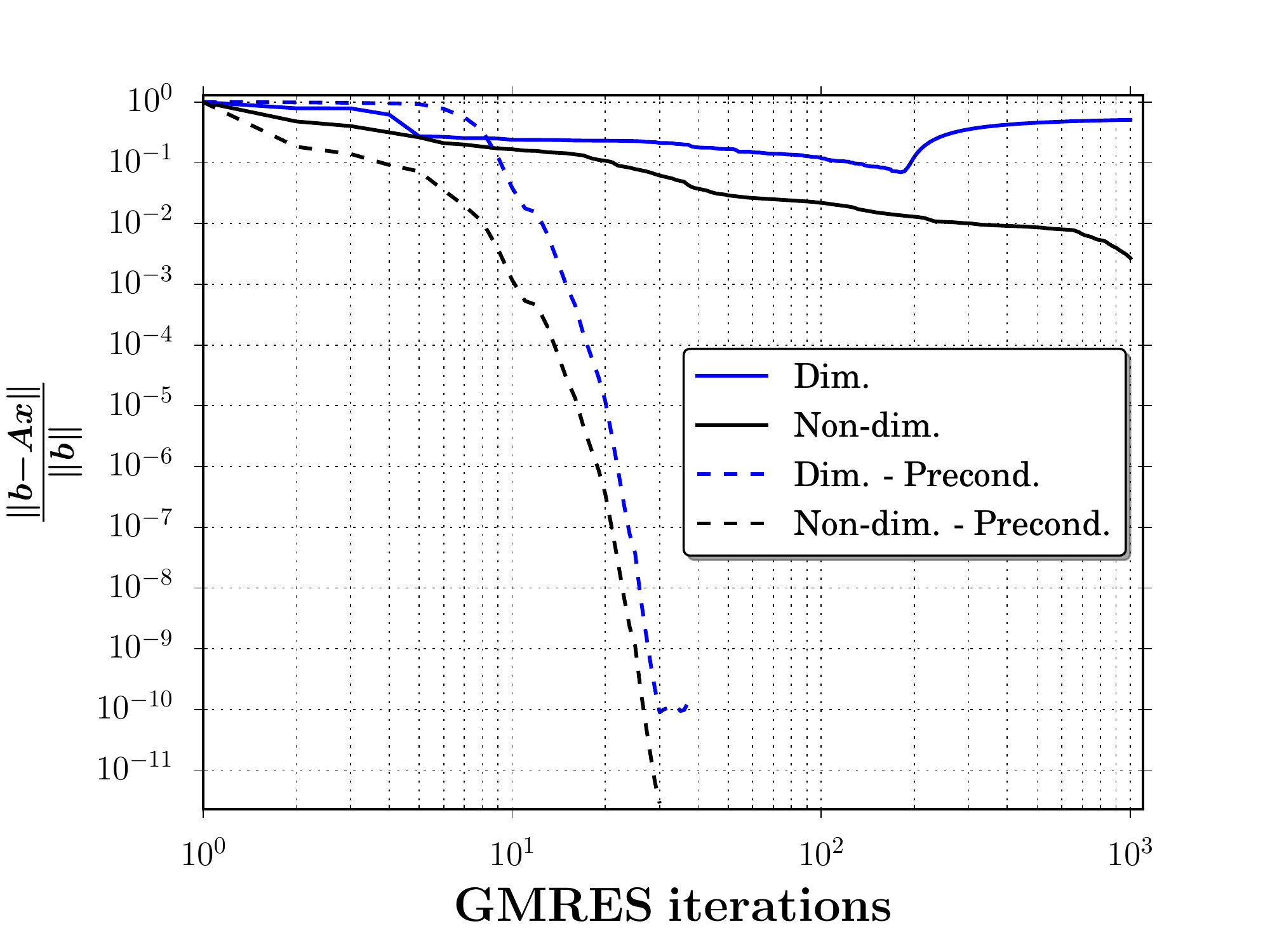}}~~~
\subfloat[$128 \times 384$ elements]{
\includegraphics[trim = 1mm 4mm 13mm 14mm, clip,width=0.36\linewidth]
{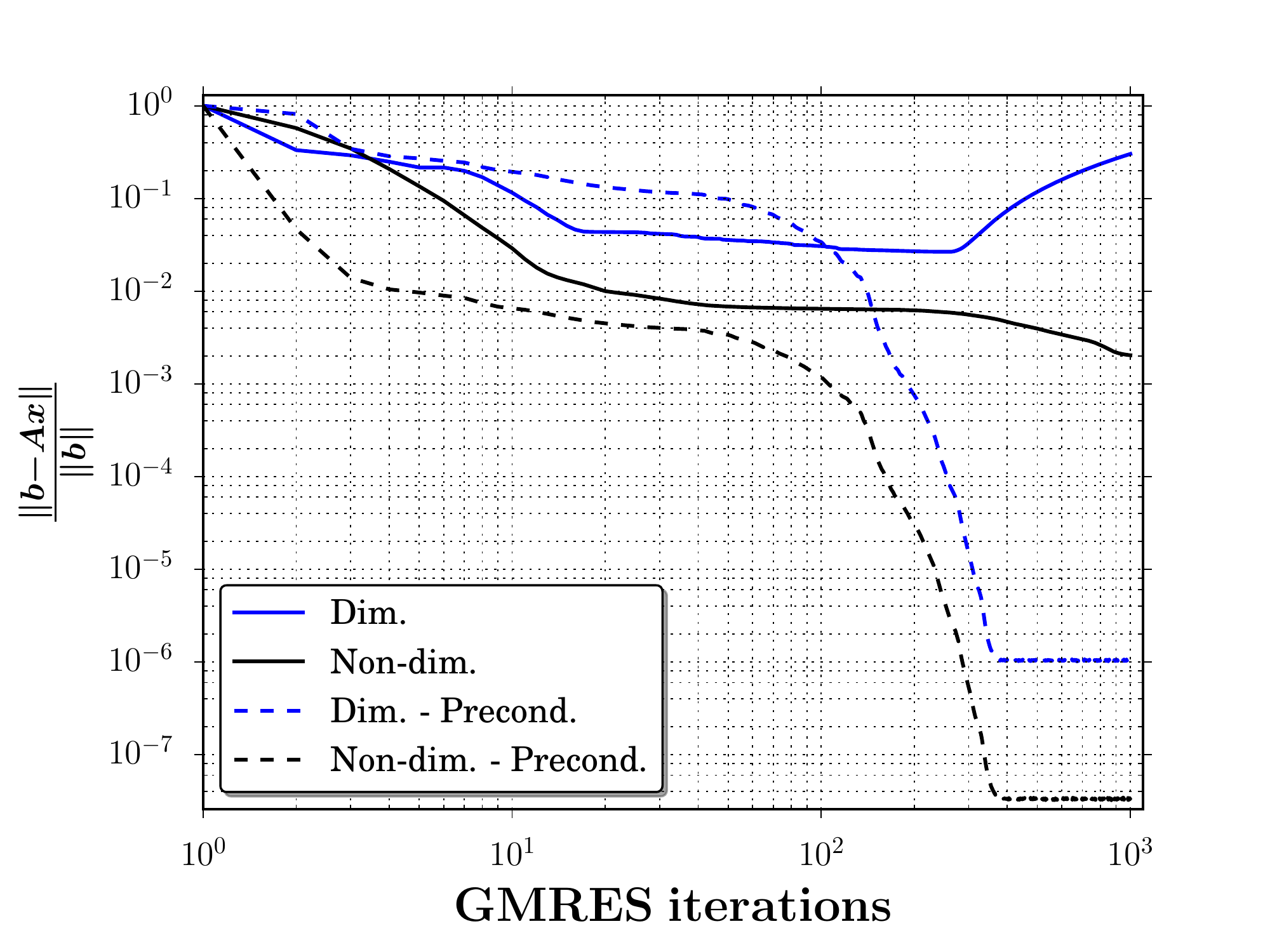}}
\caption{Relative residual convergence of GMRES  algorithm for the dimensional versus non-dimensional algebraic systems of MS-4 discretized by P1 polynomials}
\label{fig:MS-4_cond_P1}
\end{figure}

\begin{figure}[!hbt]
\centering
\subfloat[$2 \times 6$ elements]{
\includegraphics[trim = 1mm 4mm 13mm 14mm, clip,width=0.36\linewidth]
{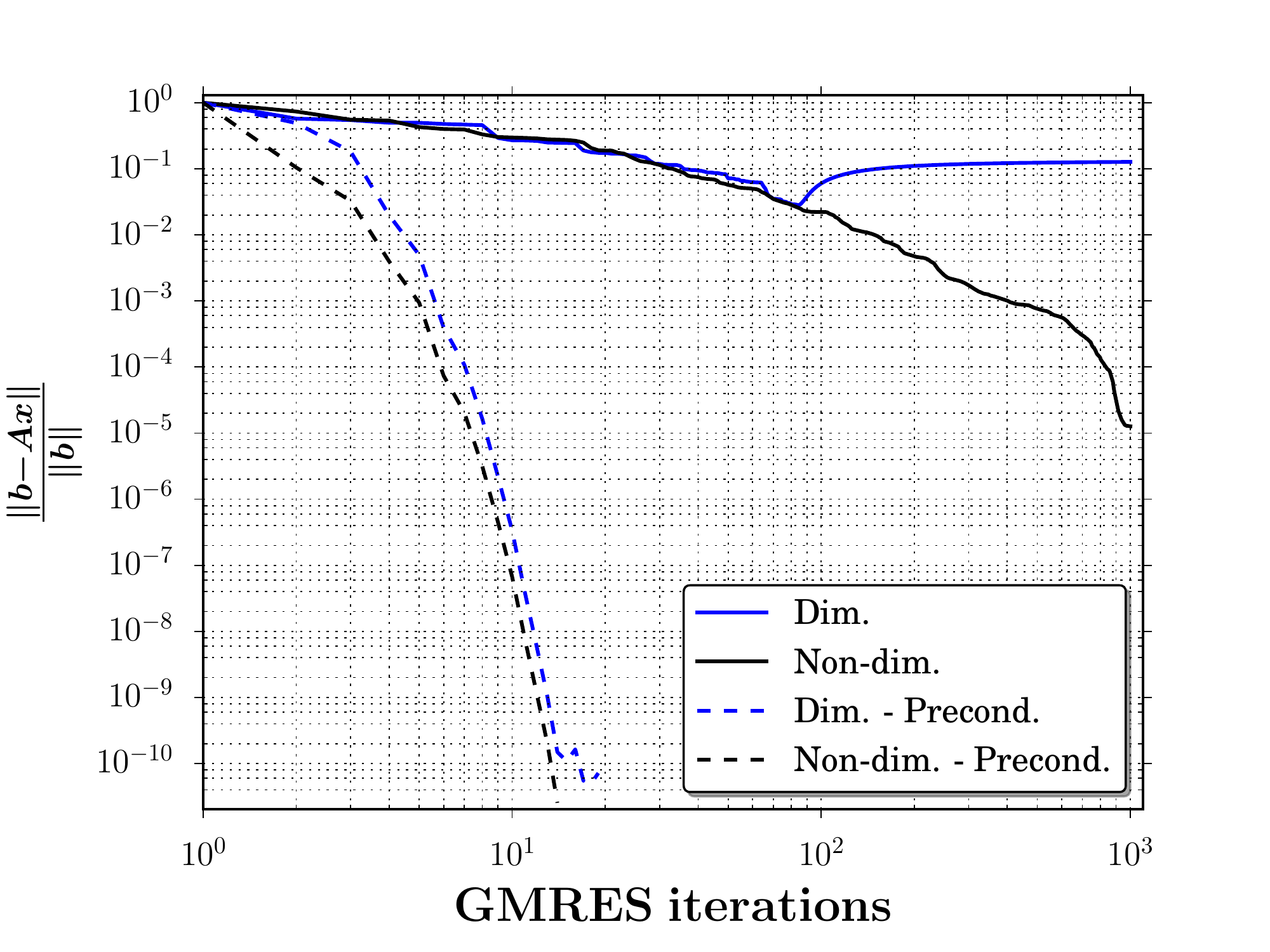}}~~~
\subfloat[$16 \times 48$ elements]{
\includegraphics[trim = 1mm 4mm 13mm 14mm, clip,width=0.36\linewidth]
{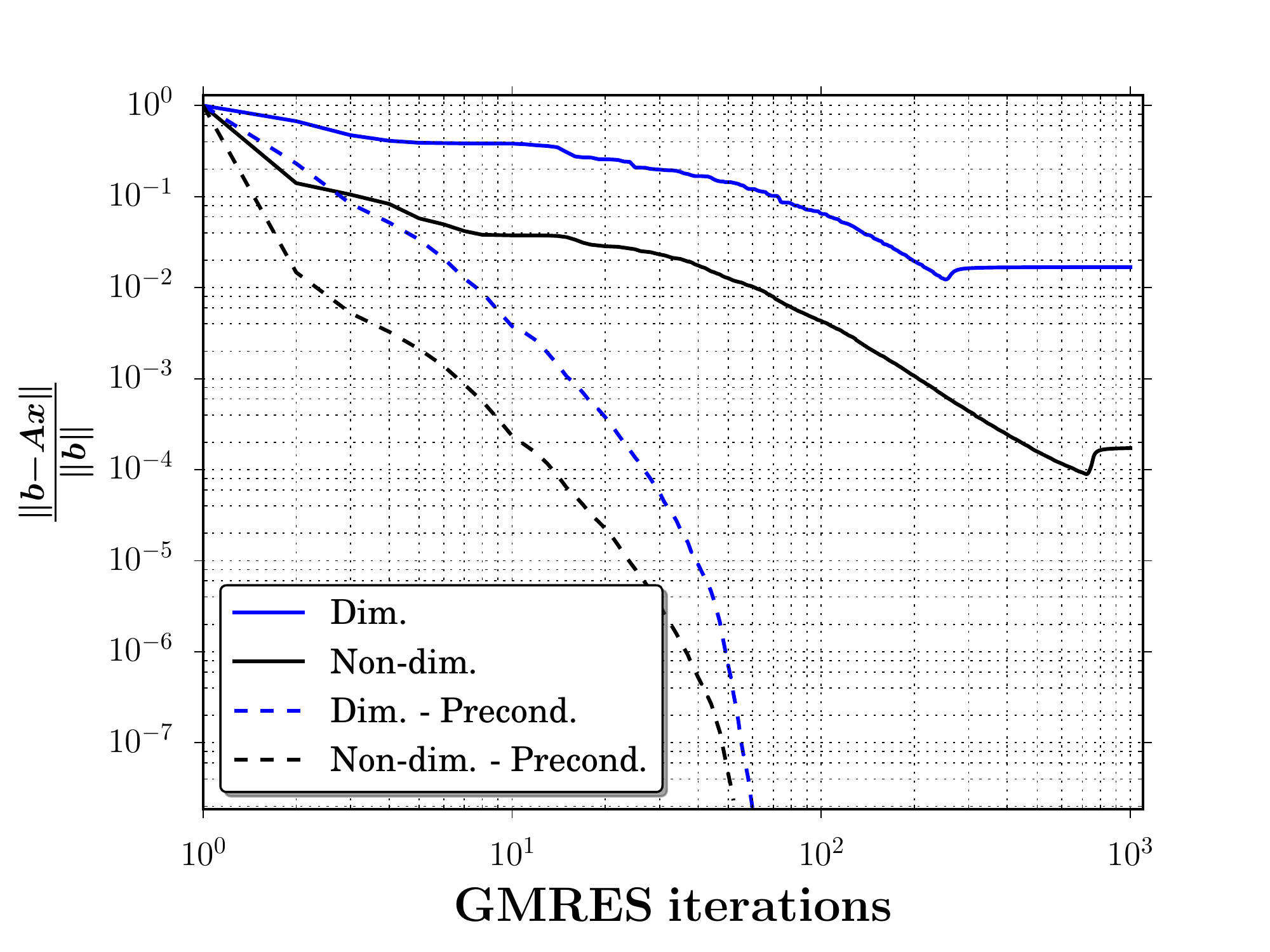}}
\caption{Relative residual convergence of GMRES  algorithm for the dimensional versus non-dimensional algebraic systems of MS-4 discretized by P5 polynomials}
\label{fig:MS-4_cond_P5}
\end{figure}

Regarding the effect of non-dimensionalization on the OOAs, let's note that as expected, the dimensional and non-dimensional versions produce the same error versus mesh size curves up to the multiplicative constant corresponding to the reference value $Q_0$ and consequently the same exact progressions of OOAs versus mesh refinement. This is for example shown in Fig. \ref{fig:Ms-dim_vs_nondim_error} for the variable $\rho\tilde \nu$.
\begin{figure}[!hbt]
\centering
\subfloat[Dimensional]{
\includegraphics[trim = 5mm 2mm 18mm 13mm clip,width=0.33\linewidth]
{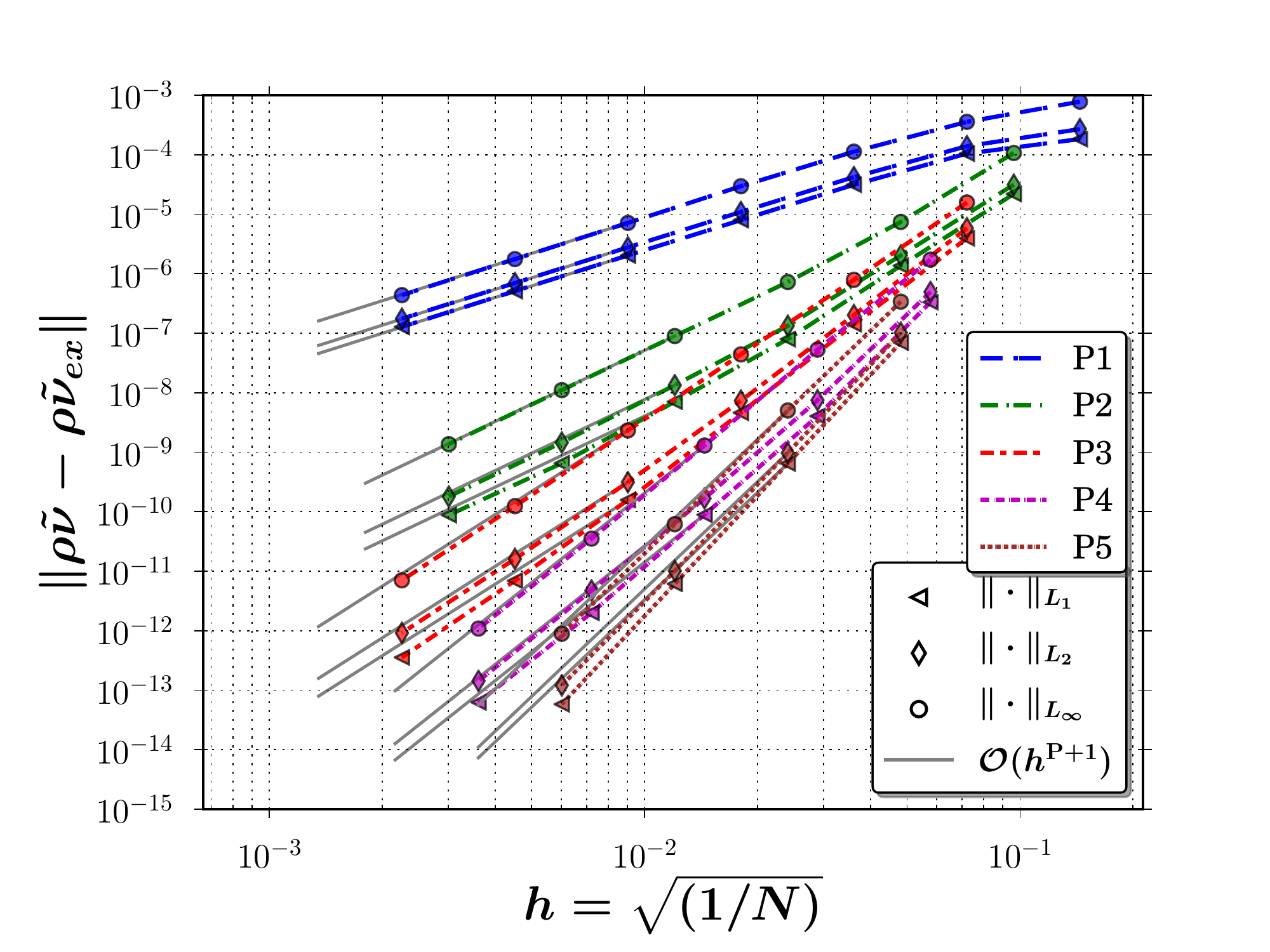}}~~~
\subfloat[Non-dimensional]{
\includegraphics[trim = 5mm 2mm 18mm 13mm clip,width=0.33\linewidth]
{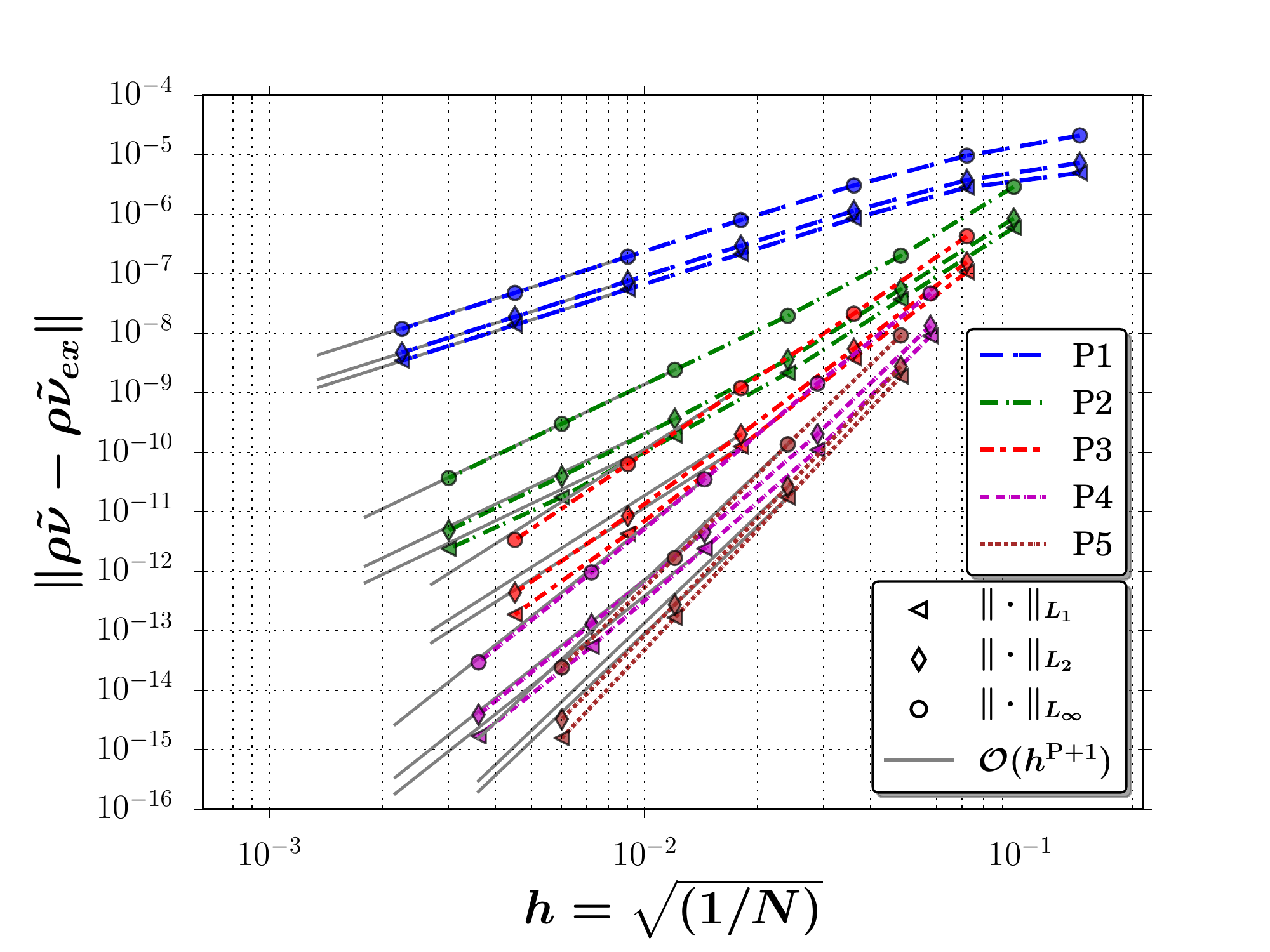}}
\caption{Evolution of the discretization error in $L_1$, $L_2$ and $L_\infty$ norms versus mesh refinement for $\rho\tilde \nu$, polynomial degrees $\mathrm{P}1$--$\mathrm{P}5$ and the dimensional compared to the non-dimensional versions of MS-4}
\label{fig:Ms-dim_vs_nondim_error}
\end{figure}

\section{Conclusions}\label{concs}

In the continuation of the work presented in \cite{Navah2017a}, a series of MSs is considered which targets the verification of high-order solvers for the solution of inviscid and laminar flows on curved wall-bounded domains in a first step, followed then by the verification for boundary layer RANS flows. We choose a numerical framework composed of the RANS-SA system of equations discretized by the CPR/FR scheme with DG correction functions.
 
Regarding the first goal, two new manufactured solutions are introduced. The first one is meant for the verification of inviscid transonic regime along with a slip wall condition on curved domains. This MS clearly displays the loss of OOAs in $L_2$ and $L_\infty$ norms due to insufficient accuracy of no-penetration condition of slip wall if based on wall normals computed by an isoparametric mapping of curved geometries. This MS is then extended to laminar flows bounded by curved, no-slip, adiabatic wall conditions and its adequacy for the verification of discretizations up to the sixth OOAs is demonstrated. Regarding code verification for the solution of RANS-modelled turbulent wall-bounded flows, two existing realistic MSs from the literature are visited and their applicability to the realization of high OOAs is discussed. The first one is the "MS2" from \cite{Eca-et-al_2007}, the adequacy of which to high-order verification of compressible solvers is assessed and a number of modifications in this regard are proposed. The addition of the SA field to the baseline Navier-Stokes portion of this MS resulted in the appearance of a high error region in the middle of the domain that delayed the appearance of the asymptotic range on grids with element clustering at the wall. This MS furthermore suffers from deficiencies in matching the characteristics of turbulent solutions such as the logarithmic velocity profile and the near-wall behaviour of some of the SA source terms. Hence a second MS, from \cite{Oliver-et-al_2012}, is as well considered which is devised to be compliant with essential characteristics of RANS-modelled flow in the vicinity of walls. This MS served to conduct a number of numerical investigations leading to the following conclusions:

\begin{itemize}
\item Explicit expressions for the evaluation of SA source terms with in-determinate form at the wall are presented which need to be incorporated into the forcing functions of physically realistic MSs for wall-bounded flow. Otherwise, the in-determination leads to NaN residual evaluations at the wall when the wall condition is weakly enforced as is customary in compact high-order solvers;
\item The modified vorticity term of the modified SA equation, left out by trigonometric MSs, is shown to be correctly verified with this wall-bounded MS;
\item A grid sensitivity analysis in terms of first element size at the wall and expansion ratio in wall normal direction is conducted showing that the realization of the OOAs requires smaller $y^+$ values than prescribed by engineering best practices ($y^+\approx 5)$. Also, observations suggest that numerical stiffness increases with mesh clustering at the wall but that the clustering intensification does not reduce the overall discretization error after a certain point.

\item It is concluded that solution verification via the estimation of numerical uncertainties associated to an output functionals such as the drag coefficient can not fully reflect the soundness of the implementation and hence should not replace a proper code verification;
\item A non-dimensional version of the turbulent wall-bounded MS is proposed and it is mathematically proven and numerically validated that a proper non-dimensionalization of the unknowns acts as an intrinsic preconditioning mechanism for iterative solution techniques, thus contributing to the minimization of iterative and round-off errors.
\end{itemize}

To assist the interested reader in the application of the presented verification methodology, Python scripts and C code routines are provided as accompanying material (see \cite{navah017_github}).


\section*{Acknowledgments}
The authors gratefully acknowledge the generous support from the Fonds de recherche du Qu\'ebec -- Nature et technologies (FRQNT), the Natural Sciences and Engineering Research Council (NSERC), and the Department of Mechanical Engineering of McGill University.
%
%
%
%
%


\appendix

\clearpage

\section{Definition of the grid convergence metrics}
\label{sec:norms}
The norms employed throughout this work to measure the discretization error are defined as follows: 

\vspace{0.1cm}
\begin{itemize}
\item $L_\infty$ norm:
\begin{equation}
\| \mathcal{E}_Q \|_\infty =
\mathrm{max} \lvert Q_i - Q_i^{\mathrm{ex}} \lvert \;\; \mathrm{for} \;\; {i\,\in \,\left[ 1\,..\,{N_\mathrm{DOF}} \right] }
\end{equation}

\item $L_1$ norm:
\begin{equation}
\| \mathcal{E}_Q \|_{L_1} =
\frac{\int_{\Omega} \, \lvert Q - Q^{\mathrm{ex}}\lvert \, d\Omega}{{\int_{\Omega} d\Omega}}
\end{equation}

\item $L_2$ norm:
\begin{equation}
\| \mathcal{E}_Q \|_{L_2} = \left( \frac{ \int_{\Omega} \, \left( Q - Q^{\mathrm{ex}}\right)^2d\Omega}{{\int_{\Omega} d\Omega}} \right)^{\frac{1}{2}}
\end{equation}
\item $H_1$ norm:
\begin{equation}
\| \mathcal{E}_Q \|_{H_1} = \left( \frac{\int_{\Omega} \, \left( Q - Q^{\mathrm{ex}}\right)^2  d\Omega\,+ \,\int_{\Omega}\,\sum_{q=1}^{N_\mathrm{d}}\left(\partial_q Q - (\partial_q Q)^{\mathrm{ex}}\right)^2 \, d\Omega}{{\int_{\Omega} d\Omega}} \right)^{\frac{1}{2}}
\end{equation}
\item $H_1$ semi-norm of uncorrected gradients:
\begin{equation}
\abs[\big]{ \mathcal{E}_Q }_{H_1} =\left(\frac{\int_{\Omega}\,\sum_{q=1}^{N_\mathrm{d}}\left(\partial_q Q - (\partial_q Q)^{\mathrm{ex}} \right)^2 \, d\Omega}{{\int_{\Omega} d\Omega}}\right)^{\frac{1}{2}}
\end{equation}
\item $H_1$ semi-norm of fully corrected gradients:
\begin{equation}
\abs[\big]{ \mathcal{E}_Q }_{\overline{\overline{H_1}}} =\left(\frac{\int_{\Omega}\,\sum_{q=1}^{N_\mathrm{d}}\left({\overline{\overline{\partial_q Q}}} - (\partial_q Q)^{\mathrm{ex}}\right)^2 \, d\Omega}{{\int_{\Omega} d\Omega}} \right)^{\frac{1}{2}}
\end{equation}
\end{itemize}

The integrals are computed by GLL quadratures and the typical element size, $h$, is estimated via 
\begin{equation}
h=\sqrt[{-N_d}]{N_{\mathrm{DOF}}},
\label{eq:elem_size}
\end{equation}
where $N_{\mathrm{DOF}}=\sum_{e_i}(N_\mathrm{node})_{e_i}$ is the total number of DOFs per equation.

\clearpage

\section{Manufactured solution results}
\subsection{MS-1}

\begin{figure}[!hbt]
\centering
\subfloat[$\rho^{\mathrm{MS}}$]{
\includegraphics[trim = 0mm 0mm 0mm 0mm, clip,width=0.4\linewidth]
{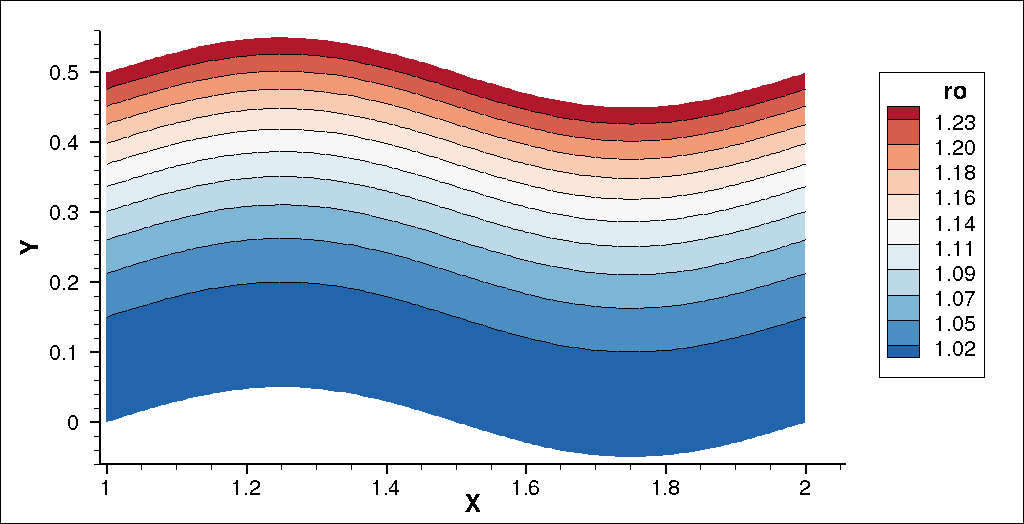}}
~~~
\subfloat[$u^{\mathrm{MS}}$]{
\includegraphics[trim = 0mm 0mm 0mm 0mm, clip,width=0.4\linewidth]
{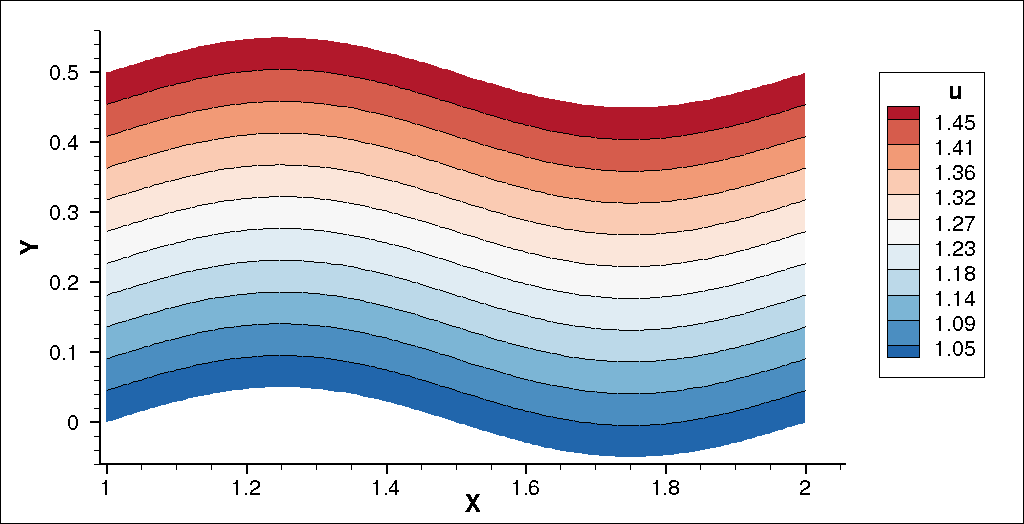}}
\vfill
\subfloat[$v^{\mathrm{MS}}$]{
\includegraphics[trim = 0mm 0mm 0mm 0mm, clip,width=0.4\linewidth]{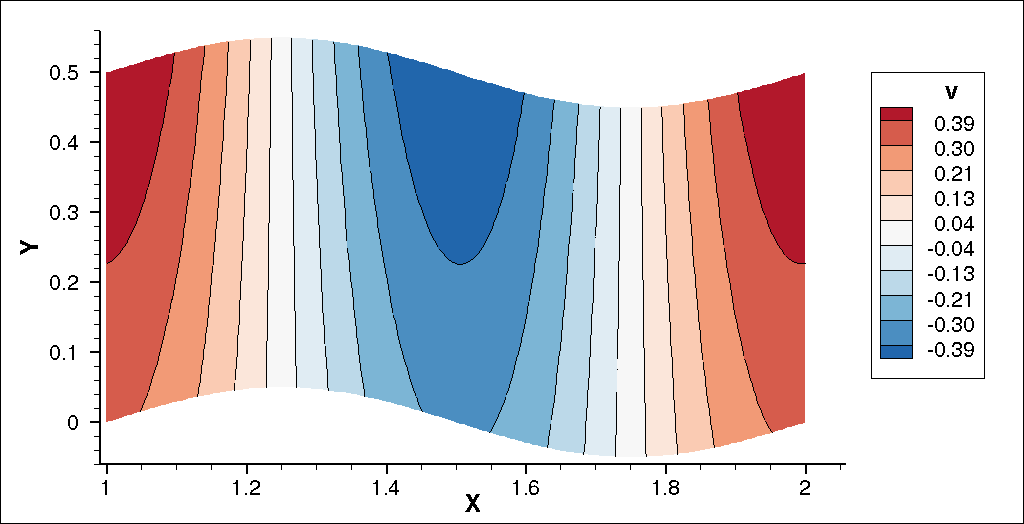}}
~~~
\subfloat[$p^{\mathrm{MS}}$]{
\includegraphics[trim = 0mm 0mm 0mm 0mm, clip,width=0.4\linewidth]{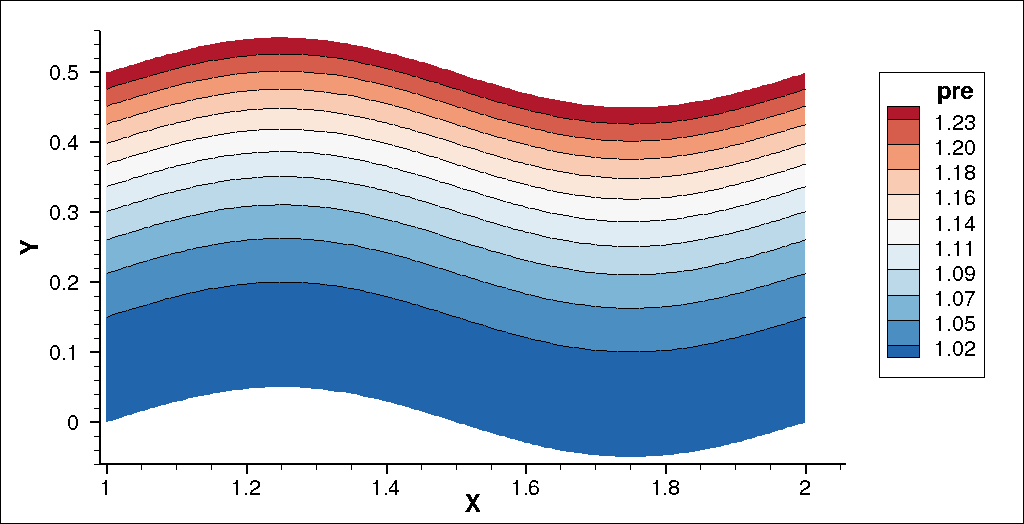}}
\vfill
\subfloat[$\mbox{Ma}^{\mathrm{MS}}$ and ${\bm{u}}^{\mathrm{MS}}$]{
\includegraphics[trim = 0mm 0mm 0mm 0mm, clip,width=0.4\linewidth]{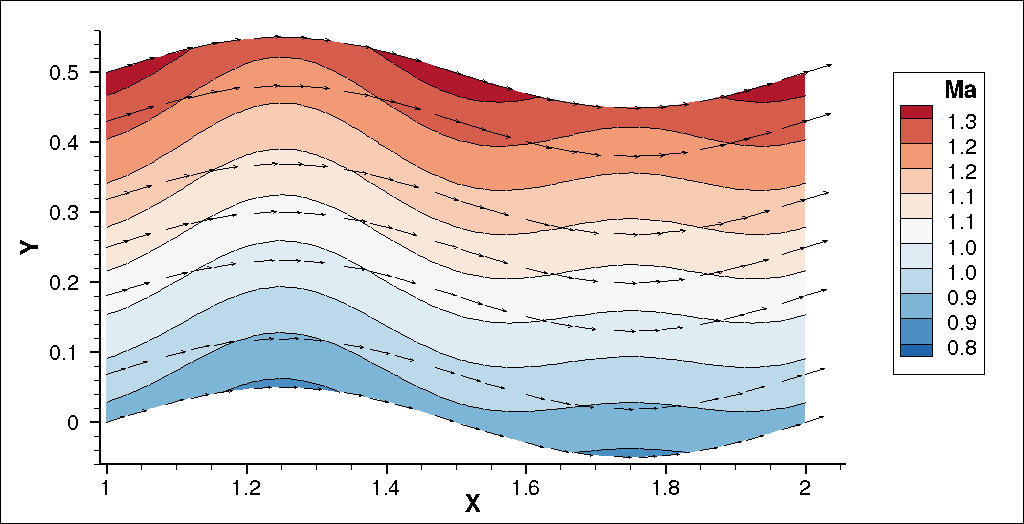}}
\caption{Manufactured solution MS-1}
\label{fig:MS-1}
\end{figure}

\begin{figure}[!hbt]
\centering
\subfloat[$\rho$]{
\includegraphics[trim = 5mm 2mm 18mm 13mm, clip,width=0.32\linewidth]
{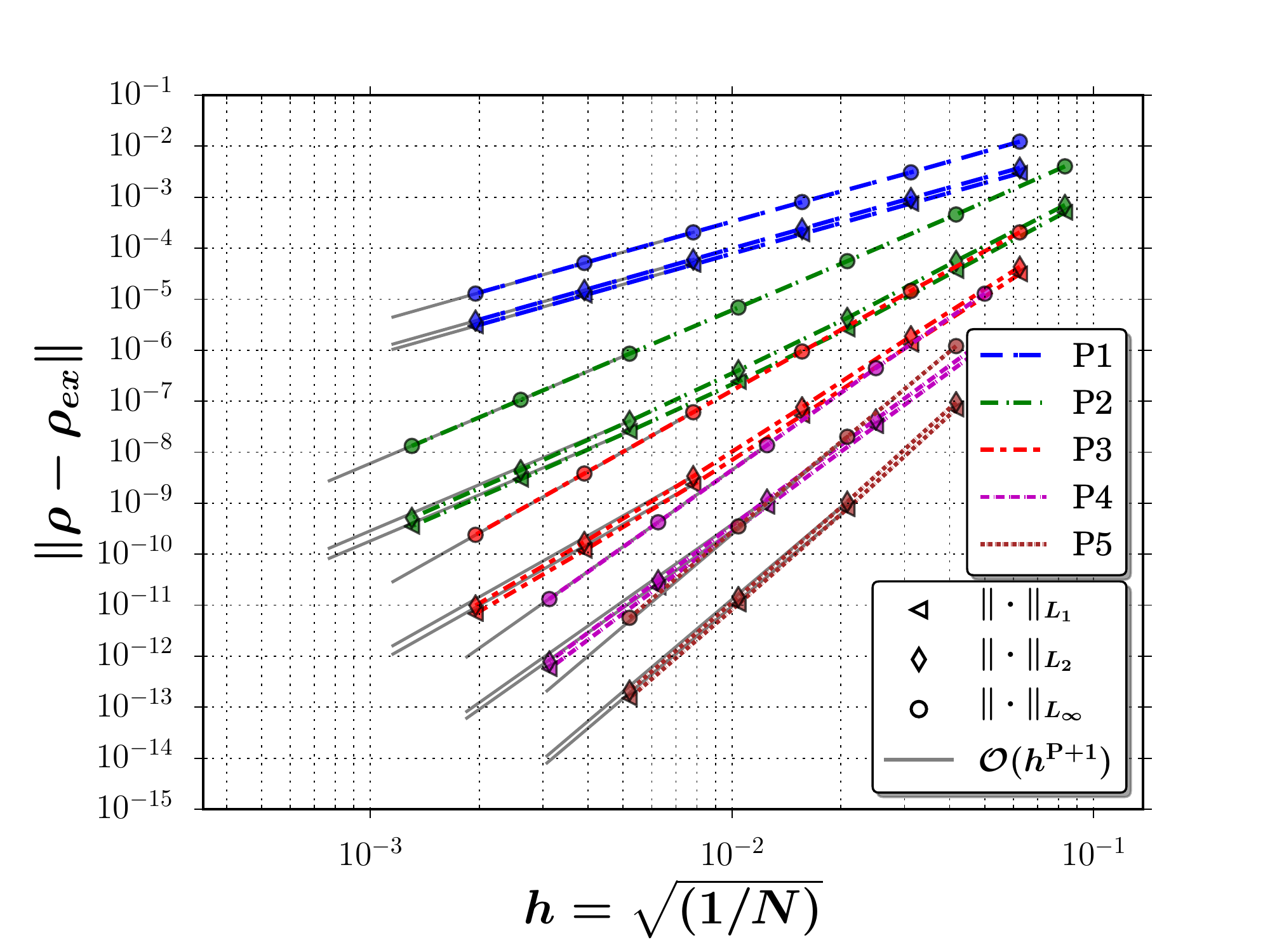}}
~~~
\subfloat[$\rho u$]{
\includegraphics[trim = 5mm 2mm 18mm 13mm, clip,width=0.32\linewidth]{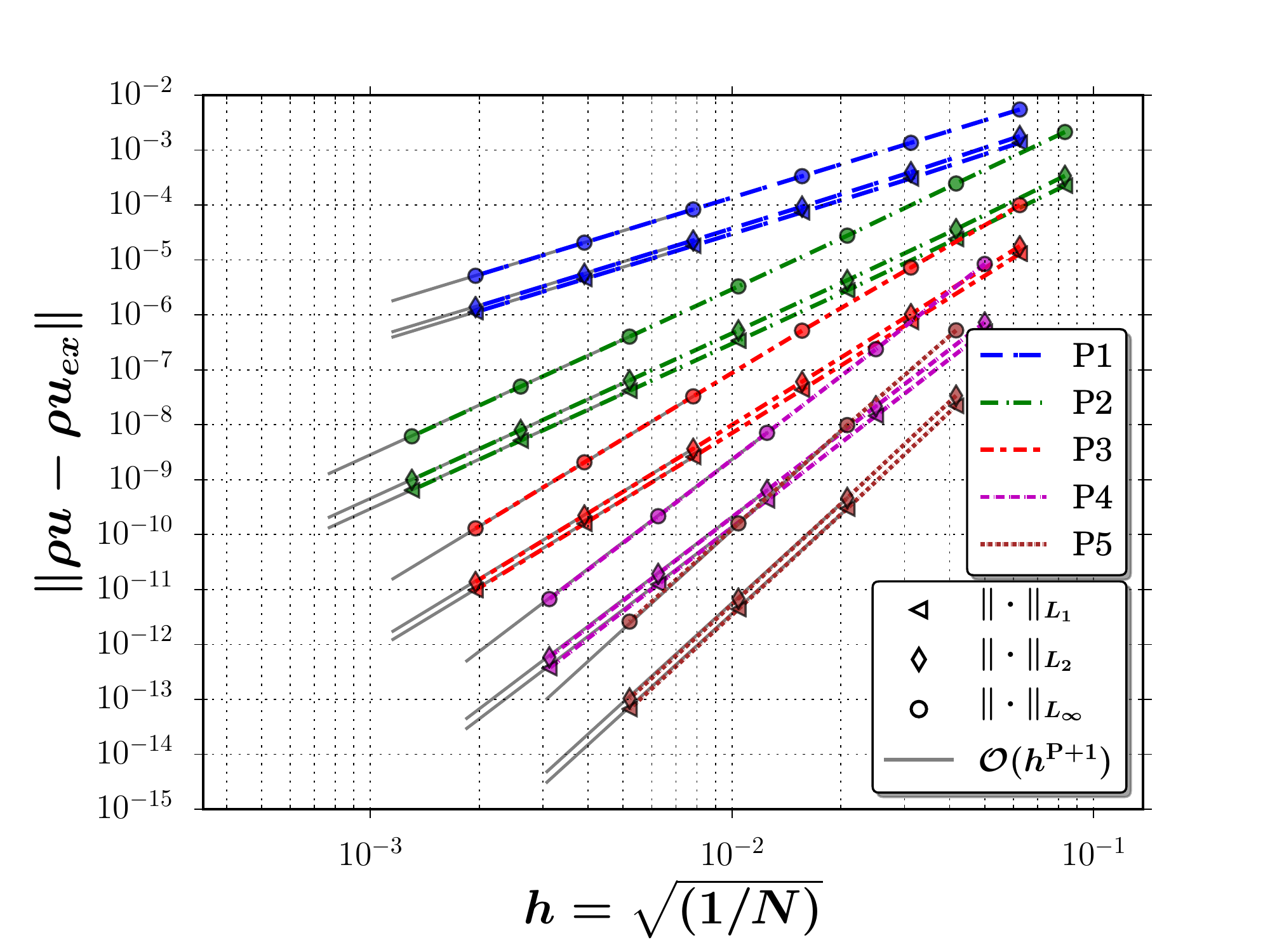}}
\vfill
\subfloat[$\rho v$]{
\includegraphics[trim = 5mm 2mm 18mm 13mm, clip,width=0.32\linewidth]
{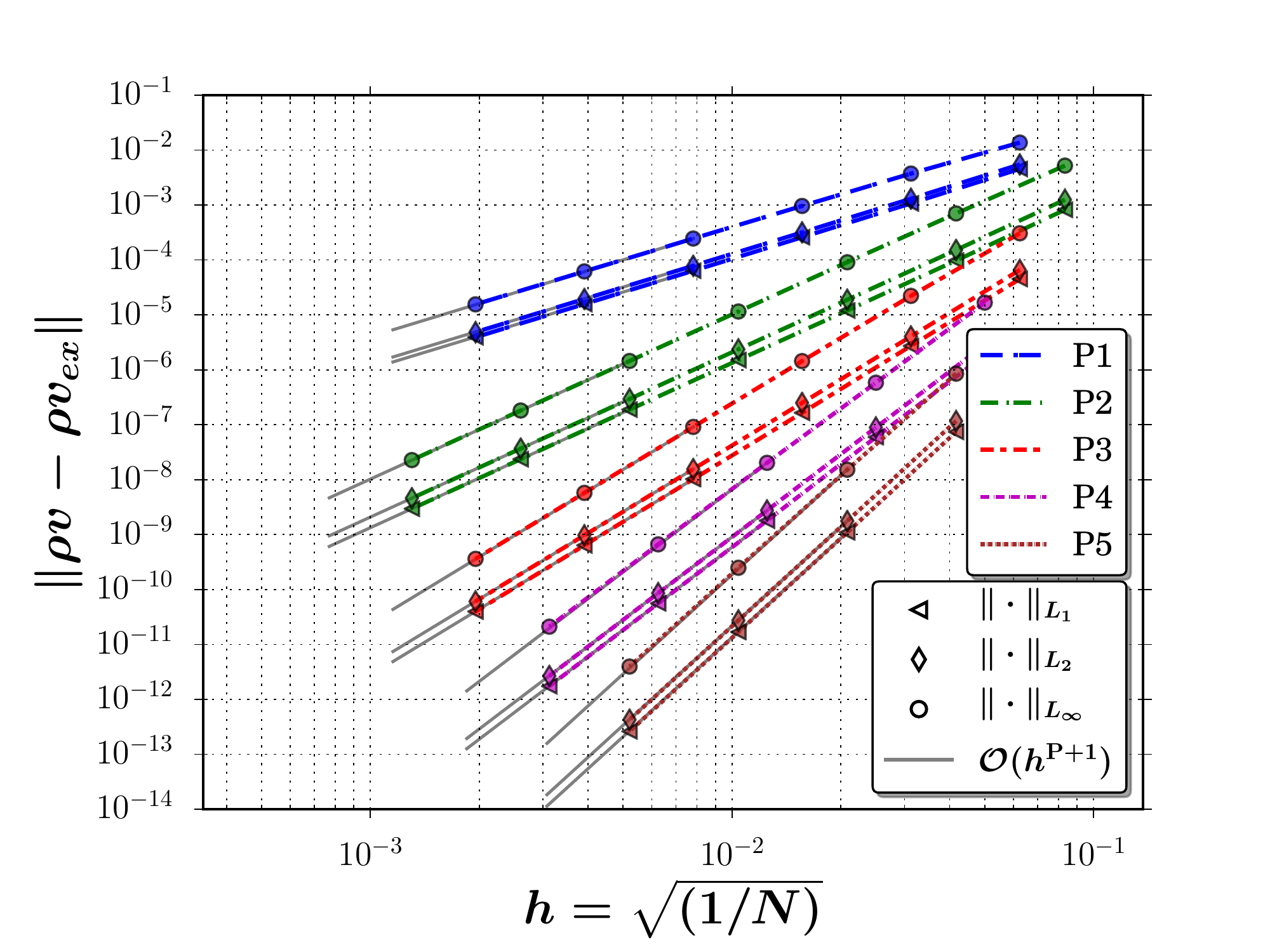}}
~~~
\subfloat[$\rho E$]{
\includegraphics[trim = 5mm 2mm 18mm 13mm, clip,width=0.32\linewidth]{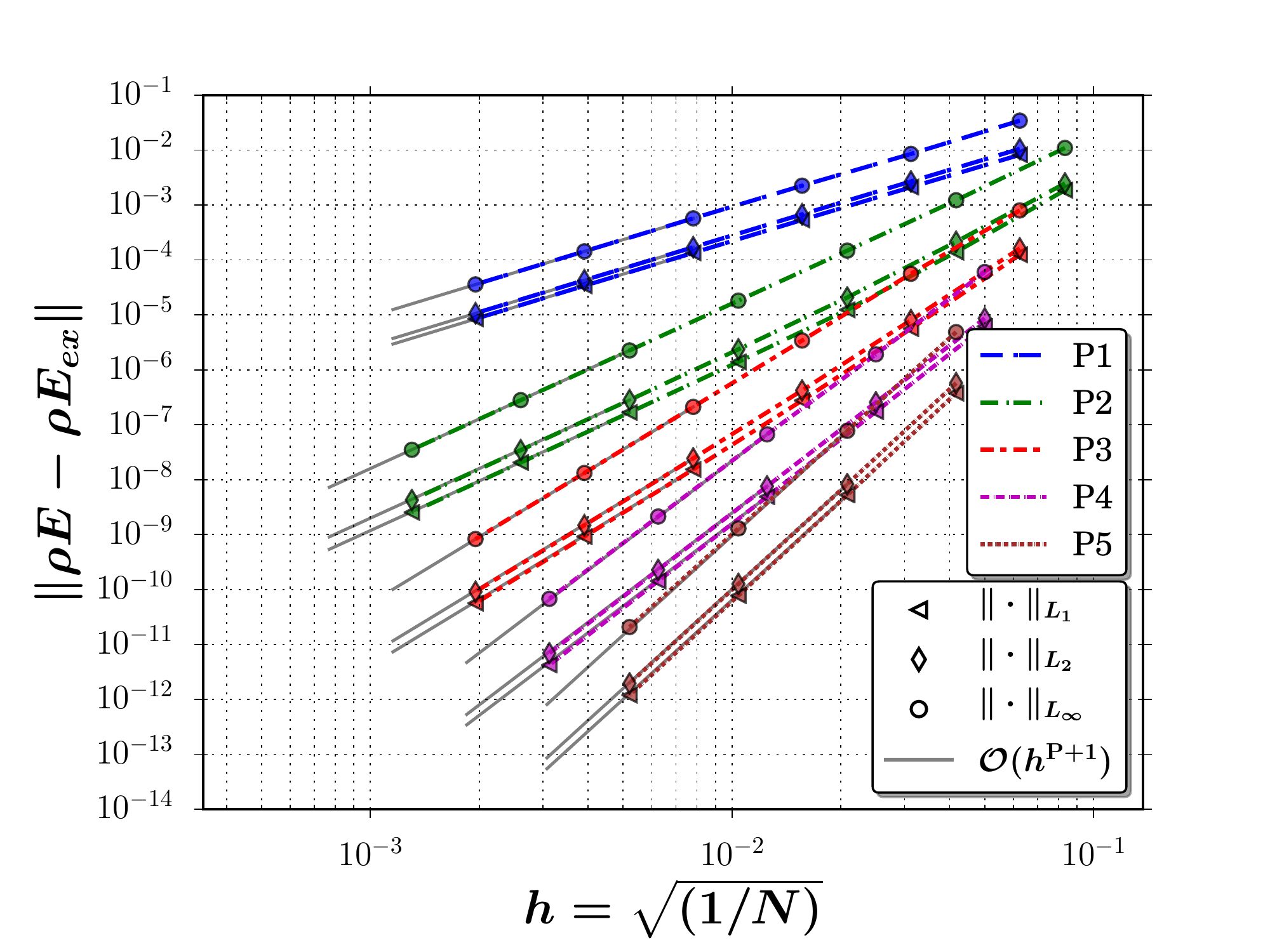}}
\caption{Evolution of the discretization error in $L_1$, $L_2$ and $L_\infty$ norms versus mesh refinement for MS-1 with analytical normal at the wall and  $\mathrm{P}1$--$\mathrm{P}5$}
\label{fig:Err_allE_allP_MS-1}
\end{figure}
\begin{figure}[!hbt]
\centering
\subfloat[$\rho$]{ 
\includegraphics[trim = 16mm 3mm 18mm 13mm, clip,width=0.3\linewidth]
{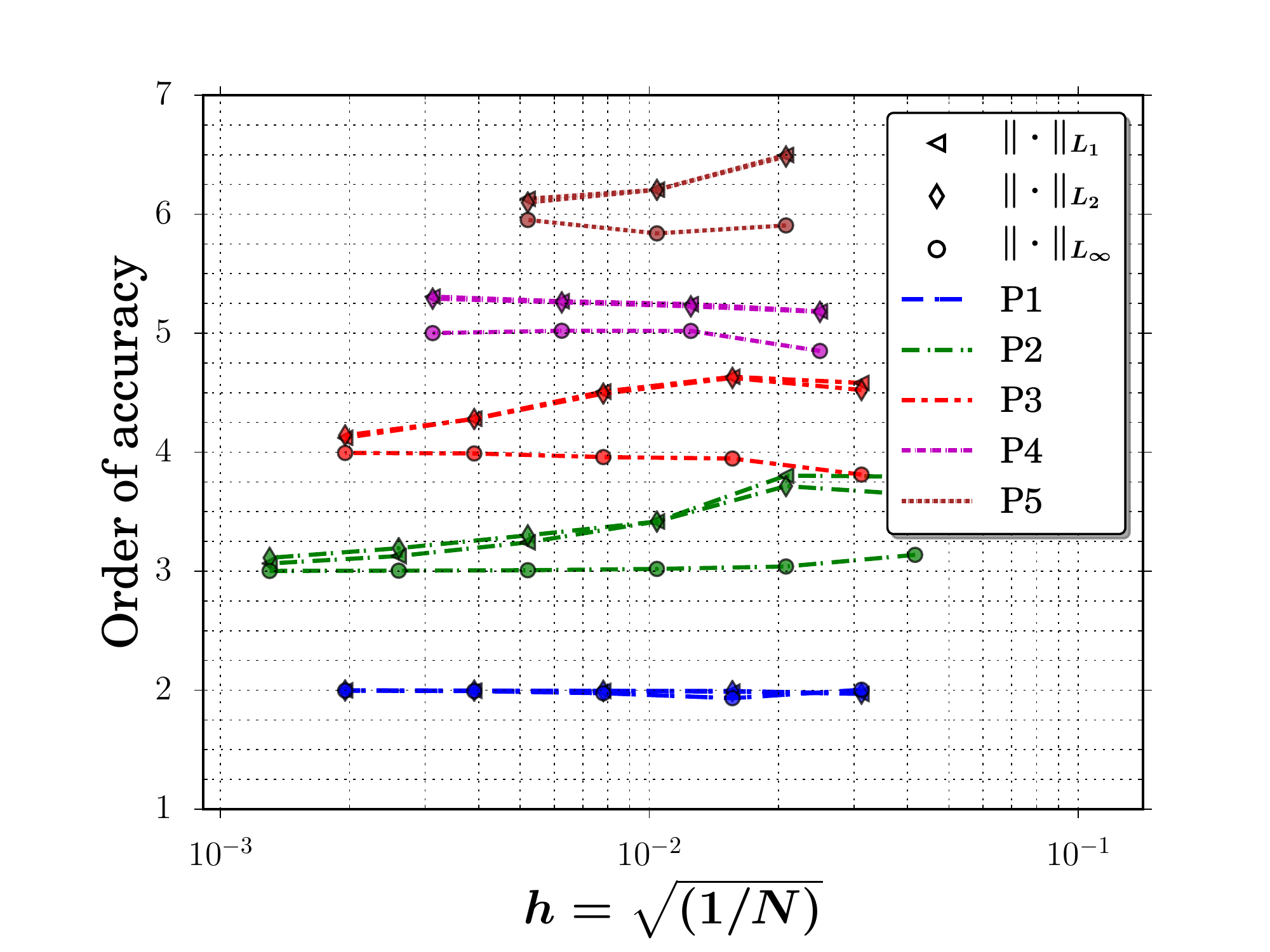}}
~~~
\subfloat[$\rho u$]{
\includegraphics[trim = 16mm 3mm 18mm 13mm, clip,width=0.3\linewidth]{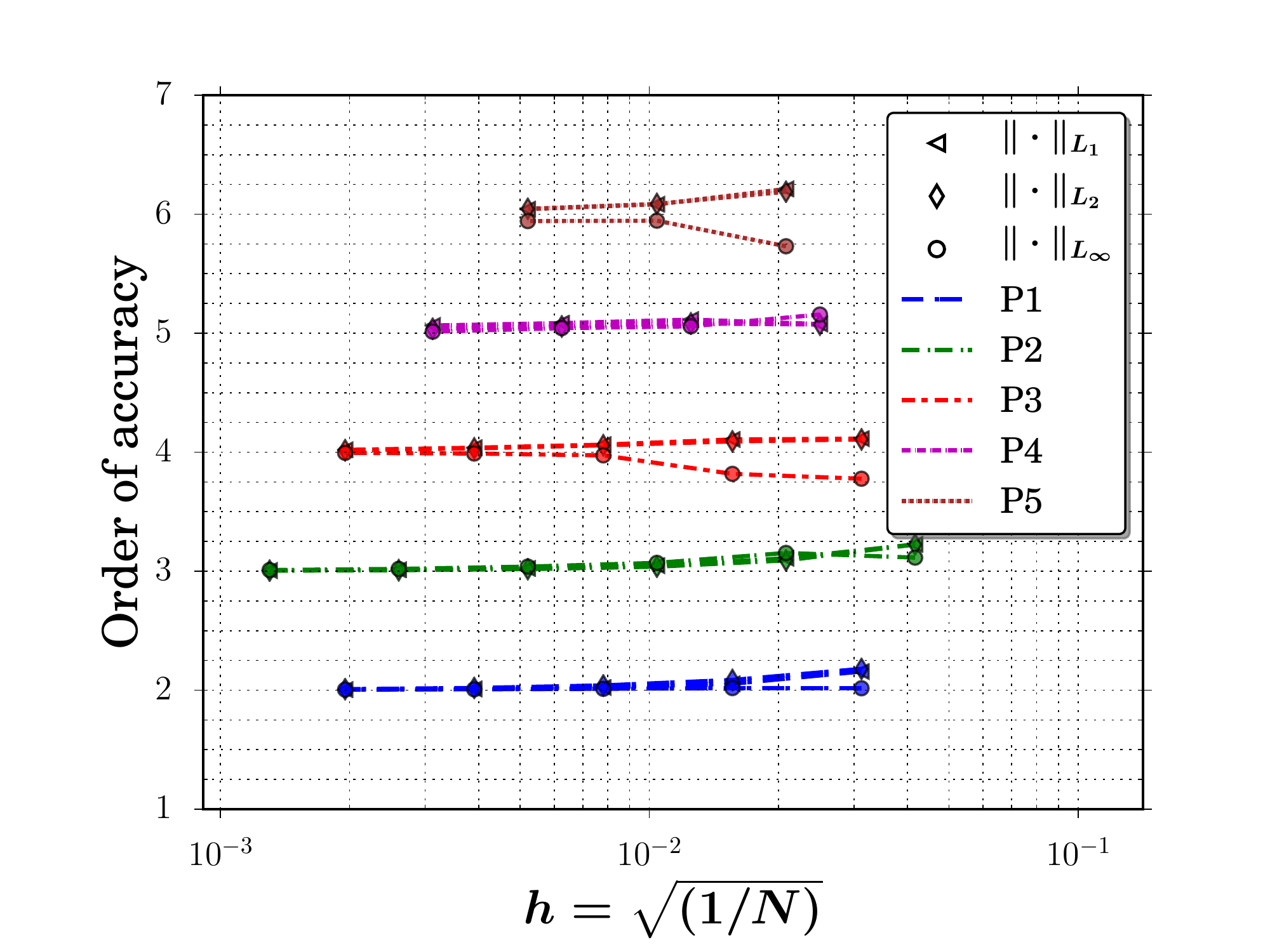}}
\vfill
\subfloat[$\rho v$]{
\includegraphics[trim = 16mm 3mm 18mm 13mm, clip,width=0.3\linewidth]
{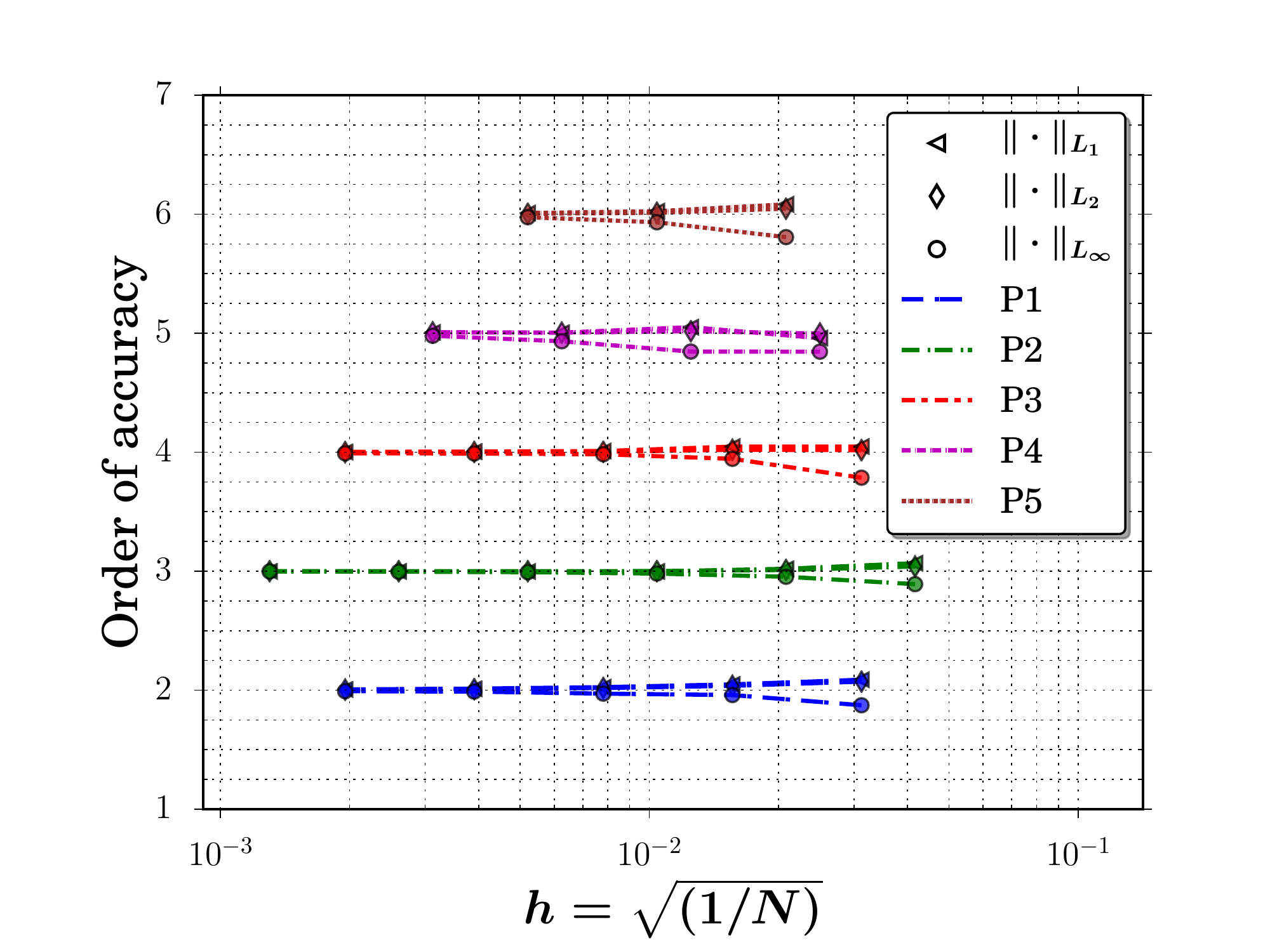}}
~~~
\subfloat[$\rho E$]{
\includegraphics[trim = 16mm 3mm 18mm 13mm, clip,width=0.3\linewidth]{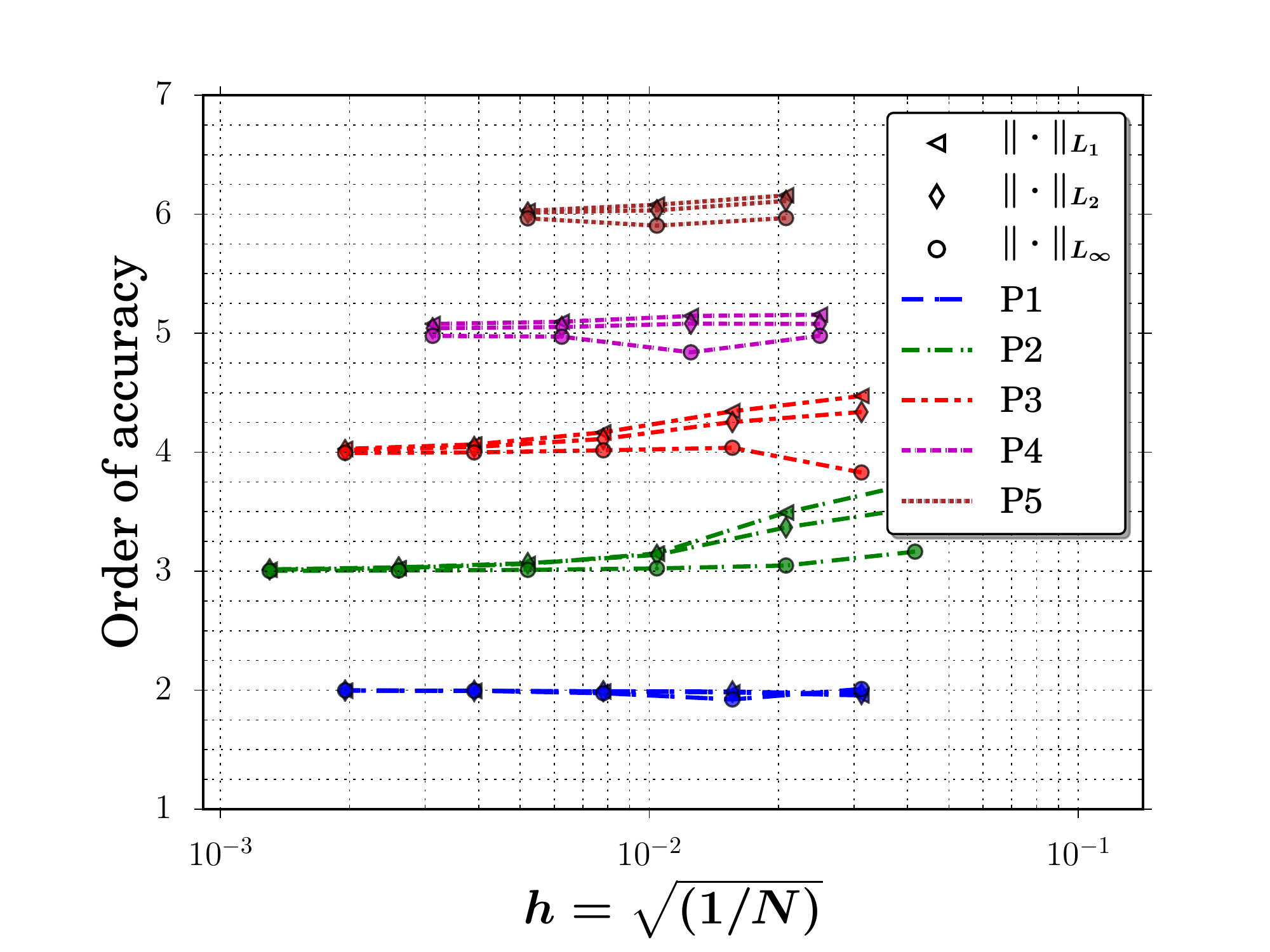}}
\caption{Evolution of the OOAs in $L_1$, $L_2$ and $L_\infty$ norms versus mesh refinement for MS-1 with analytical normal at the wall  and  $\mathrm{P}1$--$\mathrm{P}5$}
\label{fig:Orders_MS-1}
\end{figure}

\clearpage

\subsection{MS-2}

\begin{figure}[!hbt]
\centering
\vspace{1mm}
\subfloat[$\rho^{\mathrm{MS}}$]{
\includegraphics[trim = 0mm 0mm 0mm 0mm, clip,width=0.4\linewidth]
{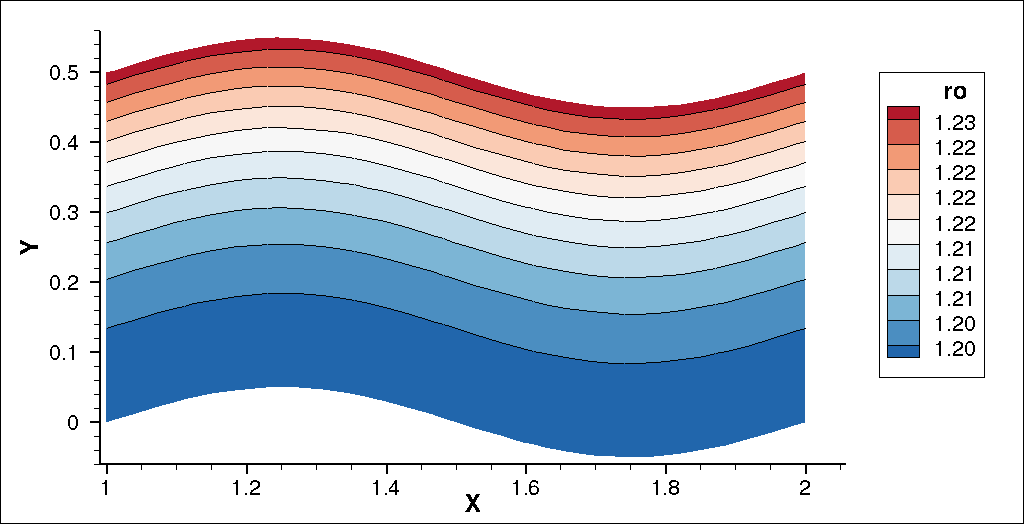}}
~~~
\subfloat[$u^{\mathrm{MS}}$]{
\includegraphics[trim = 0mm 0mm 0mm 0mm, clip,width=0.4\linewidth]
{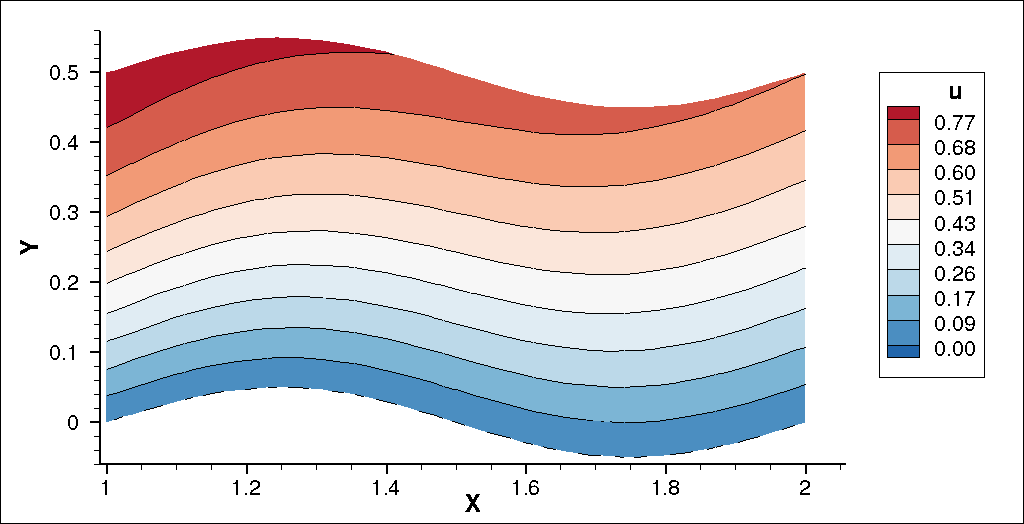}}
\vfill
\subfloat[$v^{\mathrm{MS}}$]{
\includegraphics[trim = 0mm 0mm 0mm 0mm, clip,width=0.4\linewidth]{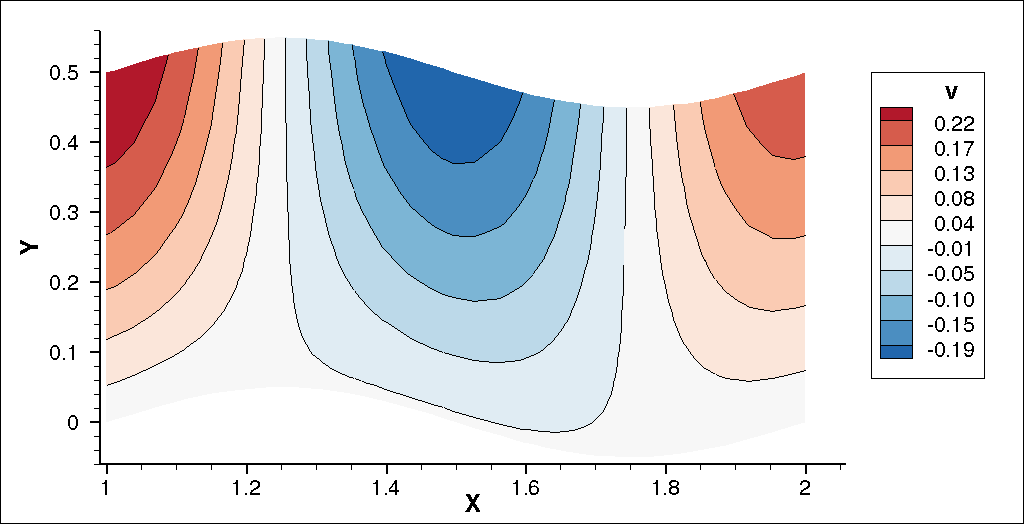}}
~~~
\subfloat[$p^{\mathrm{MS}}$]{
\includegraphics[trim = 0mm 0mm 0mm 0mm, clip,width=0.4\linewidth]{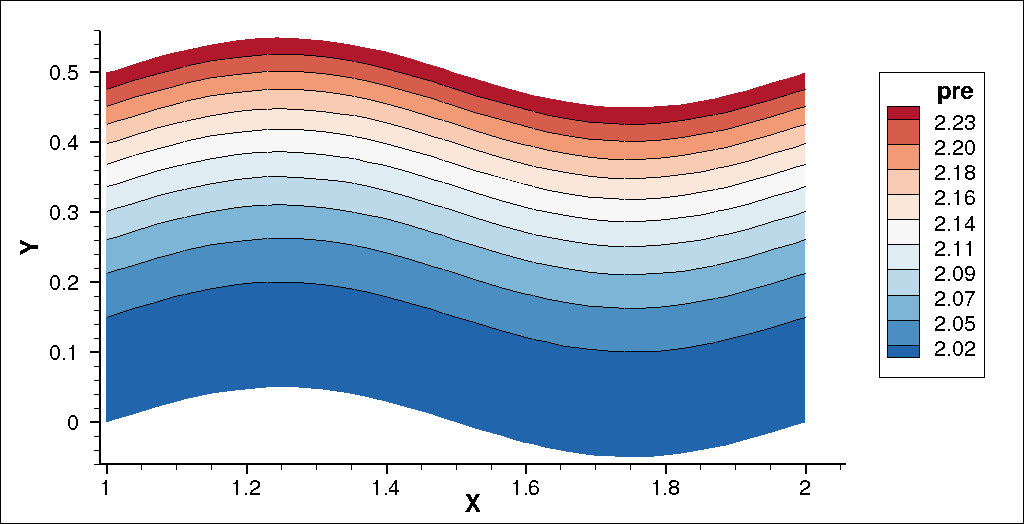}}
\vfill
\subfloat[$\mbox{Ma}^{\mathrm{MS}}$ and ${\bm{u}}^{\mathrm{MS}}$]{
\includegraphics[trim = 0mm 0mm 0mm 0mm, clip,width=0.4\linewidth]{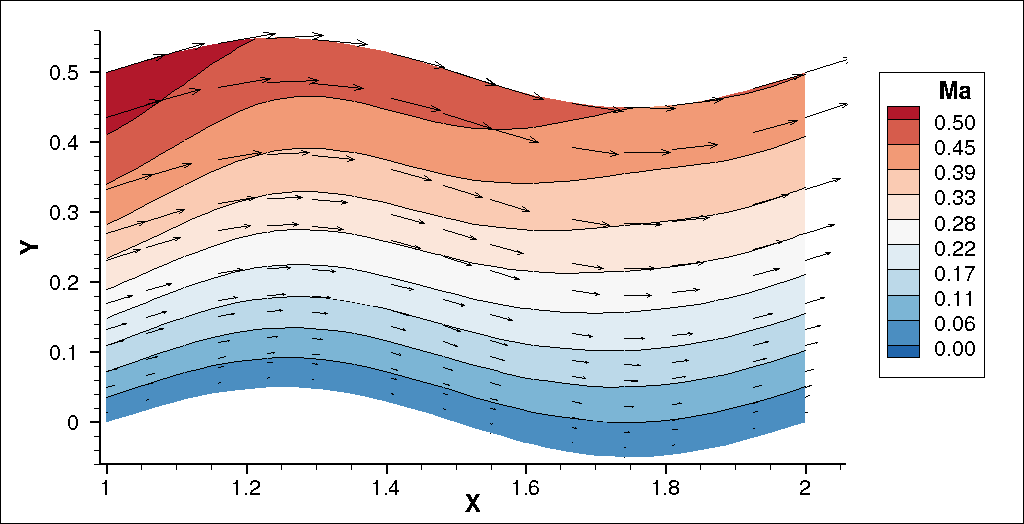}}
\caption{Manufactured solution MS-2}
\label{fig:MS-2}
\end{figure}

\begin{figure}[!hbt]
\centering
\subfloat[$\rho$]{
\includegraphics[trim = 5mm 2mm 18mm 13mm, clip,width=0.32\linewidth]
{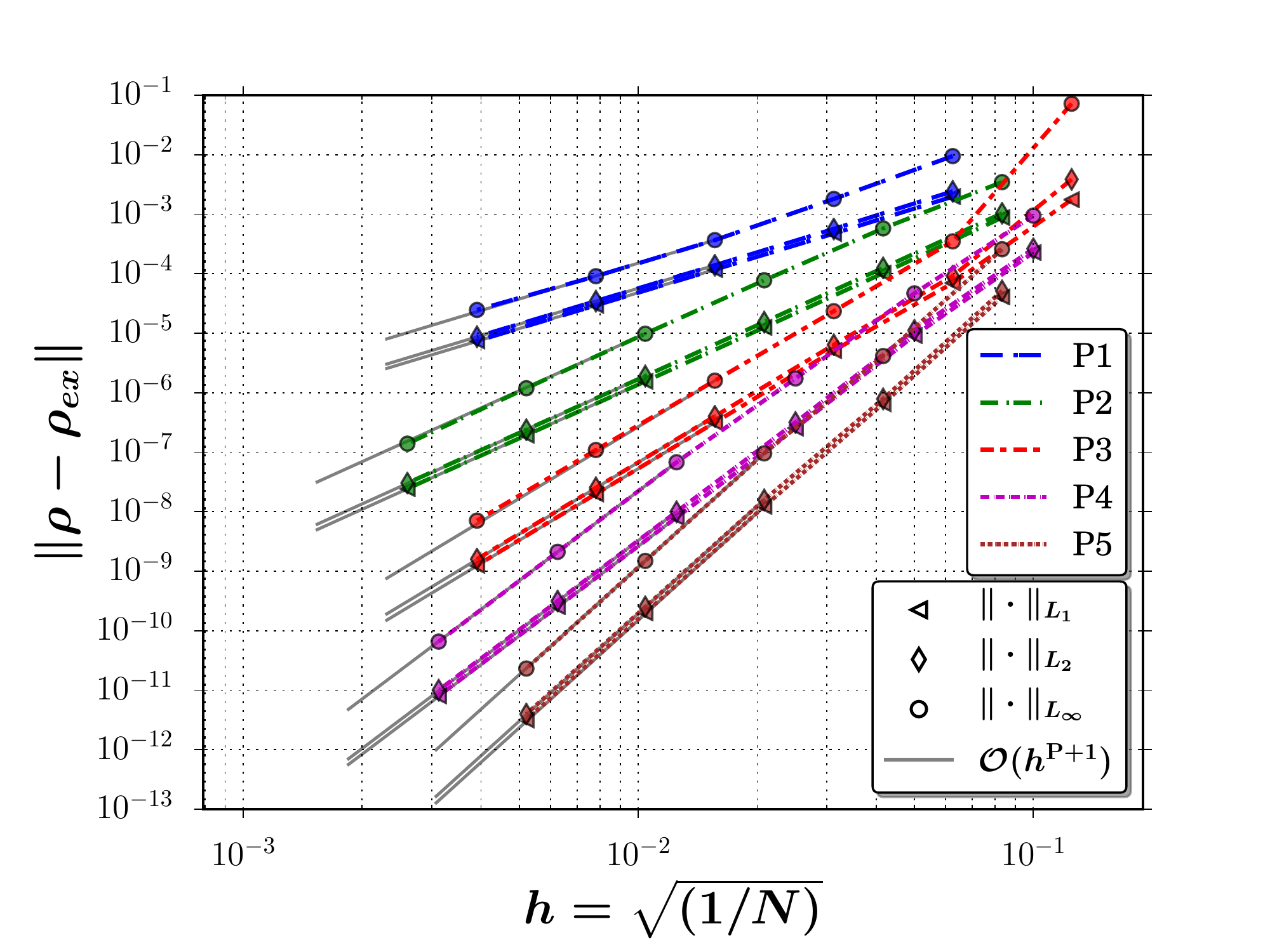}}
~~~
\subfloat[$\rho u$]{
\includegraphics[trim = 5mm 2mm 18mm 13mm, clip,width=0.32\linewidth]{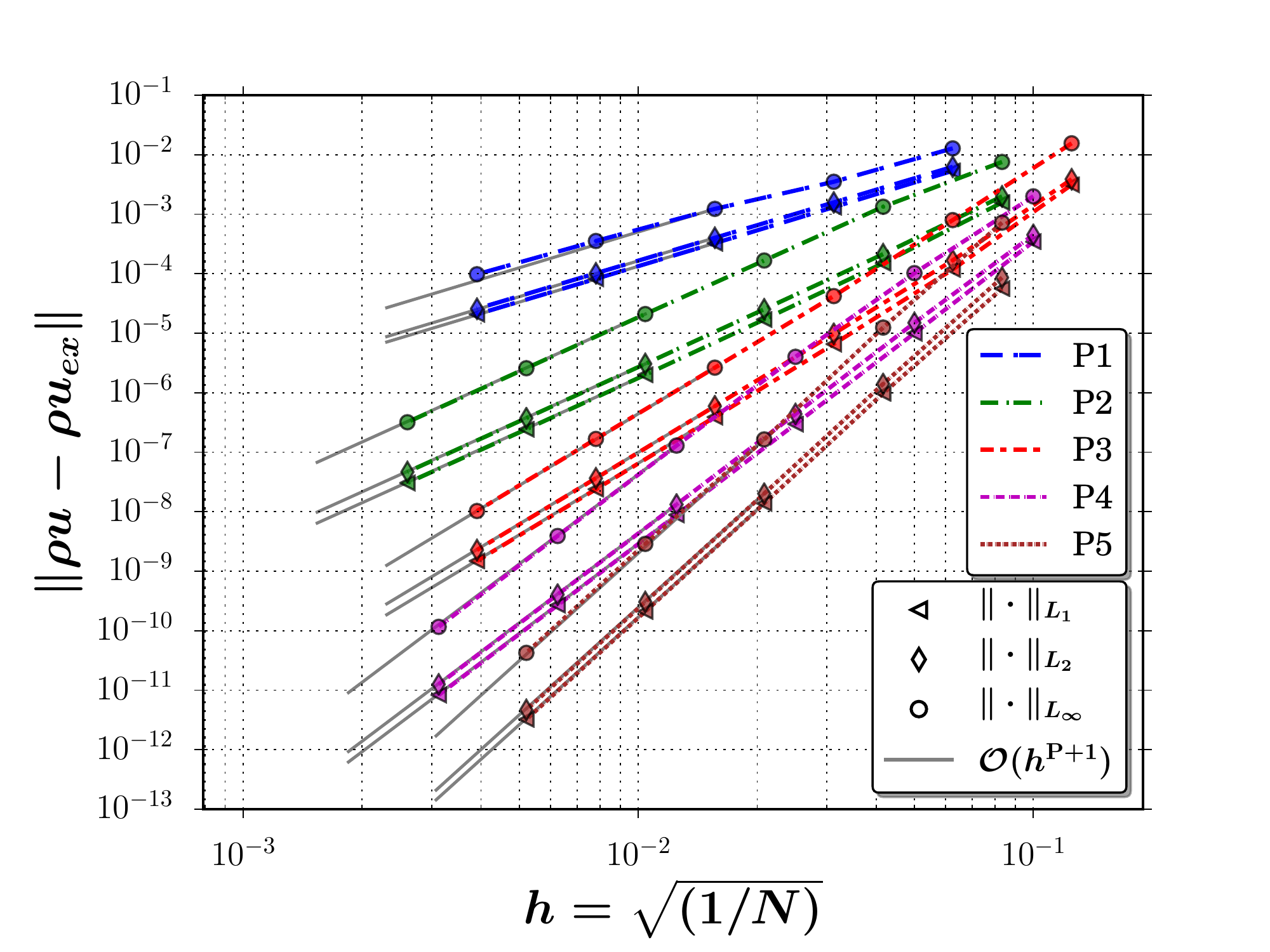}}
\vfill
\subfloat[$\rho v$]{
\includegraphics[trim = 5mm 2mm 18mm 13mm, clip,width=0.32\linewidth]
{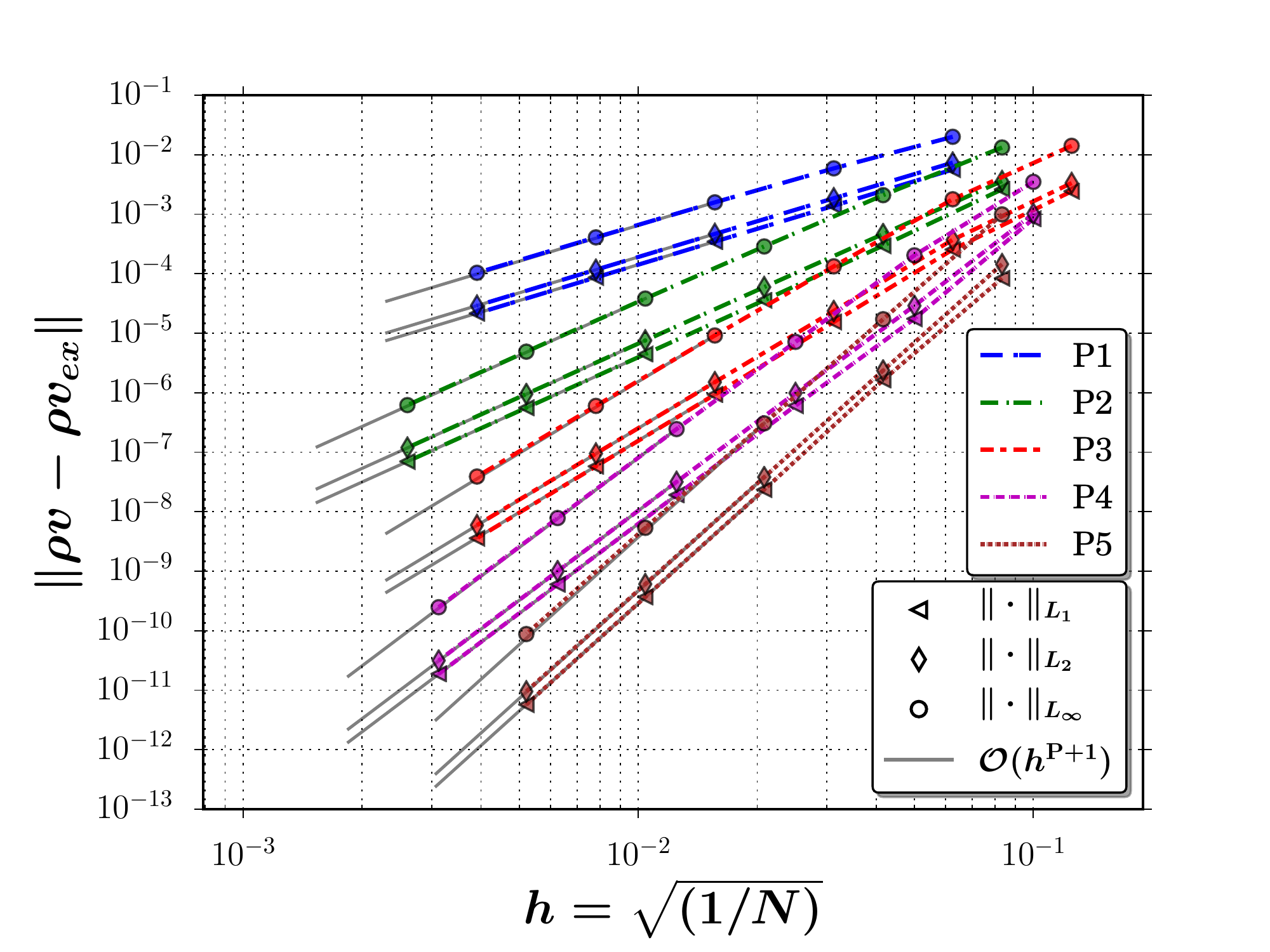}}
~~~
\subfloat[$\rho E$]{
\includegraphics[trim = 5mm 2mm 18mm 13mm, clip,width=0.32\linewidth]{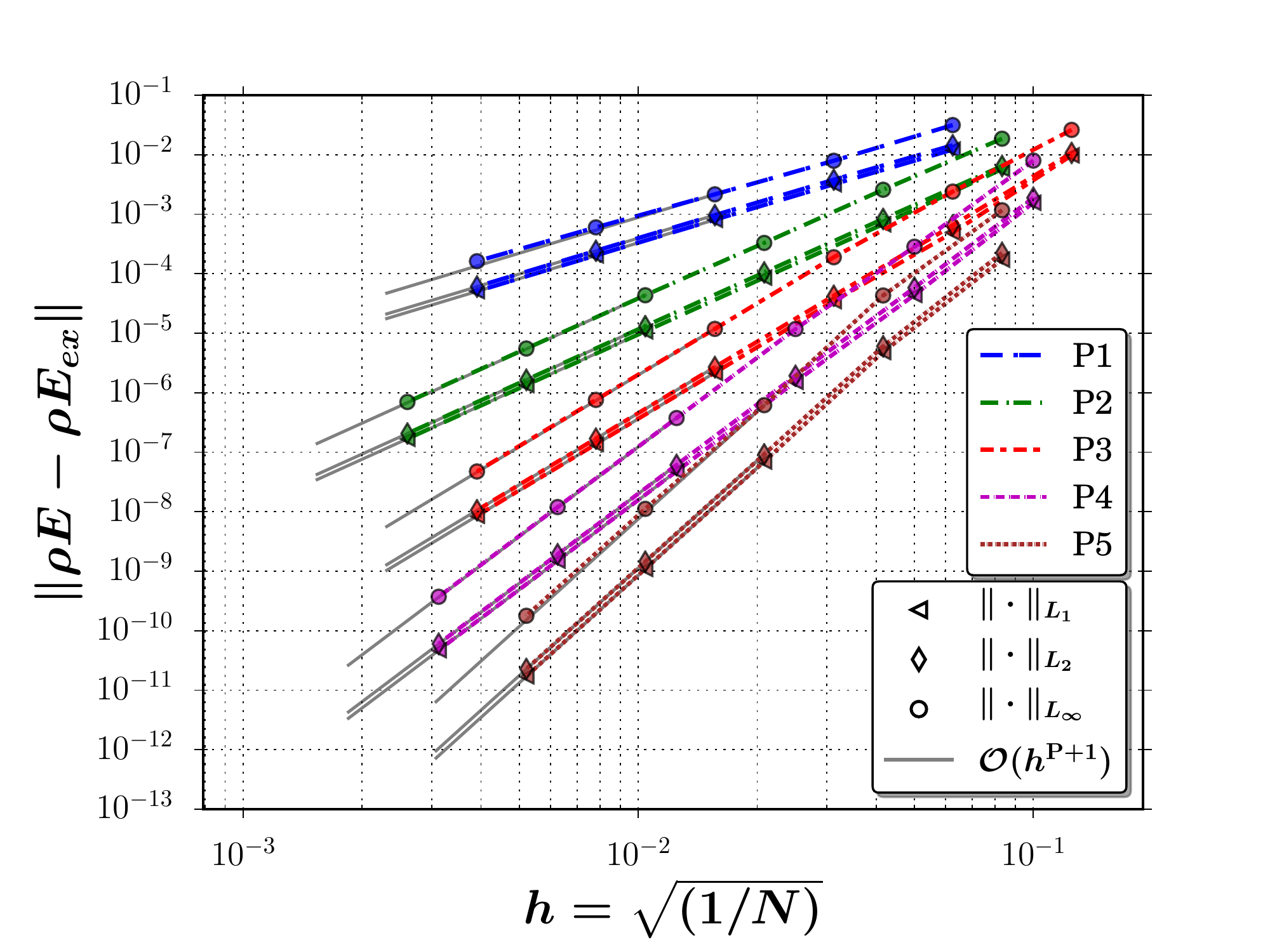}}
\caption{Evolution of the discretization error in $L_1$, $L_2$ and $L_\infty$ norms versus mesh refinement for MS-2 with polynomial degrees $\mathrm{P}1$--$\mathrm{P}5$}
\label{fig:Err_allE_allP_MS-2}
\end{figure}
\begin{figure}[!hbt]
\centering
\subfloat[$\rho$]{ 
\includegraphics[trim = 16mm 3mm 18mm 13mm, clip,width=0.3\linewidth]
{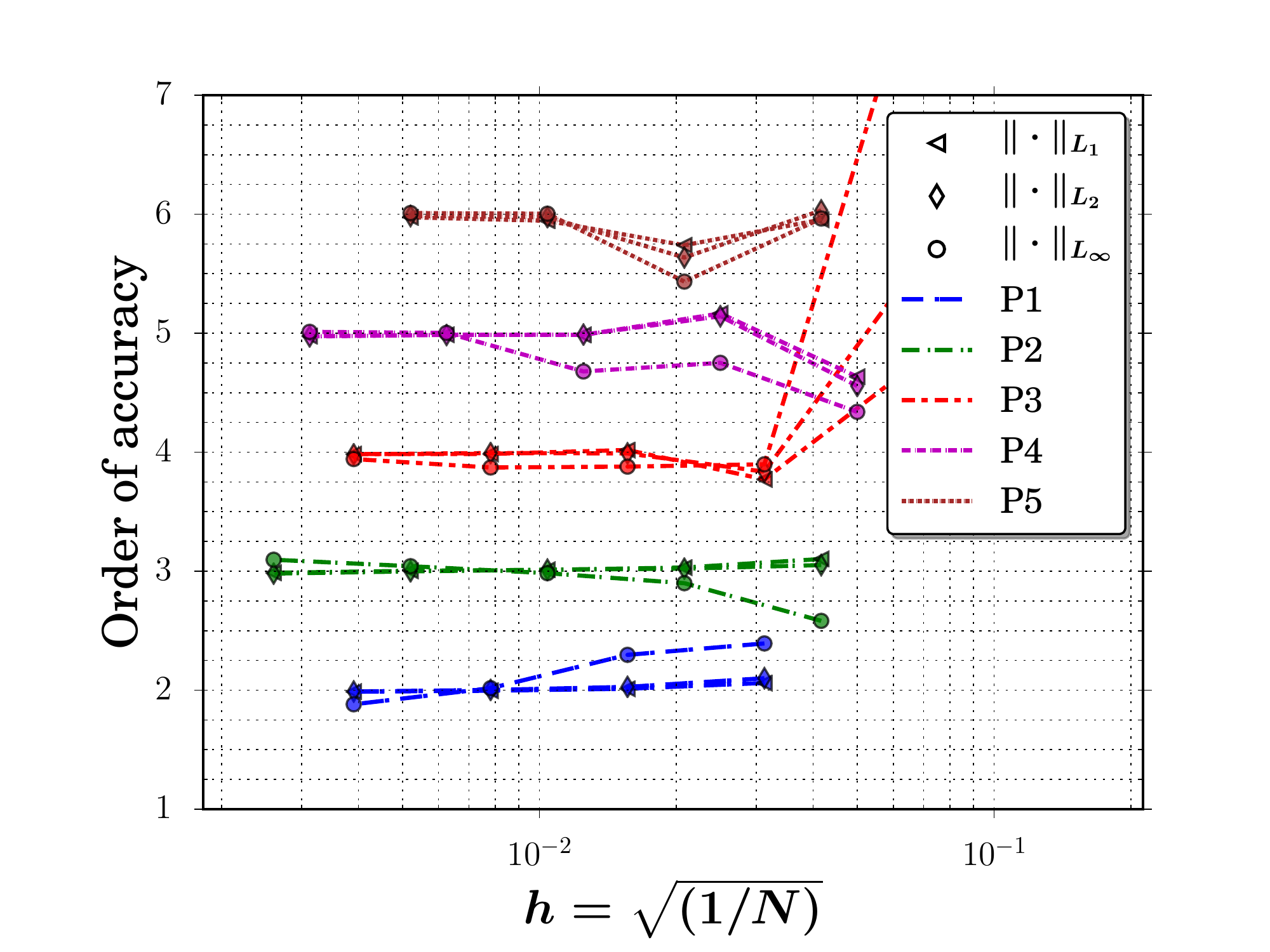}}
~~~
\subfloat[$\rho u$]{
\includegraphics[trim = 16mm 3mm 18mm 13mm, clip,width=0.3\linewidth]{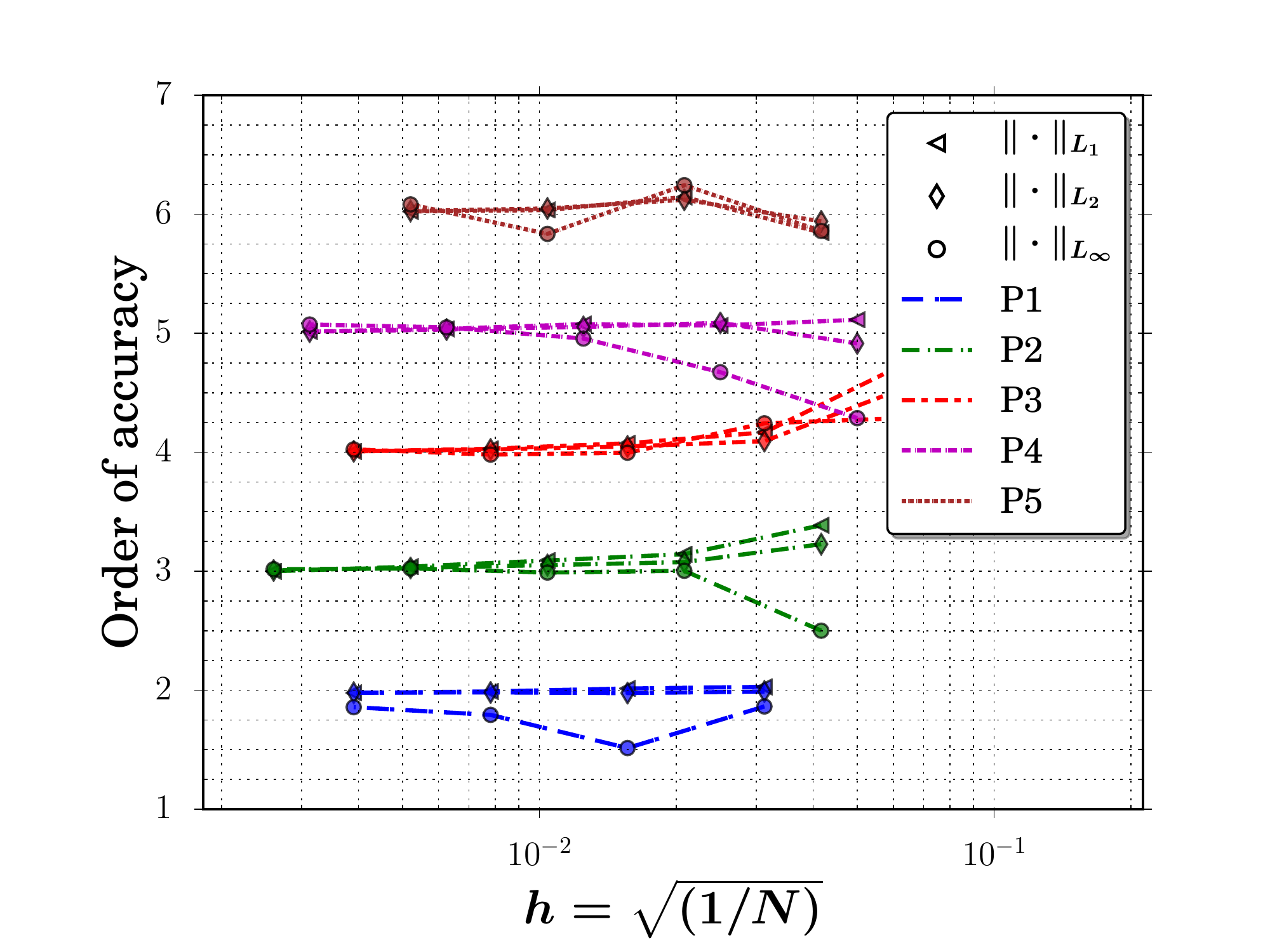}}
\vfill
\subfloat[$\rho v$]{
\includegraphics[trim = 16mm 3mm 18mm 13mm, clip,width=0.3\linewidth]
{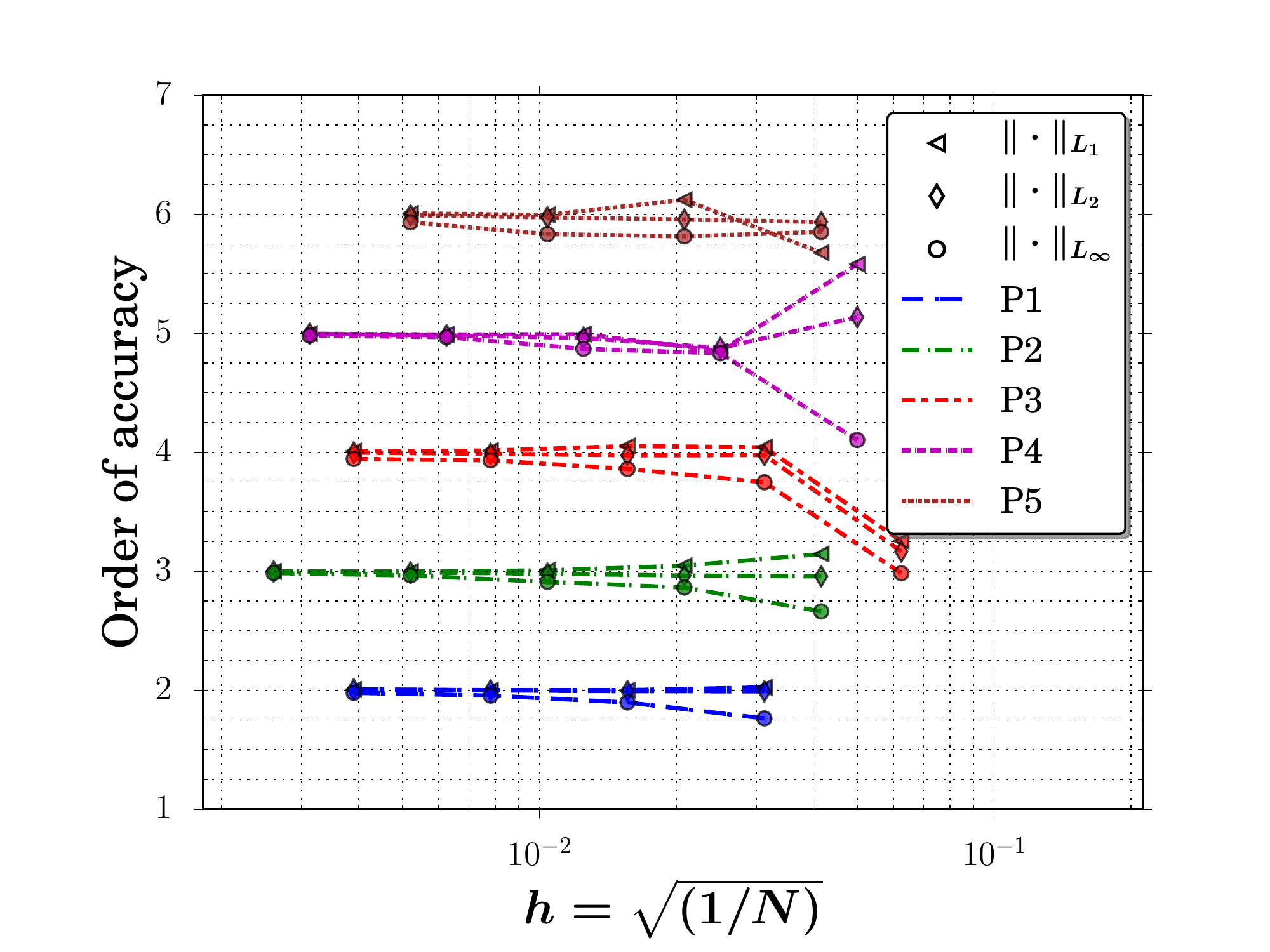}}
~~~
\subfloat[$\rho E$]{
\includegraphics[trim = 16mm 3mm 18mm 13mm, clip,width=0.3\linewidth]{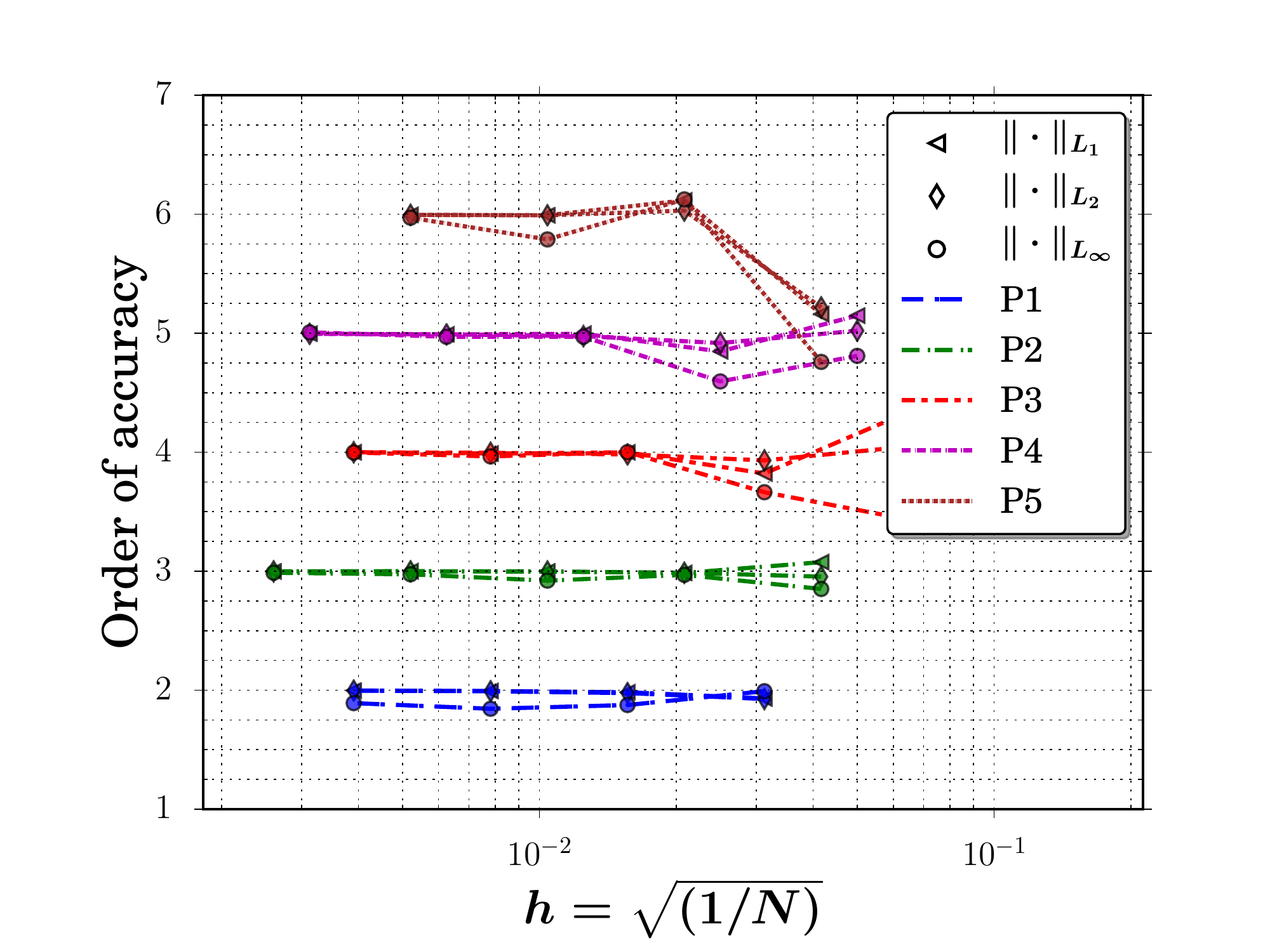}}
\caption{Evolution of the OOAs in $L_1$, $L_2$ and $L_\infty$ norms versus mesh refinement for MS-2 with polynomial degrees $\mathrm{P}1$--$\mathrm{P}5$}
\label{fig:Orders_MS-2}
\end{figure}

\begin{figure}[!hbt]
\centering
\subfloat[$\rho$]{
\includegraphics[trim = 5mm 2mm 18mm 13mm, clip,width=0.32\linewidth]
{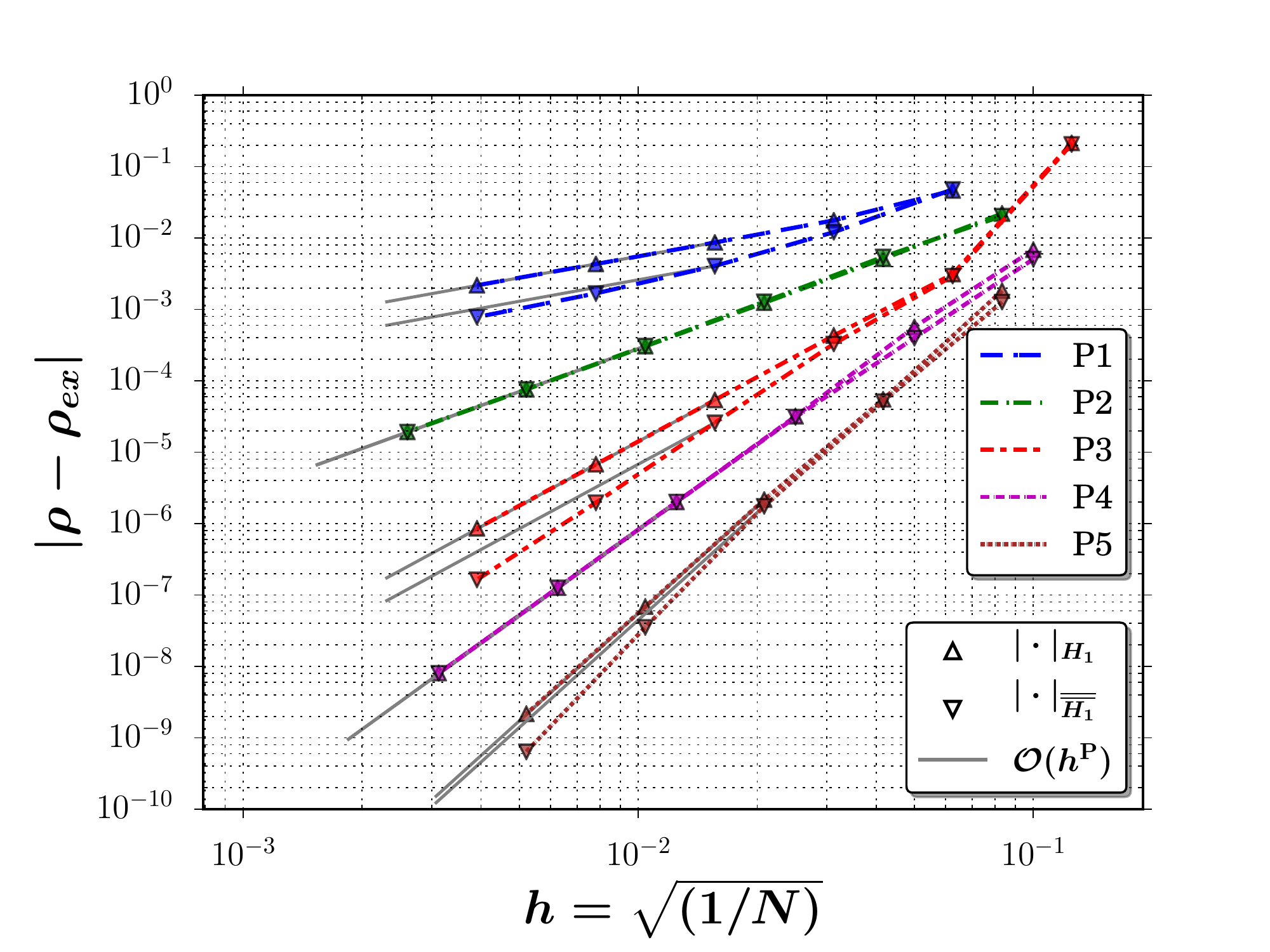}}
~~~
\subfloat[$\rho u$]{
\includegraphics[trim = 5mm 2mm 18mm 13mm, clip,width=0.32\linewidth]{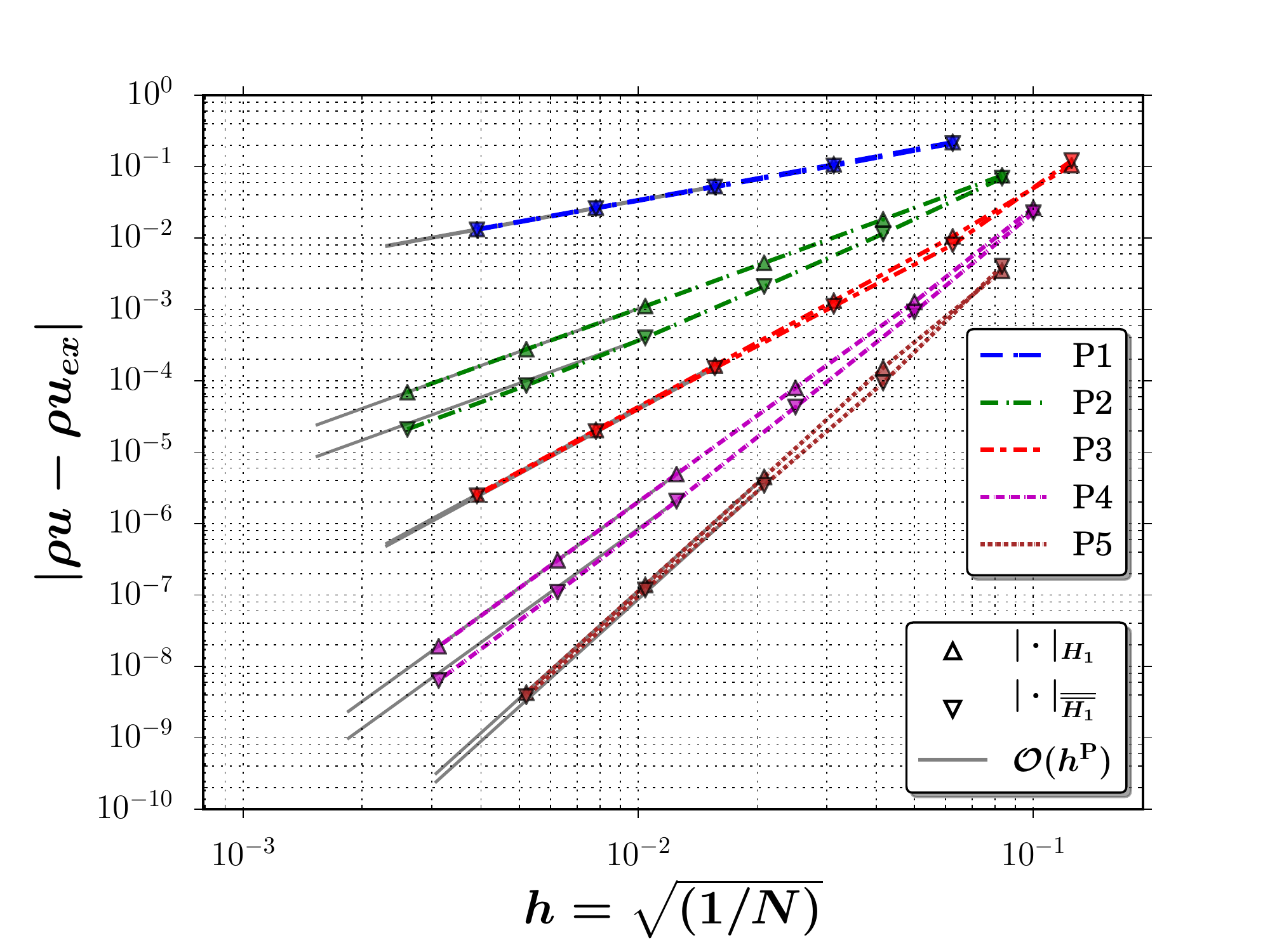}}
\vfill
\subfloat[$\rho v$]{
\includegraphics[trim = 5mm 2mm 18mm 13mm, clip,width=0.32\linewidth]
{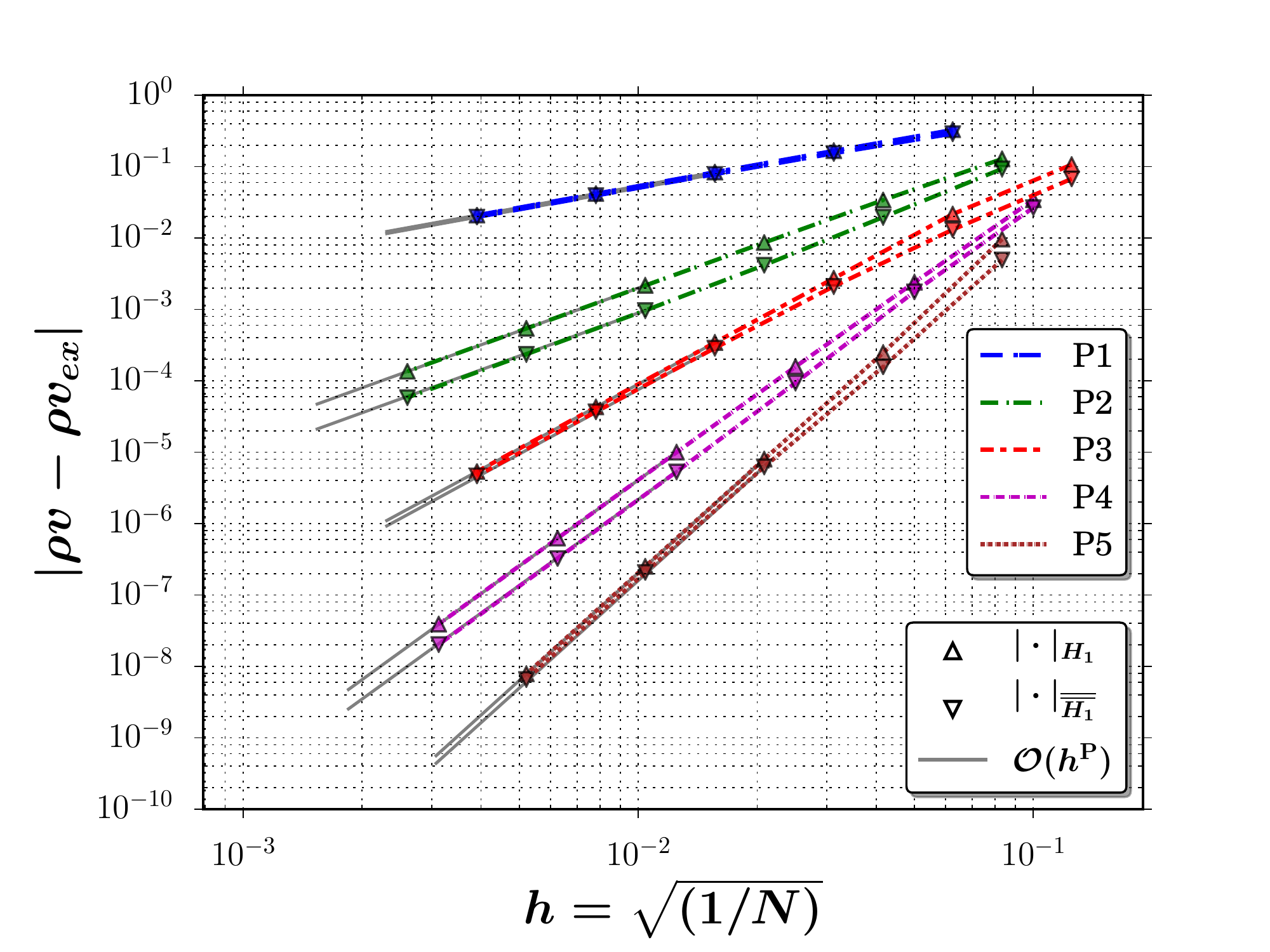}}
~~~
\subfloat[$\rho E$]{
\includegraphics[trim = 5mm 2mm 18mm 13mm, clip,width=0.32\linewidth]{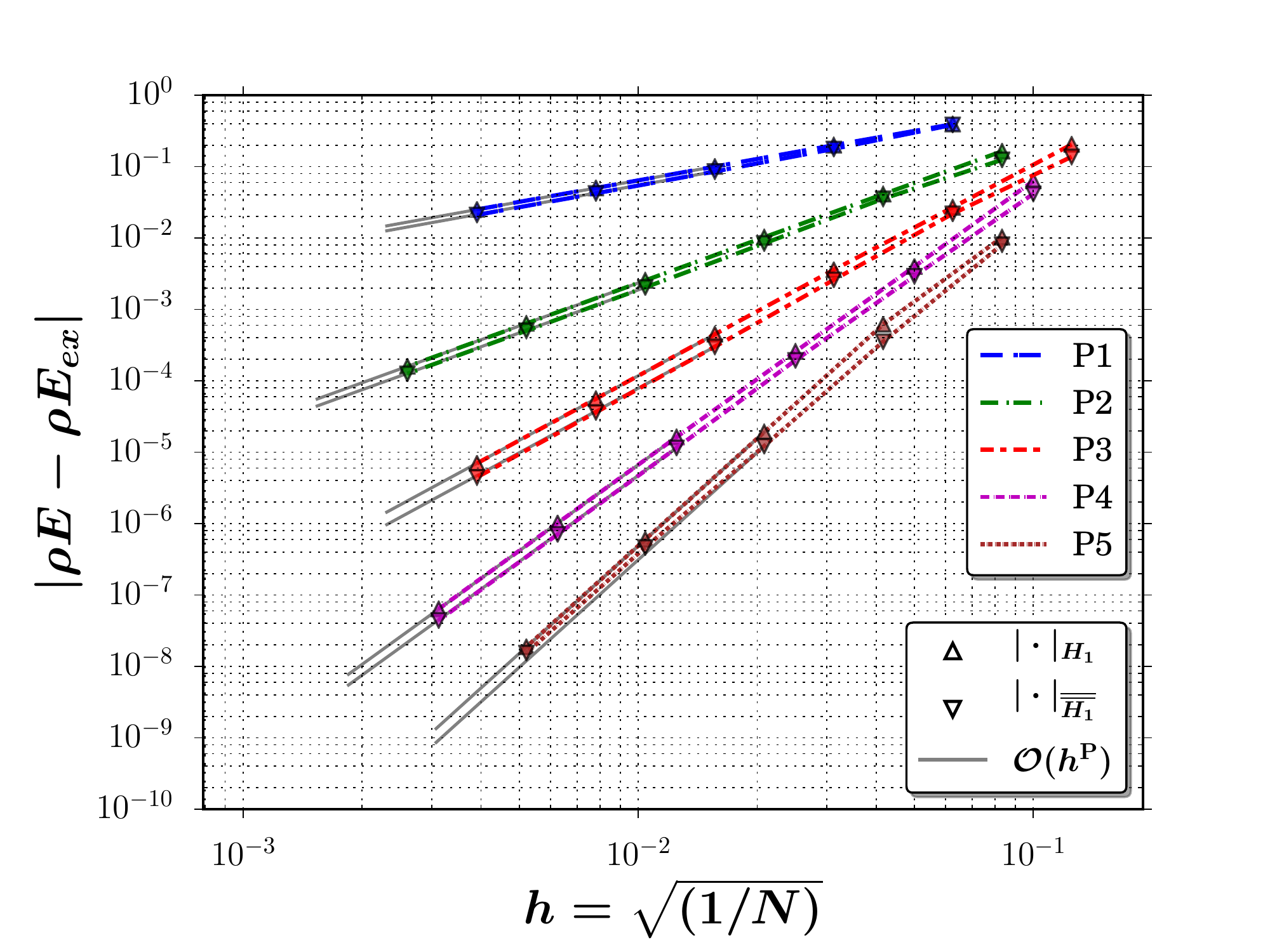}}
\caption{Evolution of the discretization error in $H_1$ semi-norm (for uncorrected and  fully corrected derivatives) versus mesh refinement for MS-2 and  $\mathrm{P}1$--$\mathrm{P}5$}
\label{fig:Err_allE_allP_H_MS-2}
\end{figure}

\begin{figure}[!hbt]
\centering
\subfloat[$\rho$]{ 
\includegraphics[trim = 16mm 3mm 18mm 13mm, clip,width=0.3\linewidth]
{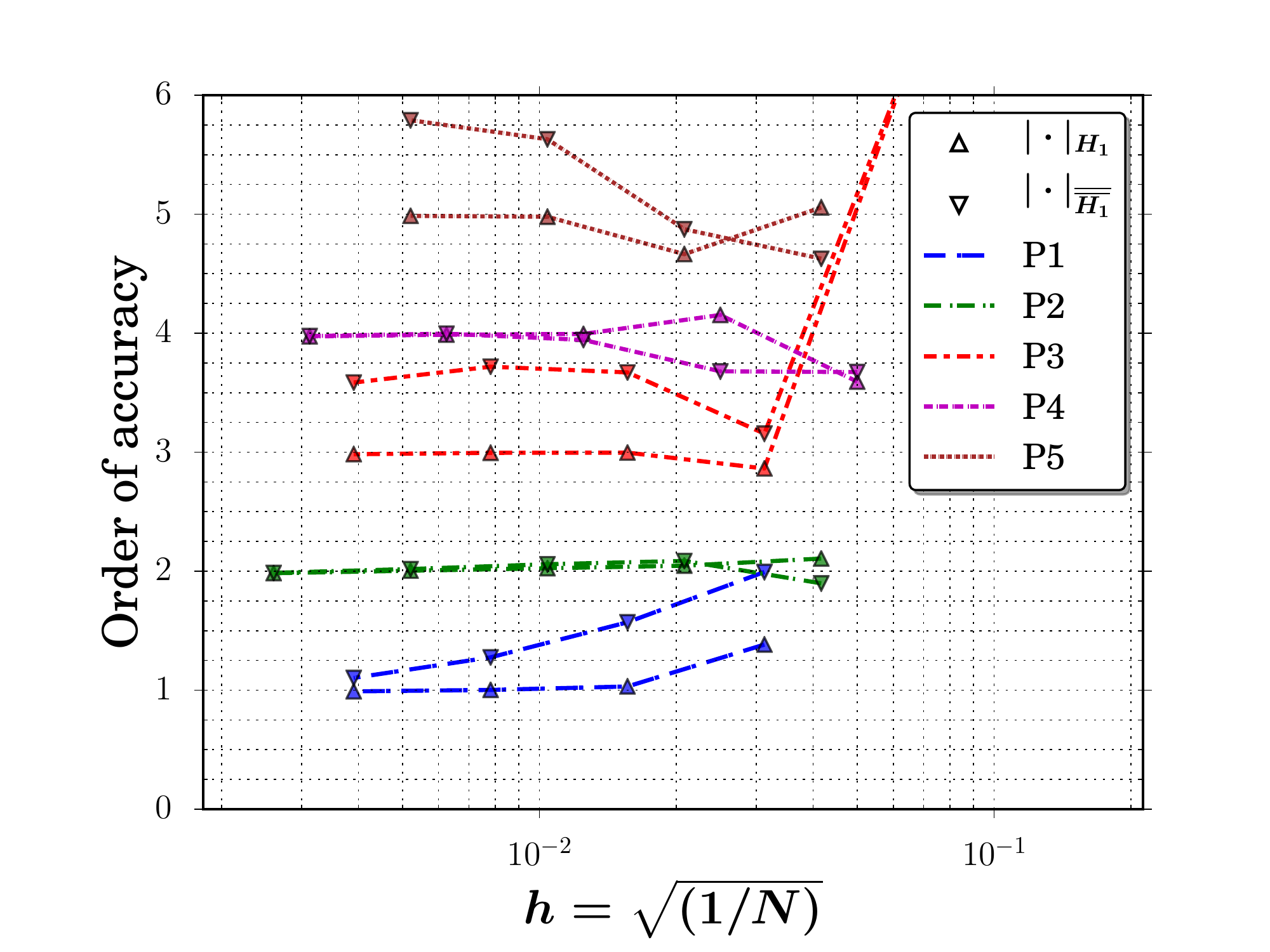}}
~~~
\subfloat[$\rho u$]{
\includegraphics[trim = 16mm 3mm 18mm 13mm, clip,width=0.3\linewidth]{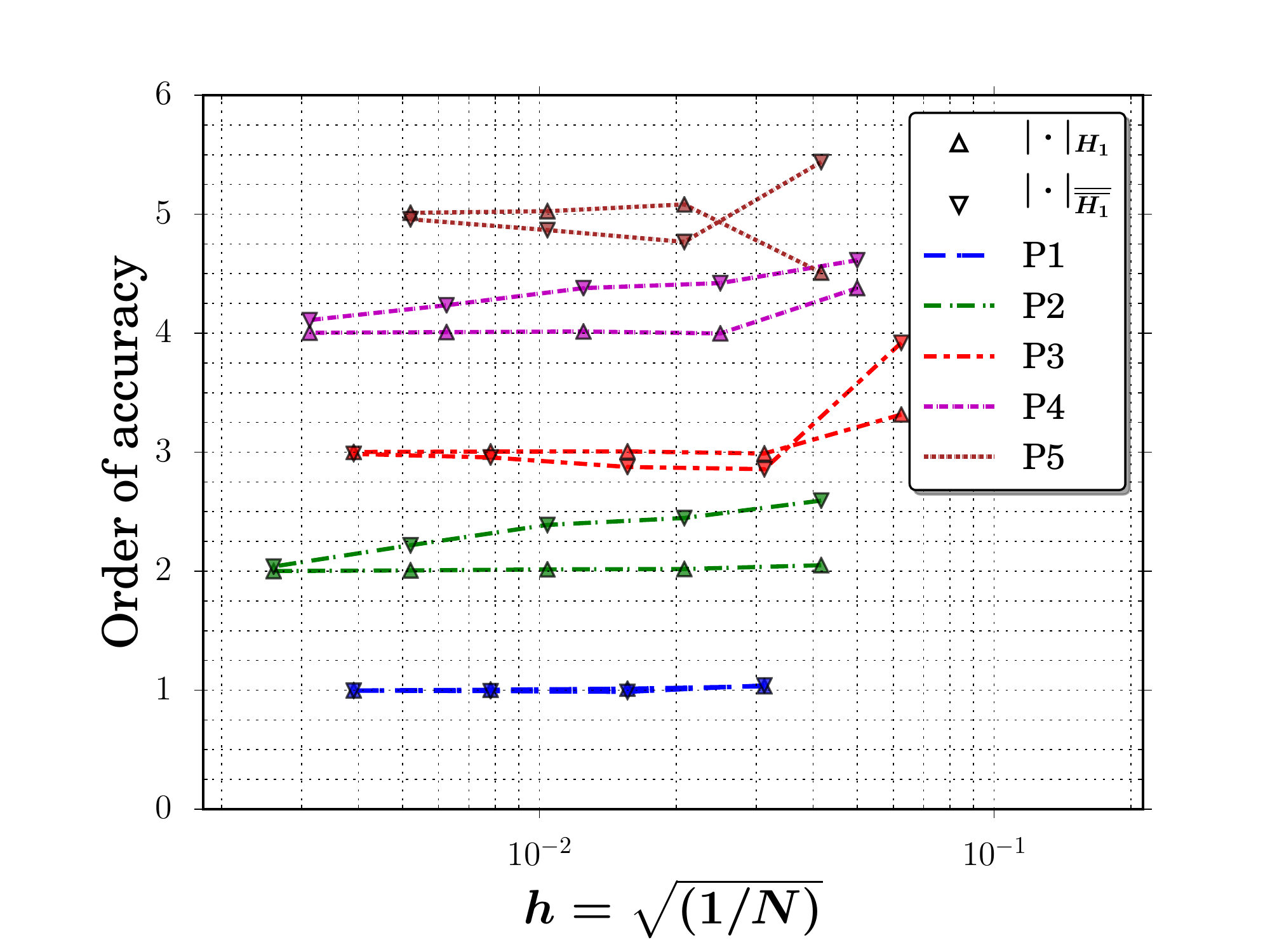}}
\vfill
\subfloat[$\rho v$]{
\includegraphics[trim = 16mm 3mm 18mm 13mm, clip,width=0.3\linewidth]
{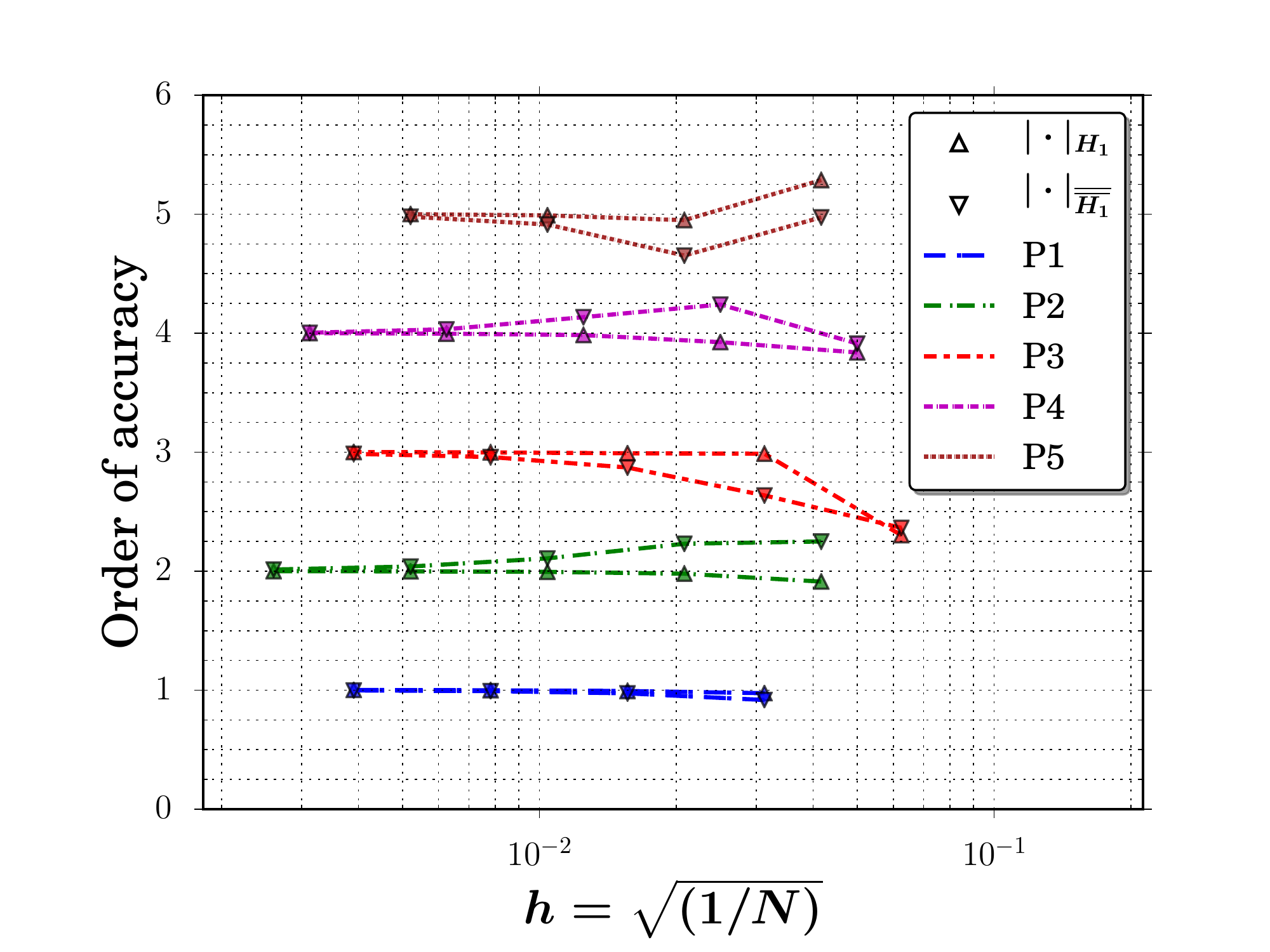}}
~~~
\subfloat[$\rho E$]{
\includegraphics[trim = 16mm 3mm 18mm 13mm, clip,width=0.3\linewidth]{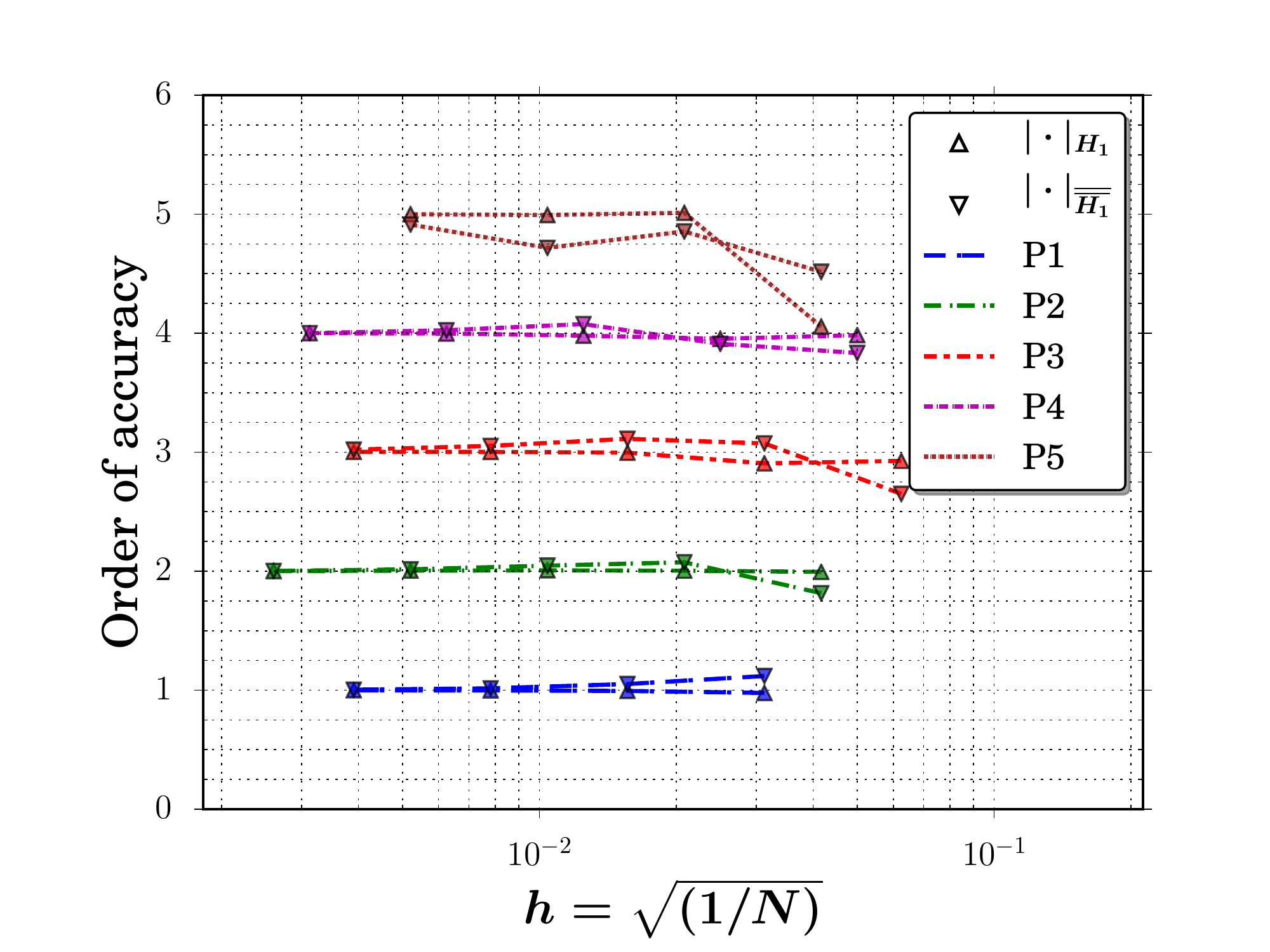}}
\caption{Evolution of the OOAs in $H_1$ semi-norm (for uncorrected and fully corrected derivatives) versus mesh refinement for MS-2 and  $\mathrm{P}1$--$\mathrm{P}5$}
\label{fig:Orders_H_MS-2}
\end{figure}

\clearpage
\subsection{MS-3}
\begin{figure}[!hbt]
\centering
\vspace{1mm}
\subfloat[$\rho^{\mathrm{MS}}$]{
\includegraphics[trim = 0mm 0mm 0mm 0mm, clip,width=0.33\linewidth]
{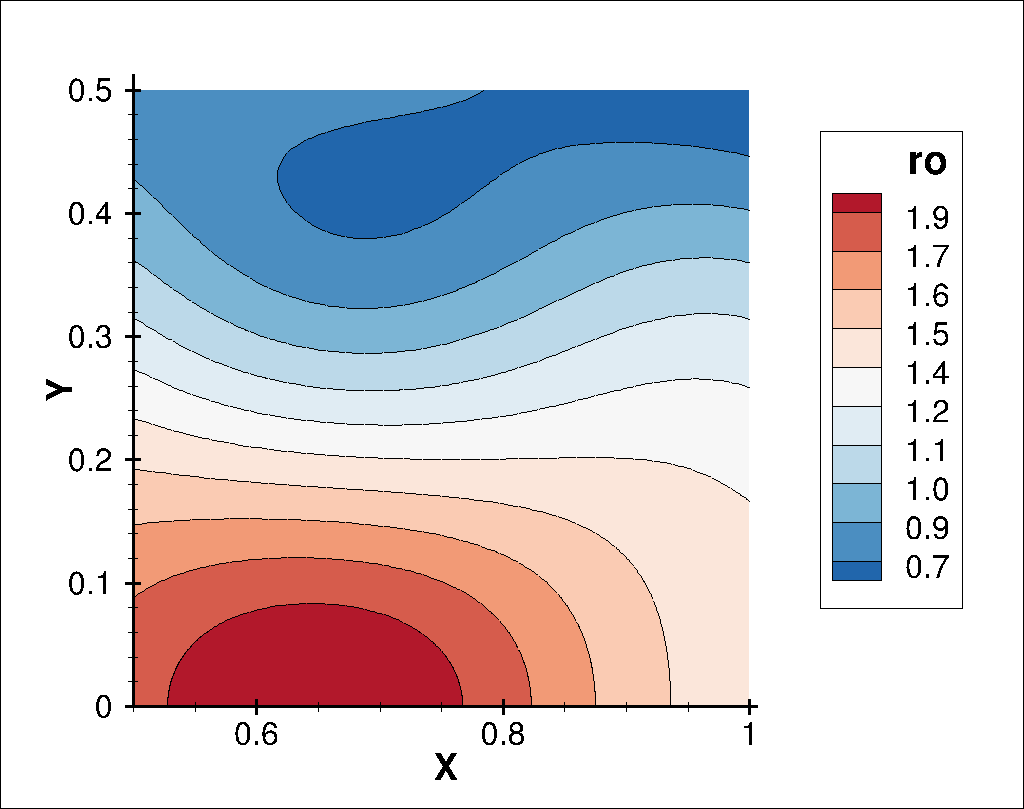}}
~~~
\subfloat[$u^{\mathrm{MS}}$]{
\includegraphics[trim = 0mm 0mm 0mm 0mm, clip,width=0.33\linewidth]
{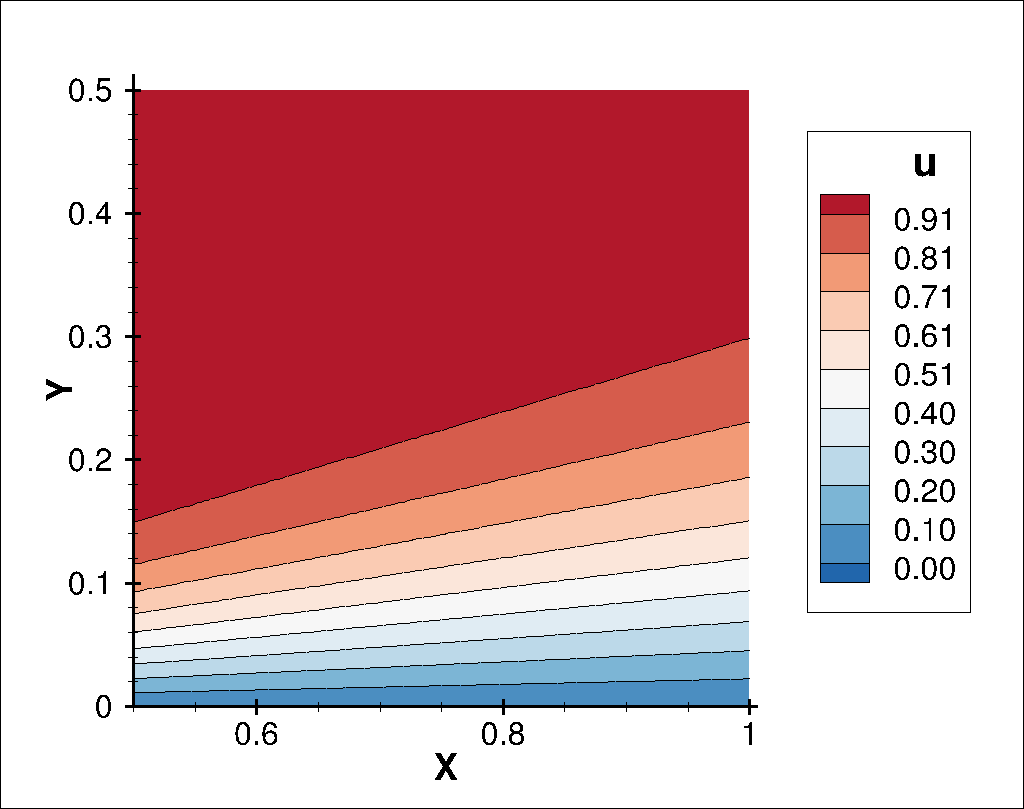}}
\vfill
\subfloat[$v^{\mathrm{MS}}$]{
\includegraphics[trim = 0mm 0mm 0mm 0mm, clip,width=0.33\linewidth]{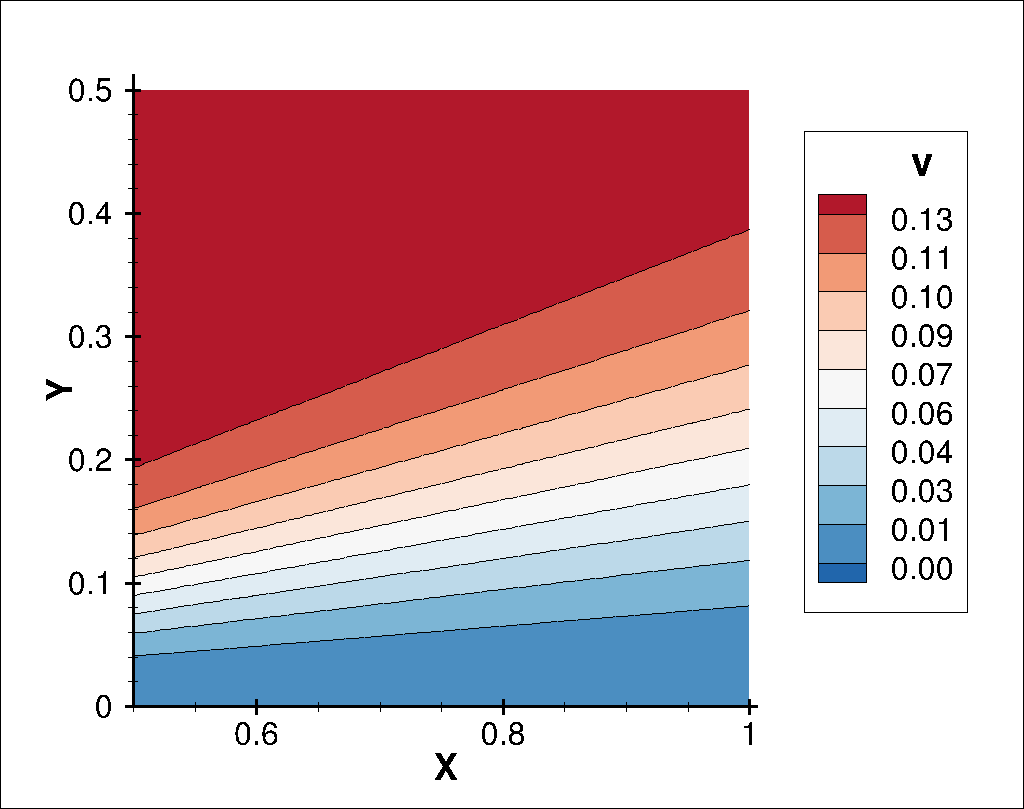}}
~~~
\subfloat[$p^{\mathrm{MS}}$]{
\includegraphics[trim = 0mm 0mm 0mm 0mm, clip,width=0.33\linewidth]{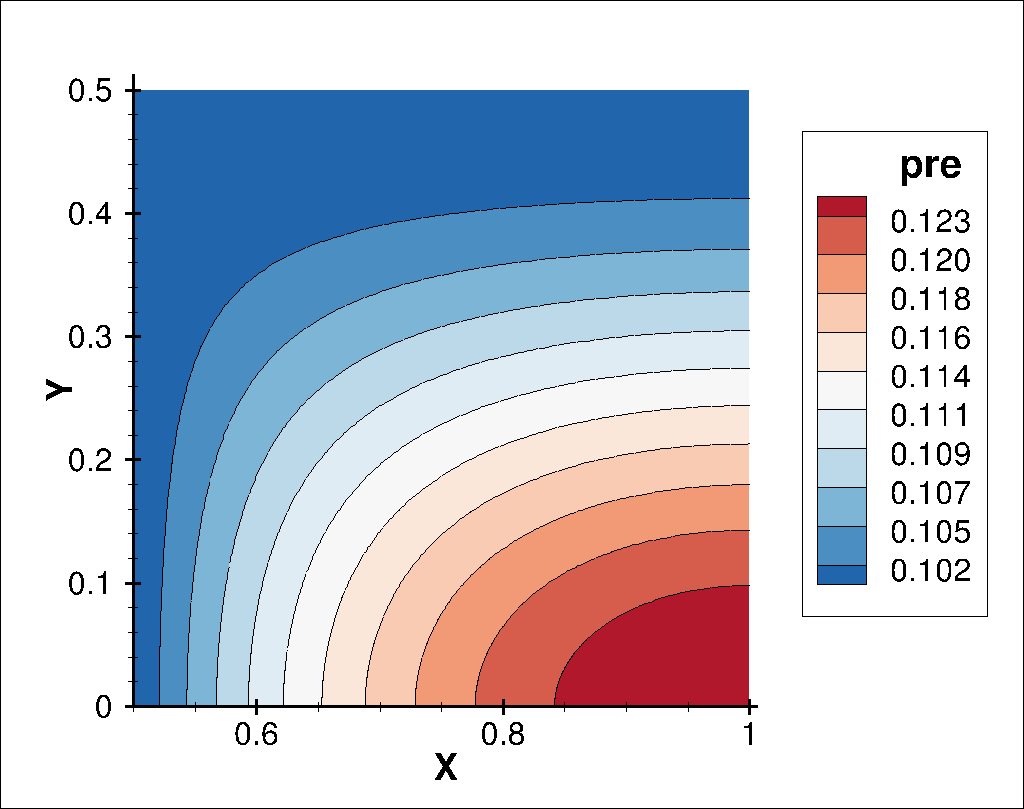}}
\vfill
\subfloat[${\tilde \nu}^{\mathrm{MS}}$]{
\includegraphics[trim = 0mm 0mm 0mm 0mm, clip,width=0.33\linewidth]{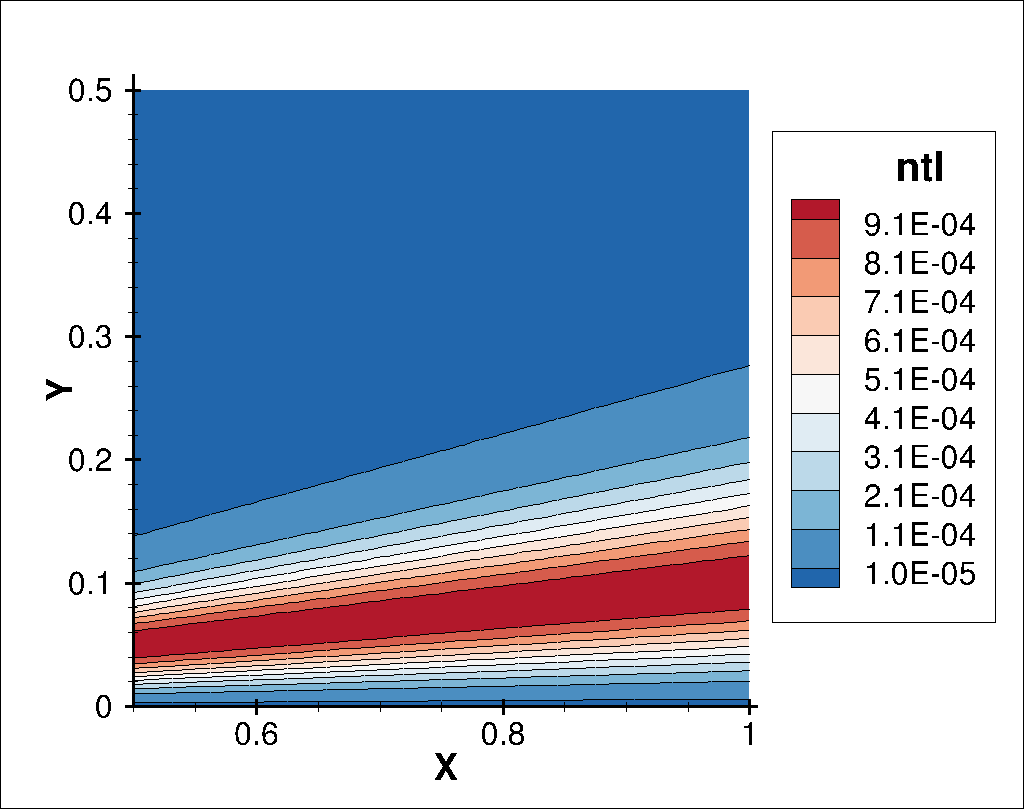}}
~~~
\subfloat[${Ma}^{\mathrm{MS}}$ and ${\bm{u}}^{\mathrm{MS}}$]{
\includegraphics[trim = 0mm 0mm 0mm 0mm, clip,width=0.33\linewidth]{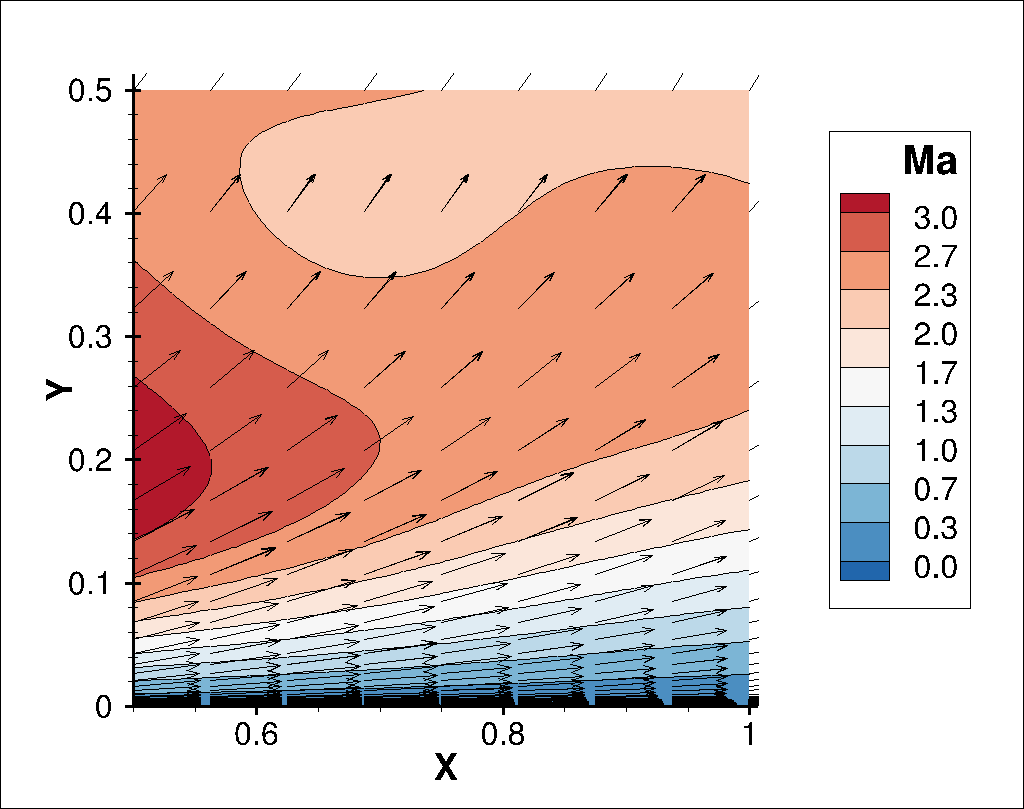}}
\caption{Manufactured solution MS-3}
\label{fig:MS-3}
\end{figure} 
 
\begin{figure}[!hbt]
\centering
\subfloat[$\rho$]{
\includegraphics[trim = 5mm 2mm 18mm 13mm, clip,width=0.32\linewidth]
{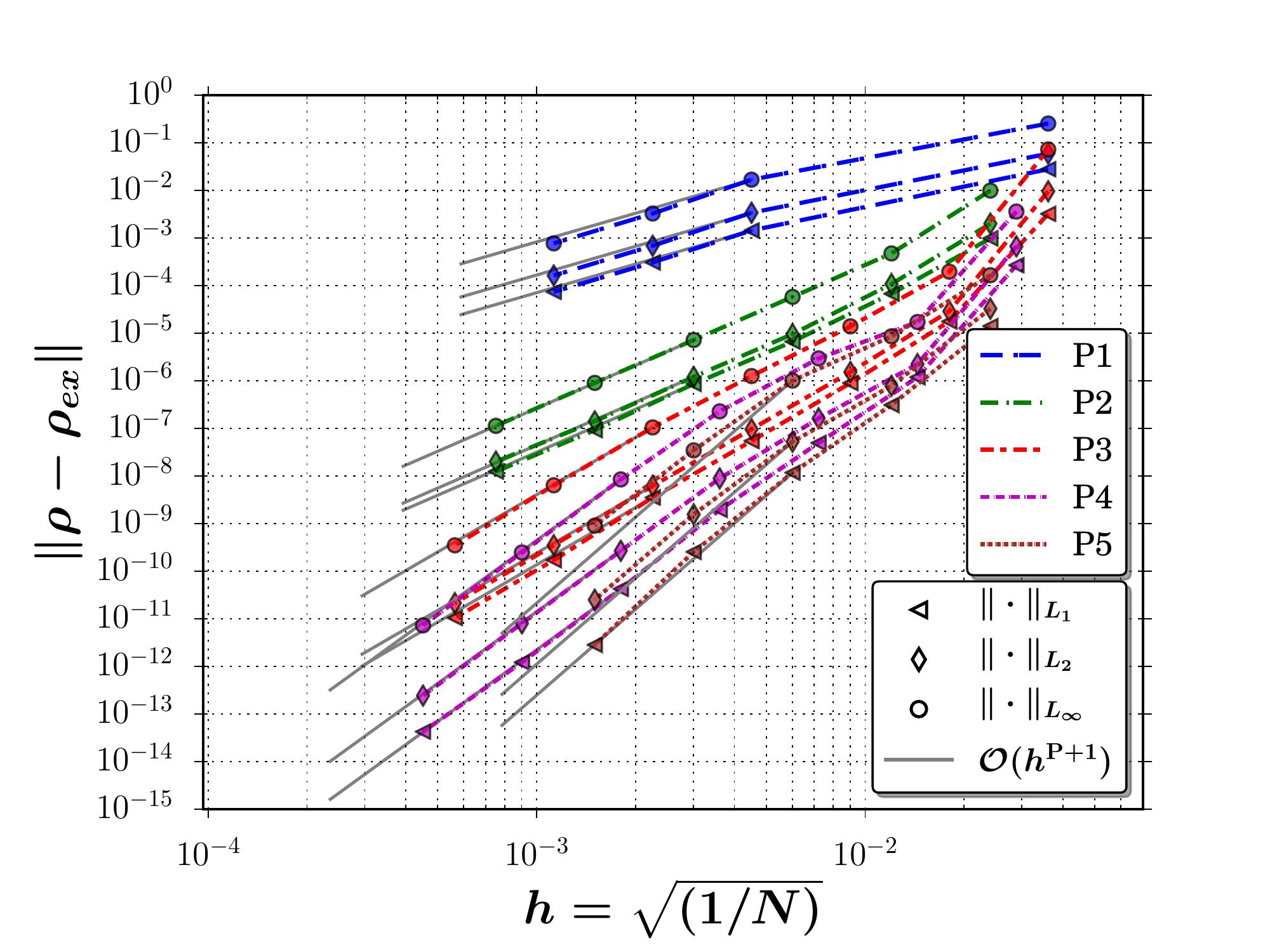}}
~~~
\subfloat[$\rho u$]{
\includegraphics[trim = 5mm 2mm 18mm 13mm, clip,width=0.32\linewidth]{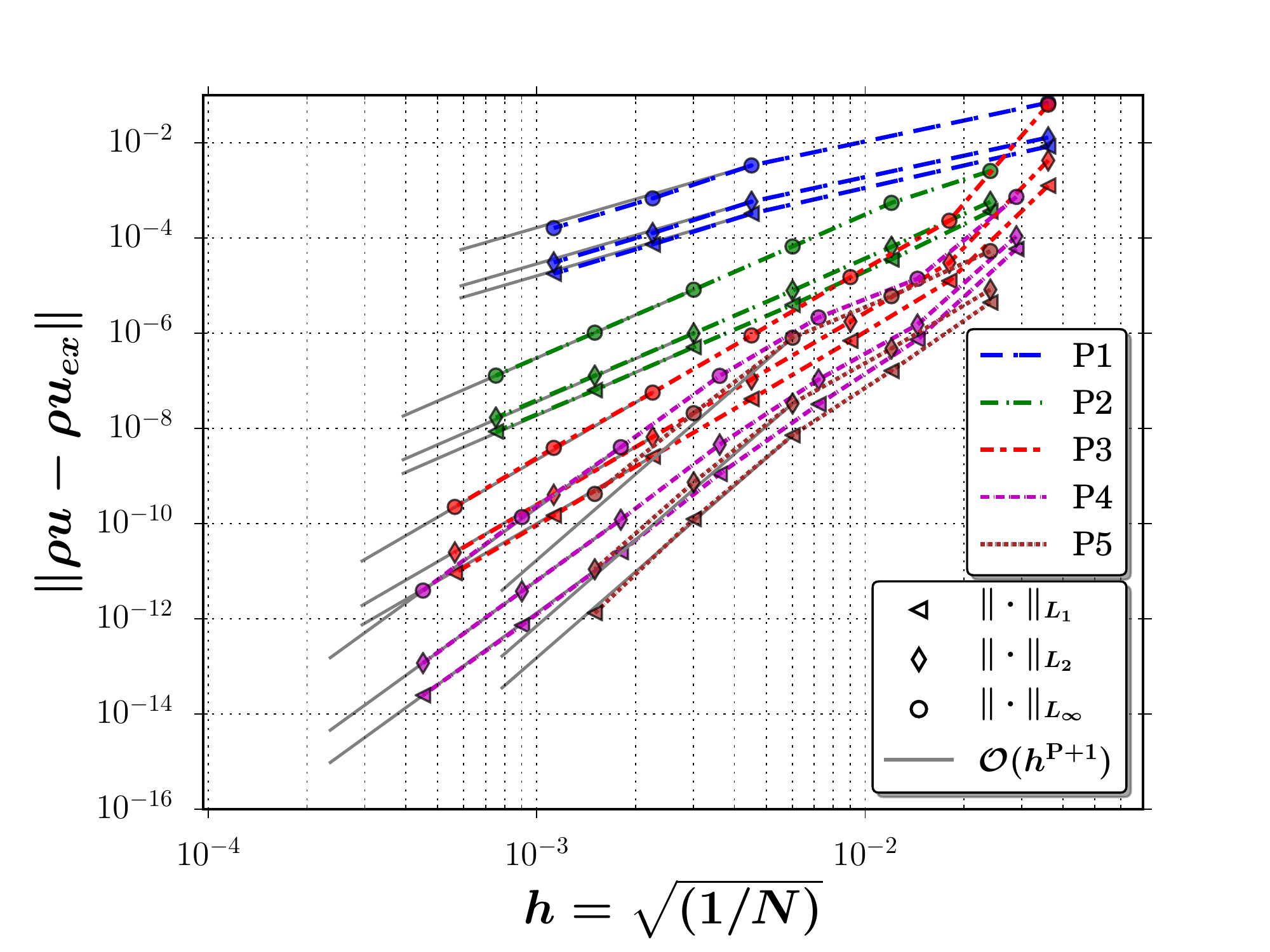}}
\vfill
\subfloat[$\rho v$]{
\includegraphics[trim =5mm 2mm 18mm 13mm, clip,width=0.32\linewidth]
{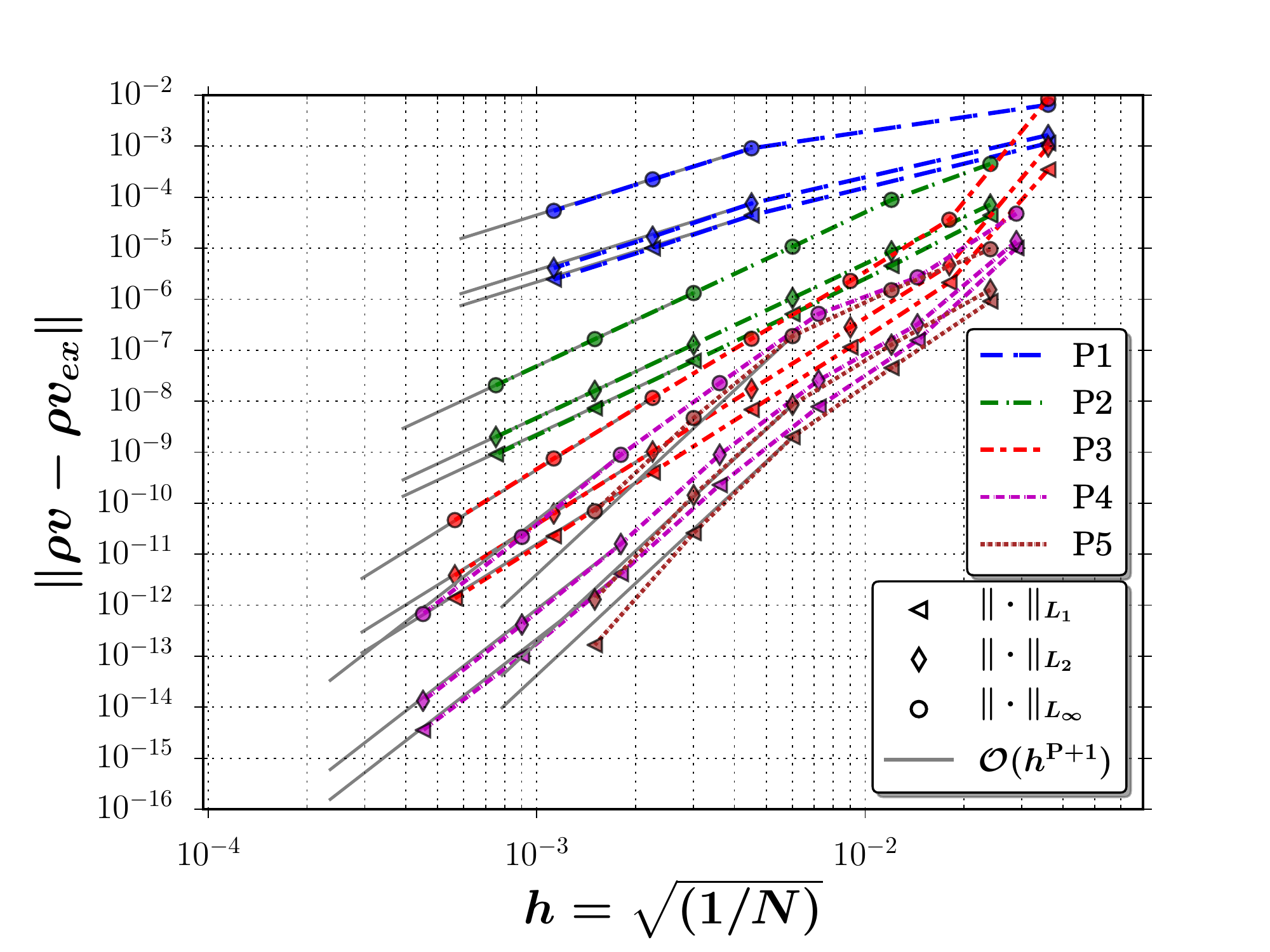}}
~~~
\subfloat[$\rho E$]{
\includegraphics[trim = 5mm 2mm 18mm 13mm, clip,width=0.32\linewidth]{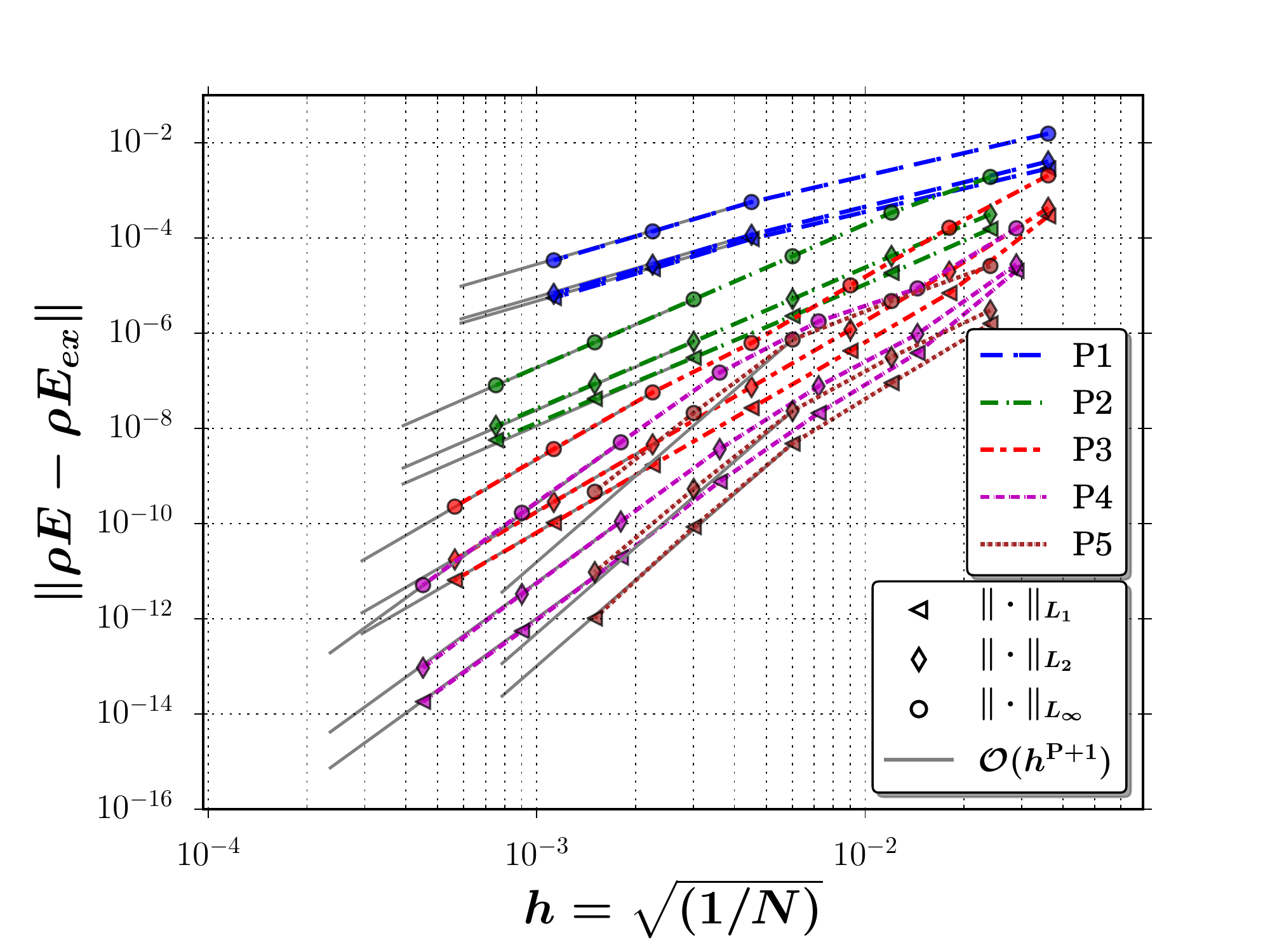}}
\vfill
\subfloat[$\rho \tilde{\nu}$]{
\includegraphics[trim = 5mm 2mm 18mm 13mm, clip,width=0.32\linewidth]
{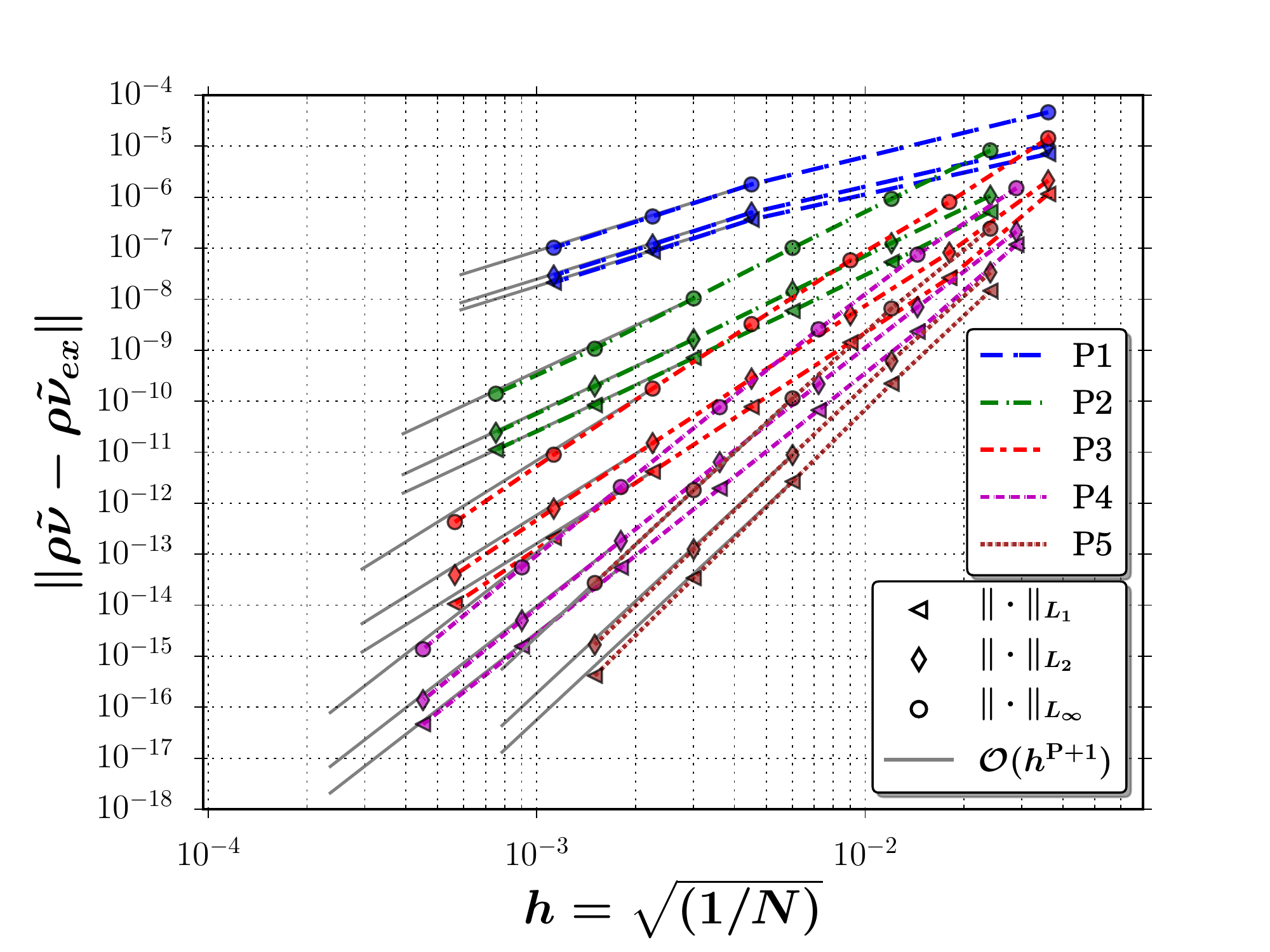}}
\caption{Evolution of the discretization error in $L_1$, $L_2$ and $L_\infty$ norms versus mesh refinement for MS-3 and  $\mathrm{P}1$--$\mathrm{P}5$}
\label{fig:Err_allE_allP_MS-3}
\end{figure}

\begin{figure}[!hbt]
\centering
\subfloat[$\rho$]{ 
\includegraphics[trim = 16mm 3mm 18mm 13mm, clip,width=0.3\linewidth]
{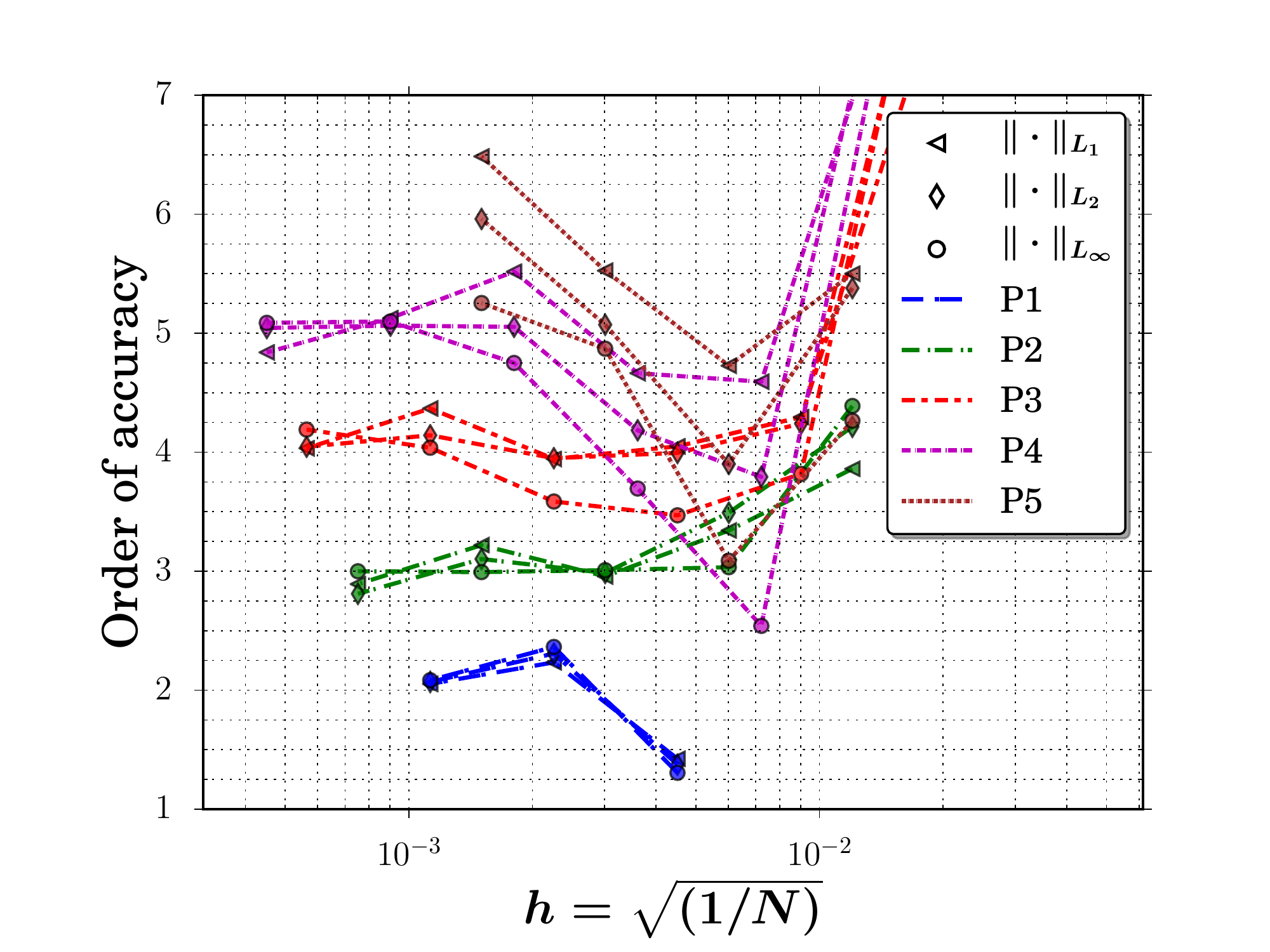}}
~~~
\subfloat[$\rho u$]{
\includegraphics[trim = 16mm 3mm 18mm 13mm, clip,width=0.3\linewidth]{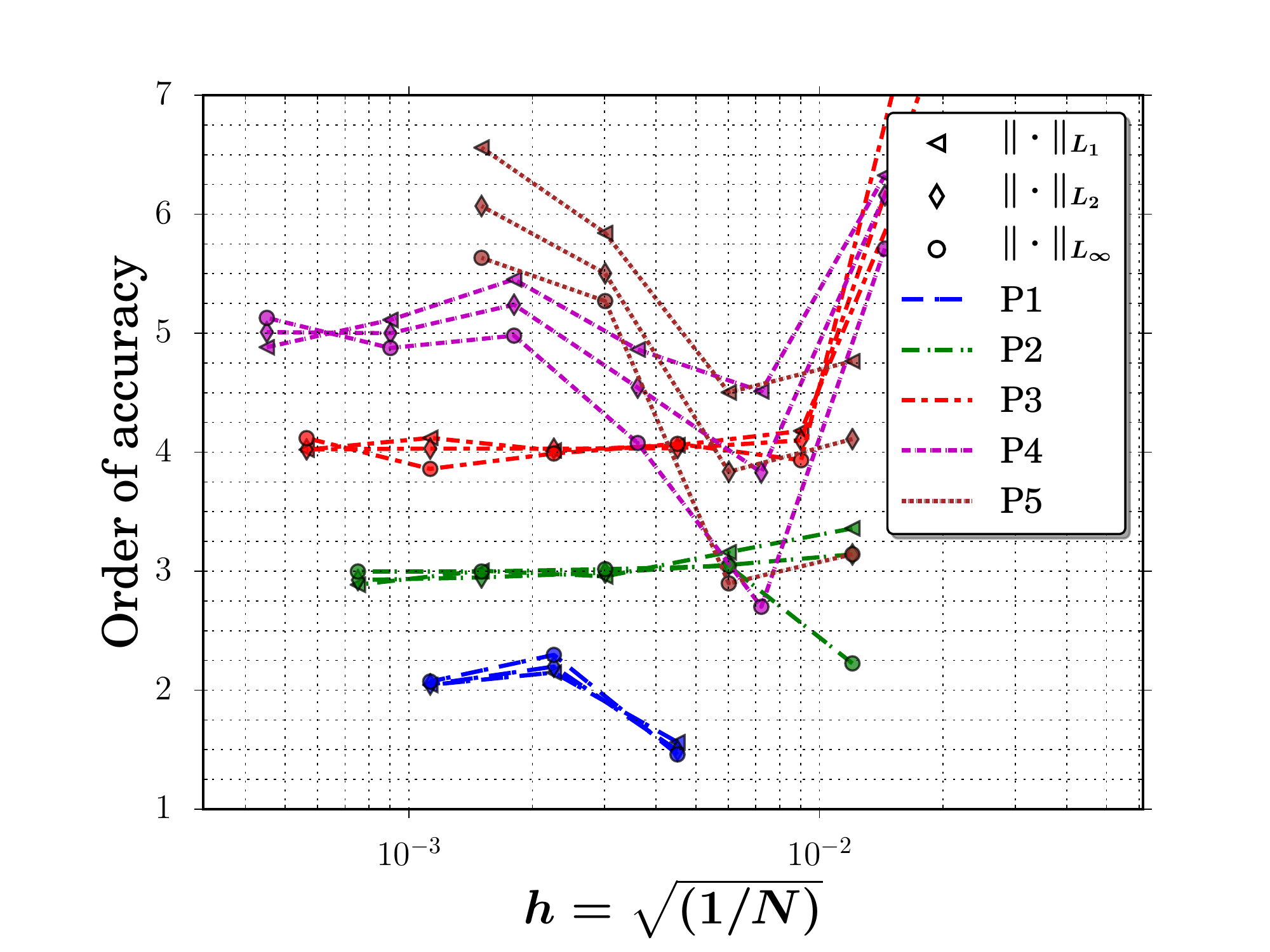}}
\vfill
\subfloat[$\rho v$]{
\includegraphics[trim = 16mm 3mm 18mm 13mm, clip,width=0.3\linewidth]
{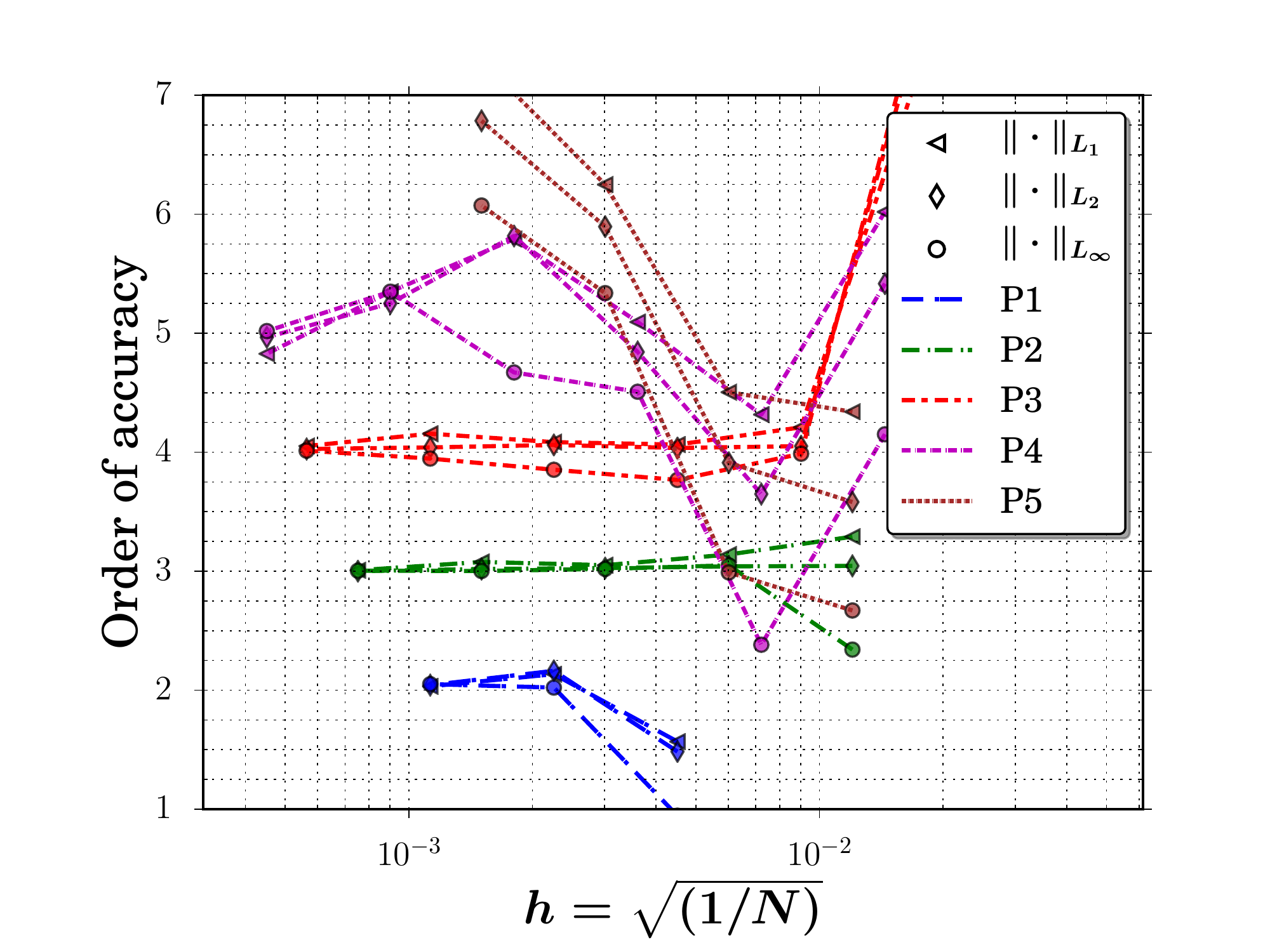}}
~~~
\subfloat[$\rho E$]{
\includegraphics[trim = 16mm 3mm 18mm 13mm, clip,width=0.3\linewidth]{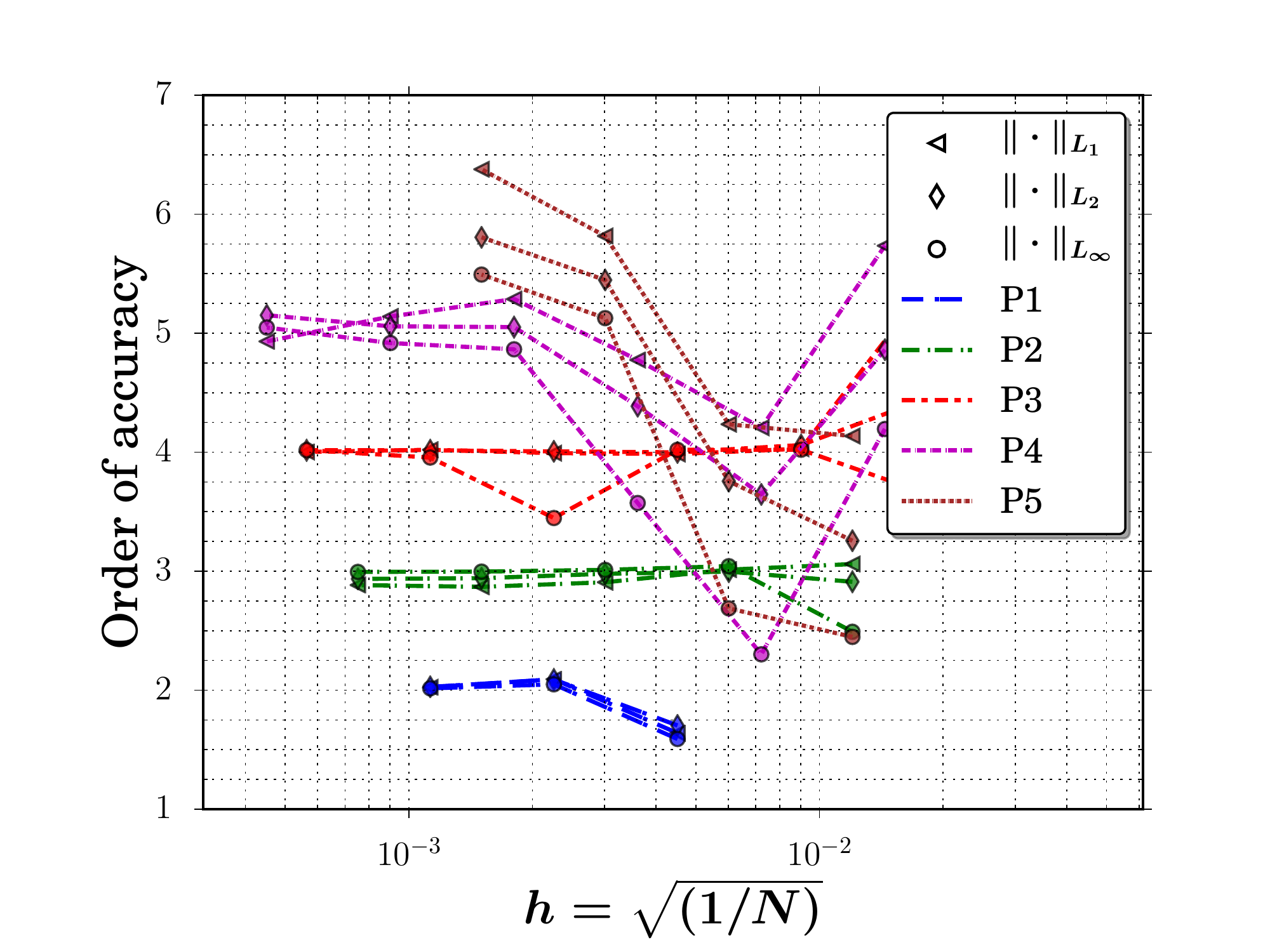}}
\vfill
\subfloat[$\rho \tilde{\nu}$]{
\includegraphics[trim = 16mm 3mm 18mm 13mm, clip,width=0.3\linewidth]{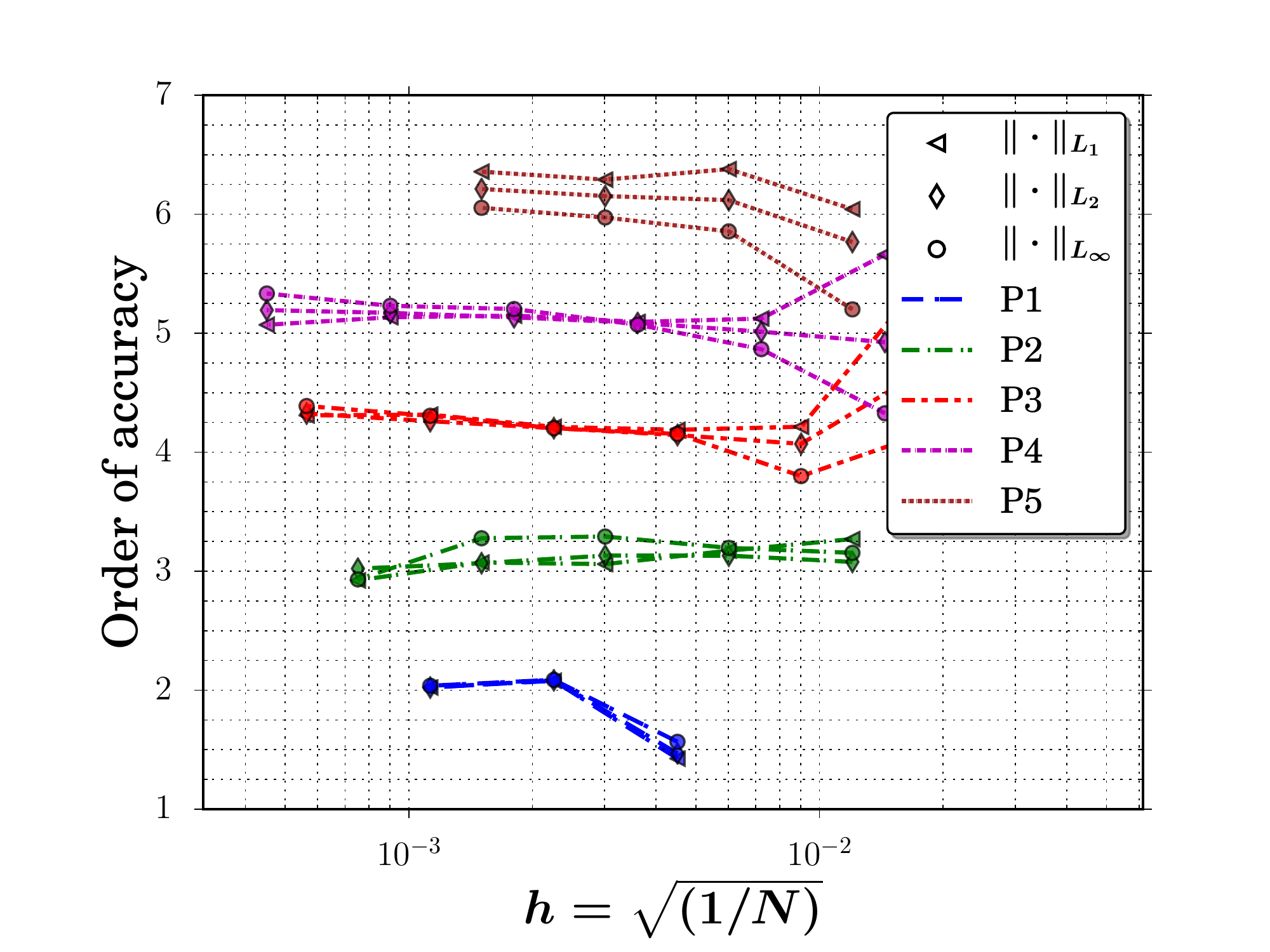}}
\caption{Evolution of the OOAs in $L_1$, $L_2$ and $L_\infty$ norms versus mesh refinement for MS-3 and  $\mathrm{P}1$--$\mathrm{P}5$}
\label{fig:Orders_MS-3}
\end{figure}

\begin{figure}[!hbt]
\centering
\subfloat[$\rho$]{
\includegraphics[trim = 5mm 2mm 18mm 13mm, clip,width=0.32\linewidth]
{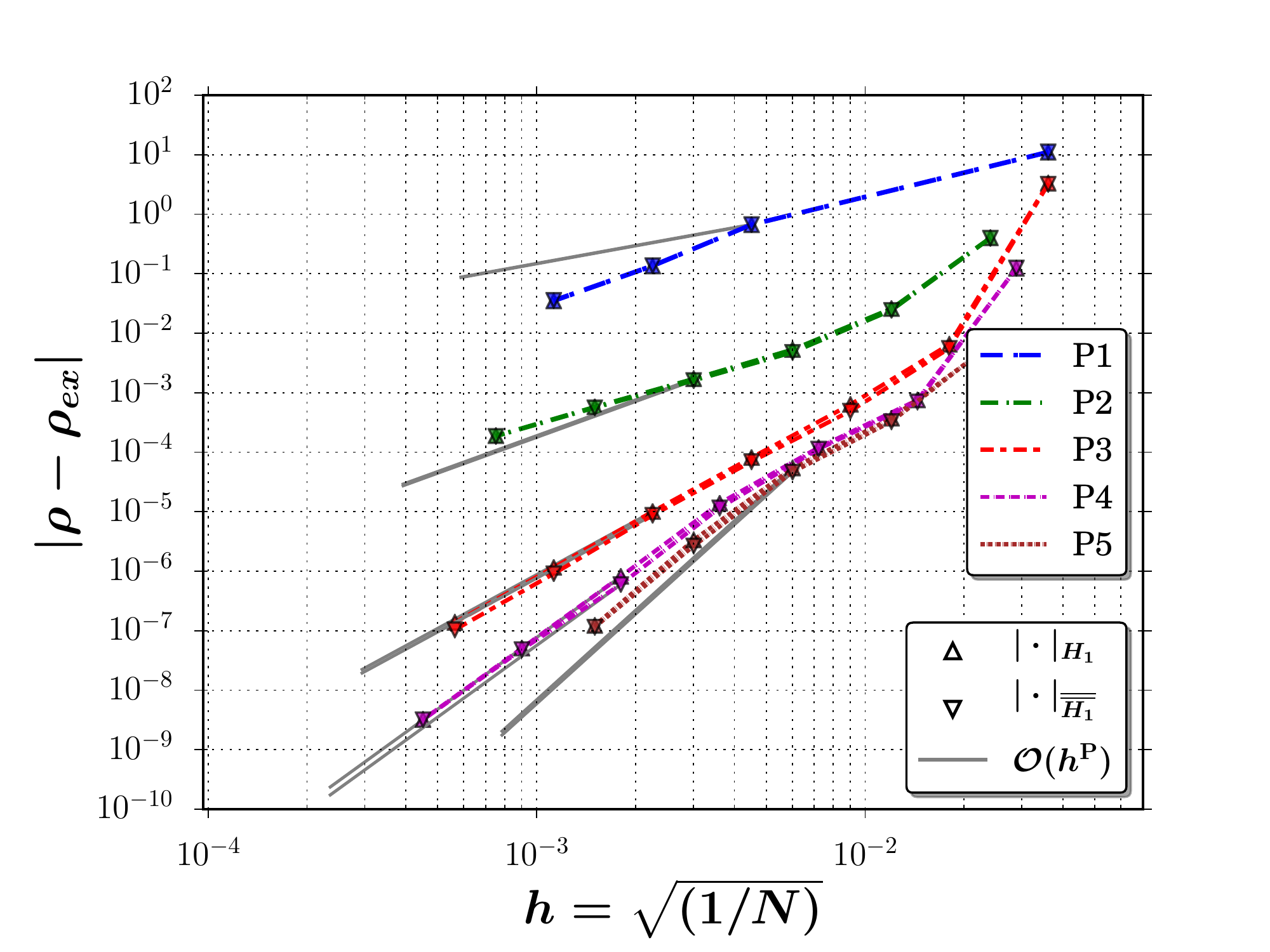}}~~~
\subfloat[$\rho u$]{
\includegraphics[trim = 5mm 2mm 18mm 13mm, clip,width=0.32\linewidth]{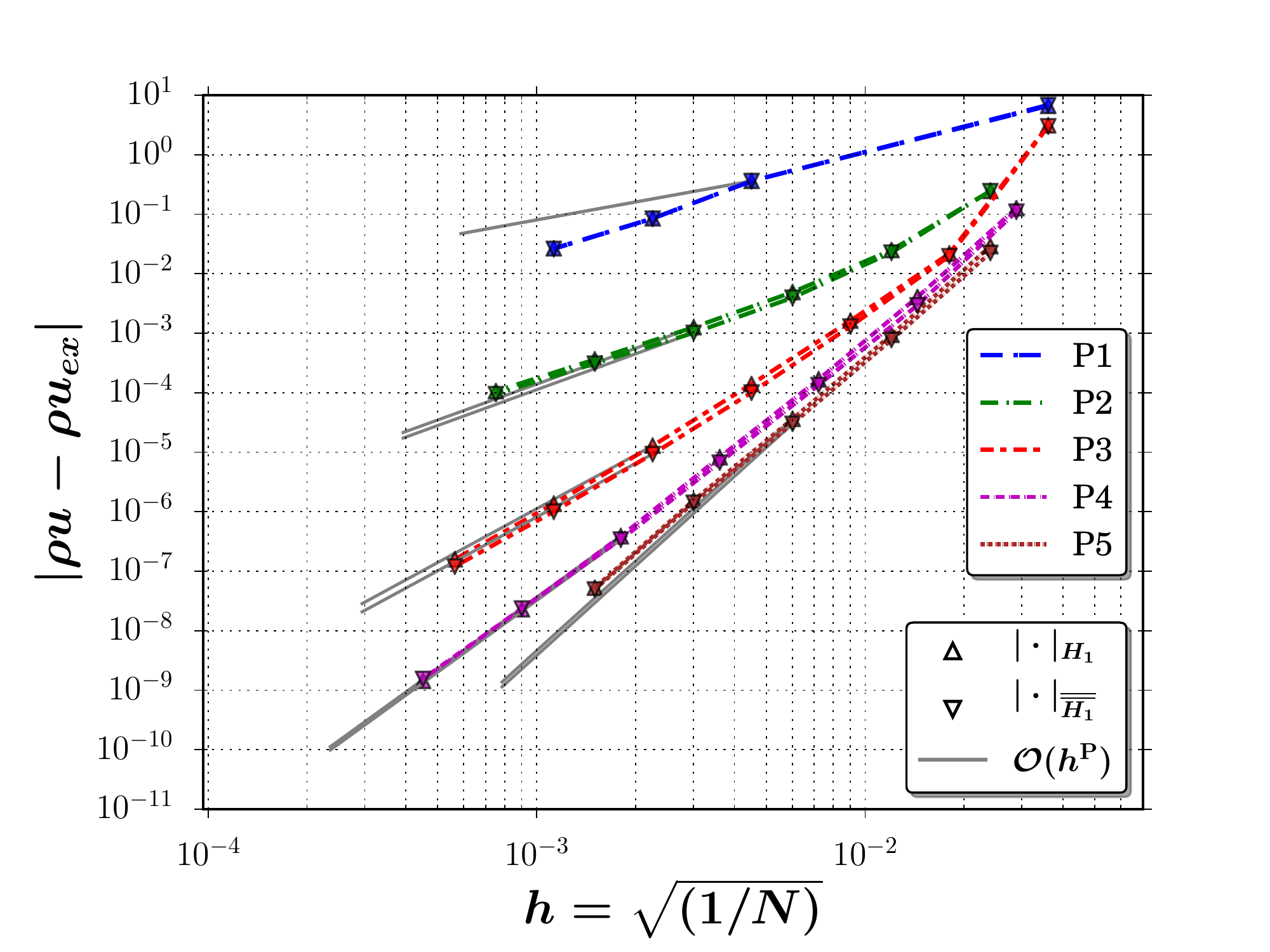}}
\vfill
\subfloat[$\rho v$]{
\includegraphics[trim = 5mm 2mm 18mm 13mm, clip,width=0.32\linewidth]
{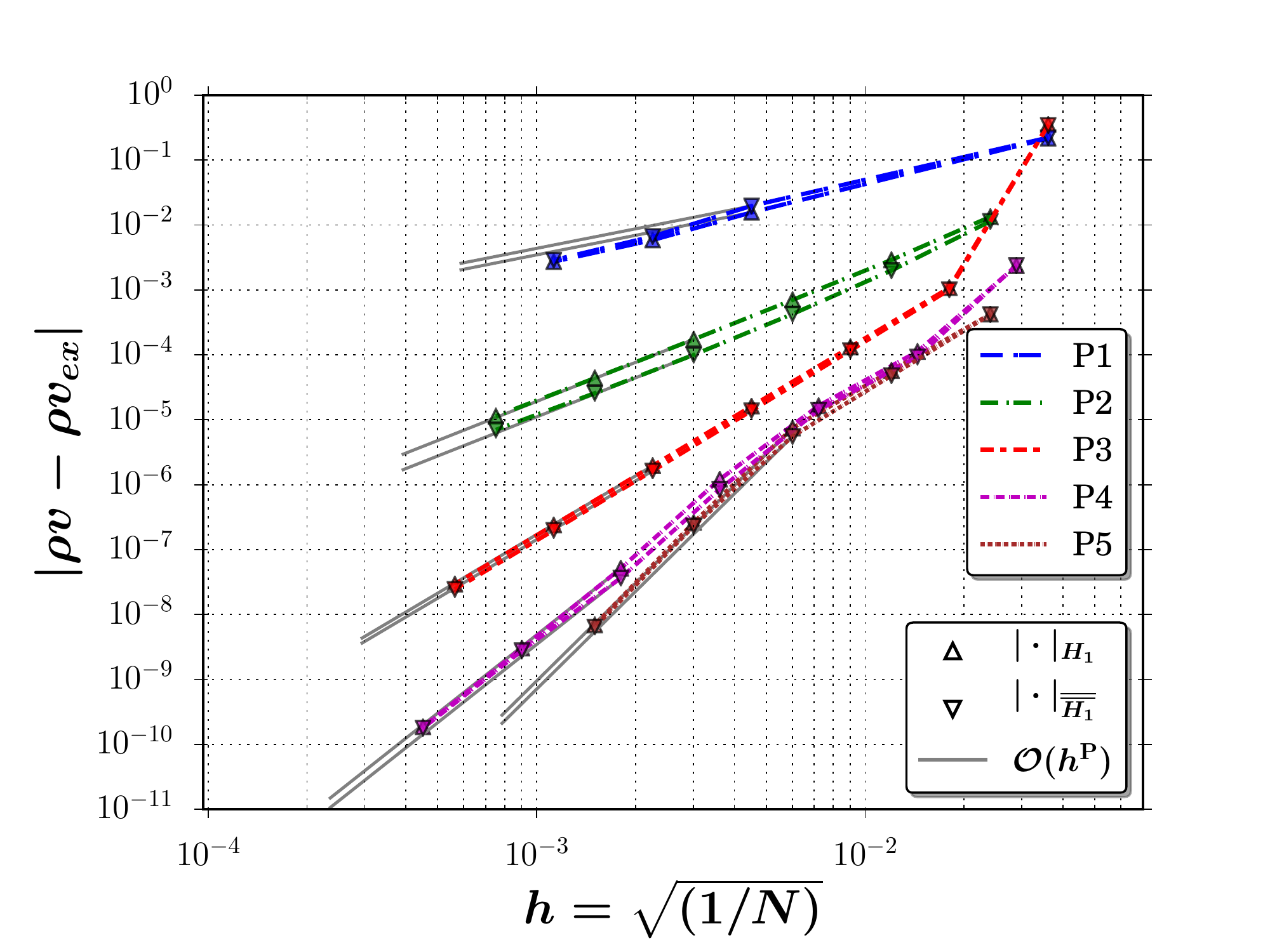}}
~~~
\subfloat[$\rho E$]{
\includegraphics[trim = 5mm 2mm 18mm 13mm, clip,width=0.32\linewidth]{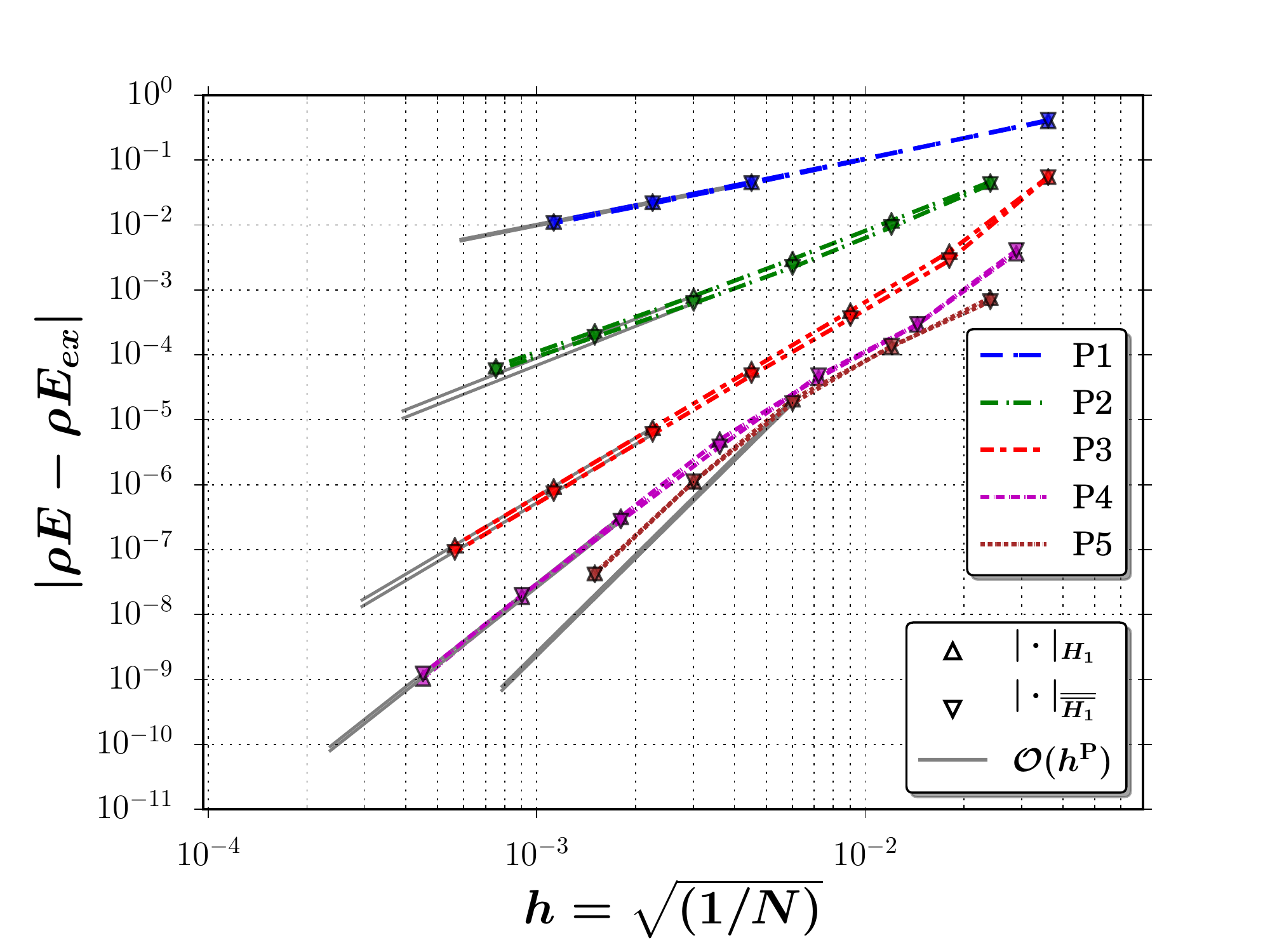}}
\vfill
\subfloat[$\rho \tilde{\nu}$]{
\includegraphics[trim = 5mm 2mm 18mm 13mm, clip,width=0.32\linewidth]{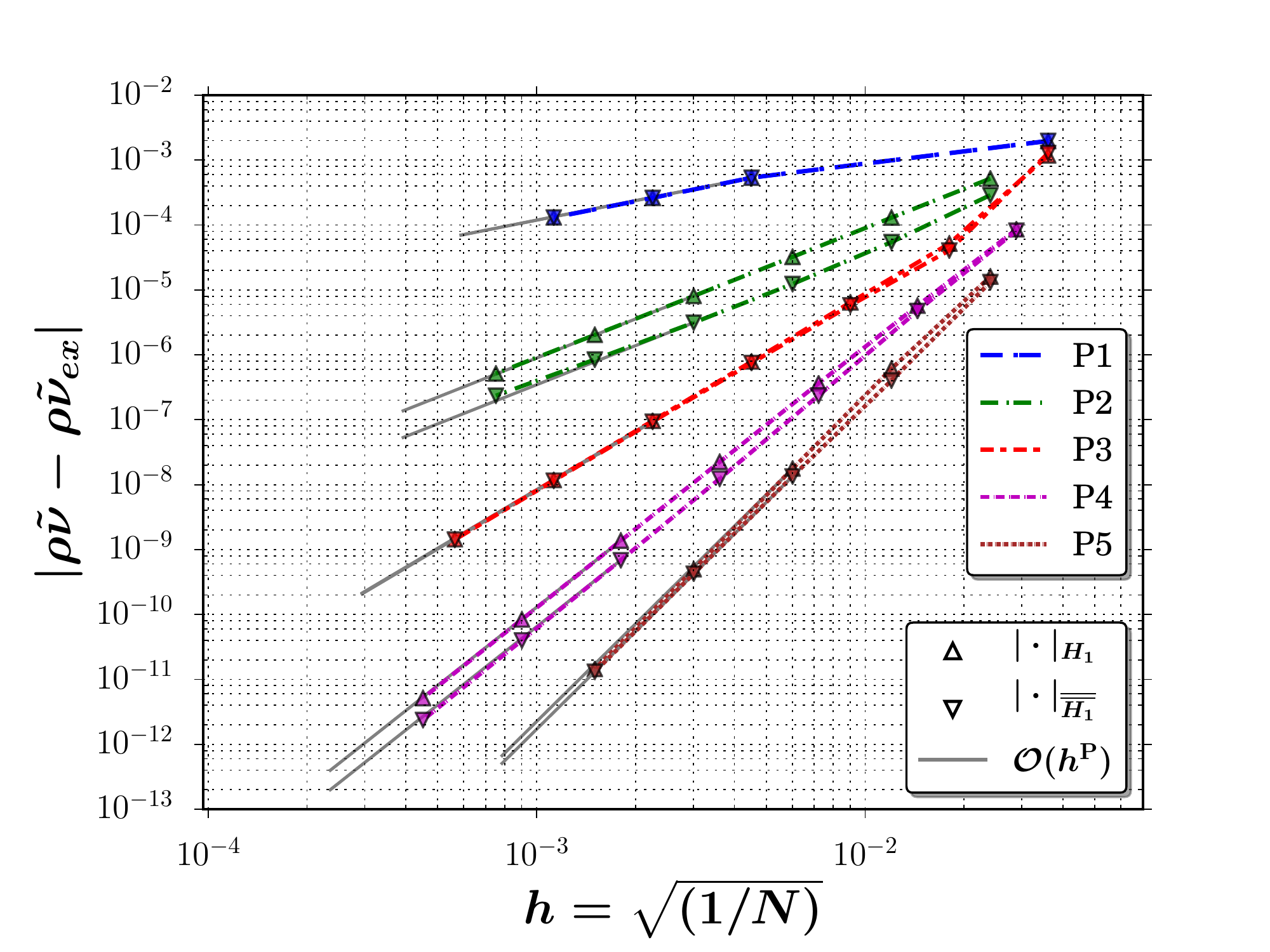}}
\caption{Evolution of the discretization error in $H_1$ semi-norm (for uncorrected and  fully corrected derivatives) versus mesh refinement for MS-3 and  $\mathrm{P}1$--$\mathrm{P}5$}
\label{fig:Err_allE_allP_H_MS-3}
\end{figure}

\begin{figure}[!hbt]
\centering
\subfloat[$\rho$]{ 
\includegraphics[trim = 16mm 3mm 18mm 13mm, clip,width=0.3\linewidth]
{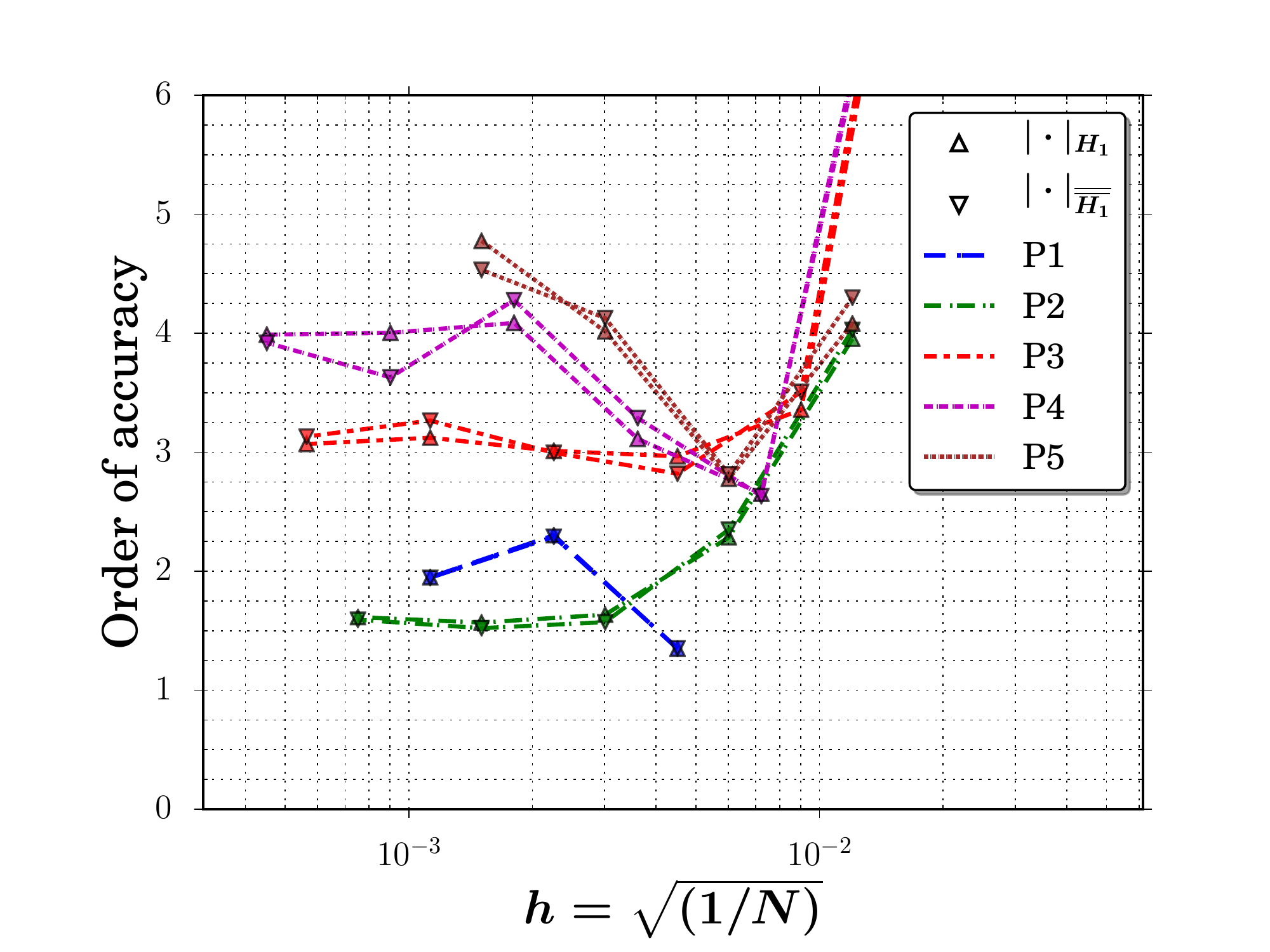}}
~~~
\subfloat[$\rho u$]{
\includegraphics[trim = 16mm 3mm 18mm 13mm, clip,width=0.3\linewidth]{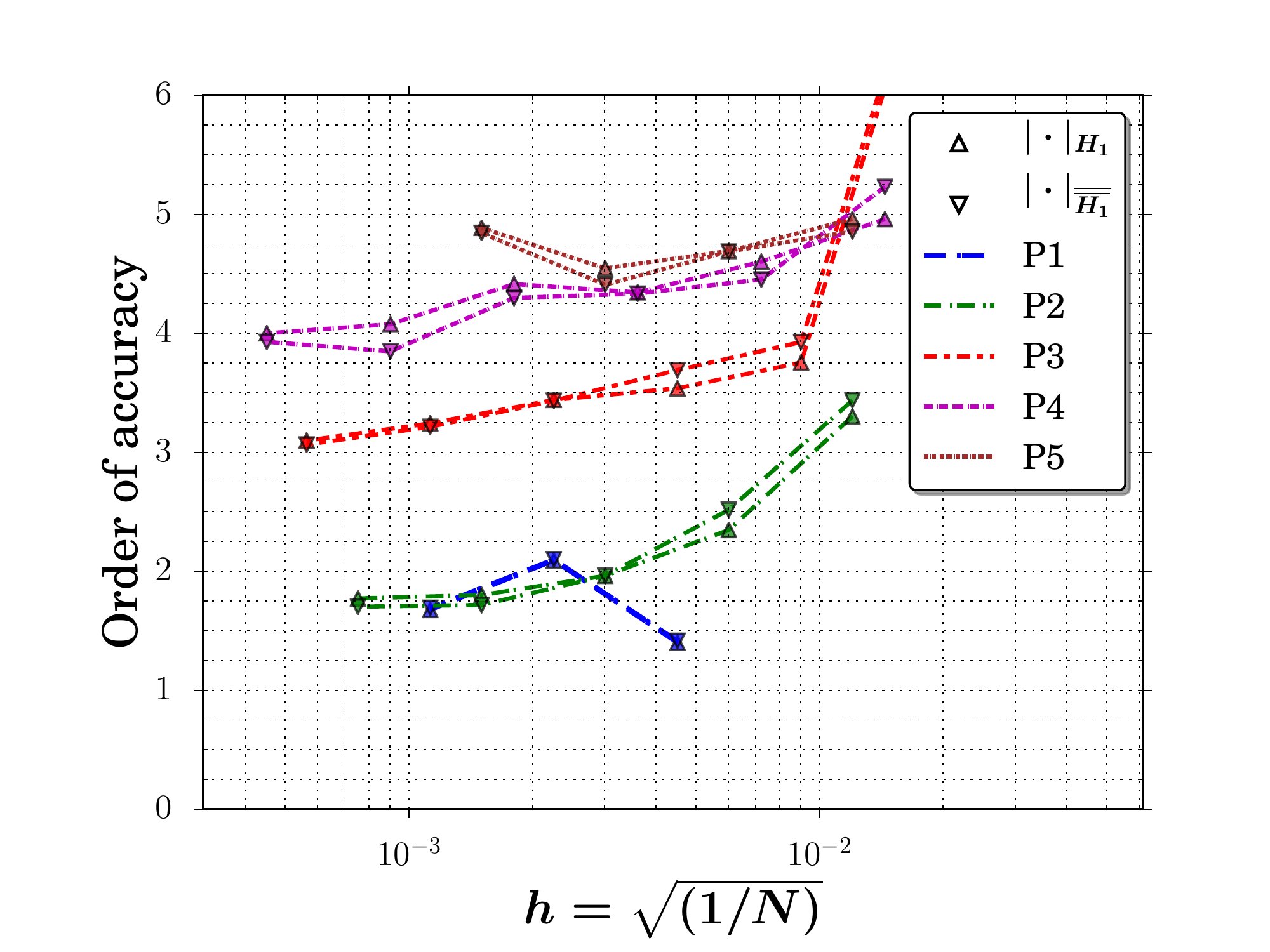}}
\vfill
\subfloat[$\rho v$]{
\includegraphics[trim = 16mm 3mm 18mm 13mm, clip,width=0.3\linewidth]
{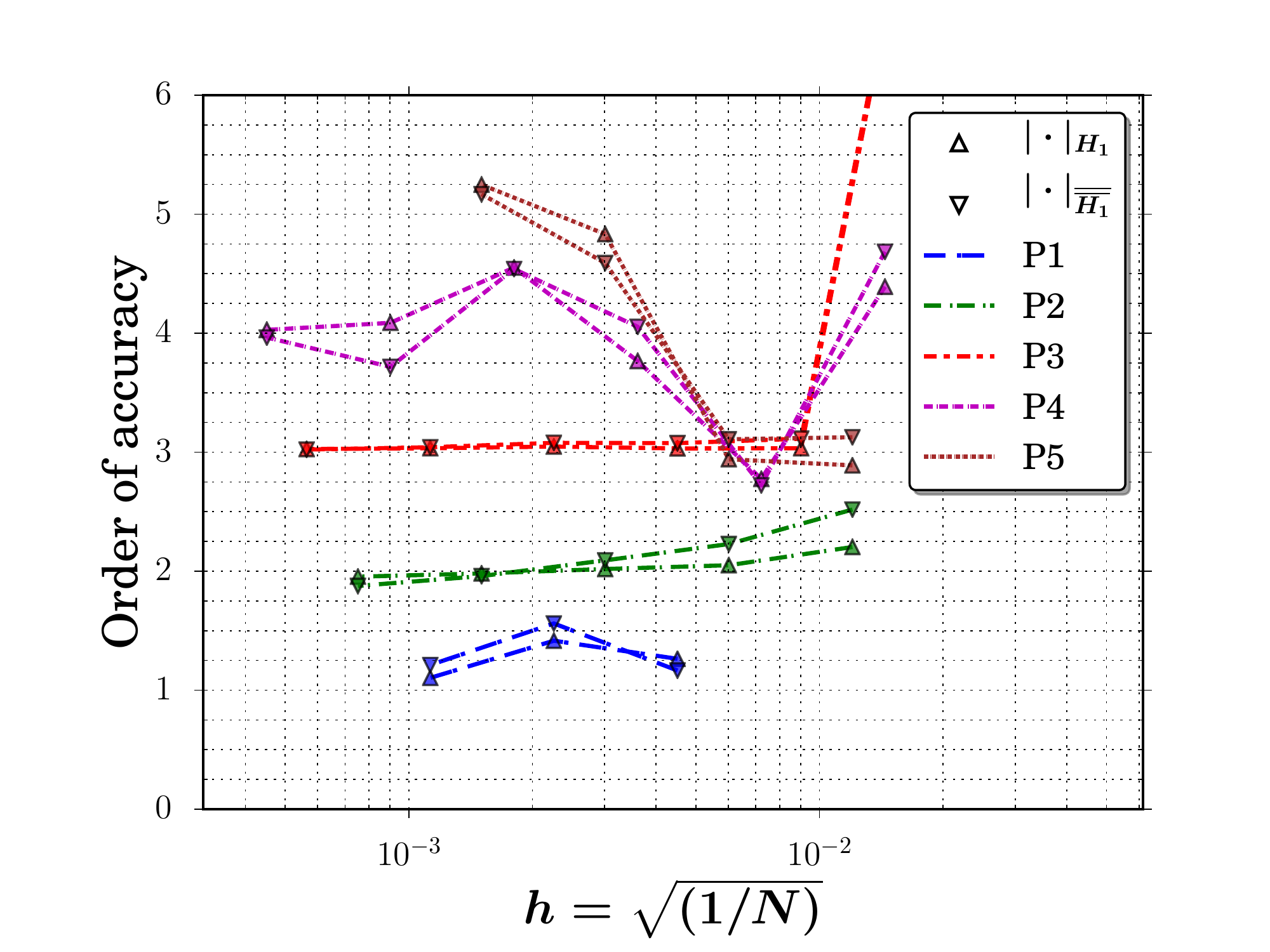}}
~~~
\subfloat[$\rho E$]{
\includegraphics[trim = 16mm 3mm 18mm 13mm, clip,width=0.3\linewidth]{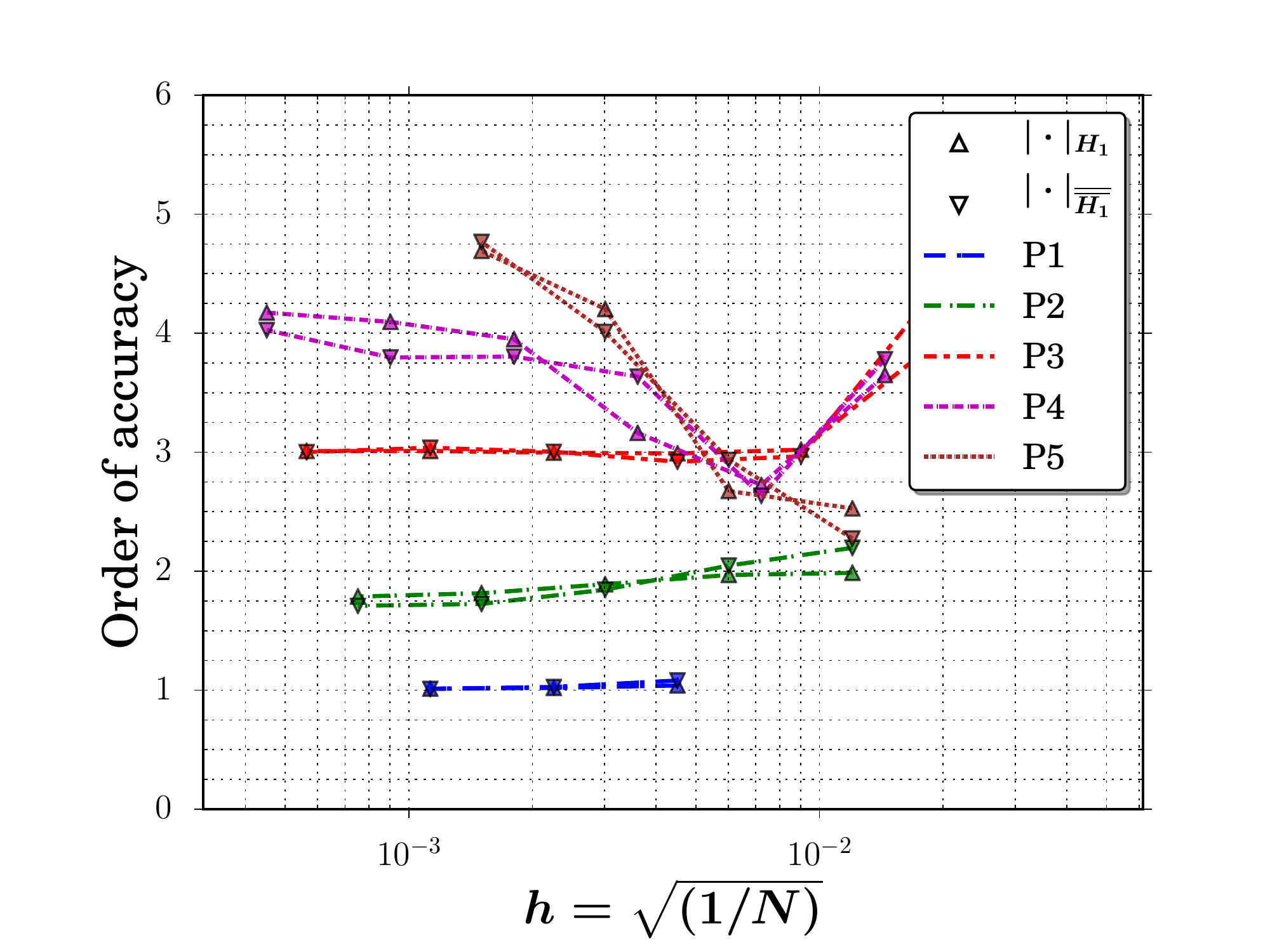}}
\vfill
\subfloat[$\rho \tilde{\nu}$]{
\includegraphics[trim = 16mm 3mm 18mm 13mm, clip,width=0.3\linewidth]{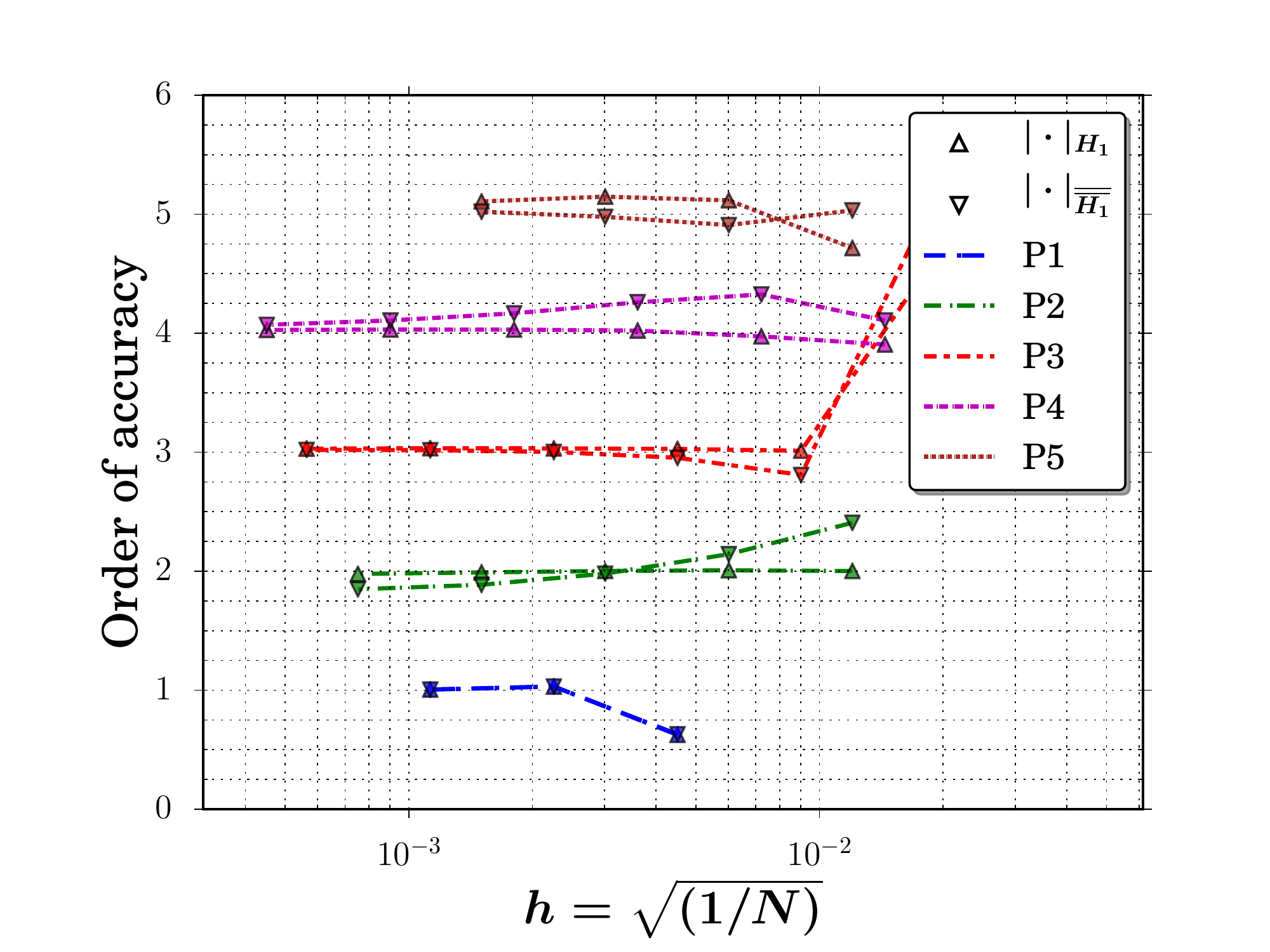}}
\caption{Evolution of the OOAs in $H_1$ semi-norm (for uncorrected and  fully corrected derivatives) versus mesh refinement for MS-3 and  $\mathrm{P}1$--$\mathrm{P}5$}
\label{fig:Orders_H_MS-3}
\end{figure} 
 
\clearpage
\subsection{MS-4}
\begin{figure}[!hbt]
\centering
\subfloat[$\rho^{\mathrm{MS}}$]{
\includegraphics[trim = 0mm 0mm 0mm 0mm, clip,width=0.4\linewidth]
{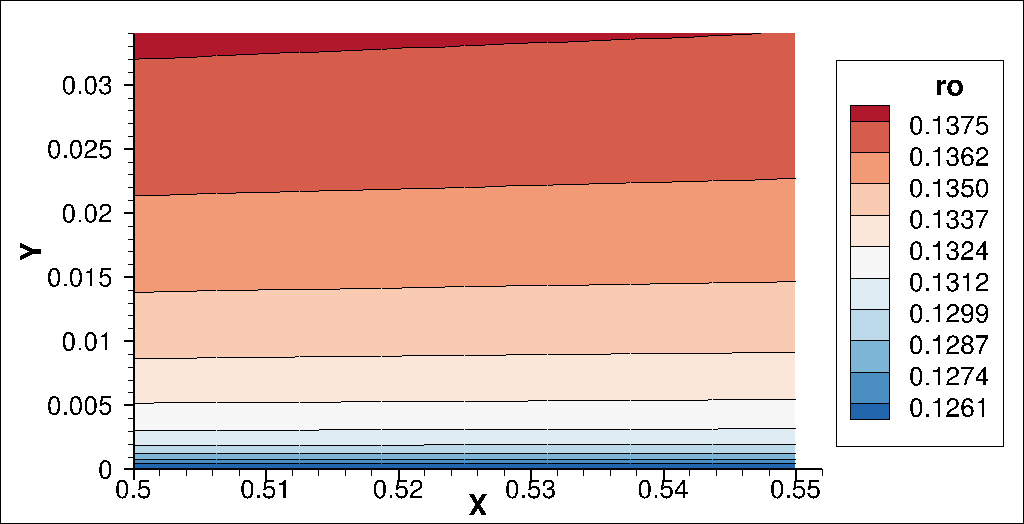}}~
\subfloat[$u^{\mathrm{MS}}$]{
\includegraphics[trim = 0mm 0mm 0mm 0mm, clip,width=0.4\linewidth]
{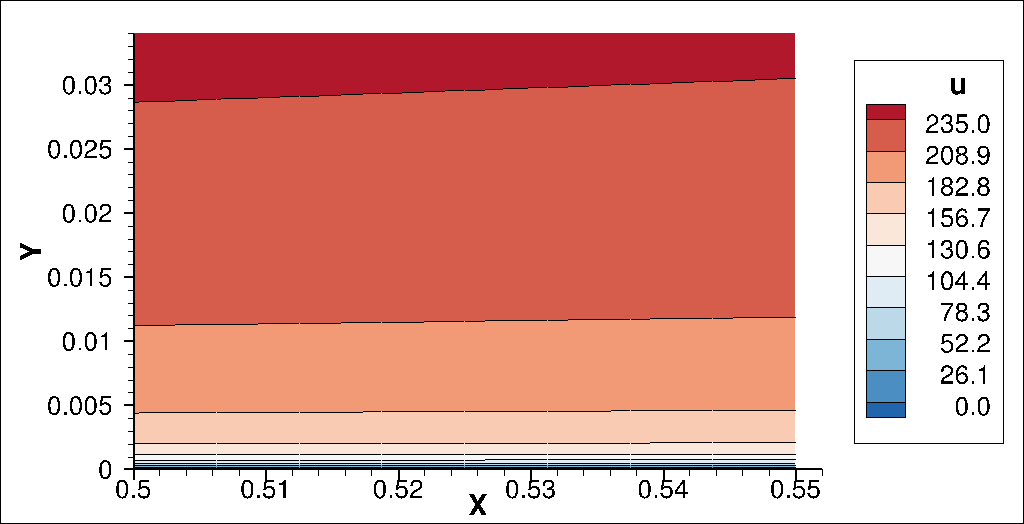}}
\vfill
\subfloat[$v^{\mathrm{MS}}$]{
\includegraphics[trim = 0mm 0mm 0mm 0mm, clip,width=0.4\linewidth]{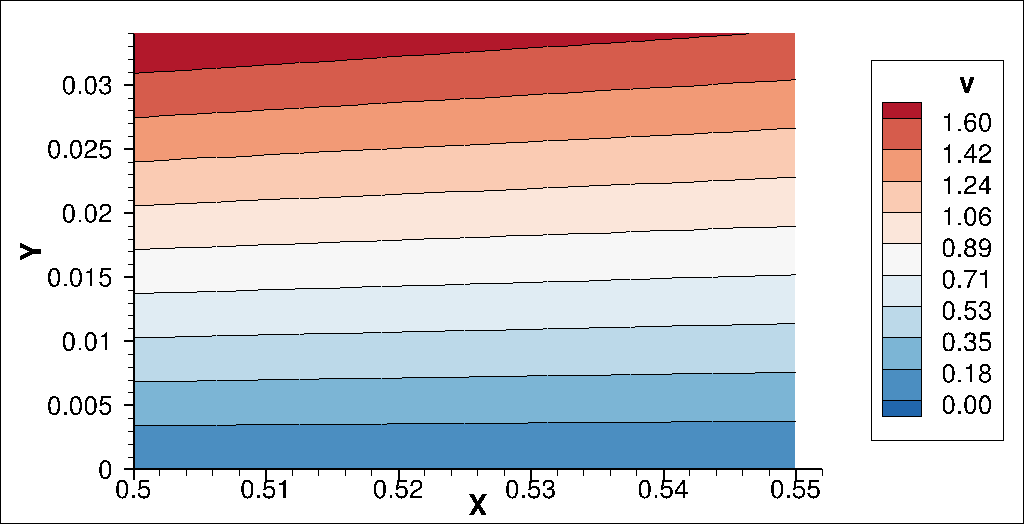}}~
\subfloat[${\tilde \nu}^{\mathrm{MS}}$]{
\includegraphics[trim = 0mm 0mm 0mm 0mm, clip,width=0.4\linewidth]{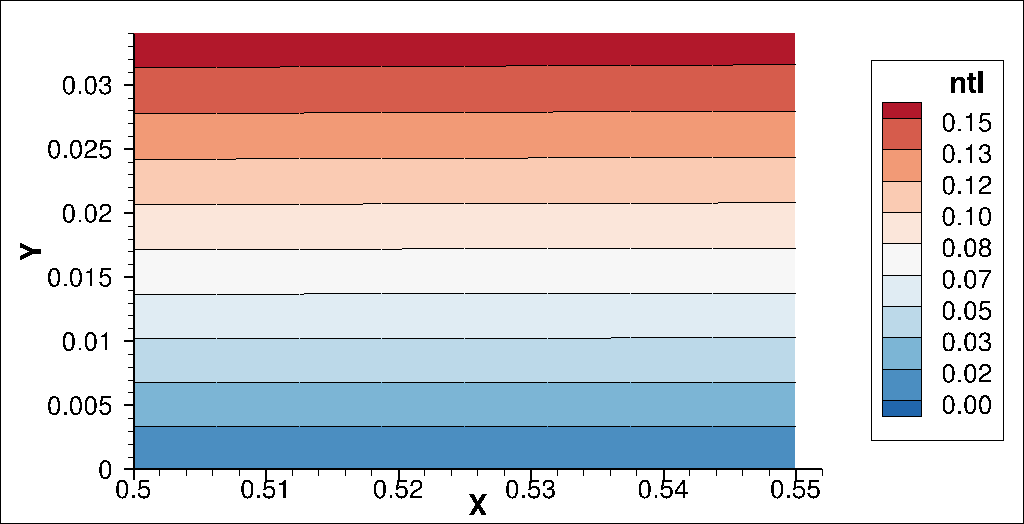}}
\vfill
\subfloat[${Ma}^{\mathrm{MS}}$ and ${\bm{u}}^{\mathrm{MS}}$]{
\includegraphics[trim = 0mm 0mm 0mm 0mm, clip,width=0.4\linewidth]{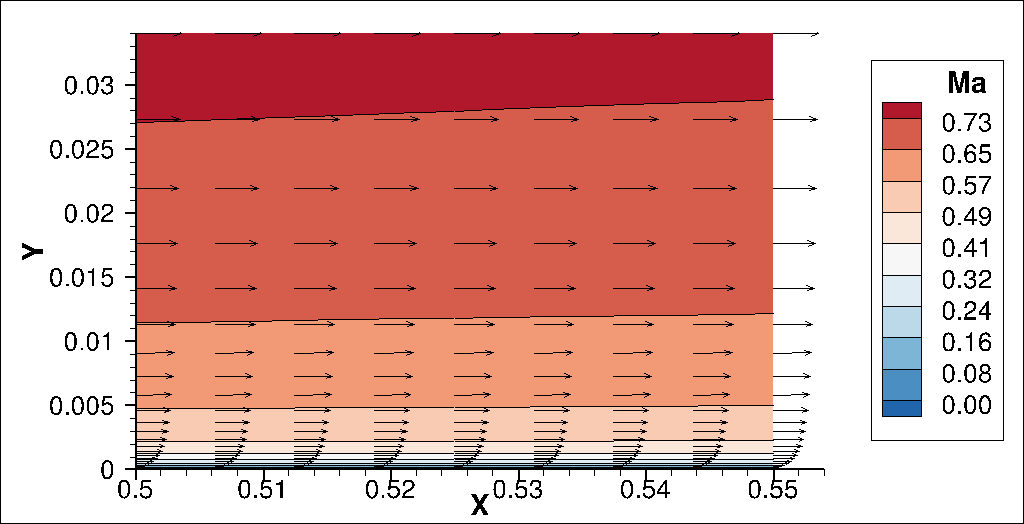}}
\caption{Manufactured solution MS-4 (dimensional)}
\label{fig:MS-4}
\end{figure} 
 
\begin{figure}[!hbt]
\centering
\subfloat[$\rho$]{
\includegraphics[trim = 5mm 2mm 18mm 13mm, clip,width=0.32\linewidth]
{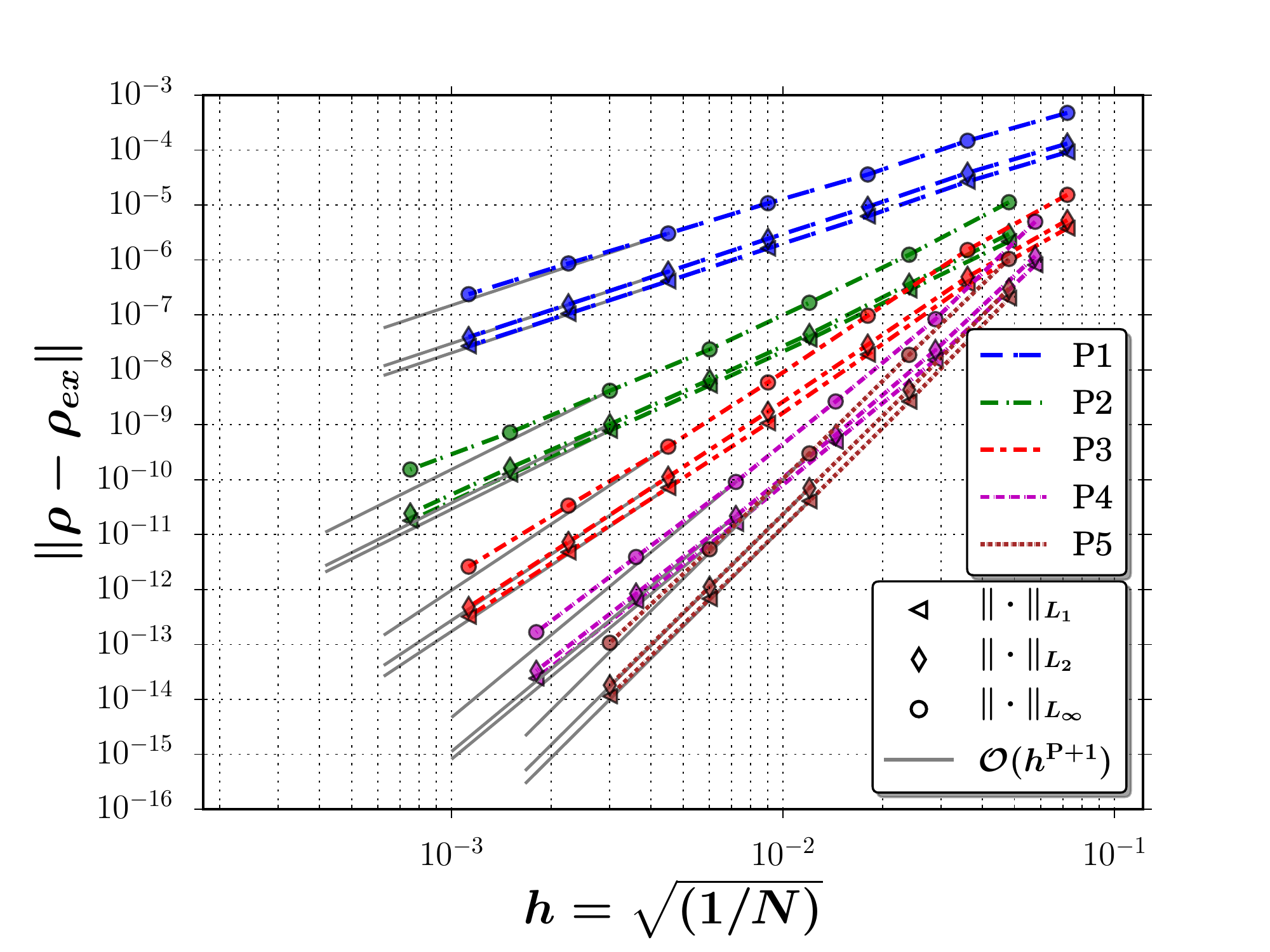}}
~~~
\subfloat[$\rho u$]{
\includegraphics[trim = 5mm 2mm 18mm 13mm, clip,width=0.32\linewidth]{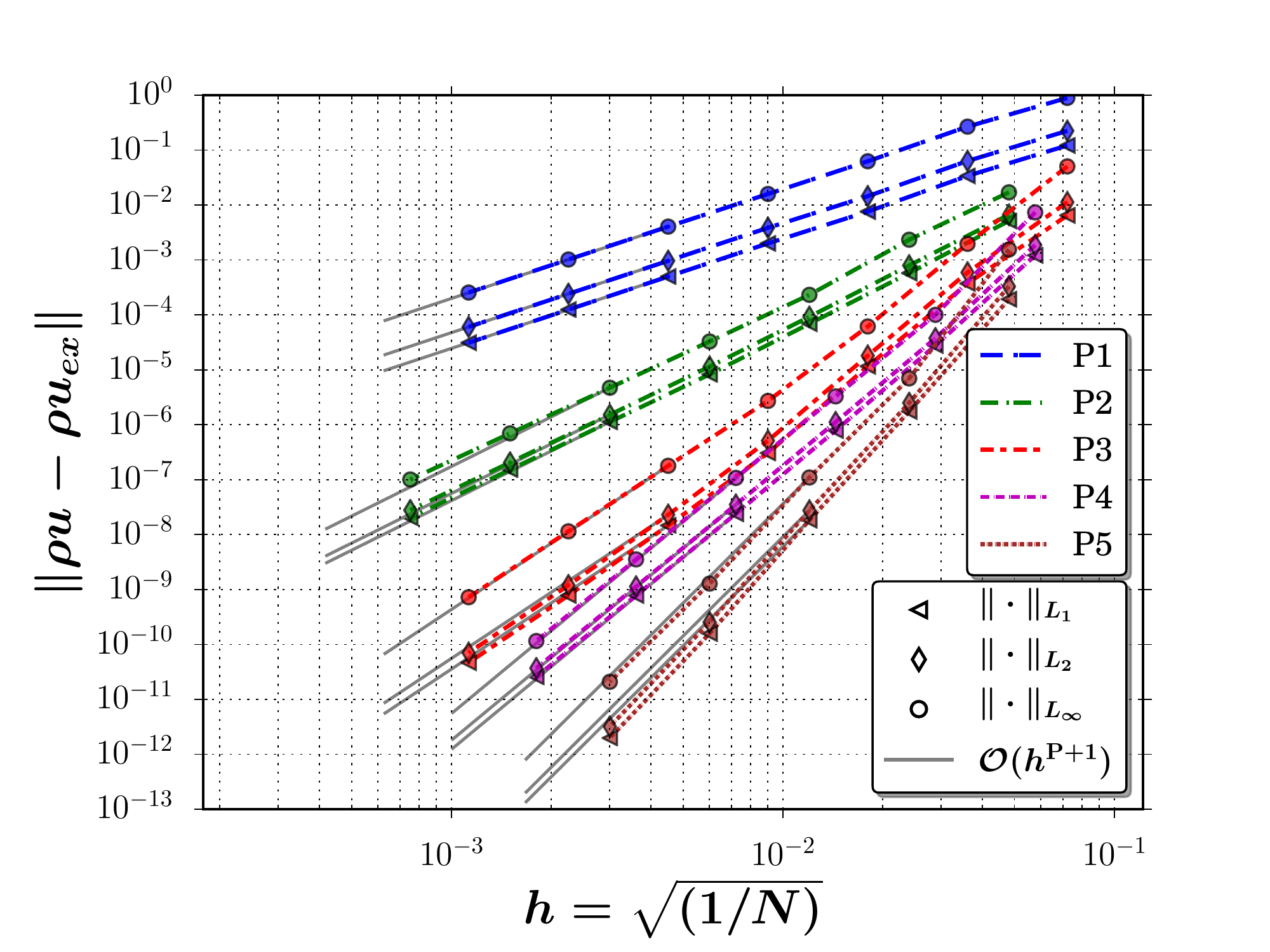}}
\vfill
\subfloat[$\rho v$]{
\includegraphics[trim =5mm 2mm 18mm 13mm, clip,width=0.32\linewidth]
{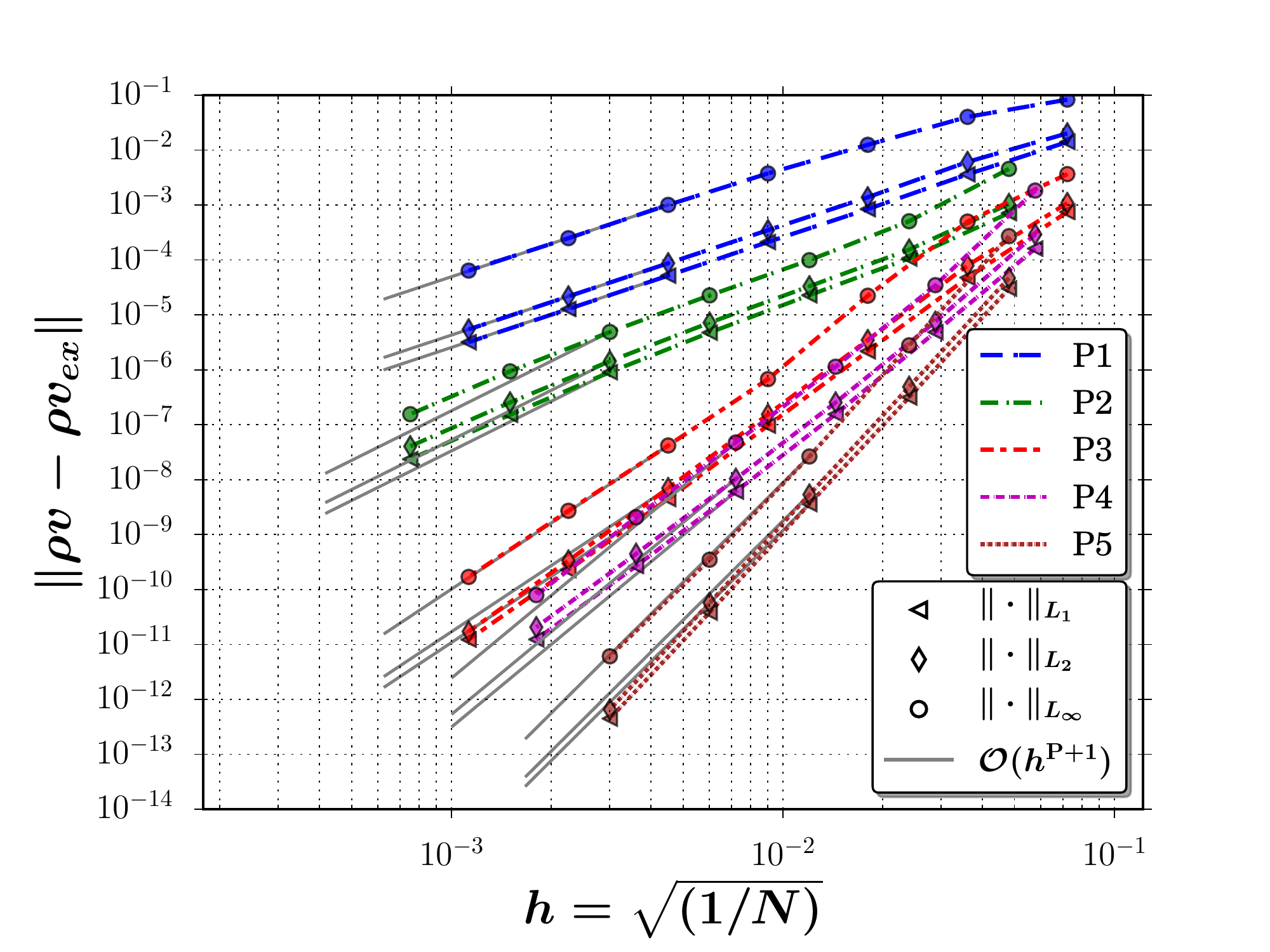}}
~~~
\subfloat[$\rho E$]{
\includegraphics[trim = 5mm 2mm 18mm 13mm, clip,width=0.32\linewidth]{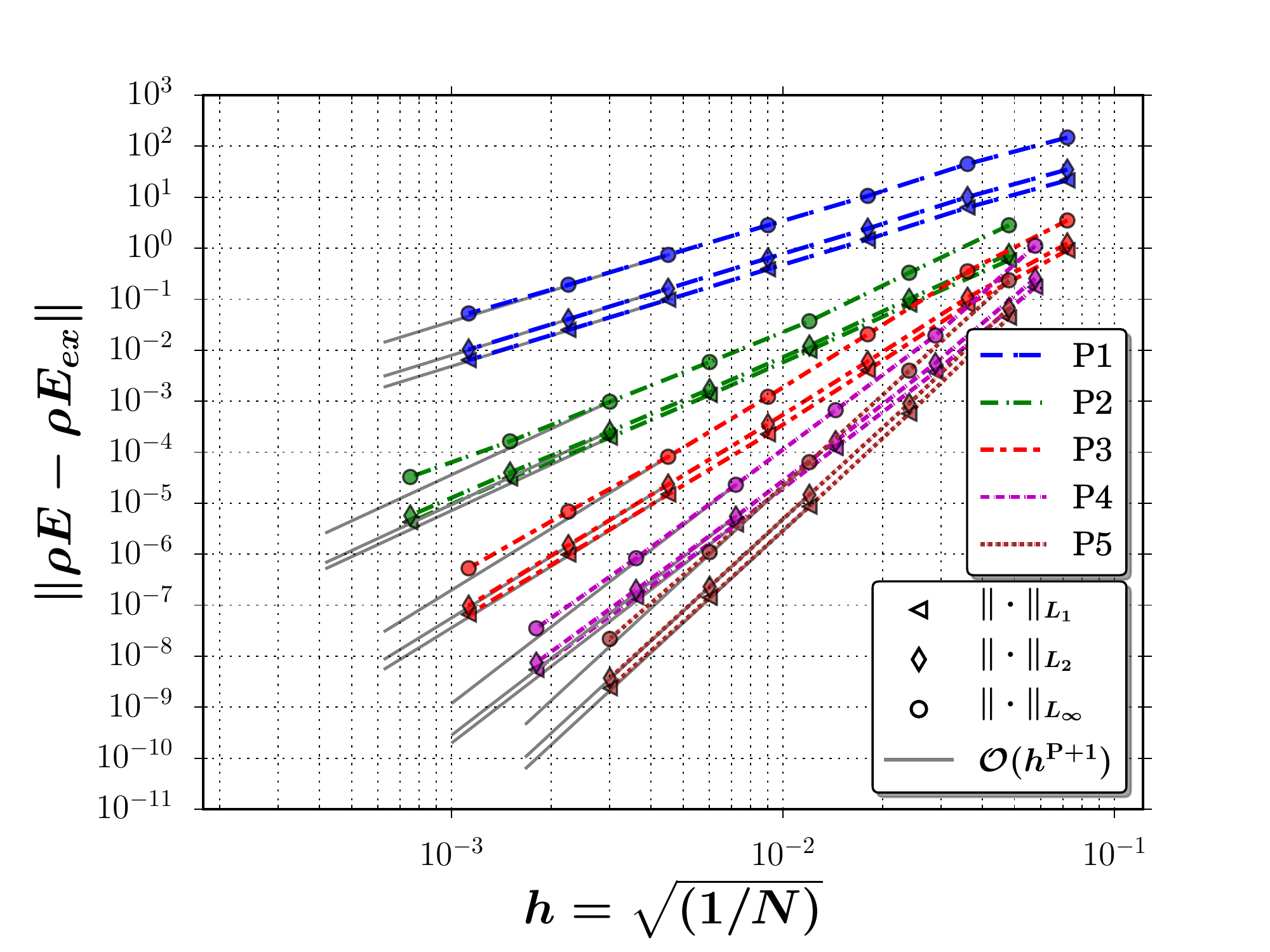}}
\vfill
\subfloat[$\rho \tilde{\nu}$]{
\includegraphics[trim = 5mm 2mm 18mm 13mm, clip,width=0.32\linewidth]
{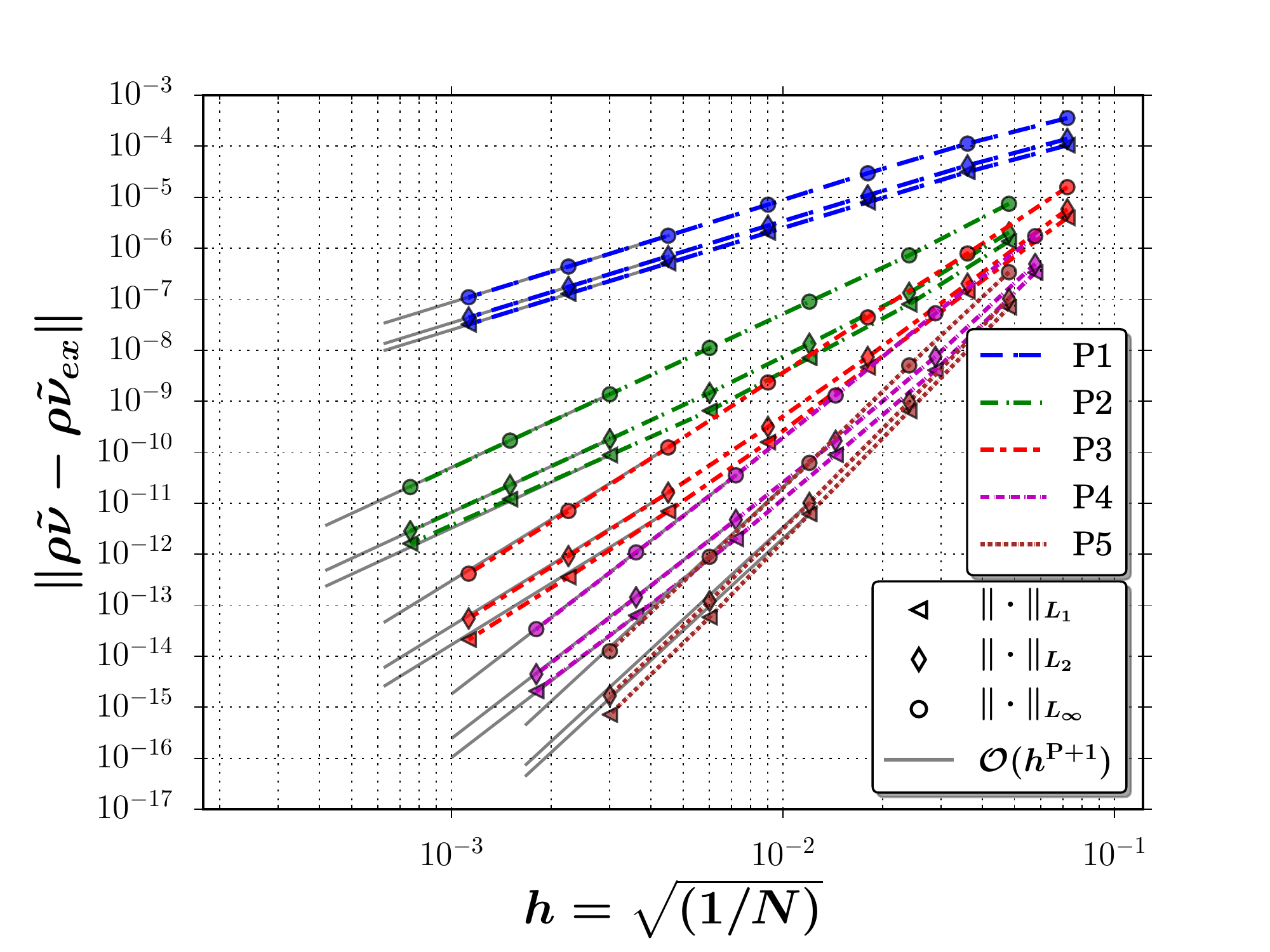}}
\caption{Evolution of the discretization error in $L_1$, $L_2$ and $L_\infty$ norms versus mesh refinement for dimensional MS-4 and  $\mathrm{P}1$--$\mathrm{P}5$}
\label{fig:Err_allE_allP_MS-4}
\end{figure}

\begin{figure}[!hbt]
\centering
\subfloat[$\rho$]{ 
\includegraphics[trim = 16mm 3mm 18mm 13mm, clip,width=0.3\linewidth]
{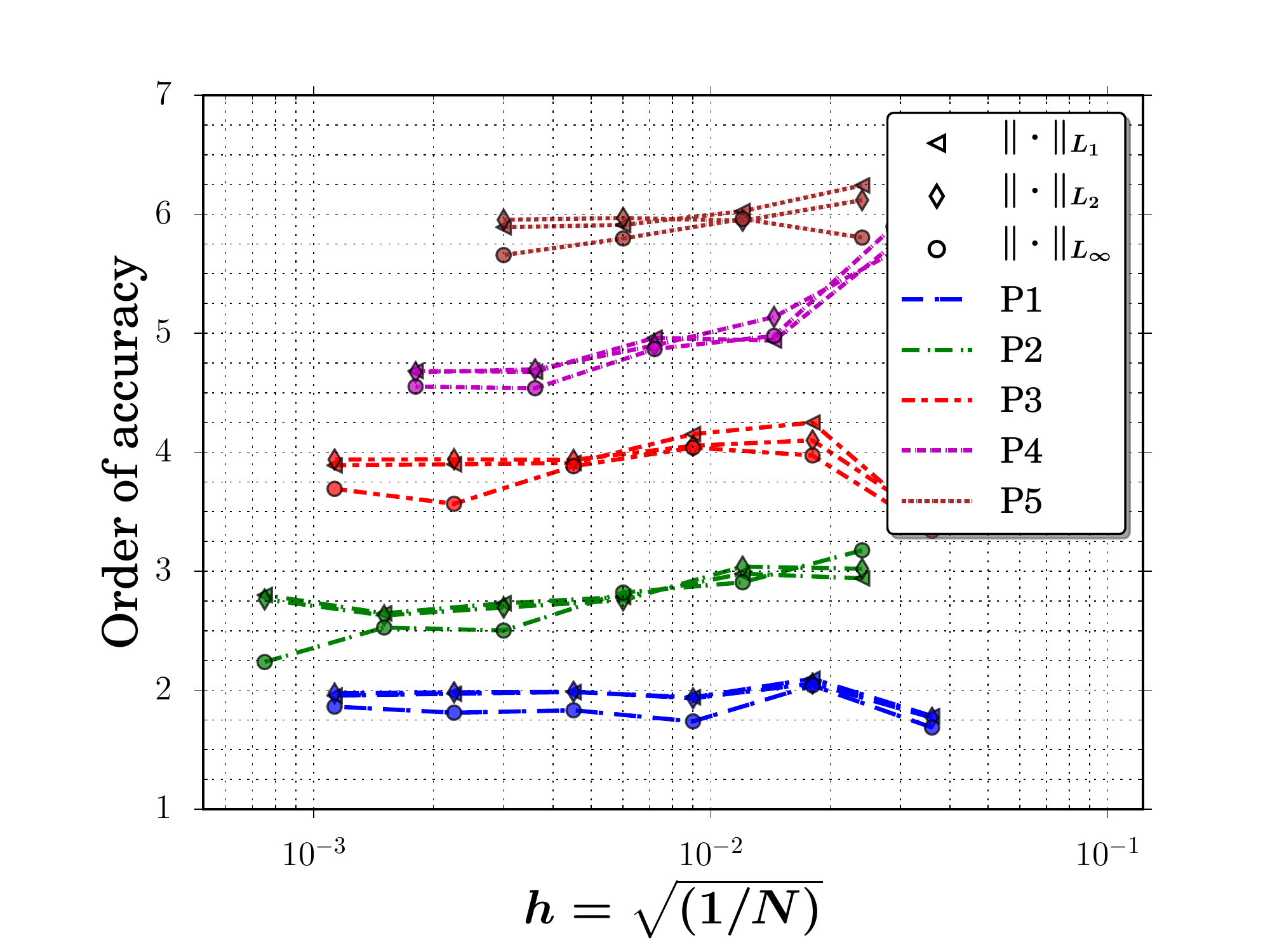}}
~~~
\subfloat[$\rho u$]{
\includegraphics[trim = 16mm 3mm 18mm 13mm, clip,width=0.3\linewidth]{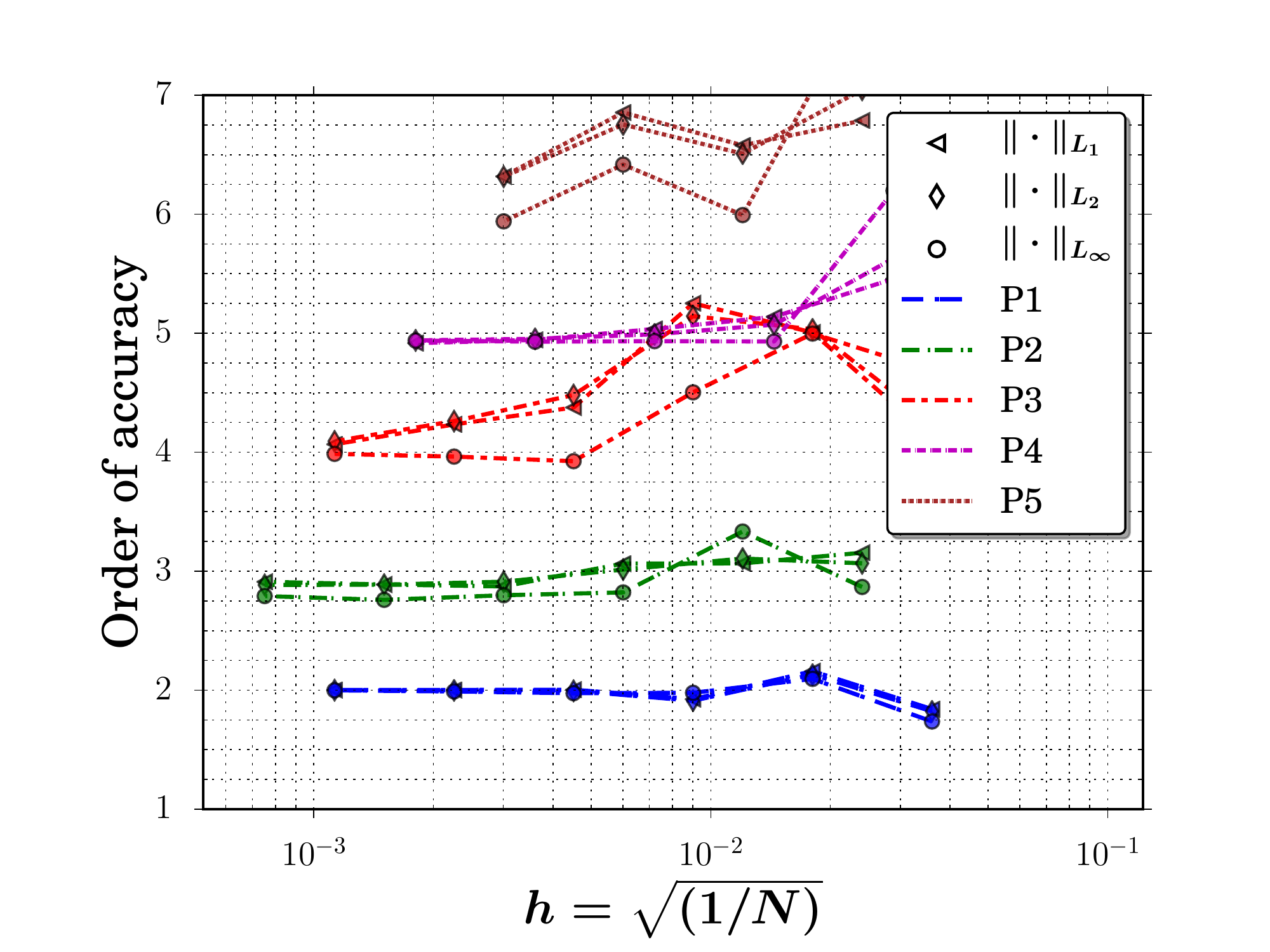}}
\vfill
\subfloat[$\rho v$]{
\includegraphics[trim = 16mm 3mm 18mm 13mm, clip,width=0.3\linewidth]
{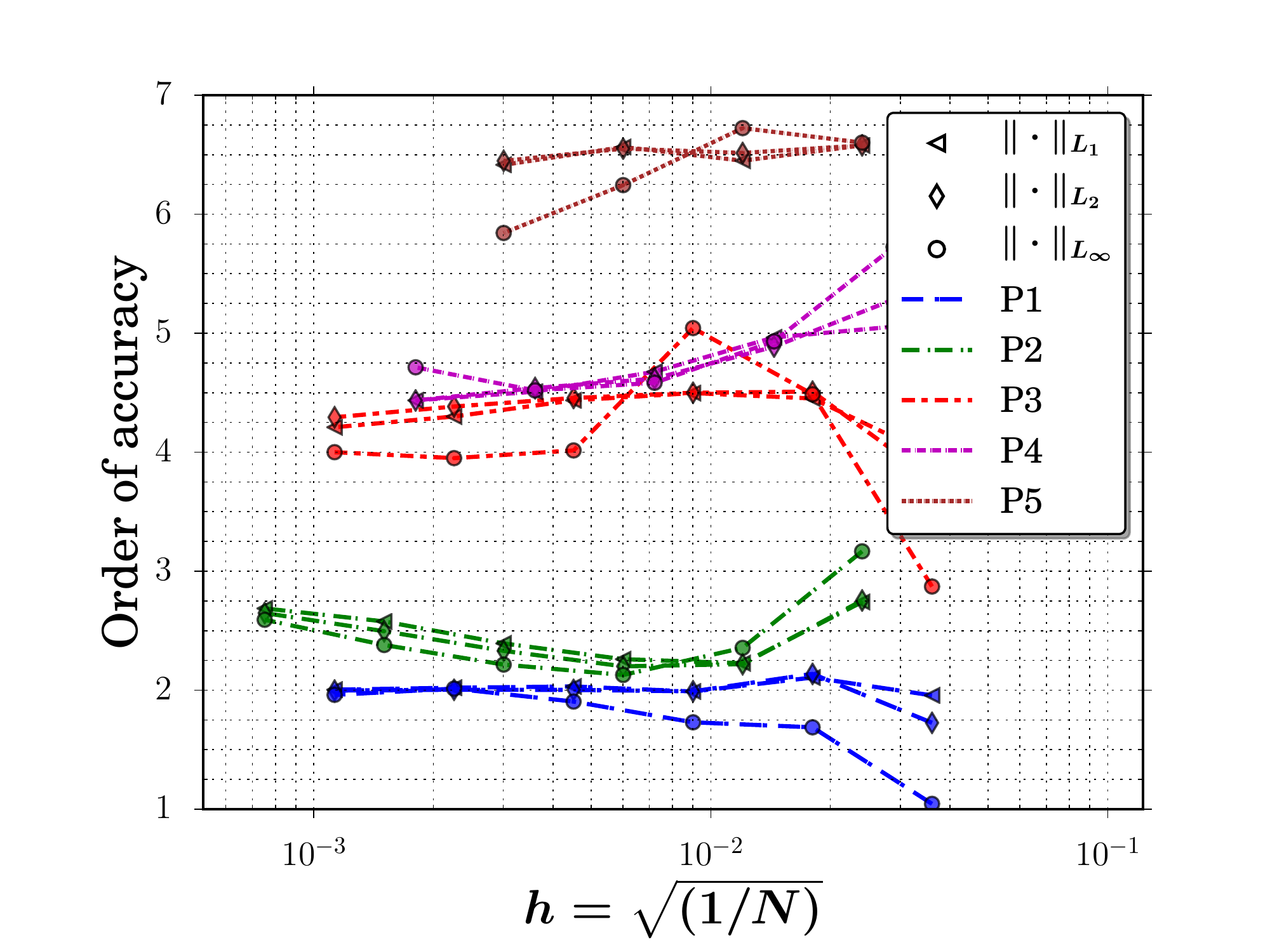}}
~~~
\subfloat[$\rho E$]{
\includegraphics[trim = 16mm 3mm 18mm 13mm, clip,width=0.3\linewidth]{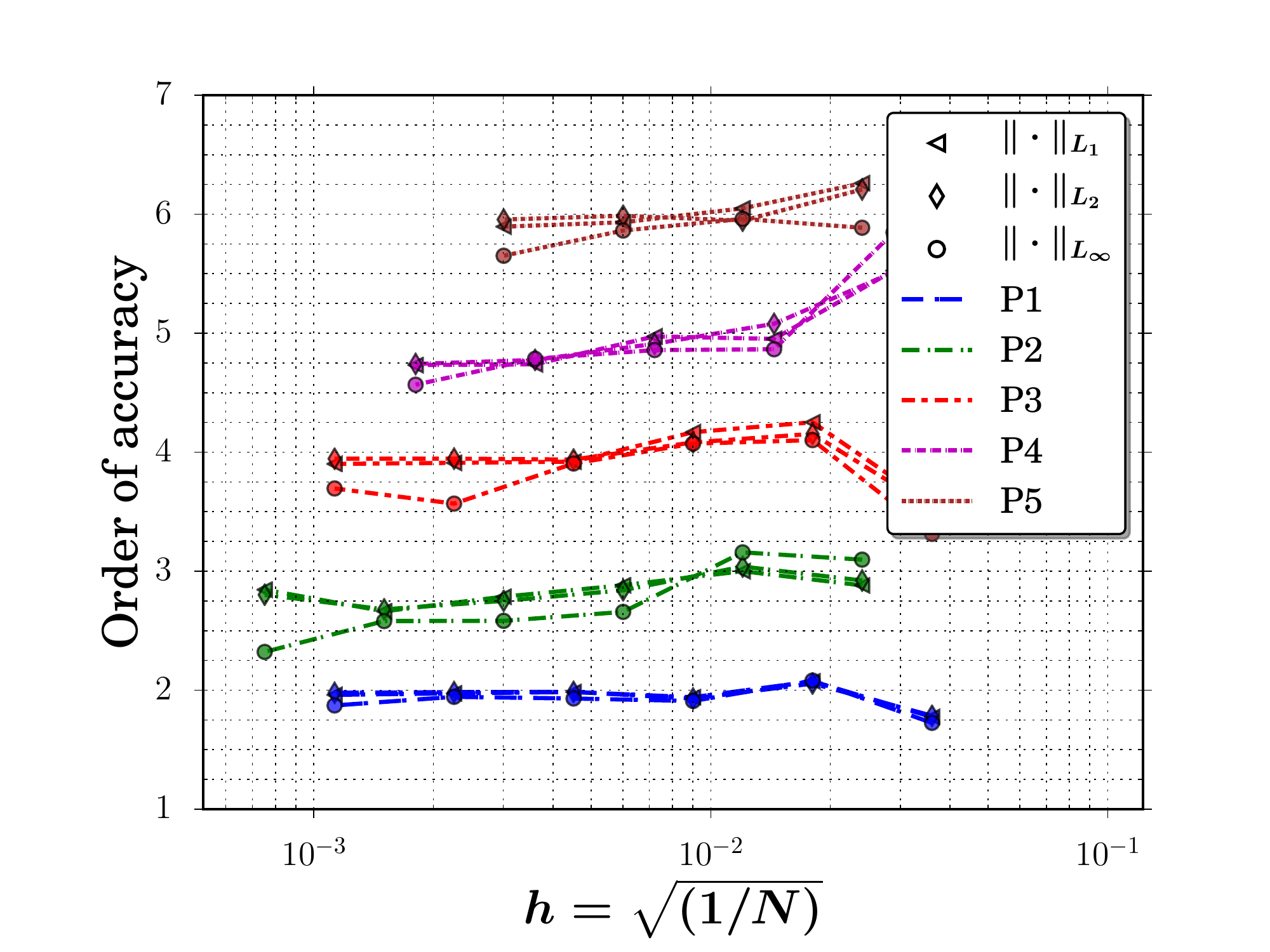}}
\vfill
\subfloat[$\rho \tilde{\nu}$]{
\includegraphics[trim = 16mm 3mm 18mm 13mm, clip,width=0.3\linewidth]{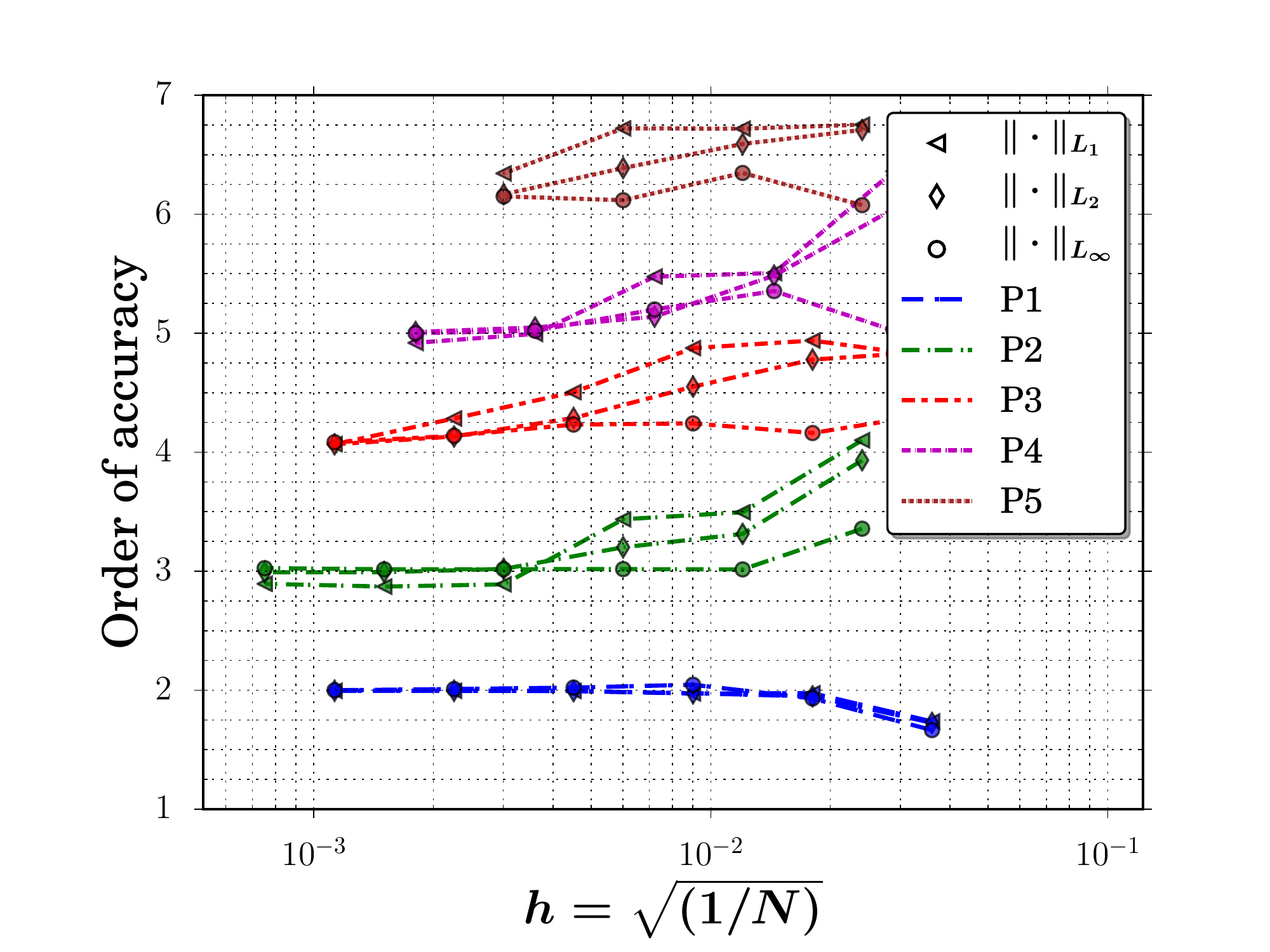}}
\caption{Evolution of the OOAs in $L_1$, $L_2$ and $L_\infty$ norms versus mesh refinement for dimensional MS-4 and  $\mathrm{P}1$--$\mathrm{P}5$}
\label{fig:Orders_MS-4}
\end{figure}

\begin{figure}[!hbt]
\centering
\subfloat[$\rho$]{
\includegraphics[trim = 5mm 2mm 18mm 13mm, clip,width=0.32\linewidth]
{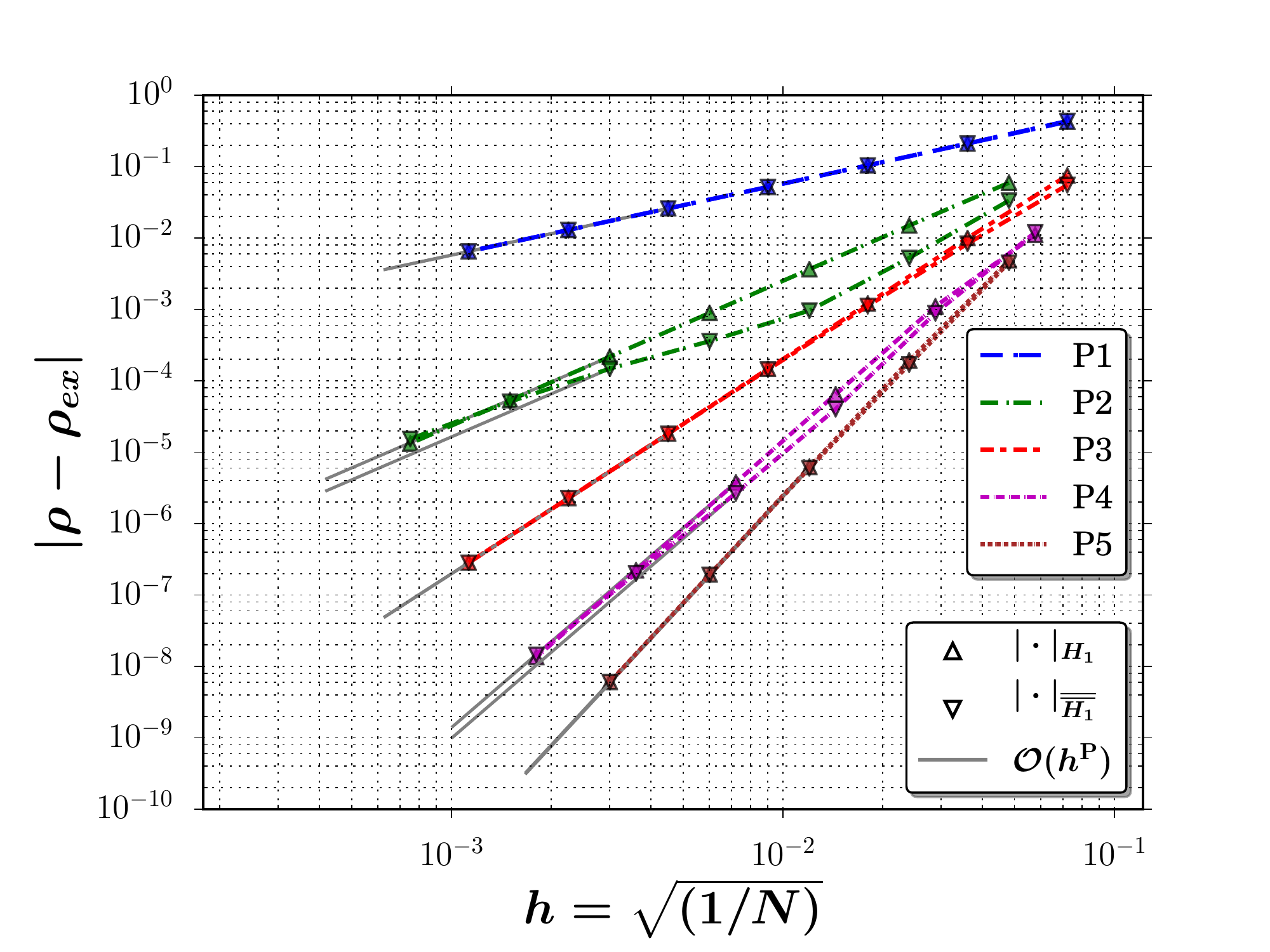}}
~~~
\subfloat[$\rho u$]{
\includegraphics[trim = 5mm 2mm 18mm 13mm, clip,width=0.32\linewidth]{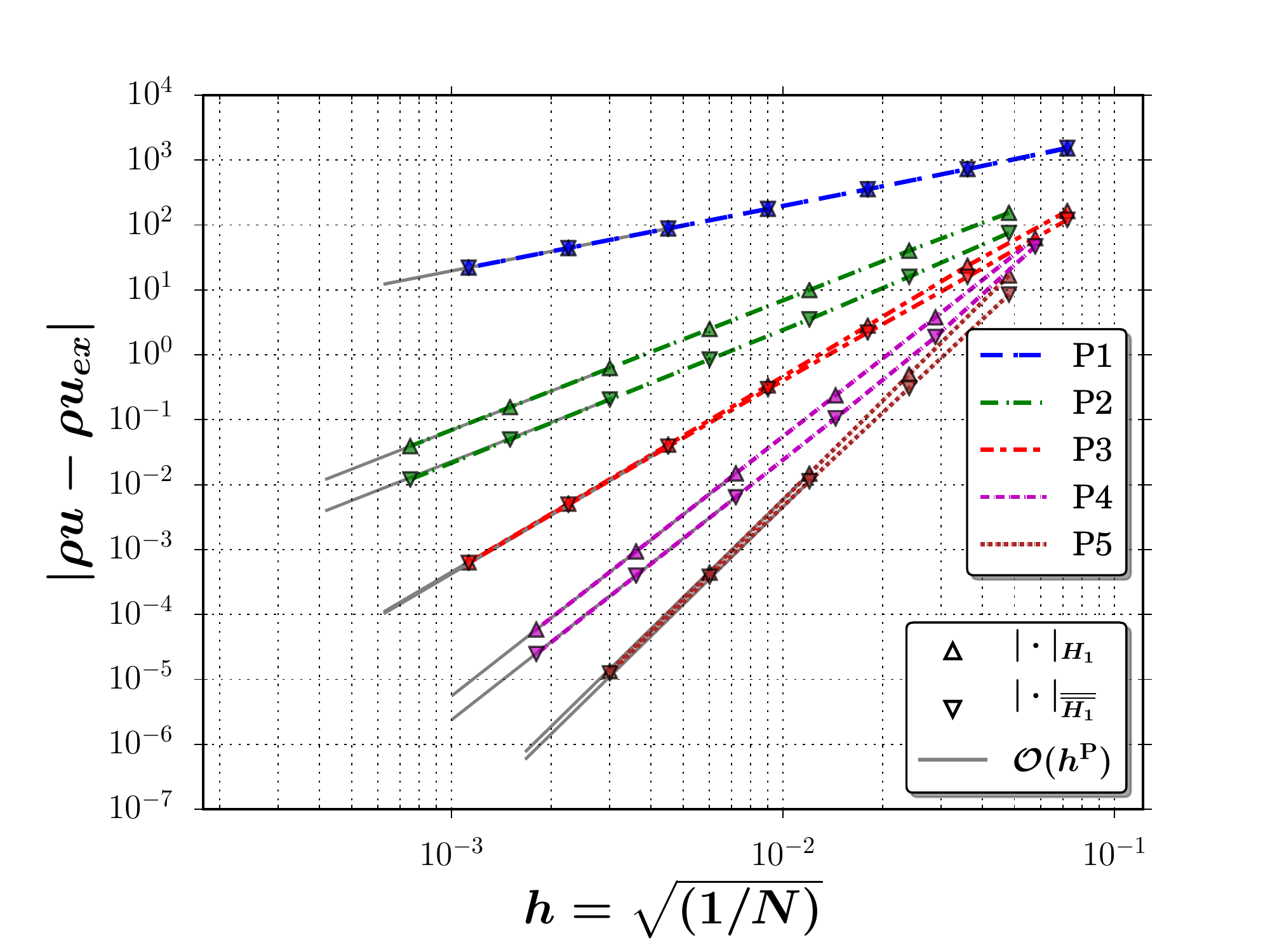}}
\vfill
\subfloat[$\rho v$]{
\includegraphics[trim = 5mm 2mm 18mm 13mm, clip,width=0.32\linewidth]
{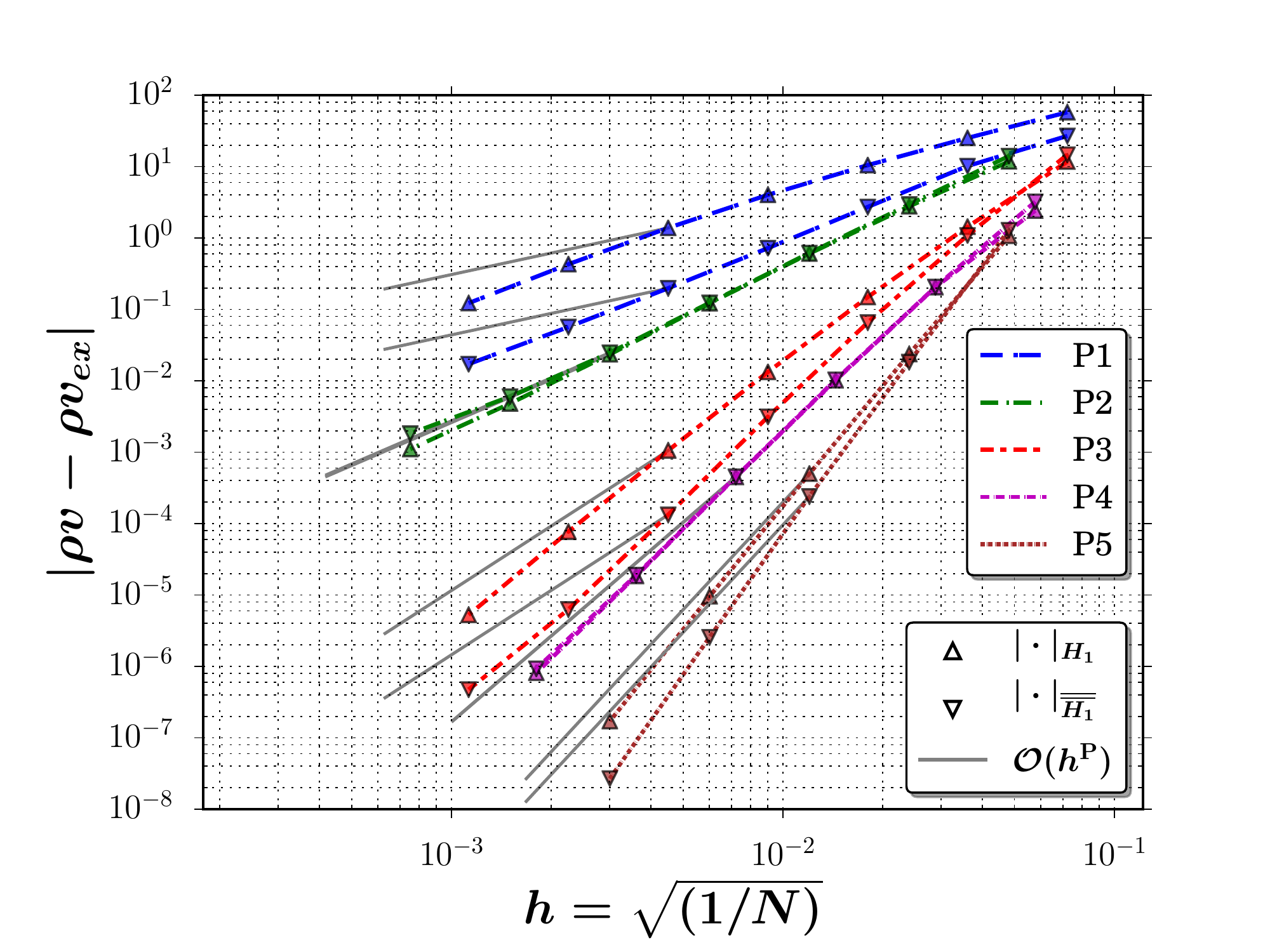}}
~~~
\subfloat[$\rho E$]{
\includegraphics[trim = 5mm 2mm 18mm 13mm, clip,width=0.32\linewidth]{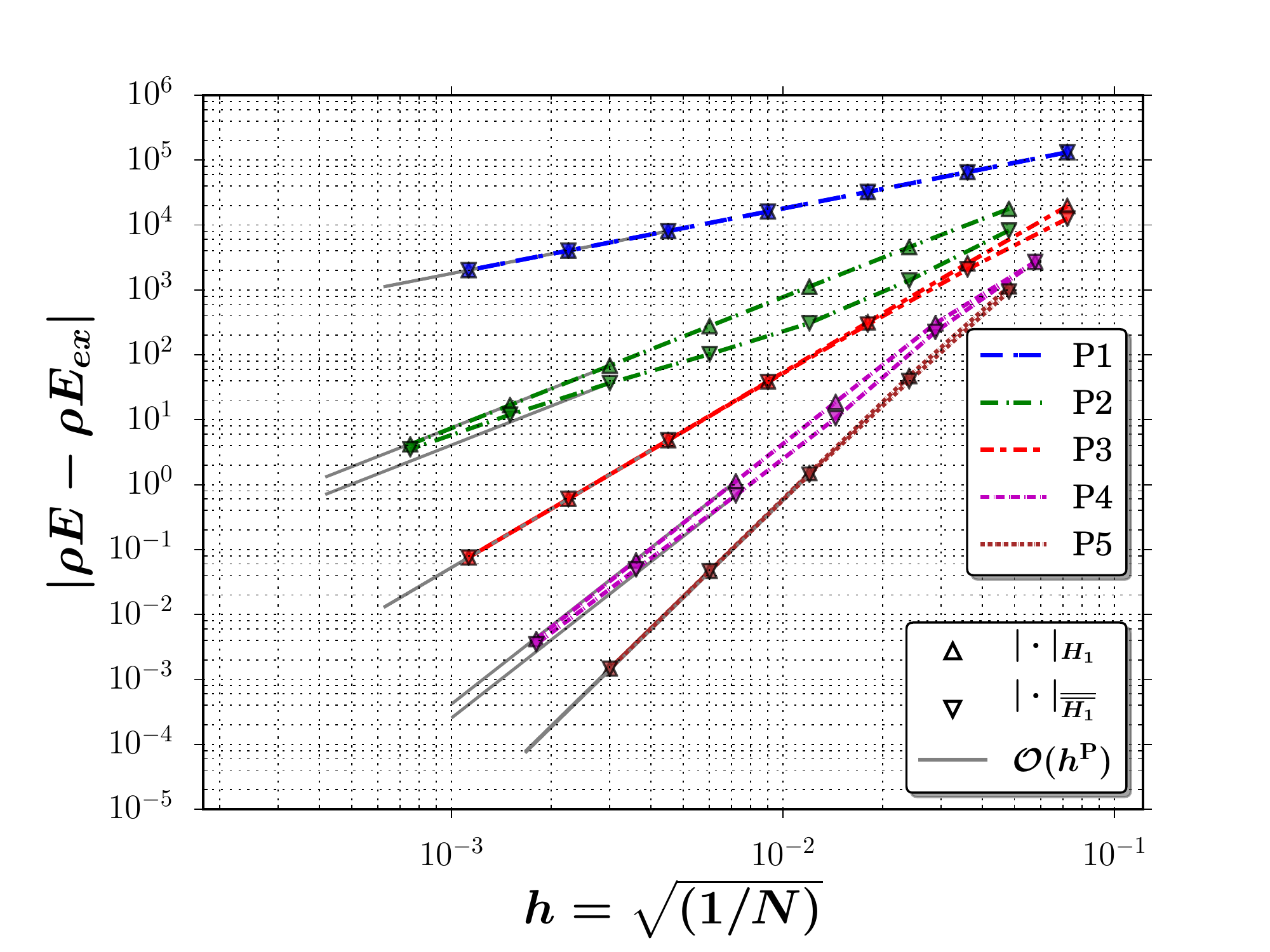}}
\vfill
\subfloat[$\rho \tilde{\nu}$]{
\includegraphics[trim = 5mm 2mm 18mm 13mm, clip,width=0.32\linewidth]{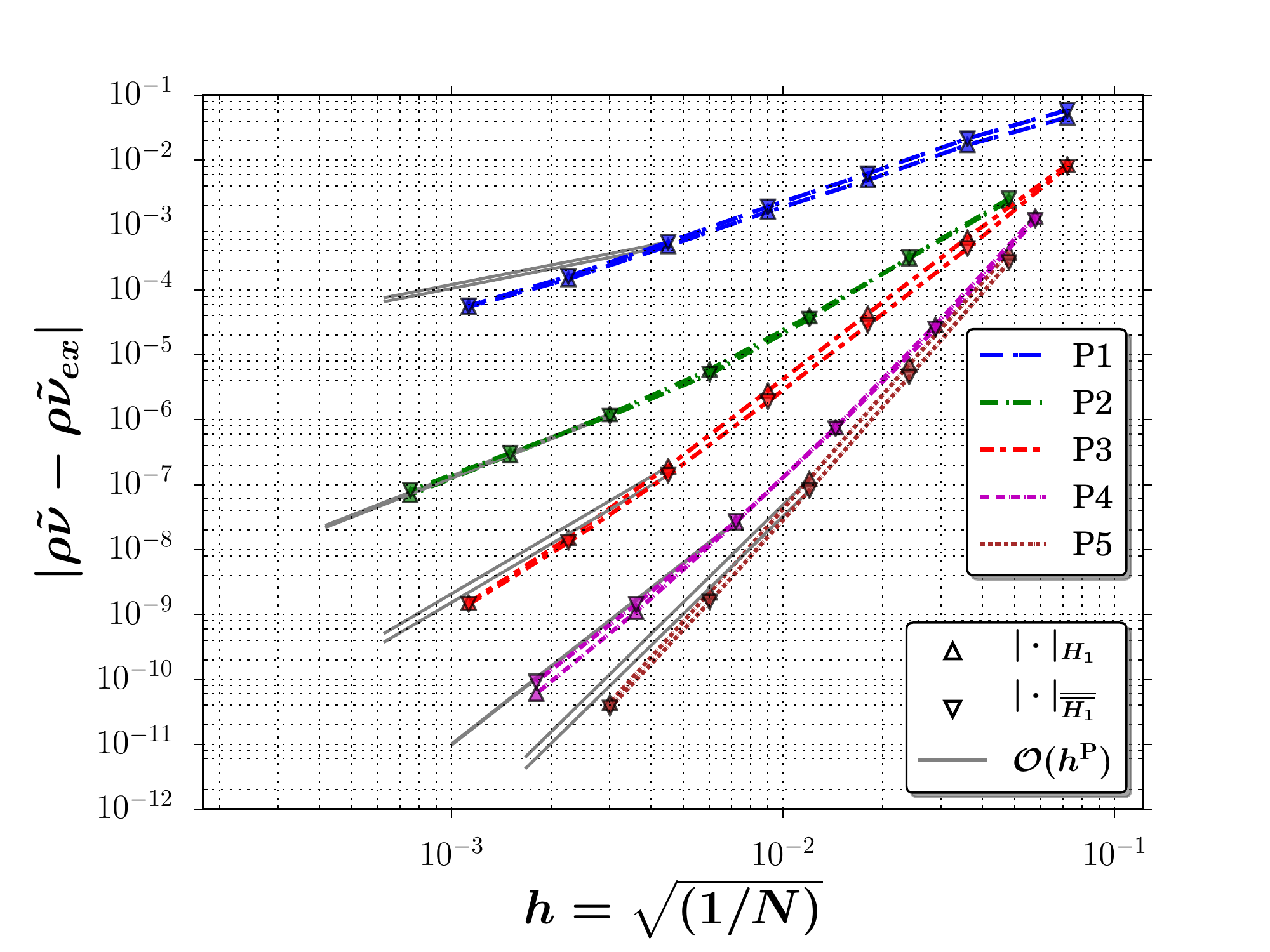}}
\caption{Evolution of the discretization error in $H_1$ semi-norm (for uncorrected and  fully corrected derivatives) versus mesh refinement for dimensional MS-4 and  $\mathrm{P}1$--$\mathrm{P}5$}
\label{fig:Err_allE_allP_H_MS-4}
\end{figure}

\begin{figure}[!hbt]
\centering
\subfloat[$\rho$]{ 
\includegraphics[trim = 16mm 3mm 18mm 13mm, clip,width=0.3\linewidth]
{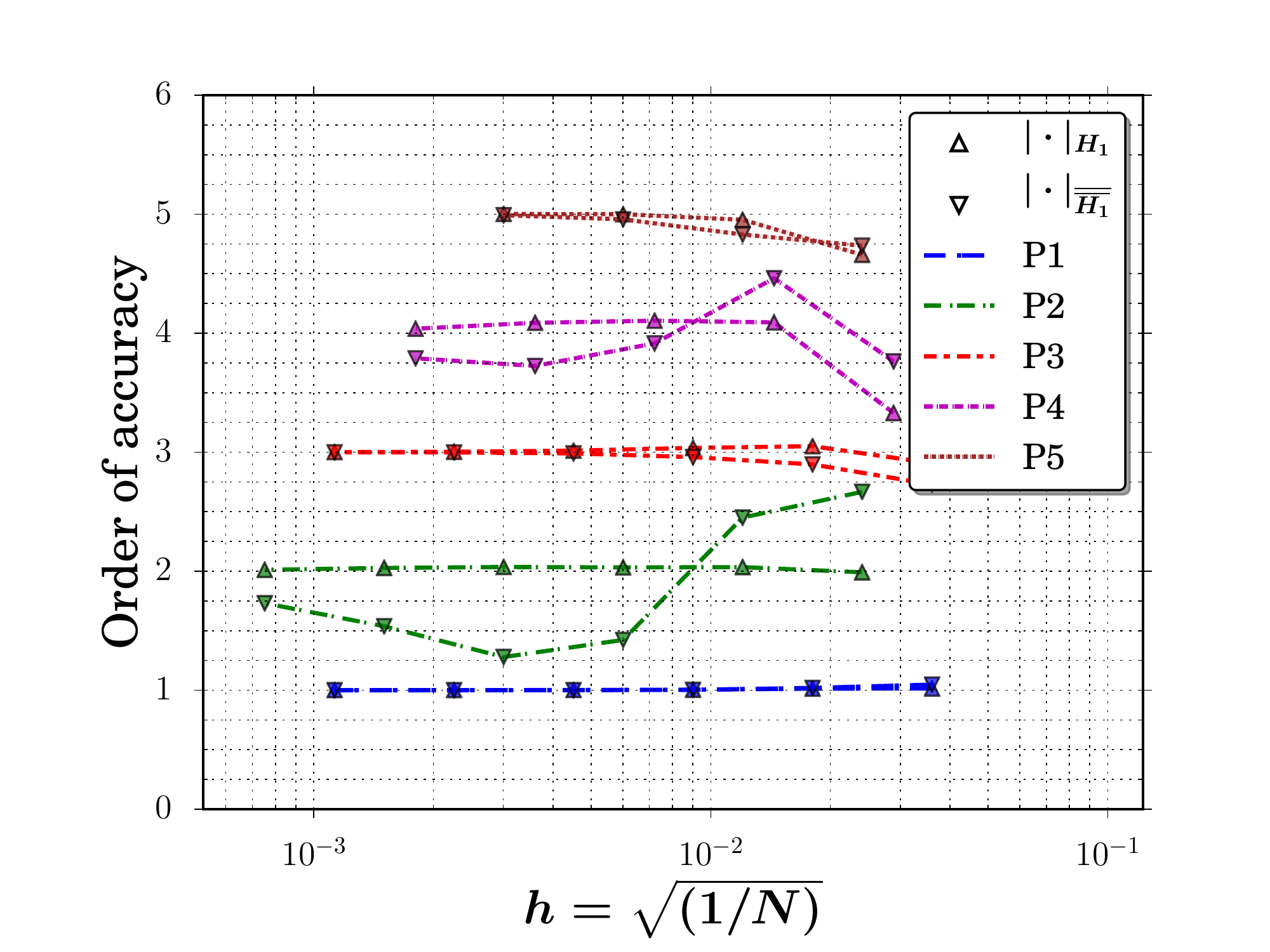}}
~~~
\subfloat[$\rho u$]{
\includegraphics[trim = 16mm 3mm 18mm 13mm, clip,width=0.3\linewidth]{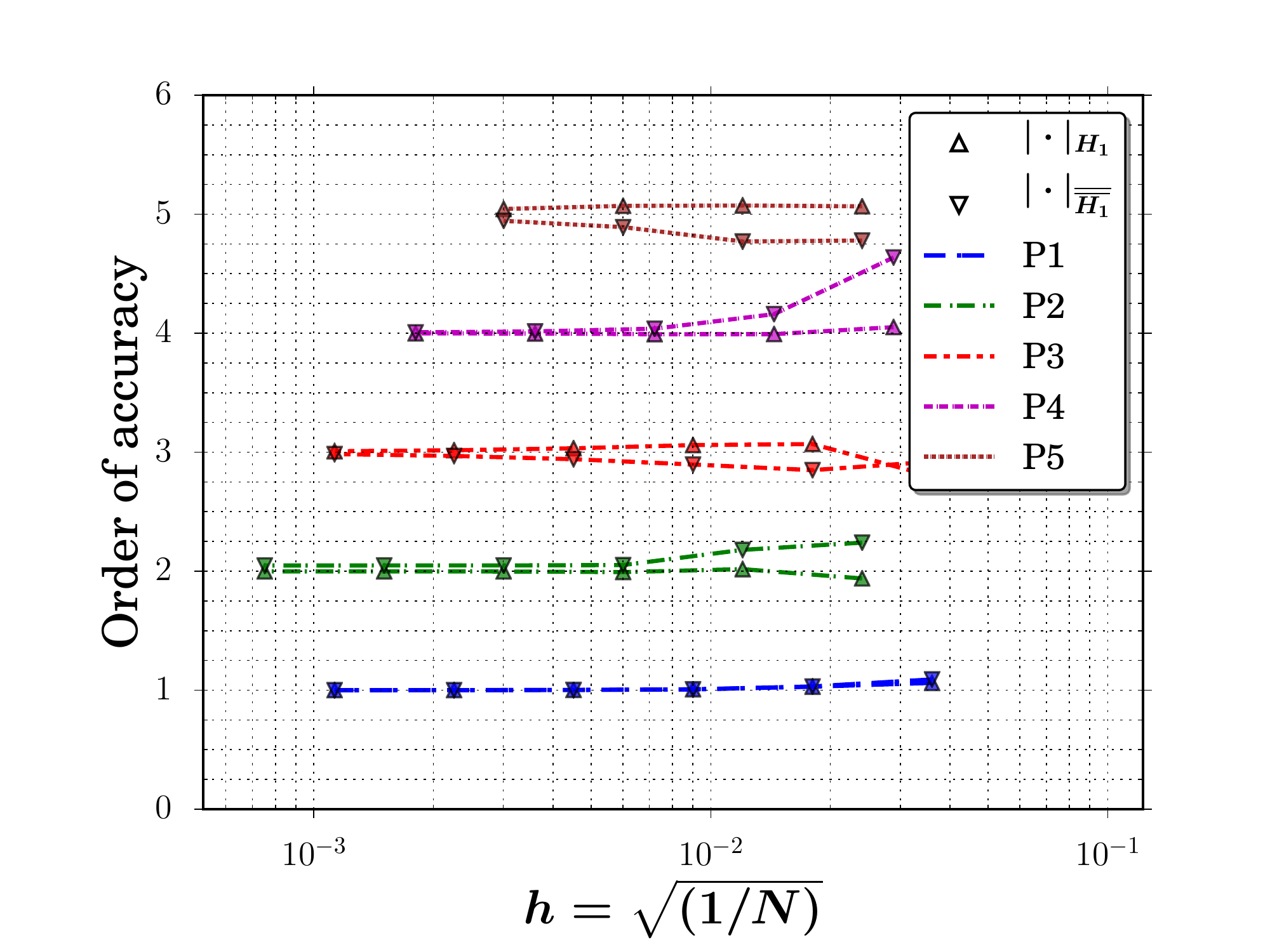}}
\vfill
\subfloat[$\rho v$]{
\includegraphics[trim = 16mm 3mm 18mm 13mm, clip,width=0.3\linewidth]
{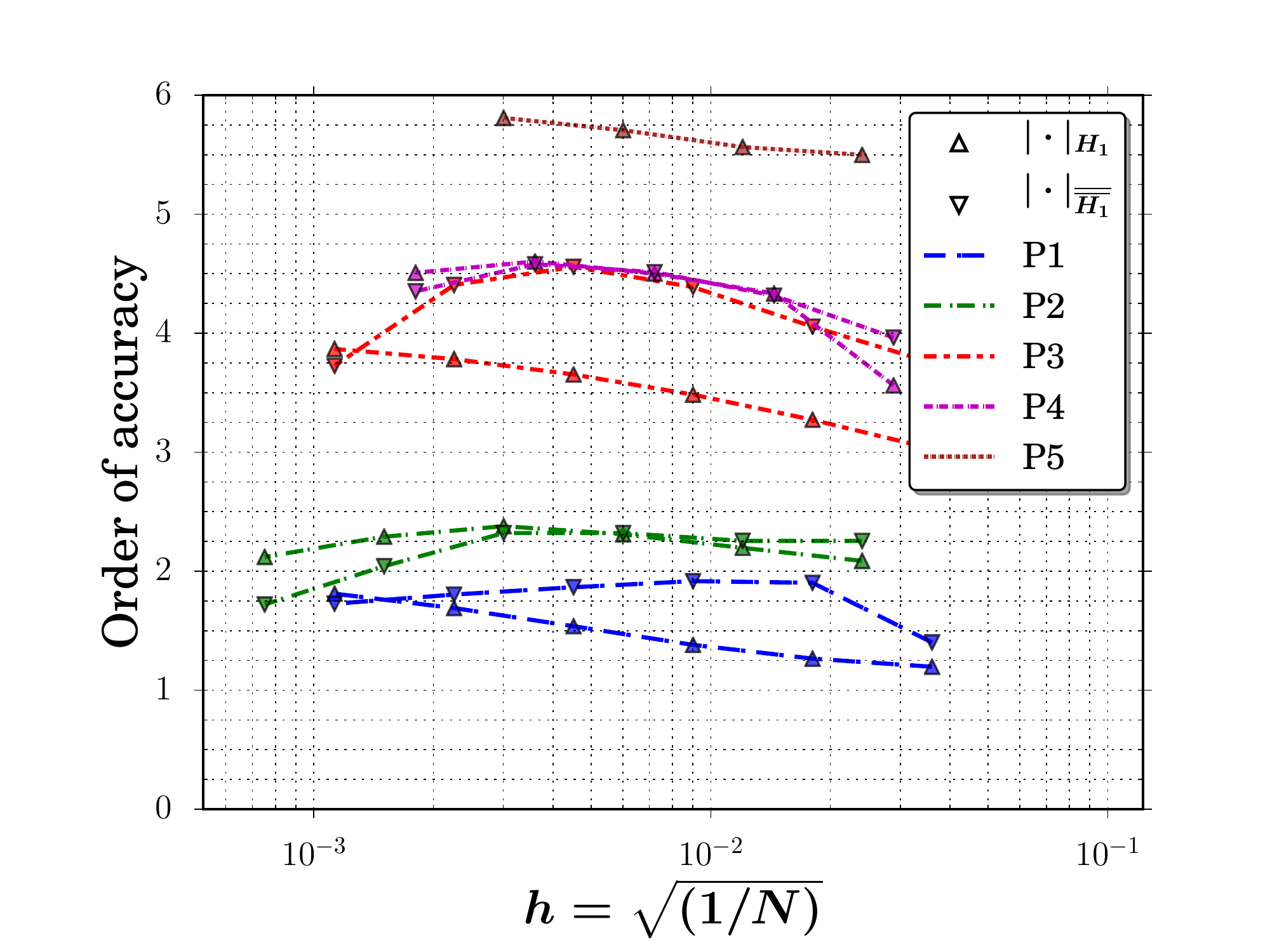}}
~~~
\subfloat[$\rho E$]{
\includegraphics[trim = 16mm 3mm 18mm 13mm, clip,width=0.3\linewidth]{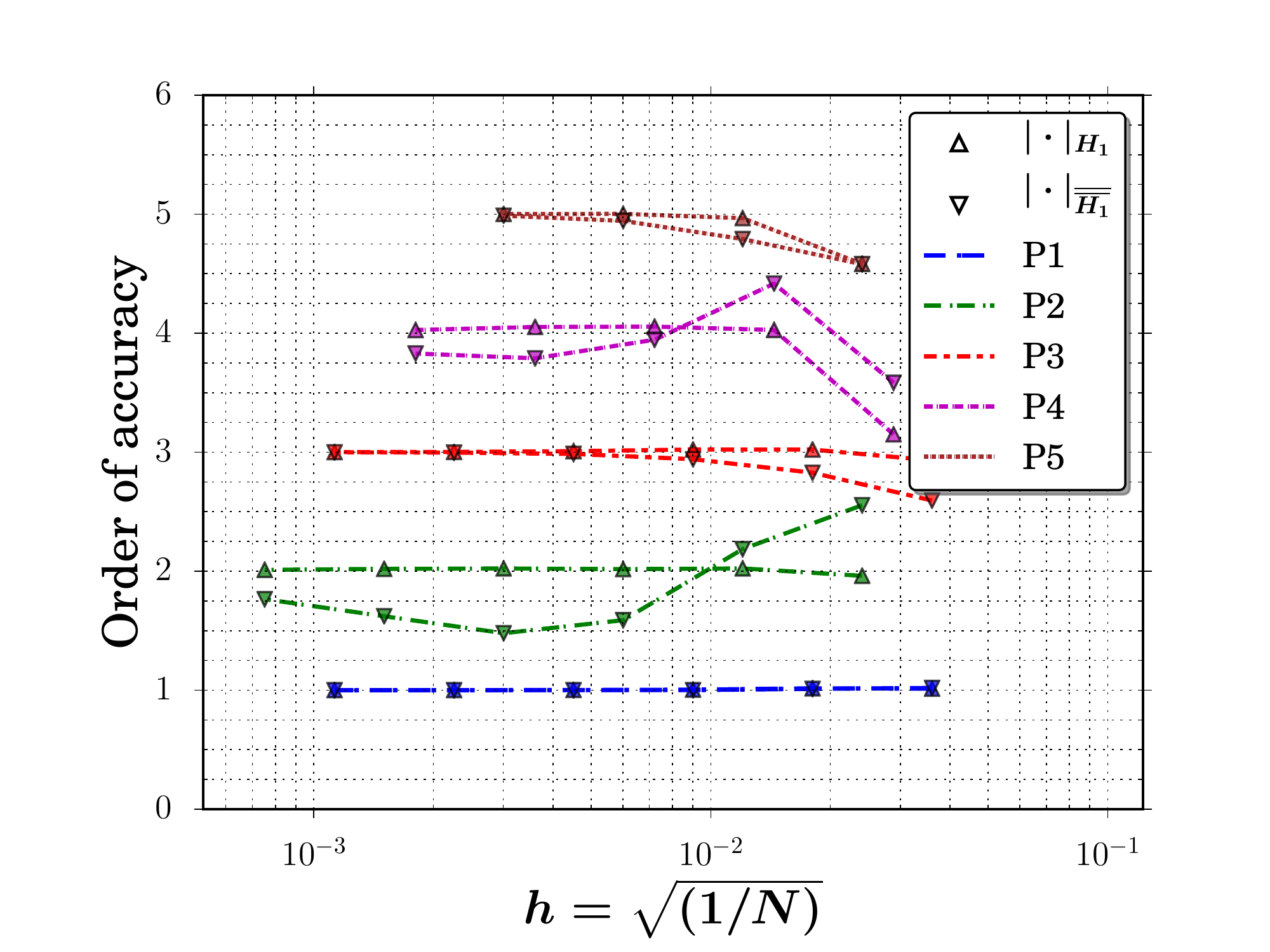}}
\vfill
\subfloat[$\rho \tilde{\nu}$]{
\includegraphics[trim = 16mm 3mm 18mm 13mm, clip,width=0.3\linewidth]{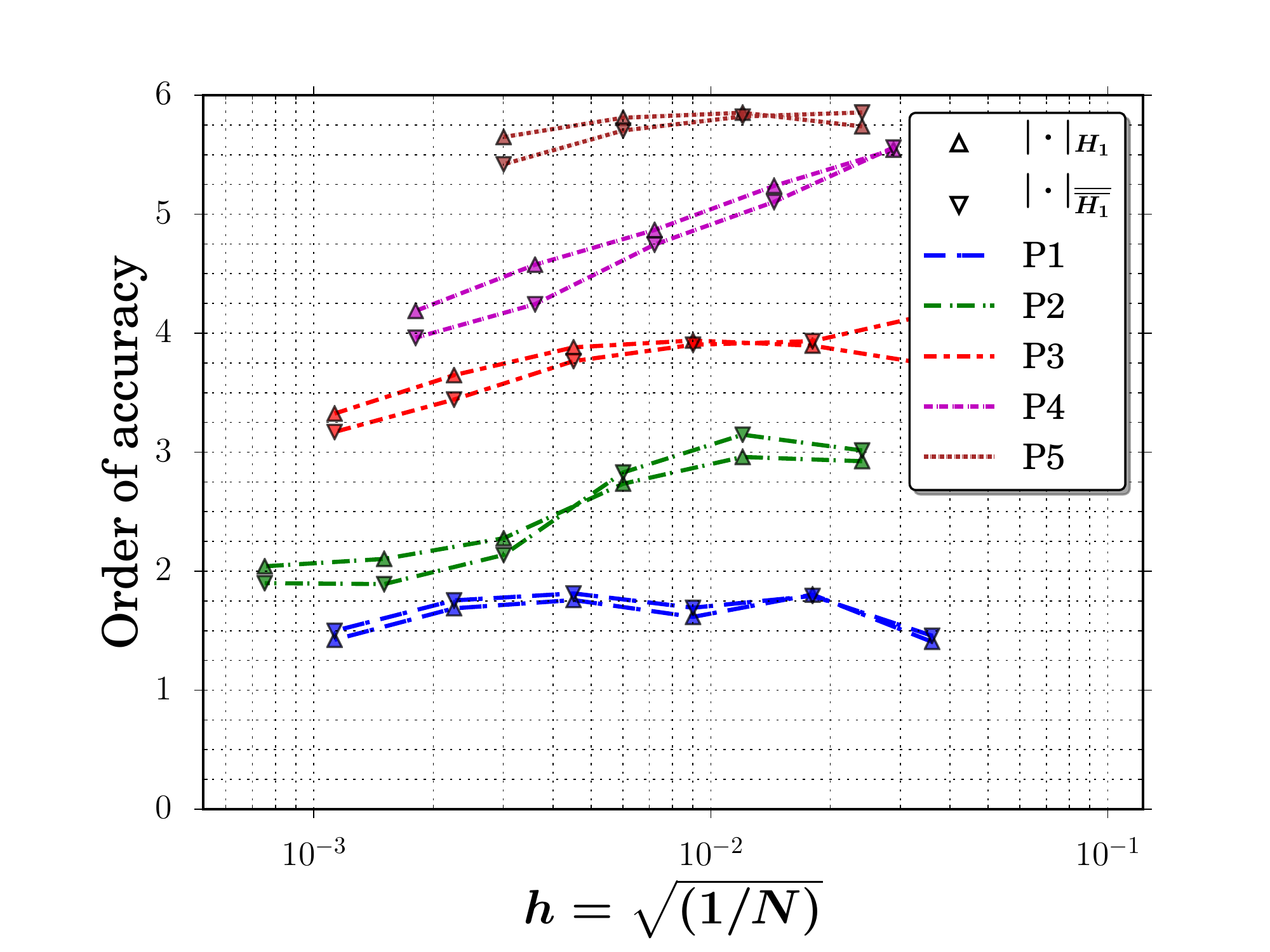}}
\caption{Evolution of the OOAs in $H_1$ semi-norm (for uncorrected and  fully corrected derivatives) versus mesh refinement for dimensional MS-4 and  $\mathrm{P}1$--$\mathrm{P}5$}
\label{fig:Orders_H_MS-4}
\end{figure}

\clearpage
\bibliography{library}


%

\end{document}